\documentclass[11pt,leqno,textwidth=10cm]{amsart}
\usepackage{amssymb}
\allowdisplaybreaks
\usepackage{mathrsfs}
\usepackage{amsmath, amssymb, mathtools, nicefrac}
\usepackage{csquotes}
\usepackage{abstract}
\usepackage[hidelinks]{hyperref}
\usepackage{marginnote}
\usepackage{xcolor}
\usepackage{graphicx}

\usepackage{etoolbox}
\makeatletter
\patchcmd{\@maketitle}{\newpage}{}{}{} 
\makeatother
\numberwithin{equation}{section}

\theoremstyle{definition}
\newtheorem{definition}{Definition}[section]

\newtheorem{remark}[definition]{Remark}
\theoremstyle{plain}
\newtheorem{theorem}[definition]{Theorem}
\newtheorem{lemma}[definition]{Lemma}
\newtheorem{corollary}[definition]{Corollary}
\newtheorem{prop}[definition]{Proposition}
\newtheorem{assumption}[definition]{Assumption}
\newtheorem{conjecture}[definition]{Conjecture}

\newcommand{\E}{\mathcal{E}}

\newcommand{\g}{\overline{g}}

\newcommand{\I}{\mathbb{I}}

\renewcommand{\L}{\mathcal{L}}
\newcommand{\M}{\overline{M}}

\newcommand{\N}{\mathbb{N}}

\newcommand{\nabbar}{\overline{\nabla}}
\renewcommand{\O}[1]{\mathcal{O}\left(#1\right)}

\newcommand{\R}{\mathbb{R}}

\newcommand{\Ric}{\text{\normalfont{Ric}}}
\newcommand{\Riem}{\text{\normalfont{Riem}}}

\renewcommand{\S}{\mathbb{S}}

\newcommand{\vol}[1]{{\text{\normalfont{vol}}}_{#1}}

\newcommand{\Z}{\mathbb{Z}}
\renewcommand{\epsilon}{\varepsilon}

\newcommand{\phibar}{{\phi}_{FLRW}}
\newcommand{\del}{\partial}
\newcommand{\Lap}{\Delta}
\renewcommand{\div}{\text{\normalfont{div}}}
\newcommand{\curl}{\text{\normalfont{curl}}}

\newcommand{\LG}{\L_{\fg,\gamma}}

\newcommand{\Gamhat}{\hat{\Gamma}}
\newcommand{\nabhat}{\hat{\nabla}}
\newcommand{\Ltilde}{\tilde{\L}}

\newcommand{\curve}{\bm{\alpha}}

\newcommand{\numberthis}{\addtocounter{equation}{1}\tag{\theequation}}

\renewcommand{\theequation}{\arabic{section}.\arabic{equation}}

\usepackage{bm}
\newcommand{\RE}{\bm{E}}
\newcommand{\RB}{\bm{B}}
\newcommand{\epsilonLC}{\bm{\epsilon}}

\newcommand{\epsilonnew}{\delta}
\newcommand{\fg}{\bm{g}}
\newcommand{\fk}{\bm{\Sigma}}
\newcommand{\fn}{\bm{n}}
\newcommand{\fN}{\hat{\fn}}
\newcommand{\fX}{\bm{X}}
\newcommand{\fLap}{\Lap_{\fg}}
\newcommand{\fdel}{\widetilde{\del_0}}

\newcommand{\Lie}{\mathcal{L}}
\newcommand{\phim}{\overline{\phi}}

\newcommand{\fE}{\mathbb{E}_{SF}}
\newcommand{\fC}{\mathcal{C}_{SF}}
\newcommand{\fEg}{E_{\text{geom}}}

\newcommand{\change}[1]{#1}
\newcommand{\changefinal}[1]{#1}
\newcommand{\changediss}[1]{#1}
\usepackage{ulem}
\newcommand{\delete}[1]{}
\usepackage{cancel}
\newcommand{\deletemath}[1]{}

\title{Cosmic Censorship near FLRW spacetimes with negative spatial curvature}
\author[D.~Fajman, L.~Urban]{David Fajman, Liam Urban}
\address{
\begin{tabular}[h]{l@{\extracolsep{8em}}l} 
David Fajman  & Liam Urban \\
Faculty of Physics & Faculty of Mathematics\\ 
University of Vienna & University of Vienna \\
Boltzmanngasse 5 & Oskar-Morgenstern-Platz 1 \\
1090 Vienna, Austria & 1090 Vienna, Austria\\
david.fajman@ univie.ac.at & liam.urban@ univie.ac.at 
\end{tabular}
}


\begin{document}
\maketitle

\begin{abstract}
We consider \change{general initial data }for the Einstein scalar-field system on a closed \delete{orientable }$3$-manifold \change{$(M,\gamma)$ }which is close to data for a Friedman-Lema{\^\i}tre-Robertson-Walker solution with \change{homogeneous }scalar field matter and a negative Einstein metric $\gamma$ as spatial geometry. We prove that \change{the maximal globally hyperbolic development of such initial data in the Einstein scalar-field system is past incomplete in the contracting direction and exhibits stable collapse into a Big Bang curvature singularity}. Under an additional condition on the first positive eigenvalue of \change{$-\Delta_\gamma$ }satisfied, for example, \change{by closed hyperbolic 3-manifolds of small diameter}, we prove that the data evolves to a future complete spacetime in the expanding direction which asymptotes to a vacuum Friedman solution with \change{$(M,\gamma)$ }as the expansion normalized spatial geometry. In particular, the Strong Cosmic Censorship conjecture holds for this class of solutions in the $C^{2}$-sense.
\end{abstract}

\setcounter{tocdepth}{2}

\section{Introduction}\label{sec:intro}

\subsection{Setting and main results}

We consider the Einstein scalar-field system 
\begin{subequations}
\begin{align}
\label{eq:ESF1}\Ric[\g]_{\mu\nu}-\frac12R[\g]\g_{\mu\nu}=&\,8\pi T_{\mu\nu}[\g,\phi]\\
\label{eq:EMT}T_{\mu\nu}=&\,\change{\nabbar_{\mu}\phi}\nabbar_{\nu}\phi-\frac12\g_{\mu\nu}\nabbar^\alpha\phi\nabbar_\alpha\phi\\
\label{eq:ESF2}\square_{\g}\phi=&\,0\,
\end{align}
\end{subequations}
with initial data $(g_0,k_0,\pi_0,\psi_0)$ on a closed\delete{ orientable} 3-manifold $M$ that admits a negative Einstein metric $\gamma$.\footnote{Here and throughout, $\pi_0$ and $\psi_0$ prescribe data for $\nabla\phi\vert_{\Sigma_{t_0}}$  and $\del_0\phi\vert_{\Sigma_{t_0}}$ respectively.} In this paper, we determine the maximal globally hyperbolic development emanating from such initial data given that it is sufficiently close to the initial data of a homogeneous solution with a non-trivial scalar field. \\
In the collapsing direction, we prove a stable Big Bang formation and curvature blow-up result, which requires the presence of a non-trivial scalar field. \change{The results complement those in \cite{Rodnianski2014,Speck2018}, which cover flat and spherical spatial geometry}. In the expanding direction, we prove a nonlinear future stability result of the corresponding vacuum background solution, which is the Milne model, under a mild condition for the first positive eigenvalue of \changefinal{$-\Lap_{\gamma}$ }(see Definition \ref{def:spatial-mf-spectral}). As discussed in more detail \change{in Remark \ref{rem:weeks-and-friends}}, numerical studies (see \cite{Cornish99, Ino01}) show that this condition holds for an analogue of Weeks space, and suggest that this may hold for all \change{closed }hyperbolic $3$-manifolds with sectional curvature $-\frac19$ .\\

Connecting the two regions, we prove the global stability \change{(i.e.,~ past and future stability) }of the spacetime
\begin{subequations}
\begin{equation}\label{eq:intro-ref1}
\left([0,\infty)\times M,-dt^2+a(t)^2\gamma\right),
\end{equation}
given \change{a negative Einstein manifold $(M,\gamma)$ obeying }the aforementioned spectral condition, with
\begin{equation}\label{eq:intro-ref2}
a(0)=0,\ \dot{a}=\sqrt{\frac19+\frac{4\pi}3C^2a^{-4}}
\end{equation}
for some given constant $C>0$, and the scalar field given by 
\begin{equation}\label{eq:intro-ref3}
\del_t\phi=Ca^{-3},\ \nabla\phi=0\,.
\end{equation}
The scale factor consequently exhibits the following asymptotic behaviour:
\begin{equation}\label{eq:intro-ref4}
a(t)\simeq t^{\frac13} \mbox{ as }t\searrow0 \mbox{ and } a(t)\simeq t \mbox{ as }t\nearrow \infty  
\end{equation}
\end{subequations}

\noindent The main result can be split into two parts:

\begin{theorem}[Big Bang stability -- rough version]\label{thm:main-past}
Let $(M,g_0,k_0,\pi_0,\psi_0)$ be initial data for the Einstein scalar-field system that is sufficiently close to $(M,a(t_0)^2\gamma,-\dot{a}(t_0)a(t_0)\gamma,0,Ca(t_0)^{-3})$, where $C>0$ and $(M,\gamma)$ is a closed \delete{orientable }Riemannian 3-manifold with $\Ric[\gamma]=-\frac29\gamma$ \change{(i.e.,~ a closed negative Einstein manifold with scalar curvature $-\frac23$)}.\\

Then, the past maximal globally hyperbolic development $((0,t_0]\times M,\g,\phi)$ of the initial data within the Einstein scalar-field system \eqref{eq:ESF1}-\eqref{eq:ESF2} admits a foliation by CMC hypersurfaces $\Sigma_s=t^{-1}(\{s\})$ \change{with zero shift}. This development remains close to the FLRW solution described in \eqref{eq:intro-ref1}-\eqref{eq:intro-ref3} in the past of the initial data slice $\Sigma_{t_0}$. In particular, the solution exhibits curvature blow-up of order $t^{-4}$ and \change{every causal geodesic becomes incomplete }as $t$ approaches $0$.
\end{theorem}
\begin{theorem}[Global stability]\label{thm:main-full}
Let $(M,g_0,k_0,\pi_0,\psi_0)$ be initial data as in Theorem \ref{thm:main-past}. In addition, we \change{suppose }that the smallest positive eigenvalue of $-\Lap_\gamma$ acting on scalar functions is strictly greater than $\frac19$.\\ Then, the initial data admits a \change{maximal globally hyperbolic development }$((0,\infty)\times M,\g,\phi)$ \change{solving }the Einstein scalar-field system that, in addition to the results of Theorem \ref{thm:main-past}, is future (causally) complete. As $t\nearrow\infty$, the solution is attracted by Milne spacetime in the sense that the expansion normalized variables $(\fg,\bm{k},\nabla\phi,\phi^\prime)$ \change{(see Definition \ref{def:fut-rescaled}) }converge toward $(\gamma,\frac13\gamma,0,0)$. 
\end{theorem}
A more detailed statement of Theorem \ref{thm:main-past} is provided in Theorem \ref{thm:main}. The additional spectral condition in Theorem \ref{thm:main-full} is discussed at the end of Subsection \ref{subsec:prev}, and the statement itself is proven in Section \ref{sec:full-stab} to be an extension of the Milne future stability result in Theorem \ref{thm:fut-stab-simple}.


\change{
\subsection{Background material}

We now provide context for the previously discussed setting and the results in Theorems \ref{thm:main-past} - \ref{thm:main-full}:\\

\subsubsection{Initial data to the Einstein scalar-field equations}\label{subsubsec:initial-data}

It is well known that the Einstein equations can, via the 3+1 decomposition, be viewed as an elliptic-hyperbolic system of PDEs (see, for example, \cite{AM03B}). This reduces solving the Einstein equations to two problems: finding admissible Einstein initial data in physical space, and then solving the corresponding initial value problem. Regarding the former, initial data to the Einstein scalar-field system takes the form
\[(M,\mathring{g},\mathring{k},\mathring{\pi},\mathring{\psi}),\]
where $\mathring{g}$ and $\mathring{k}$ are symmetric $(0,2)$-tensors on $M$,  $\mathring{\pi}$ is an exact $(0,1)$-tensor (corresponding to $\nabla\phi$) and $\mathring{\psi}$ is a scalar function (corresponding to the future directed normal derivative $\del_0\phi$ of the scalar field). The initial data must satisfy the Hamiltonian and momentum constraints
\begin{subequations}
\begin{align}
\text{R}[\mathring{g}]+\left({\mathring{k}^{a}}_{\ a}\right)^2-\left({\mathring{k}^{a}}_{\ b}{\mathring{k}^b}_{\ a}\right)=&\,8\pi\left[\lvert\mathring{\psi}\rvert^2+\lvert\mathring{\pi}\rvert_{\mathring{g}}^2\right] \label{eq:init-Hamilton}\,,\\
\div_{\mathring{g}}\mathring{k}=&\,-8\pi\cdot\mathring{\pi}\cdot\mathring{\psi}\,\label{eq:init-momentum}
\end{align}
\end{subequations}
(see \eqref{eq:Hamilton} and \eqref{eq:Momentum}), where the indices of $\mathring{k}$ in the first line are raised with respect to $\mathring{g}$.\\ 
We note that, in our argument, we will additionally assume that our initial data has constant mean curvature so that our gauges can be satisfied initially -- this is enforced on the level of initial data by requiring
\changefinal{\begin{equation*}
\text{tr}_{\mathring{g}}\mathring{k}=-3\frac{\dot{a}(t_0)}{a(t_0)}\,
\end{equation*}
(see \eqref{eq:CMC}).} We will argue in Remark \ref{rem:CMC-hypersurface} why the initial data being near-FLRW allows us to assume the initial hypersurface to be CMC without loss of generality.\\

The results of \cite{FB52, CBGer69} show that there exists an embedding\footnote{We usually ignore the embedding in notation.} $\iota:M\hookrightarrow\iota(M)\subset\M$ and a \changefinal{maximal }solution $(\M,\g,\nabla\phi,\del_0\phi)$ to the Einstein scalar-field equations such that $\iota(M)=\Sigma_{t_0}$ is a Cauchy hypersurface and such that
\[\iota^\ast\g=\mathring{g},\,\iota^\ast k=\mathring{k},\,\iota^\ast\pi=\mathring{\pi},\,\text{and }\iota^\ast \del_0\phi=\mathring{\psi_0}\,.\]
We will perturb around initial \changefinal{data corresponding to data for an FLRW spacetime at time $t=t_0$, i.e.,~
\[(M\cong\Sigma_{t_0},a(t_0)^2\gamma,-\dot{a}(t_0)\,a(t_0)\,\gamma,0,C\,a(t_0)^{-3})\,.\]} Furthermore, the maximal globally hyperbolic development (MGHD) is unique (up to diffeomorphism), and thus we can assume $(\M,\g,\nabla\phi,\del_0\phi)$ to be globally hyperbolic. However, these statements provide little information on the properties of the MGHD in the future and past of the initial data slice.

\subsubsection{Strong Cosmic Censorship}

In their groundbreaking papers on singularity theorems, Hawking \cite{Hawk67} and Penrose \cite{Pen65} established very general criteria for the MGHD of spacetimes to become causally geodesically incomplete. Many spacetimes of physical relevance satisfy these criteria, including the spacetimes considered in this article. While giving us more information on the MGHD than the existence and uniqueness results mentioned above, a key issue in the application of this mathematical result to General Relativity is that no statement is made on how precisely the singularity comes about: In particular, such incompleteness (within a given regularity class) could either mean that the geodesic is inextendible -- which must be caused by the blow-up of some geometric quantity -- or that there exist multiple inequivalent extensions. While the latter behaviour is exhibited even for some cosmological spacetimes (see, for example, the Taub solutions discussed in \cite{ChrIs93}), such behaviour is usually considered to be unphysical since it would imply a breakdown of determinism. The \textit{Strong Cosmic Censorship Conjecture} (SCCC) posits in its most general form that, for generic solutions to the Einstein equations, this incompleteness instead manifests as inextendibility at a given level of regularity (e.g., $C^0, C^2, C^\infty,\dots$).\\

\noindent In certain frameworks in the homogeneous cosmological setting -- i.e. for homogeneous initial data on a closed spatial hypersurface --, it was shown in fundamental works by Chrusciel-Rendall \cite{CR95} and Ringström \cite{Ring09} that the so called Kretschmann scalar $R_{\alpha\beta\gamma\delta}R^{\alpha\beta\gamma\delta}$ is unbounded where incompleteness manifests. Thus, it is the driving force behind geodesic incompleteness in these cases, forcing $C^2$-inextendibility of the MGHD. For the purposes of analyzing cosmologically relevant spacetimes, the SCCC is hence often rephrased as follows:
\begin{conjecture}[Cosmological SCCC](See e.g. \cite[Chapter 17]{Ring09})
For generic initial data, the Kretschmann scalar is unbounded where causal geodesics become incomplete.
\end{conjecture}
Theorem \ref{thm:main-past}, in short, shows that this conjecture is rigorously supported in the case of FLRW spacetimes with negative spatial curvature. More precisely, the past asymptotics of such spacetimes, determined by initial data on $\Sigma_{t_0}$ as discussed above, are generic in the following sense: There exists an open neighbourhood of said FLRW data within the set of Einstein scalar-field initial data  such that the solutions past directed causal geodesics become incomplete, and the incompleteness is driven by blow-up of Kretschmann scalar with the same asymptotics as the FLRW solution. The global result in Theorem \ref{thm:main-full} portrays the other side of Cosmic Censorship -- as with the past evolution, near-FLRW data fully determines the future of the spacetime in the sense that the MGHD is future complete, again showing that this feature of FLRW spacetimes with negative spatial sectional curvature is generic.

\subsubsection{FLRW and generalized Kasner spacetimes with scalar fields}\label{subsec:FLRW-Kasner}

On a large scale, the universe is often viewed as spatially homogeneous and isotropic, i.e., no point in space and no direction are distinguishable from any other point and direction \changefinal{(referred to as the \enquote{Cosmological Principle})}. In 1935, it was shown by Robertson and Walker that, under a few very natural additional assumptions, this restricts the class of potential spacetimes to the FLRW class
\[\left(I\times\tilde{M},\,\tilde{g}_{FLRW}=-dt^2+a(t)^2\tilde{\gamma}\right)\,,\]
where $(\tilde{M},\tilde{\gamma})$ is a manifold of constant sectional curvature $\kappa$ and where the scale factor $a$ depends smoothly on $t$. 
This holds before taking the Einstein equations into consideration -- when doing so, the matter model determines how space expands within the cosmological model via $a$. We refer to Lemma \ref{lem:FLRW} for the scalar-field solution for $\kappa=-\frac19$, but note that the scale factor behaves like $t^\frac13$ for scalar-field matter, regardless of spatial geometry, and that the Kretschmann scalar blows up at order $\O{t^{-4}}$ toward the Big Bang ($t\downarrow 0$).\\
Spatially flat FLRW spacetimes are a subclass of the closely related \textit{generalized Kasner spacetimes}, which are still spatially homogeneous but anisotropic in general. For scalar field matter, the spacetime metric is given by
\changefinal{\begin{gather*}
\g_{Kasner}=-dt^2+\sum_{i=1}^Dt^{2p_i}dx^i\otimes dx^i,\quad
\sum_{i=1}^D{p_i}=1,\,\sum_{i=1}^Dp_i^2=1-8\pi A^2,\quad\overline{\phi}_{Kasner}(t)=A\log(t)\,.
\end{gather*}}
The standard Kasner family is obtained by considering the vacuum case ($A=0$), and the spatially flat FLRW spacetime by setting \changefinal{$D=3, p_i=\frac13, A=\sqrt{\frac1{12\pi}}$}. If more than one of the Kasner exponents is non-zero, the generalized Kasner family satisfies the SCCC, also by exhibiting Kretschmann scalar blow-up of order $t^{-4}$as $t\downarrow 0$ (see \cite[(1.8)]{Rodnianski2014}).\\

Kasner spacetimes are of particular relevance to cosmology due to their relationship with the \textit{BKL conjecture}: Heuristically, this conjecture states that the dynamics of cosmological spacetimes \changefinal{near a spacelike singularity }generically exhibit chaotic and highly oscillatory behaviour, often referred to as \enquote{Mixmaster} behaviour. This behaviour is driven by velocity terms within the Einstein equations and is locally comparable to that of (vacuum) Kasner solutions. However, even if the BKL picture is to be believed in general, scalar-field (or, more generally, stiff-fluid) solutions seem to form an exception to it: They have a dampening effect on said oscillations, thus generating Big Bang stability as shown rigorously in \cite{Rodnianski2014, RodSpFou20} for Kasner spacetimes (for more details, see Section \ref{subsec:prev}). This scenario, often referred to as \textit{quiescent cosmology}, was studied in, for example, \cite{BK73,Bar78,AnRen01}. With this in mind, both the aforementioned Kasner results and the results within this article, along with the prior FLRW results \cite{Rodnianski2014,Speck2018}, confirm this quiescent effect of scalar fields in \changefinal{cosmology.}\\\

We note that one can view this as a scalar field ensuring a specific scenario in the very early universe given a class of initial data, namely matching the asymptotic behaviour of the Big Bang singularity. This fits into the recent use of nonlinear scalar fields in string cosmology, where specific choices of field are made to specific behaviours (e.g.,~inflation) in the early universe. For a recent review, we refer to \cite{StringCosm23}.
}\\

\subsection{Relation to previous work}\label{subsec:prev}

Theorem \ref{thm:main-full} is the first theorem about the full global structure of FLRW spacetimes with negatively curved spatial geometry. For such solutions, \changefinal{prior }results exclusively concern future stability, which we further discuss below. Besides \cite{Speck2018} covering the $\mathbb{S}^3$-case, it is the only open set of \change{initial data for cosmological spacetimes }(i.e.,\,\,without symmetry assumptions) with $\Lambda=0$ and in absence of accelerated expansion for which the global (future and past) dynamics are now fully understood.\footnote{\change{For a related future stability result in accelerated expansion, see \cite{Ring08} which considers scalar fields with a non-trivial potential.}}\\

Scalar field matter (and, more generally, \changefinal{matter obeying }semilinear wave equations or fluid matter) and their asymptotic behaviour on fixed cosmological backgrounds have been studied extensively, for example in \cite{AlRen10, Franzen18, Bach19, Ring19, BO24, Ring21linear, Ring21wave, Wang21}. While many of the results, in particular \cite{Ring21linear}, manage to analyze very general classes of equations and spacetime geometries\change{, }including the wave equation on the FLRW backgrounds studied in \cite{Franzen18, FU22}, the methods used are often difficult to apply to the full Einstein scalar-field system. In \cite{Urban22}, we extended the approach of \cite{Franzen18} to be able to deal with various warped product spacetimes, and in particular FLRW spacetimes \change{with negatively curved spatial geometry}, by using the spatial Laplace operator to control high order derivatives. The perturbation-adapted analogue of this strategy is at the basis of the energy method in this paper.\\

\indent We also note that, by the results of \cite{Girao19}, there are non-trivial waves on fixed FLRW backgrounds that converge toward the Big Bang singularity, even if, as demonstrated in \cite{Franzen18,Urban22}, this behaviour is non-generic. Such waves can give rise to convergent asymptotics on cosmological backgrounds as studied in \cite{Ring21linear}. Thus, it will likely be difficult to replace \eqref{eq:intro-ref3} with an arbitrary non-trivial reference wave while keeping past stability intact. However, by restricting to an open neighbourhood near the solution described in \eqref{eq:intro-ref1}-\eqref{eq:intro-ref3}, \changefinal{potential non-generic solutions of this type are excluded. For the more general conditions on initial data that lead to quiescent asymptotics, we refer to \cite{GPR23}, which will be discussed further below.}\\

Theorem \ref{thm:main-past} forms the counterpart to the pioneering works by Rodnianski-Speck \cite{Rodnianski2018,Rodnianski2014} and Speck \cite{Speck2018}, which cover nonlinear Big Bang stability for \change{FLRW spacetimes with }spatial geometry $\mathbb{T}^3$ and $\S^3$ respectively.
\change{These results were extended to Kasner spacetimes in \cite{RodSp22} with $\lvert q_i\rvert<\frac16$, and to the full subcritical regime in \cite{RodSpFou20}, i.e., (generalized) Kasner spacetimes as discussed in Section \ref{subsec:FLRW-Kasner} with $\changefinal{\max_{i,j,k=1,\dots,D}(p_i+p_j-p_k)<1}$. The former necessitates considering $1+D$-dimensional Kasner spacetimes with $D\geq 38$, while the latter result also can be satisfied in $D=3$ for generalized Kasner spacetimes. Recall that this means, in contrast to our setting, that the reference spacetime can be anisotropic, even \changefinal{if the }conditions on Kasner exponents rule out extremely anisotropic regimes. As a result, the analysis therein becomes significantly more involved, especially at top order, since approximately monotonic energy identities as used in our work as well as in \cite{Rodnianski2014,Speck2018} have not been found in these anisotropic settings.\\

\changefinal{We note that the argument in \cite{RodSpFou20} relies on identifying an almost-diagonal structure for the asymptotics of (combined) connection coefficients for an adapted frame that is carried along by Fermi-Walker transport; this is precisely where subcriticality enters. Given that these no longer can vanish in a reference frame adapted to near-hyperbolic spatial geometry, it is a priori unclear whether this structure is sufficiently maintained.\\
The impressive recent preprint \cite{GPR23} by Oude Groeniger, Petersen and Ringström circumvents this issue and uses the equations considered in \cite{RodSpFou20} to establish general conditions for initial data to the Einstein (non-linear) scalar-field equations to give rise to quiescent singularities (see \cite[Theorem 12]{GPR23}). Additionally, they show that a large class of cosmological model solutions to exhibit stable Big Bang formation (see \cite[Theorem 49]{GPR23}). In particular, by only requiring that the mean curvature is sufficiently large compared to the expansion-normalised data, the rescaled connection coefficients can be made to be sufficiently small even if they are non-trivial in the reference. However, this high level of generality comes at the cost of no longer being able to ensure that the expansion-normalized solution variables themselves, in particular the generalized Kasner exponents, remain close to the reference solution, in contrast to our asymptotic results in Theorem \ref{thm:main}.\\}

Furthermore, Beyer and Oliynyk have recently shown in \cite{BeyOl21} that, over $\mathbb{T}^3$, the Big Bang formation can be localized in the sense that data given solely on a ball within the initial hypersurface must also cause stable blow-up on a (smaller) ball on the Big Bang hypersurface. While this result further indicates that blow-up behaviour of near-FLRW spacetimes might be, at least, independent of global geometric properties as it seems to be a localizable, we note that proof of localized stability crucially relies on the flatness of the conformal reference spacetime. To be more precise, the proof relies on extending the local initial data to global data for a Fuchsian system of metric and matter quantities as well as, again, connection coefficients for an adapted, Fermi-Walker transported frame. However, the derivation of the system for the former explicitly seems to use flat spatial geometry to obtain the necessary Fuchsian form. This \changefinal{form }seems to similarly be broken as soon as the connection coefficients are not perturbed around $0$, since this would lead to inhomogeneous error terms of order $t^{-1}$ for the rescaled variables which are stronger than what the method, so far, accounts for.\\

By contrast, in \cite{Rodnianski2014, Speck2018}, the reference frame itself is used in the commutator method to obtain the necessary energy identities at high orders. In all of these works, it hence is a priori unclear how one could extend these methods to the negative spatial Einstein geometry of $(M,\gamma)$. }We provide an alternative approach that, besides establishing the complementary stability result to \cite{Rodnianski2014,Speck2018}, does not rely on any information on the spatial geometry of the reference manifold in its methodology (although it is of course relevant in determining the FLRW reference solution that we are studying). Instead, we rely on differential operators adapted to the evolved spatial metric. 
Hence, we believe that our approach may also prove useful for stability problems in spatially inhomogeneous \change{(and hence also anisotropic) }settings. In light of \cite{RodSp22, RodSpFou20} in particular, the main challenge in achieving this would either be to find approximately monotonic energy identities with our Bel-Robinson approach that have not been observed previously, or to also find ways to circumvent the lack thereof. \\

To obtain Theorem \ref{thm:main-past}, we use the Laplace-Beltrami-operator (acting, respectively, on scalar functions and tensor fields) with respect to the (rescaled) evolved metric as our commutating operator instead of a fixed reference frame. This, in turn, leads us to replacing the wave-like system for metric and second fundamental form exploited in \cite{Rodnianski2014, Speck2018} by an evolutionary system in the second fundamental form and Bel-Robinson variables. \changefinal{The latter technique dates back to the fundamental works by Christodoulou-Klainerman \cite{ChrKl90,ChrKl93}, where it was used to analyse field equations on Minkowski space and then to show global stability of Minkowski space itself. It has also been applied to the future stability of Milne spacetimes in the vacuum Einstein equations by Anderson-Moncrief in \cite{AM03} and, more recently, within the massive Einstein Klein-Gordon system by Wang in \cite{Wang19}. }As far as we are aware, this method has not yet been applied to solutions that are not near-vacuum or in the context of Big Bang singularity formation.\\

\indent Toward the Big Bang, the solutions exhibit asymptotically velocity dominated (AVTD) behaviour in the sense that they behave, to leading order, like solutions to the Einstein scalar-field equations in CMC gauge with zero shift with all terms involving spatial derivatives set to zero (the \enquote{velocity term dominated} (VTD) equations). This behaviour also matches results obtained by studying high regularity solutions (e.g.,\,\cite{AnRen01}), or related works using Fuchsian methods that prescribe a behaviour at the singularity and then develop it locally, often under additional symmetry assumptions (e.g.,\,\cite{DHRW02, CBIM04, IsMon02, FL23}). \change{In particular, this asymptotic behaviour leads to the same types of \enquote{Kasner footprint states} as in \cite{Rodnianski2018,Rodnianski2014}: As one approaches the Big Bang, the rescaled variables converge toward tensor fields on the Big Bang hypersurface that precisely solve the truncated VTD equations. Further,  the distance between the footprints of the FLRW and the perturbed solution are controlled by the initial data. For example, the rescaled Weingarten map $a^3{k^a}_{b}$ converges to ${(K_{Bang})^a}_b$ on the Big Bang hypersurface, which is close to $\frac{\sqrt{4\pi}}3C\I^a_b$, the rescaled FLRW footprint (see \eqref{eq:asymp-K} and \eqref{eq:footprint-K}).}\\

What remains to be considered to obtain Theorem \ref{thm:main-full} is future stability, which we can reduce to future stability of the vacuum solution in the Einstein scalar-field system. This solution, called the Milne spacetime, has been shown to be stable within the set of vacuum solutions -- see \cite{AM11} -- and a range of other Einstein systems -- see, for example\,,  \cite{Wang19,AndFaj20,FajWy21,FOW24,BaFaj20,BraFajKr19} and related work in lower dimensions, e.g. \cite{AMT97, Mon08, Faj17, Faj20, Mondal20}. As such, our contribution to the study of future stability of Milne spacetimes is that we deal with the massless scalar field matter via corrected energy estimates which are inspired by work of Choquet-Bruhat and Moncrief in \cite{CBM01} for vacuum Einstein equations with $U(1)$-symmetry. \change{Out of the works listed above, only \cite{Wang19, FajWy21} deal with scalar field matter at all, namely the massive case. These fields exhibit stronger decay toward the future, making the matter components easier to deal with than in our analysis.}\\
The additional spectral condition is needed to ensure coercivity of the corrected scalar field energy. Numerical work, e.g. \cite{Cornish99, Ino01}, does not suggest that this condition is violated by any \change{closed }3-manifold with constant sectional curvature $\kappa=-\frac19$, and verifies that is is satisfied, for example, by an analogue of Weeks space in which the metric is appropriately scaled to have the required sectional curvature. \changefinal{The latter is also verified by the recent result \cite{BoMaPa25} that, amongst considering more general related settings, sufficiently constrains the spectrum of the Laplacian on Weeks space. }We refer to Remark \ref{rem:weeks-and-friends} where this discussed in more detail.

\subsection{Challenges in the proof}

The contracting and expanding regimes of near-FLRW spacetime are analyzed in two separate and methodologically independent parts. Before providing an overview of both arguments, we summarize the challenges that arise:

\subsubsection{Big Bang stability}

The main difficulties in establishing Big Bang stability are three-fold:\\

Firstly, we have to expect that the solutions are asymptotically velocity term dominated (as argued in Remark \ref{rem:AVTD}, we end up proving that this is the case)\change{, and thus that rescaled variables at best exhibit the same asymptotic behaviour as their counterparts in FLRW spacetime, up to a small perturbation in the asymptotic footprint.  For example, note that, in the reference FLRW spacetime, one has \[{(k_{FLRW})^i}_j=-3\frac{\dot{a}}a\I^{i}_j\approx-\frac1t\I^{i}_j\,.\] At best, the shear ${\hat{k}^i_j}$ of the perturbed solution then behaves like $\frac{\epsilon}t$. In fact, we show that this is the case in \eqref{eq:APSigma}. This implies that the contraction rescaled metric $G_{ij}=a^{-2}g_{ij}$ can only be controlled up to $\O{t^{-c\sqrt{\epsilon}}}$ (see \eqref{eq:APmidG}), since one has $\del_tg_{ij}\approx -2g_{il}{k^l}_j$ and thus
\[\del_t G_{ij}\approx G_{il}{\hat{k}^l}_{\ j}\approx \frac{\epsilon}t\ast G\,.\]}
However, to be able to use the structure of the evolution equations to cancel terms in our energy arguments, we have to work with adapted quantities. For example, we need to use integration by parts with respect to $(\Sigma_t,G_t)$ to cancel high order scalar field terms with help of the (rescaled) wave equation that contains $\Lap_G$, or to obtain elliptic estimates from the lapse equation via the operator $\Lap_G$ or from the adapted div-curl-system for $\Sigma$ arising from the constraint equations. \\

\indent As a result, even the rescaled solution variables \change{will diverge at order $\O{t^{-c\sqrt{\epsilon}}}$} toward the singularity, so we need to track and control their rate of divergence within the bootstrap argument. This significantly complicates dealing with nonlinear terms, where the bootstrap assumptions often cannot be inserted naively. This in turn makes coercivity of the energies more involved to establish \change{(see Lemma \ref{lem:Sobolev-norm-equivalence-improved} and Remark \ref{rem:Sobolev-norm-equivalence-improved}), since this only holds up to curvature errors that also diverge and thus need to be carefully tracked.}\\ 

\change{Secondly, and in contrast to \cite{Rodnianski2014,Speck2018}, replacing the wave structure of the geometric evolution in the Einstein equations with our less geometry dependent Bel-Robinson framework seems to lose regularity at first glance: The \changefinal{energy estimates for the }evolution system for the scalar field energy and the geometric energies can be caricatured as follows\changefinal{: 
\begin{align*}
-\frac{d}{dt}\E^{(L)}(\phi,\cdot)\lesssim&\,\frac{\epsilon^\frac18}t\left[\E^{(L)}(\phi,\cdot)+\E^{(L)}(\Sigma,\cdot)\right]+\dots\\
-\frac{d}{dt}\left[\E^{(L)}(\Sigma,\cdot)+\E^{(L)}(W,\cdot)\right]+\dots\lesssim&\,\frac {\epsilon^\frac18}t\left[\E^{(L)}(\Sigma,\cdot)+\E^{(L)}(W,\cdot)\right]+\frac{\epsilon^{-\frac18}}t\cdot a^{4}\E^{(L+1)}(\phi,\cdot)+\dots\\
\end{align*}
Herein, the superscript refers to the order of derivatives, while $\E^{(L)}(\phi,\cdot), \E^{(L)}(\Sigma,\cdot)$ and $\E^{(L)}(W,\cdot)$ refer to energies for the scalar field, the rescaled tracefree part $\Sigma$ of second fundamental form and the Bel-Robinson variables respectively. }Thus, it seems that we lose derivatives in the scalar field and are not able to close the argument. This is remedied using the div-curl-system in $\Sigma$, see \eqref{eq:comeq-mom-div} and \eqref{eq:comeq-mom-curl}, which yields a weak estimate of the form
\[a^4\E^{(L+1)}(\Sigma,\cdot)\lesssim \E^{(L)}{(\phi,\cdot)}+\E^{(L)}(W,\cdot)+\E^{(L)}(\Sigma,\cdot)+\dots\,.\]
Combining these estimates to improve the bootstrap assumptions then necessitates an intricately constructed total energy to balance these different types of estimates against one another.\\}

Finally, given \eqref{eq:intro-ref3}, the rescaled time derivative of the scalar field is not small and does not become so toward the Big Bang. This leads to various terms within the core linearized evolutionary system of both matter and geometry that, if estimated naively, could lead to exponential blow-up toward the singularity. \change{When such terms occur in the scalar field energy evolution, this can be dealt with along similar lines as in \cite{Rodnianski2014, Speck2018}, but we incur additional large terms in our geometric evolution that only cancel using the explicit form of the Friedman equations, which we highlight in Lemma \ref{lem:en-error-cancellation} and its proof.}

\subsubsection{Future and global stability}

For Milne stability, the canonical Sobolev energies for the scalar field variables\change{, i.e.,~
\[\int_M\lvert\phi^\prime\rvert_{\fg}^2+\lvert\nabla\phi\rvert_{\fg}^2\,\vol{\fg}\]
and higher order analogues, }do not obey useful energy estimates. This can be overcome by adding an indefinite correction term \change{of the type 
\[\int_M\phi^\prime(\phi-\overline{\phi})\vol{\fg}\]
}to the canonical energy\change{, see Definition \ref{def:fut-stab}. This is similar to what was done in \cite{CBM01} in a $2+1$-dimensional setting, as well as similar to the indefinite terms we introduce in our geometric energy to control the wave system in the metric variables, as in previous work on Milne stability in different matter models, including \cite{AndFaj20, FajWy21}. That this corrected energy controls Sobolev norms relies on the aforementioned spectral condition. As a result, and unlike for past stability, the specific spatial geometry is crucial in generating decay from energy estimates, even before considering the geometric evolution.}\\

Moreover, we need to transition from the near-FLRW data used to analyze the contracting regime to data in the expanding regime on a distant enough future hypersurface such that it is near-Milne and the future stability result applies. \change{This requires \changefinal{a gauge }switch from CMC gauge with zero shift to CMCSH gauge, as well as careful control of the solution variables over a finite time interval using continuous dependence on initial data. For the former, close inspection of \cite{FajKr20} gives us a diffeomorphism close to the identity that maps the initial data for the metric to new data satisfying the spatially harmonic gauge condition, thus allowing us to switch gauges without losing proximity to the reference solution. }This is discussed in detail in Section \ref{sec:full-stab}.

\subsection{Proof outline}\label{subsec:intro-pf-outline}\changediss{\phantom{m}\\}
\subsubsection{Big Bang stability}\phantom{m}\\

\textbf{The big picture.} The key argument in our Big Bang stability proof is a hierarchized series of energy estimates that establishes the asymptotic behaviour of solution variables toward the singularity. We rely on a bootstrap argument which establishes that energies $\E^{(L)}$ (see Definition \ref{def:energies}) \change{for the scalar field, the rescaled shear, the Bel-Robinson variables, the lapse and the curvature }at worst only diverge slightly. Here, $0\leq L\leq \change{18}$ denotes the order of \change{derivatives considered}. To this end, we make a bootstrap assumption on the solution \changefinal{norm $\mathcal{C}$ }(see Definition \ref{def:sol-norm}) which controls the distance of \change{these rescaled variables, as well as the metric itself, }to their FLRW counterparts in \change{terms of supremum norms with respect to $G$}, where \change{$G=a^{-2}g$ }is the rescaled \textit{adapted} spatial metric (see Definition \ref{def:rescaled}). We refer to Assumption \ref{ass:bootstrap} and Remark \ref{rem:bs-strategy} for the detailed bootstrap assumptions and improvements\change{, as well as to Lemma \ref{lem:lwp} for the underpinning local well-posedness result. }That this bootstrap argument implies Theorem \ref{thm:main-past} follows from a straightforward adaptation of the arguments in \cite[Theorem 15.1]{Rodnianski2014}.\\

We work with evolution-adapted \change{norms }even though $G(t,x)$ degenerates toward the Big Bang singularity. Indeed, since we need to exploit the structure of the evolutionary equations, it is more convenient to have these adapted quantities controlled by the solution norms $\mathcal{H}$ and $\mathcal{C}$ directly instead of having to perform changes of metric at that point. Once the improved energy estimates are shown, a (time-scaled) coercivity notion (see Lemma \ref{lem:Sobolev-norm-equivalence-improved} and the proof of Corollary \ref{cor:H-imp}) and Sobolev embeddings with respect to the reference metric $\gamma$ then ensure that these improved estimates translate to $\mathcal{H}$ and $\mathcal{C}$. This then closes the bootstrap. To actually achieve this improved energy behaviour, we derive elliptic energy estimates or integral-type estimates that, once suitably combined and scaled, yield the desired improvements by straightforwardly applying the Gronwall lemma. \change{Additionally, note that we assume that the initial data is close to FLRW data not just in $\mathcal{H}$, which contains precisely the norms needed to control $\mathcal{C}$ by Sobolev embedding, but also scaled smallness assumptions at one order higher, contained in the top order semi-norm $\mathcal{H}_{top}$ (see Assumption \ref{ass:init}). This is needed to ensure that the top order energy is small initially, and thus to close the bootstrap.}\\

\textbf{Scale factor \changefinal{$a(t)$}.} The precise structure of the Friedman equations \eqref{eq:Friedman}-\eqref{eq:Friedman2} is crucial not only to control time integral quantities up to the Big Bang hypersurface (see Lemma \ref{lem:scale-factor}), but also to ensure that certain terms in the evolution that would otherwise cause large divergences contribute with favourable sign (see the arguments in Lemma \ref{lem:en-est-SF} as well as Lemma \ref{lem:en-error-cancellation}). It turns out that the sectional curvature entering the Friedman equations actually is not of key importance to large parts of the Big Bang stability analysis\changefinal{: The leading order behaviour }of the scale factor toward the Big Bang singularity is determined via the Friedman equation \eqref{eq:intro-ref2} by the matter term, not the sectional curvature. This indicates that our method might extend to different settings.\\

\textbf{Gauge choice, commutation method and Bel-Robinson variables.} We commute the resulting elliptic-hyperbolic Einstein system with the Laplace-Beltrami operator $\Lap_G$ with respect to the rescaled evolved spatial metric \change{$G(t,x)$ }to obtain higher order energy control. Commuting with this operator has the advantage of leaving many integration-by-parts \change{identities }intact. These are needed to provide specific cancellations, e.g., \change{to cancel $\Lap^{\frac{L}2+1}\phi$-terms arising from the wave equation when computing $\del_t\E^{(L)}(\phi,\cdot)$. We also note that the only feature of the adapted metric we use is that it is close to $\gamma$, and do not use any further information on the geometry, e.g., by choosing a specific reference frame in our commutation method. Further, we employ CMC gauge with zero shift to avoid badly behaved shift terms (see Remark \ref{rem:why-not-CMCSH}).}\\

We still, however, need to deal with the Ricci term in the evolution equation for the second fundamental form. To \change{this end}, we consider the Bel-Robinson variables $E$ and $B$ which are $\Sigma_t$-tangent symmetric tracefree $(0,2)$-tensors and contain all information of the spacetime Weyl tensor $W[\g]$ (see Subsection \ref{subsec:BR}). Suitably projecting the Gauss-Codazzi equations admits additional constraint equations in terms of $E$ and $B$ that allow us to replace the Ricci tensor at the \enquote{cost} of introducing Bel-Robinson energies into the formalism\change{, see \eqref{eq:constr-E} and the rescaled version \eqref{eq:REEqConstrE}}. Further, $E$ and $B$ satisfy a Maxwell-type system (see Lemma \ref{lem:EEqBR}) that can be exploited to obtain energy estimates and, as with the other evolution equations, is well adapted to commutation with $\Lap_G$.\\

\textbf{A priori low order $C_G$-control.}  By applying the bootstrap assumptions on $\mathcal{C}$ to the evolution equations, we can immediately deduce improved low order estimates in $C_G^{l}$ for $l\geq 10$ for the solution variables by inserting them into the respective evolution equations (see Lemma \ref{lem:AP}), as well as via the maximum principle for the lapse (see Lemma \ref{lem:lapse-maxmin}). These usually still diverge slightly, mostly due to the asymptotic behaviour of $G$. However and crucially to our argument, at order $0$, the renormalized time derivative \change{$\Psi$ }of the wave, the rescaled tracefree part \change{$\Sigma$ }of the second fundamental form and the rescaled Bel-Robinson variable $\RE$ are in fact $K\epsilon$-small in $C^0_G$ on the bootstrap interval (see Lemma \ref{lem:APzero}). \change{If these estimates did not hold, it would lead to terms that diverge at order $\O{a^{-3-c\sqrt{\epsilon}}}$ in the differential inequalities, and thus cause exponential energy blow-up of order $\O{e^{a^{-c\sqrt{\epsilon}}}}$ that we could no longer control. This behaviour is closely related to the fact that $\Psi$ and $\Sigma$ converge toward footprint states on the Big Bang hypersurface that remain $K\epsilon$-small (see \eqref{eq:asymp-Psi} and \eqref{eq:asymp-K}), and then pass this convergence on to $\lvert\RE\rvert_G$ (see \eqref{eq:asymp-E}).}\\

\textbf{Energy estimates and hierarchy.} The main part of the analysis is establishing various energy estimates.
\begin{itemize}
\item For the \underline{lapse} (see Section \ref{sec:lapse}), the relevant estimates are direct results of the elliptic lapse equations \change{\eqref{eq:REEqLapse1}-\eqref{eq:REEqLapse2}. The non-lapse terms on the right hand side \changefinal{of \ref{eq:REEqLapse1} }only diverge slightly toward the Big Bang, in contrast to the divergence at \changefinal{order }$a^{-4}$ in \eqref{eq:REEqLapse1}, and thus allows one to show that, at lower derivative order, the lapse converges to $1$. However, since the right hand side of \eqref{eq:REEqLapse2} contains the scalar curvature of $G$, this estimate loses derivatives. On the other hand, \eqref{eq:REEqLapse1} does not lose derivatives, and the elliptic nature in fact allows one to \changefinal{estimate }lapse energies of order $L+2$ by energies in $\Sigma$ and the scalar field of \changefinal{order }$L$. This makes it possible to control the higher order lapse term occurring, for example, in \eqref{eq:REEqSigma}, without losing regularity. Conversely, both of these gains in regularity are at the cost of losing powers of $a$. In short, \eqref{eq:REEqLapse2} is needed to establish the asymptotic behaviour of the lapse, and \eqref{eq:REEqLapse1} to obtain improved energy bounds as a whole.}
\item The core \underline{matter} energy estimate (see Lemma \ref{lem:en-est-SF}) relies on delicate cancellations when computing the time derivative of $\E^{(L)}(\phi,\cdot)$. While we derive this in a fashion that differs from the energy flux method used in \cite{Speck2018}, the necessary cancellations to arrive at Lemma \ref{lem:en-est-SF} are similar. 
\item The (rescaled) tracefree component of the \underline{second fundamental form} $\Sigma$ (see Lemma \ref{lem:en-est-Sigma}) and the (rescaled) \underline{Bel-Robinson variables} $\RE$ and $\RB$ (see Lemma \ref{lem:en-est-BR}) need to be treated simultaneously to deal with the leading curvature term in the evolution of the former by inserting a constraint equation in which $\RE$ occurs as the leading term (see \eqref{eq:comeq-Ham-BR}). However, the matter terms within the evolution of $\RE$ and $\RB$ contain, firstly, terms where we again need very precise estimates to show that they do not contribute large $a^{-3}$-divergences, and, secondly, matter terms that lose one order of derivative.\\
\change{This order of regularity can be regained using }the momentum constraint equation \eqref{eq:comeq-mom-div} and its Bel-Robinson counterpart \eqref{eq:comeq-mom-curl} containing $\RB$, {which leads }to a div-curl-system for $\Sigma$ \change{(see Lemma \ref{lem:en-est-Sigma-top}). This is, again, at the cost of losing powers of $a$.}
\item As a result, the \underline{core Gronwall argument} performed in Proposition \ref{prop:en-bs-imp} combines energies for the matter variables, $\Sigma$ and the Bel-Robinson variables, as well as energies for $\Ric[G]$. \change{In particular, the curvature energies are necessary to handle commutation errors within the energy estimates, and improved bounds on them need to be obtained to apply the coercivity results in Lemma \ref{lem:Sobolev-norm-equivalence-improved} -- else, none of energy improvements would extend to improved Sobolev norm bounds and the bootstrap argument would not close. }\\
As many of the a priori $C_G$-norm estimates add small additional divergences, it is necessary to perform an induction over derivative orders within this mechanism to deal with lower order error terms. Since $\Lap_G$ is elliptic, it is sufficient to perform this for even orders. \change{Along with energies at order $L\in 2\N_0$, the total energy also includes the energy controlling $\Sigma$ as well as the scalar field and curvature energies at order $L+1$, appropriately scaled to account for the degenerate elliptic estimate for $\Sigma$ from Lemma \ref{lem:en-est-Sigma-top}. This remedies the derivative loss in the Bel-Robinson energy and allows one to improve the total energy at each order until reaching $L=18$, at which point the bootstrap argument can be closed.}
\item Note that \textit{the \underline{metric} itself does not enter the core energy mechanism}. In fact, trying to replace control of the Ricci tensor by control of $G$ is likely too imprecise in dealing with high order curvature errors. Instead, control of $G-\gamma$ \change{and }$\Gamma[G]-\Gamhat[\gamma]$ is a consequence of a simple integral energy inequality and the improvements achieved for $\Sigma$ and matter variables (see Lemma \ref{lem:norm-est-G} and Corollary \ref{cor:H-imp}). Since we cannot utilize any additional structure in dealing with the metric, we have to construct our argument carefully to allow for the metric control to be weaker than what one gets for the core variables, while still being sufficiently strong to constitute an improvement and allowing to switch between $H_G$ and $H_\gamma$ (and, respectively, $C_G$ and $C_\gamma$) norms.
\end{itemize}
We also point to Remark \ref{rem:en-est-strat} for a more detailed sketch of how the integral inequalities for the core Gronwall argument are structured and how this leads to the bootstrap improvement for the energies.

\subsubsection{Future stability and connecting the regions}

We follow similar lines as in \cite{AndFaj20, FajWy21} to prove that near-FLRW spacetimes in negative spatial geometry are future stable. \change{Since $\del_t\phi$ decays like $a^{-3}\simeq t^{-3}$ in the reference spacetime, the sectional curvature becomes dominant in the Friedman equations and the scale factor approaches that of Milne spacetime as $t$ approaches $\infty$. Hence, if one moves sufficiently far \changefinal{to the future}, choosing near-FLRW data with a homogeneous scalar field is equivalent to choosing near-vacuum data. Thus, }what we prove first in Section \ref{sec:fut} is future stability of near-Milne spacetimes under the Einstein scalar-field system. Once this is established, we argue in Section \ref{sec:full-stab} how early near-FLRW initial data evolves to data that is sufficiently close to Milne for large enough times, which is essentially a consequence of the scale factor and the (physical) mean curvature approaching that of Milne, up to a multiplicative constant.\\

In terms of dealing with geometric and elliptic estimates, we can essentially carry over the results of \cite{AndFaj20}, as was also done in \cite{FajWy21}, by working in CMCSH gauge and verifying that the matter components are indeed only perturbative terms within the geometric evolution.

This leaves only the scalar field to be examined. Here, we introduce corrective terms to the energies (see Definition \ref{def:fut-stab}) which yield decay estimates for the corrected scalar field energy (see Lemmas \ref{lem:fut-en-est-ESF0} and \ref{lem:fut-en-est-ESF}). That these energies are coercive (see Lemmas \ref{lem:fut-ESF-coercivity} and \ref{lem:fut-Sob-est}) requires the aforementioned lower bound for the first positive eigenvalue of \changefinal{$-\Lap_\gamma$}.

\begin{remark}[Why not use CMCSH gauge to prove Big Bang stability?]\label{rem:why-not-CMCSH}
\change{One might consider applying this gauge to Big Bang stability as well since this is precisely the choice of gauge turning the geometric evolution into a wave-like system in $(g,k)$, which seems simpler than our chosen approach in CMC gauge with zero shift. }In particular, this would also not rely on any choice of reference frame, and keep the wave structure of the geometric evolution intact, unlike when using Bel-Robinson variables. However, the issue with this approach lies in the shift equation, which would take the following form for the rescaled shift vector $X=a^3\tilde{X}$:
\begin{align*}
\Lap_GX^l+\Ric[G]^l_mX^m=&-2(N+1)(G^{-1})^{im}(G^{-1})^{jn}\Sigma_{ij}\left(\Gamma_{mn}^l-\Gamhat_{mn}^l\right) \numberthis\label{eq:REEqShift}\\
&+2(G^{-1})^{im}\nabla_iX^n\left(\Gamma_{mn}^l-\Gamhat_{mn}^l\right)\\
&\,+\langle\text{error terms in lapse and matter}\rangle
\end{align*}
As a result, the first term has to be expected to diverge at the same rate as the metric, i.e., we expect even low order norms of $\tilde{X}$ to behave like $a^{-3-c\sqrt{\epsilon}}$ at best up to small prefactors. However, computing the time derivative of an integral over $\lvert G-\gamma\rvert_G^2$ (or derivatives thereof) becomes the integral over the $(\del_t-\Lie_{\tilde{X}})$-derivative of this quantity, and hence we get explicit terms of the form $\Lie_{\tilde{X}}\gamma$ which always exist at highest order and diverge worse than $t^{-1}$. In short, the fact that the metric cannot be expected to converge to a footprint state leads to leading order terms in the differential energy estimates to carry strongly divergent pre-factors in CMCSH gauge. This obstructs improvements in a tentative bootstrap argument.
\end{remark}

\subsection{Paper outline}\label{subsec:intro-paper-outline}

\begin{itemize}
\item Sections \ref{sec:prelim}-\ref{sec:main-thm} cover the proof of Big Bang stability:
\begin{itemize}
\item In Section \ref{sec:prelim}, we introduce notation and provide the necessary information on the FLRW background solution as well as the equations relevant to the subsequent analysis.
\item Then, in Section \ref{sec:norm-en-bs}, we discuss the solution norms and energies and state the initial data and bootstrap assumptions.
\item In Section \ref{sec:ap}, improved low order $C_G$-norm estimates that follow directly from the bootstrap assumptions are established, along with additional formulas and a priori estimates.
\item Section \ref{sec:lapse} concerns the elliptic estimates for the lapse. 
\item \change{In Section \ref{sec:en-est}, }we discuss the energy and Sobolev norm estimates for all other variables, \change{all of which are integral estimates except for the aforementioned elliptic estimate for $\Sigma$, as well as a norm bound for $\nabla\phi$ that is not needed for the energy improvement.}
\item These are all combined in Section \ref{sec:bs-imp} to improve the bootstrap assumptions -- first for the energies, then for $\mathcal{H}$ and finally $\mathcal{C}$. \item In Section \ref{sec:main-thm}, we show how this bootstrap argument implies the main Big Bang stability result (see Theorem \ref{thm:main}, which is the formal version of Theorem \ref{thm:main-past}).
\end{itemize}
\item Section \ref{sec:fut} contains the proof of near-Milne future stability. 
\item \change{In Section \ref{sec:full-stab}, }we show that this is sufficient for future stability of near-FLRW spacetimes, proving Theorem \ref{thm:main-full}.
\item The appendices (Sections \ref{sec:appendix}-\ref{sec:appendix-fut}) collect various basic formulas and commutator expressions as well as error terms and how these can be estimated.
\end{itemize}

\noindent \textbf{Acknowledgements.} \changefinal{This research was funded in whole or in part by the Austrian Science Fund (FWF) 10.55776/Y963 and 10.55776/P34313. For open access purposes, the author has applied a CC BY public copyright license to any author-accepted manuscript version arising from this submission. }\change{Liam Urban is a recipient of a DOC Fellowship of the Austrian Academy of Sciences at the Faculty of Mathematics at the University of Vienna. }Liam Urban also thanks the German Academic Scholarship Foundation (Studienstiftung des deutschen Volkes) for their scholarship. The authors thank \changefinal{Ian Agol, Klaus Kröncke, Michael Lipnowski, Dalimil Mazac and Roman Prosanov }for their help in seeking out \changefinal{numerical and analytic }evidence for the spectral condition used in Section \ref{sec:fut}, and Michael Eichmair \change{and the anonymous referees }for \changefinal{their detailed, constructive and warm }feedback on a previous version of this manuscript. \changefinal{The authors would like to thank the Erwin Schrödinger International Institute for Mathematics and Physics in Vienna for hosting the authors during the Thematic Programs \enquote{Mathematical Perspectives of Gravitation beyond the Vacuum Regime}, \enquote{Spectral Theory and Mathematical Relativity} and \enquote{Nonlinear Waves and General Relativity} during which research for this work was done and parts of this paper were written.}


\section{Big Bang stability: Preliminaries}\label{sec:prelim}

\subsection{Notation}\label{subsec:notation}

\subsubsection{Foliations}\label{subsubsec:notation-foliation}

\change{On a spacetime manifold $(\M,\g)$, we assume the existence of a spacelike Cauchy hypersurface $\Sigma_{t_0}$ that is diffeomorphic to $M$. As \changefinal{we argue }in Remark \ref{rem:CMC-hypersurface}, we can assume without loss of generality that it has constant mean curvature. We will ultimately show that there exists a time function $t$ such that the past of $\Sigma_{t_0}=t^{-1}(t_0)$ can be foliated by $\Sigma_{s}=t^{-1}(s)$ for $s\in (0,t_0)$, and that where the solution exists, this is at least possible up to some $T\in(0,t_0)$. These constant time surfaces are then also spacelike Cauchy hypersurfaces diffeomorphic to $M$ and CMC. We will use this notation throughout with little comment and often simply view $\Sigma_s$ as $\{s\}\times M$.}

\subsubsection{Metrics}\label{subsubsec:notation-metric}

The spacetime metric $\g$ on $\M$ takes the general form
\[\g=-n^2dt^2+g_{ab}dx^adx^b\]
where $n\equiv n(t,x)$ is the lapse function and $g\vert_{\Sigma_t}\equiv g\vert_{\Sigma_t}(t,x)$ is a Riemannian metric on $\Sigma_t$. We will often simply denote the spatial metric by $g$. Furthermore, we denote the rescaled spatial metric \changefinal{by $G_{ij}=a^{-2}g_{ij}$ (see Definition \ref{def:rescaled}) }and the tensor-field induced by the matrix inverse of $(G_{ij})$ by $G^{-1}$. Similarly, $\det g$ and $\det G$ are also meant as the determinants in the matrix sense. Finally, we define $\vol{g}$ and $\mu_g$ as the volume form and volume element with regard to $g$, and the same for $\gamma$ and $G$.

\subsubsection{Indices and coordinates}\label{subsubsec:notation-indices}

Greek indices $\alpha,\beta,\dots,\mu,\nu,\dots$ run from $0$ to $3$, lowercase latin indices $a,b,\dots$, $i,j,\dots$ from $1$ to $3$. The spatial indices on some coordinate neighbourhood $V\subseteq\M$ are always with regard to the local frame induced by coordinates $(x^1,x^2,x^3)$ on $M$, applied to each $V\cap\Sigma_t$ by the standard embedding where this intersection is non-empty. The index $0$ always denotes components relative to $\del_0=n^{-1}\del_t$, where $\del_t$ is the derivative associate to the time function $t$. The Levi-Civita connections associated to $\g$, respectively $g$ and $G$, are denoted by $\nabbar$, respectively $\nabla$.\footnote{Note that $g$ and $G$ have the same Levi-Civita-connection since, on every hypersurface $\Sigma_t$, they are related by a scalar multiple. }Additionally, for the hyperbolic spatial reference metric $\gamma$ on $M$ (see  Definition \ref{def:spatial-mf}), we write the Levi-Civita connection as $\nabhat$.\\

Whenever we raise or lower Greek (resp. Latin) indices without additional notation, it is with regard to $\g$ (resp. $g$). When we raise indices of a tensor $\mathfrak{T}$ with regard to the rescaled spatial metric $G$, we flag this by writing $\mathfrak{T}^\sharp$. We never raise or lower with respect to $\gamma$. Refer to Section \ref{subsubsec:notation-derivatives} as to how we distinguish taking multiple covariant derivatives from index raising.

\subsubsection{$\Sigma_t$-tangent tensors}\label{subsubsec:notation-tangent}

For any $\Sigma_t$-tangent tensor ${\xi^{\alpha_1\dots\alpha_r}}_{\beta_1\dots\beta_s}$, we write ${{\xi(t)}^{a_1\dots a_r}}_{b_1\dots b_r}$ for the $\g$-orthogonal projection of $\xi$ onto the hypersurface $\Sigma_t$. When clear from context, we will drop the time dependency in notation.

\subsubsection{Sign conventions}\label{subsubsec:notation-sign}

Within this paper, the second fundamental form with regard to $\Sigma_t$ is defined \change{as the $(0,2)$-tensor $k$ given by
\[k(X,Y)=-\g(\nabbar_X\del_0,Y)\,,\]
where $X$ and $Y$ are $\Sigma_t$-tangent vectors}. The Riemann curvature tensor of $\g$ is taken to be
\change{\[\nabbar_\alpha\nabbar_\beta Z_\gamma-\nabbar_\beta \nabbar_{\alpha}Z_\gamma={\Riem[\g]_{\alpha\beta\gamma}}^\delta Z_\delta\]}
for the covariant vector field $(Z_\mu)$, and the analogous convention holds for all other Riemann curvature tensors that appear.

\subsubsection{Constants}\label{subsubsec:notation-constants}

For two nonnegative scalar functions $\zeta_1,\zeta_2$, we write $\zeta_1\lesssim\zeta_2$ \change{if and only if } there exists a constant $K>0$ such that $\zeta_1\leq K\zeta_2$. This implicit constant may depend on information from the FLRW reference solution at the starting point of the evolution (in particular on $\gamma$ and $a(t_0)$, see \change{Definition \ref{def:spatial-mf}}) and combinatorial quantities. We extend this notation to a real function $\zeta_1^\prime$ by
\[\zeta_1^\prime\lesssim\zeta_2 :\Leftrightarrow \max\left(\zeta_1^\prime,0\right)\lesssim\zeta_2\,.\]
Additionally, we write $\zeta_1\simeq\zeta_2$ \change{if and only if } $\zeta_1\lesssim\zeta_2\lesssim\zeta_1$ is satisfied.\\


\subsubsection{Tensor contractions}\label{subsubsec:notation-contract}

We denote by $\epsilonLC_{\alpha\beta\gamma\delta}$ the Levi Civita tensor with regard to $\g$ and define the Levi-Civita tensor on spatial hypersurfaces $\Sigma_t$ by $\epsilonLC[g]_{ijk}=\epsilonLC_{0ijk}$. Notice that this corresponds to the Levi-Civita tensor associated to $g$. Further, $\epsilonLC[G]_{ijk}=a^{-3}\epsilonLC[g]_{ijk}$ is the Levi-Civita tensor with respect to the rescaled metric $G$ (see \eqref{eq:rescalingGK}).

For $\Sigma_t$-tangent $(0,2)$-tensors $A,\tilde{A}$ and vector field $v$, we define the following objects \change{as in \cite[Section A.2]{AM03}}:
\begin{align*}
A\cdot\tilde{A}&=\changefinal{A_{ab}\tilde{A}^{ab}}=\langle A,\tilde{A}\rangle_g\\
(A\odot_g\tilde{A})_{ij}&=\change{A_{ik}{\tilde{A}^k}_{\ j}}\\
(A\wedge \tilde{A})_i&={\epsilonLC_i}^{jp}{A_j}^q\tilde{A}_{qp}\\
(v\wedge A)_{ab}&={\epsilonLC_a}^{cd}v_cA_{db}+{\epsilonLC_b}^{cd}v_cA_{ad}\\
(A\times \tilde{A})_{ij}&={\epsilonLC_{i}}^{ab}{\epsilonLC_j}^{pq}A_{ap}\tilde{A}_{bq}+\frac13(A\cdot\tilde{A})g_{ij}-\frac13(tr_gA\cdot tr_g\tilde{A})g_{ij}\\
(\curl A)_{ij}=(\curl_g A)_{ij}&=\frac12\left[{\epsilonLC_i}^{cd}\nabla_dA_{cj}+{\epsilonLC_j}^{cd}\nabla_dA_{ci}\right]\\
(\div_g A)_i&=\nabla^bA_{ib}
\end{align*}
The operations $\odot_G$, $\langle\cdot,\cdot\rangle_G$ and $\div_G$ are defined analogously, with all indices raised and lowered by $G$ instead of $g$. Finally, for two $(0,1)$-tensors $\pi,\tilde{\pi}$, we denote their symmetrized product by 
\[(\pi\otimes\tilde{\pi})_{ij}=\frac12\left(\pi_i\tilde{\pi}_j+\pi_j\tilde{\pi}_i\right)\,.\]
\change{For pointwise estimates of these quantities, refer to Lemma \ref{lem:BelRobinsonLemmas}.}

\subsubsection{Schematic term notation}\label{subsubsec:schematic-notation}
We will denote as $\mathfrak{T}_1\ast\dots\ast\mathfrak{T}_l$, \change{where $\mathfrak{T}_i$ }are $\Sigma_t$-tangent tensors, any multiple of $(\mathfrak{T}_i)$, with regard to the rescaled adapted spatial metric $G$ or as standard multiplication if no summation over indices occurs between factors. Constant prefactors and contractions with regard to $G$ are also \change{suppressed }in this notation. 

When working with terms where such notation is used, we will estimate these inner products by $\lesssim \prod_{i=1}^l\lvert\mathfrak{T}_i\rvert_G$, making any constant in front irrelevant, and further we can view any contraction with regard to $G$ as a product of the non-contracted tensor $\mathfrak{T}$ with $G$ or $G^{-1}$, and estimate that up to constant by $\lvert G\rvert_G\lvert\mathfrak{T}\rvert_G$, where the first factor is simply $\sqrt{3}$.\\
For similar products with respect to $\gamma$, we denote them by $\ast_{\gamma}$.

\subsubsection{On multiple derivatives of variables}\label{subsubsec:notation-derivatives}

For a scalar function $\zeta$, an $(r,s)$-tensor field $\mathfrak{T}$ and \textit{capitalized} integers $I,J,\change{\ldots}\in\N_0$, we denote by $\nabla^I\zeta$ and $\nabla^I\mathfrak{T}$ the tensors $\nabla_{l_1}\dots\nabla_{l_I}\zeta$ and\linebreak
$\nabla_{l_1}\dots\nabla_{l_I}{\mathfrak{T}^{i_1\dots i_r}}_{j_1\dots j_s}$. We extend this notation to other covariant derivatives analogously. To avoid potential ambiguity with an index raised by \change{$g$}, we will apply the following convention:
\begin{itemize}
\item If a covariant derivative carries an uppercase letter, a formula with more than one symbol or a positive integer in its superscript, this refers to taking a derivative of that order.
\item If a covariant derivative carries a lowercase letter or $0$ in its superscript, this refers to an index.
\end{itemize}
Further, we will only apply this notation where the precise distribution of indices is not important (e.g.\,in schematic notation, see Section \ref{subsubsec:schematic-notation}).

\subsection{FLRW spacetimes and the Friedman equations}\label{subsec:FLRW}

Herein, we collect the properties of the reference FLRW solution to the Einstein scalar-field system in CMC-transported coordinates. Our main focus will lie on the behaviour of the scale factor as determined by the Friedman equations. Before moving on to that, we collect the information on the spatial geometry we will need:

\begin{definition}[Hyperbolic reference geometry]\label{def:spatial-mf}
$(M,\gamma)$ is a three-dimensional, connected, closed, orientable Riemannian manifold with constant sectional curvature $-\frac19$, and hence Ricci tensor $\Ric[\gamma]=-\frac29\gamma_{ij}$ and scalar curvature $R[\gamma]=-\frac23$. 
\end{definition}
\change{\begin{remark}[Orientability is not a restriction]
We assume that $M$ is orientable for the sake of simplicity. If $M$ should be non-orientable, we may pass the initial data to the oriented double cover and solve the problem there. Since the result is equivariant with respect to the double covering map, this then solves the original problem.
\end{remark}}

With this in hand, we can express our classical \change{family of solutions }to the Einstein scalar-field system as follows:

\begin{lemma}[FLRW solutions and Friedman equations]\label{lem:FLRW}
Consider FLRW spacetimes $(\M,{\g}_{FLRW})$ with $\M=(0,\infty)\times M$, where $(M,\gamma)$ is as in Definition \ref{def:spatial-mf} and where
\begin{equation}\label{eq:FLRW-metric}
{\g}_{FLRW}=-dt^2+a(t)^2\gamma
\end{equation}
holds for some $a\in C^\infty((0,\infty))$, with the conventions $a(0)=0$ and $\dot{a}(T)>0$ for some arbitrary $T>0$. Further, choose a (smooth) scalar function $\phibar$ such that
\begin{equation}\label{eq:FLRW-wave}
\del_t\phibar=C\cdot a(t)^{-3},\quad \nabla\phibar=0, \quad \square_{\g_{FLRW}}\phibar=0\,.
\end{equation}
Such a pair $(\g_{FLRW},\phibar)$ solves the Einstein scalar-field system \eqref{eq:ESF1}-\eqref{eq:EMT} on $\M$ if and only if $a$ satisfies \change{the Friedman equation
\begin{equation}\label{eq:Friedman}
\dot{a}=\sqrt{\frac19+\frac{4\pi}3C^2a^{-4}}\,.
\end{equation}
In particular, one has
\begin{equation}\label{eq:Friedman2}
\ddot{a}=-\frac{8\pi}3C^2a^{-5}\,.
\end{equation}}
\end{lemma}
\begin{proof}
The first statement follows from explicitly computing $\Ric[\g]$ as in \cite[p.345]{ONeill83}. That \eqref{eq:Friedman} implies \eqref{eq:Friedman2} follows simply by computing the derivative of $\dot{a}^2$.
\end{proof}

In the subsequent analysis, the following properties of $a$ that follow from \eqref{eq:Friedman} will be crucial for our analysis:

\begin{lemma}[Scale factor analysis]\label{lem:scale-factor}Let $a$ solve \eqref{eq:Friedman} with $a(0)=0$. Then $a$ is analytic on $(0,\infty)$ and extends to a continuous function on $[0,\infty)$ with $a(t)\simeq t^{\frac13}$ being satisfied near $t=0$. Further, for any $p>0$, there exist constants $c>0$ and $K_p>0$, where $c$ is independent of $p$ and $K_p$ depends analytically on $p$, such that, for any $t\in(t,t_0]$, one has
\begin{equation}\label{eq:a-exp-est}
\exp\left(p\int_t^{t_0}a(s)^{-3}\,ds\right)\leq K_pa(t)^{-cp}
\end{equation}
and
\begin{equation}\label{eq:a-integrals}
\int_t^{t_0}a(s)^{-3-p}\,ds\changefinal{\lesssim}\frac1p a(t)^{-p},\quad \int_t^{t_0}a(s)^{-3+p}\,ds\lesssim\frac{1}p\,.
\end{equation}
Moreover, for any $t\in(0,t_0]$ and any $q>0$, \change{there exist constants $c>0$ and $K>0$ which both are independent of $q$ such that one has
\begin{equation}\label{eq:log-est}
\int_t^{t_0}a(s)^{-3}\,ds\leq \frac{K}qa(t)^{-cq}\,.}
\end{equation}
Finally, \eqref{eq:Friedman} also implies
\begin{equation}\label{eq:diff-ineq-Friedman}
\sqrt{\frac{4\pi}3}C a^{-2}\leq \dot{a}
\end{equation}
\end{lemma}
\begin{remark}\label{rem:scale-factor}
We will use the estimates in Lemma \ref{lem:scale-factor} where $p$ is a positive power of $\epsilon$ up to algebraic constants. Then, we can simply replace $K_p$ in \eqref{eq:a-exp-est} by a uniform constant. 
%
\end{remark}
\begin{proof}
For the first statement, we refer to \cite[Lemma 2.1]{Urban22} with $\gamma=2$. We also collect from there\footnote{\label{fn:t0-small}In \cite[Lemma 2.1]{Urban22}, one at first only has $\int_t^{t_0}a(s)^{-3}\,ds\lesssim \log(t_0)-\log(t)$ for $t_0$ small enough that we can estimate $a(t)$ by $t^\frac13$ up to constant. However, assuming this inequality were to hold up to $\overline{t}>0$ and one had $t_0>\overline{t}$, the contribution $\int_{\overline{t}}^{t_0}a(s)^{-3}\,ds$ only adds a constant that we can absorb into our notation. Similarly, $a(t)\simeq t^\frac13$ can be assumed to hold on $(0,t_0]$ for any fixed $t_0>0$, and we can ignore this technicality in proving the integral formulas in Lemma \ref{lem:scale-factor}.} that, for $t<t_0$, 
\[\changefinal{\int_t^{t_0}a(s)^{-3}\,ds\lesssim 1+\left\lvert\log\left(\frac{t}{t_0}\right)\right\rvert}\]
is satisfied. Hence, there exists some $c^\prime>0$ such that
\[\changefinal{\exp\left(p\int_t^{t_0}a(s)^{-3}\,ds\right)\leq K_p\exp(c^\prime\cdot p\log(t_0))\cdot\exp(-c^\prime\cdot p)\leq K_p^\prime t^{-c^\prime p}}\,.\]
Then \eqref{eq:a-exp-est} follows by applying $a(t)\simeq t^\frac13$. Noting that $a^{-3}\simeq \nicefrac{\dot{a}}a$ holds, one further has
\begin{equation}
\int_t^{t_0}a(s)^{-3-p}\,ds\lesssim\int_{a(t)}^{a(t_0)}y^{-1-p}dy=\frac1p\left(a(t)^{-p}-a(t_0)^{-p}\right)\leq\frac1p a(t)^{-p}\,,
\end{equation}
and the other inequality in \eqref{eq:a-integrals} follows analogously. Finally, \eqref{eq:log-est} follows directly from \eqref{eq:a-integrals} when assuming without loss of generality that $a\lvert_{(t,t_0)}$ only takes values in $(0,1)$.

\end{proof}

\subsection{Solutions to the Einstein scalar-field equations in CMC gauge}\label{subsec:eq}

From here on out, we impose the CMC condition\change{\footnote{\change{Recall that $k$ is negative in our sign convention, see Section \ref{subsubsec:notation-sign}.}}}
\begin{equation}\label{eq:CMC}
{k^{l}}_l(t,\cdot)=\tau(t)=-3\frac{\dot{a}(t)}{a(t)}\,.
\end{equation}
We use \eqref{eq:Friedman} and \eqref{eq:Friedman2} to collect the following formulas for the mean curvature:
\begin{align}
\label{eq:Hubble} \del_t\tau
&=12\pi C^2a^{-6}+\frac13a^{-2}\\
\label{eq:Hubble2} \tau^2&=9\frac{\dot{a}^2}{a^2}=12\pi C^2a^{-6}+a^{-2}\,.
\end{align}
We consequently define the trace-free component $\hat{k}$ of $k$ as
\begin{equation}
\hat{k}_{ij}=k_{ij}-\frac{\tau}3g_{ij}
\end{equation}
and recall that the future directed unit normal to our foliation is written as
\begin{equation}
\del_0=n^{-1}\del_t\,.
\end{equation}

With this, we can express the Einstein scalar-field equations in our gauge as follows:

\begin{prop}[The Einstein scalar-field system in CMC gauge]\label{prop:eq}
A pair $(\g,\phi)$ solves the Einstein scalar-field equations \eqref{eq:ESF1}-\eqref{eq:ESF2} on $I\times M$ in CMC gauge \eqref{eq:CMC} for some interval $I\subseteq (0,t_0]$, where the scale factor satisfies \eqref{eq:Friedman}, if and only if the following equations are satisfied on $I\times M$:\\
The \textbf{metric evolution equations}
\begin{subequations}
\begin{align*}
\del_t g_{ij}=&\,-2nk_{ij}=-2n\hat{k}_{ij}+2n\frac{\dot{a}}ag_{ij}\,, \numberthis\label{eq:EEqg}\\
\numberthis\label{eq:EEqk}\del_t \hat{k}_{ij}
=&\,-\nabla_i\nabla_jn+n\left[\Ric[g]_{ij}-\frac{\dot{a}}a\hat{k}_{ij}-2\hat{k}_{il}\hat{k}^l_j-8\pi\nabla_i\phi\nabla_j\phi\right]\\
&\quad+4\pi C^2a^{-6}(n-1)g_{ij}+\frac19\left(3n-1\right)a^{-2}g_{ij}\,,
\end{align*}
\end{subequations}
the \textbf{Hamiltonian} and \textbf{momentum constraint equations}
\begin{subequations}
\begin{align}
R[g]+\frac23\tau^2-\langle\hat{k},\hat{k}\rangle_g=&8\pi\left[\lvert \del_0\phi\rvert^2+\lvert\nabla\phi\rvert_g^2\right] \label{eq:Hamilton}\,,\\
\nabla^l\hat{k}_{lj}=&-8\pi\nabla_j\phi\cdot \del_0\phi\,, \label{eq:Momentum}
\end{align}
\end{subequations}
the \textbf{lapse equation}
\begin{subequations}
\begin{align*}
\numberthis\label{eq:EEqLapse}\Lap_g n
=&-12\pi C^2a^{-6}-\frac13a^{-2}+n\left[\frac13a^{-2}+4\pi C^2a^{-6}+\langle\hat{k},\hat{k}\rangle_g+8\pi\lvert \del_0\phi\rvert^2\right]\,,
\end{align*}
or equivalently by \eqref{eq:Hamilton}
\begin{equation}\label{eq:EEqLapse2}
\Lap_g n=-12\pi C^2a^{-6}-\frac13a^{-2}+n\left[R[g]-8\pi\lvert\nabla\phi\rvert_g^2+12\pi C^2a^{-6}+a^{-2}\right]\,,
\end{equation}
\end{subequations}
and the \textbf{wave equation}
\begin{equation}\label{eq:wave}
\square_{\g}\phi=-\del_0^2\phi+n^{-1}g^{ij}\nabla_in\nabla_j\phi+\Lap_g\phi+\tau \del_0\phi=0\,.
\end{equation}
\end{prop}
\begin{proof}
These are standard equations that follow from \cite[Chapter 2.3]{Rendall08} and applying \eqref{eq:CMC}-\eqref{eq:Hubble2}.
\end{proof}

\subsection{Bel-Robinson variables}\label{subsec:BR}
In this subsection, we briefly (re-)establish Bel-Robinson variables and how they behave within the Einstein scalar-field system.\\

\noindent\change{ Recall that the Weyl tensor $W\equiv W[\g]$ is the trace-free component of the spacetime curvature and, in the Einstein scalar-field system, takes the form
\begin{align*}
W_{\alpha\beta\gamma\delta}=&\,\Riem[\g]_{\alpha\beta\gamma\delta}-P[\g]_{\alpha\beta\gamma\delta}\,,\\
P[\g]_{\alpha\beta\gamma\delta}=&\,\frac12\left(\g_{\alpha\gamma}\Ric[\g]_{\beta\delta}-\g_{\beta\gamma}\Ric[\g]_{\alpha\delta}-\g_{\alpha\delta}\Ric[\g]_{\gamma\beta}+\g_{\beta\delta}\Ric[\g]_{\alpha\gamma}\right)-\frac16 R[\g]\left(\g_{\alpha\gamma}\g_{\beta\delta}-\g_{\alpha\delta}\g_{\beta\gamma}\right)\\
=&\,4\pi\left(\g_{\alpha\gamma}\nabbar_{\beta}\phi\nabbar_\delta\phi-\g_{\beta\gamma}\nabbar_\alpha\phi\nabbar_\delta\phi-\g_{\alpha\delta}\nabbar_\beta\phi\nabbar_{\gamma}\phi+\g_{\beta\delta}\nabbar_{\alpha}\phi\nabbar_{\gamma}\phi\right)\\
&\,-\frac{4\pi}3 \left(\nabbar^\rho\phi\nabbar_\rho\phi\right)\left(\g_{\alpha\gamma}\g_{\beta\delta}-\g_{\alpha\delta}\g_{\beta\gamma}\right)\,.
\end{align*}}
We define the dual $W^\ast$ of the Weyl tensor as
\[W^\ast_{\alpha\beta\gamma\delta}=\frac12\epsilonLC_{\alpha\beta\mu\nu}{W^{\mu\nu}}_{\gamma\delta}\,.\]
The electric and magnetic components of the Weyl tensor\change{, referred to as the Bel-Robinson variables from here on, }are now defined as
\[E(W)_{\alpha\beta}=W_{\alpha\mu\beta\nu}\del_0^{\mu}\del_0^{\nu}=W_{\alpha 0\beta 0},\ B(W)_{\alpha\beta}=W^\ast_{\alpha\mu\beta\nu}\del_0^\mu \del_0^\nu=W^{\ast}_{\alpha 0\beta 0}\,.\]
We note that, conversely, the Weyl tensor can be fully reconstructed from $E$ and $B$ since the following identities hold:
\begin{equation}\label{eq:Weyl-reconstruct}
W_{a0c0}=E_{ac},\quad W_{abc0}=-{\epsilonLC_{ab}}^mB_{mc},\quad W_{abcd}=-\epsilonLC_{abi}\epsilonLC_{cdj}E^{ij}
\end{equation}

By the symmetries of the Weyl tensor as a whole, $E$ and $B$ are symmetric and one has $E_{0\beta}=0=B_{0\beta}$. Hence, $E$ and $B$ are symmetric, tracefree $\Sigma_t$-tangent $(0,2)$-tensors which we shall simply denote as $E_{ij}$ and $B_{ij}$.\\
Further, we define
\begin{equation}
J_{\beta\gamma\delta}:=\nabbar^\alpha W_{\alpha\beta\gamma\delta},\quad J^\ast_{\beta\gamma\delta}:=\nabbar^\alpha W^\ast_{\alpha\beta\gamma\delta}\,.
\end{equation}
By applying the Bianchi identity for $\Riem[\g]$, we gain the explicit expression
\begin{equation}
J_{\beta\gamma\delta}=\frac12\left(\nabbar_\gamma \Ric[\g]_{\delta\beta}-\nabbar_\delta \Ric[\g]_{\gamma\beta}\right)-\frac1{12}\left(\g_{\beta\delta}\nabbar_\gamma R[\g]-\g_{\beta\gamma}\nabbar_\delta R[\g]\right)\,.
\end{equation}
Using \eqref{eq:ESF1}, we collect:
\begin{align*}
\change{\numberthis\label{eq:J}J_{i0j}
=}&\change{\,4\pi\Bigr[\nabla_i(\del_0\phi)\nabla_j\phi+{k^l}_i\nabla_l\phi\nabla_j\phi-\del_0\phi\nabla_i\nabla_j\phi\\}
&\change{\,\phantom{4\pi}-k_{ij}(\del_0\phi)^2-n^{-1}\nabla_in\cdot\nabla_j\phi\cdot\del_0\phi\Bigr]-\frac{2\pi}3\Bigr[\del_0(\nabbar^\alpha\phi\nabbar_\alpha\phi)\Bigr]g_{ij}\\}
\numberthis\label{eq:Jast}J^\ast_{i0j}
=&\,4\pi\epsilonLC_{lmj}\left(\nabla^l\nabla_i\phi+{k^l}_i\del_0\phi\right)\nabla^m\phi+\frac{2\pi}3\epsilonLC_{imj}\nabla^m\left(\nabbar^\alpha\phi\nabbar_\alpha\phi\right)
\end{align*}

\change{\noindent Note that expressions containing $\nabbar^\alpha\phi\nabbar_\alpha\phi$ can be ignored throughout our analysis since they are either pure trace or antisymmetric and thus will cancel in inner products with $E$, $B$, $\hat{k}$ and their rescaled analogues.\\}
\noindent The Bel-Robinson variables then behave as follows:

\begin{lemma}[Constraint and evolution equations for Bel-Robinson variables]\label{lem:EEqBR} If $(\g,\phi)$ is a classical solution to the Einstein scalar-field system \eqref{eq:ESF1}-\eqref{eq:EMT} in CMC gauge (see \eqref{eq:CMC}), $E$ and $B$ satisfy the following constraint equations:
\begin{subequations}
\begin{align}
E&\,=\Ric[g]+\frac29\tau^2g+\frac\tau3\hat{k}-\hat{k}\odot_g\hat{k}-4\pi(\nabla\phi\otimes\nabla\phi)-\left(\frac{8\pi}3\left\lvert \del_0\phi\right\rvert^2+\frac{4\pi}3\lvert\nabla\phi\rvert_g^2\right)g\label{eq:constr-E}\\
\label{eq:constr-B} B&\,=-\curl\hat{k}
\end{align}
\end{subequations}
Further, they satisfy the following evolution equations:
\begin{subequations}
\begin{align}
\del_tE_{ij}&=n\curl{B}_{ij}-(\nabla n\wedge B)_{ij}-\frac52n\left(E\times k\right)_{ij}-\frac23n(E\cdot k)g_{ij}-\frac{\tau}2n\cdot E_{ij}-\change{\frac{n}2\left(J_{i0j}+J_{j0i}\right)} \label{eq:EEqE}\\
\del_tB_{ij}&=-n\curl{E}_{ij}+(\nabla n\wedge E)_{ij}-\frac52n\left(B\times k\right)_{ij}-\frac23n(B\cdot k)g_{ij}-\frac{\tau}2n\cdot B_{ij}-\change{\frac{n}2\left(J_{i0j}^\ast+J_{j0i}^\ast\right)}\label{eq:EEqB}
\end{align}
\end{subequations}
\end{lemma}
\begin{proof}
For \eqref{eq:EEqE}-\eqref{eq:EEqB}, we refer to \cite[(3.11a)-(3.11b)]{AM03}.\change{\footnote{\change{Note that there is a minor error in the statement in \cite{AM03}, where the authors forget to symmetrize the $J$-tensors when applying the symmetrization to (3.14) of that work. This error seems to also occur in \cite[(7.2.2jk)]{ChrKl93}.}} } Equations \eqref{eq:constr-E}-\eqref{eq:constr-B} follow as in \cite[(3.63a)-(3.63b)]{Wang19} from contracting the Gauss-Codazzi constraints.
\end{proof}

\begin{remark}[Initial data for Bel-Robinson variables]\label{rem:init-BR}
Since the Weyl tensor vanishes over FLRW spacetimes, so do $E(W[\g_{FLRW}])$ and $B(W[\g_{FLRW}])$. Furthermore, note that given initial data $(M,\mathring{g},\mathring{k},\mathring{\pi},\mathring{\psi})$ on $\Sigma_{t_0}$ in the sense \change{discussed in Section \ref{subsubsec:initial-data}}, and defining $\hat{\mathring{k}}=\mathring{k}-\frac{\tau}3\mathring{g}$, we can use \eqref{eq:constr-E} and \eqref{eq:constr-B} to define the following $(0,2)$-tensors:
\begin{subequations}
\begin{align}
\label{eq:init-data-E}\mathring{E}=&\,\Ric[\mathring{g}]+\frac29\tau^2\mathring{g}+\frac{\tau}3\hat{\mathring{k}}-\hat{\mathring{k}}\odot_{\mathring{g}}\hat{\mathring{k}}-4\pi\left(\mathring{\pi}\otimes\mathring{\pi}\right)-\left(\frac{8\pi}3\left\lvert\mathring{\psi}\right\rvert_{\mathring{g}}^2+\frac{4\pi}3\left\lvert\mathring{\pi}\right\rvert_{\mathring{g}}^2\right)\mathring{g}\\
\label{eq:init-data-B}\mathring{B}=&\,-\curl_{\mathring{g}}\hat{\mathring{k}}
\end{align}
\end{subequations}
These are easily seen to be symmetric, and the constraints \eqref{eq:init-Hamilton} and \eqref{eq:init-momentum} on the initial data ensure that they are also tracefree. Hence, any choice of initial data for the Cauchy problem immediately also contains a unique choice of initial data for the Bel-Robinson variables that is consistent with solutions to the Einstein scalar-field equations in CMC gauge. 
\end{remark}

\subsection{Rescaled variables and equations}\label{subsec:REEq}

It will be more convenient to work the rescaled and shifted solution variables to measure their distance from the FLRW reference solution. In this subsection, we introduce the renormalized solution variables and restate the Einstein scalar-field system in terms of these variables.

\begin{definition}[Rescaled variables for Big Bang stability]\label{def:rescaled}
We will consider the rescaled variables
\begin{subequations}
\begin{gather}
\label{eq:rescalingGK}
G_{ij}=a^{-2}g_{ij},\quad (G^{-1})^{ij}=a^2g^{ij},\quad \Sigma_{ij}=a\hat{k}_{ij}\\
\label{eq:rescalingL} N=n-1\\
\label{eq:rescalingBR}
\RE_{ij}=a^4\cdot E_{ij},\quad \RB_{ij}=a^4\cdot B_{ij}\\
\label{eq:rescalingMatter}
\Psi=a^3\del_0\phi-C,
\end{gather}
\end{subequations}
\end{definition}
We note that the scaling of $\RB$ in \eqref{eq:rescalingBR} is \textbf{not} the asymptotic rescaling of $B$ -- in fact, we expect $B$ to have (approximate) leading order $a^{-2}$ as one can see in \eqref{eq:APmidB}. However, keeping this scaling parallel to that of $\RE$ makes the structurally very similar evolution equations significantly easier to deal with. We also do not rescale $N$ asymptotically\change{, unlike \cite{Rodnianski2014, Speck2018}, but note that $N$ converges to $0$ at an order slightly below $a^4$ at low orders (see \eqref{eq:BsN} and \eqref{eq:asymp-lapse}).}
\begin{prop}[The rescaled Einstein scalar-field system]\label{prop:REEq}
The Einstein scalar-field system in CMC gauge as in Proposition \ref{prop:eq} are solved by  $(g,\hat{k},n,\nabla\phi,\del_0\phi)$ if the rescaled variables\linebreak $(G,\Sigma,N,\nabla\phi,\Psi,\RE,\RB)$ as in Definition \ref{def:rescaled} solve the following set of equations:\footnote{\change{We refer to Lemma \ref{lem:BelRobinsonLemmas} for the scalings that occur when switching between tensor field operations like $\wedge$ and $\wedge_G$.}} The \textbf{rescaled metric evolution equations}
\begin{subequations}
\begin{align*}
\del_tG_{ij}=&-2(N+1)a^{-3}\Sigma_{ij}+2N\frac{\dot{a}}{a}G_{ij} \numberthis\label{eq:REEqG}\\
\del_t(G^{-1})^{ij}=&\,2(N+1)a^{-3}(\Sigma^\sharp)^{ij}-2N\frac{\dot{a}}{a}(G^{-1})^{ij}\numberthis\label{eq:REEqG-1}\\
\del_t\Sigma_{ij}=&\,-a\nabla_i\nabla_jN+(N+1)\left[a\Ric[G]_{ij}-2a^{-3}(\Sigma\odot_G\Sigma)_{ij}-8\pi a\nabla_i\phi\nabla_j\phi\right]\numberthis\label{eq:REEqSigma}\\
&+4\pi C^2a^{-3}\cdot NG_{ij}+\frac19\left(3N+2\right)aG_{ij}+N\frac{\dot{a}}{a}\Sigma_{ij}\\
\del_t{(\Sigma^\sharp)^a}_b=&\,\tau N{(\Sigma^\sharp)^a}_b-a\nabla^{\sharp a}\nabla_bN+(N+1)a\left[\left(\Ric[G]^\sharp\right)^a_b+\frac 29\I^a_b\right]\,\numberthis\label{eq:REEqSigmaSharp}\\
&-8\pi (N+1)a\nabla^{\sharp a}\phi\nabla_b\phi+N\left(4\pi C^2a^{-3}+\frac19 a\right)\I^a_b\,,
\end{align*}
\end{subequations}
the \textbf{rescaled Hamiltonian and momentum constraints}
\begin{subequations}
\begin{align}
R[G]+\frac23-a^{-4}\langle\Sigma,\Sigma\rangle_G&=8\pi\left[a^{-4}\Psi^2+2Ca^{-4}\Psi+\lvert\nabla\phi\rvert_G^2\right] \label{eq:REEqHam}\\
\nabla^{\sharp m}\Sigma_{ml}&=-8\pi\nabla_l\phi(\Psi+C) \label{eq:REEqMom}
\end{align}
with their Bel-Robinson analogues
\begin{align*}
\RE_{ij}=&\,a^4\left(\Ric[G]_{ij}+\frac29G_{ij}\right)+\frac{\tau}3a^3\Sigma_{ij}-(\Sigma\odot_G\Sigma)_{ij}-4\pi a^4\nabla_i\phi\nabla_j\phi\numberthis\label{eq:REEqConstrE}\\
&\,-\left[\frac{4\pi}3a^4\lvert\nabla\phi\rvert_G^2+\frac{8\pi}3\Psi^2+\frac{16\pi}3C\Psi\right]G_{ij}\\
\RB_{ij}=&\,-\change{a^2\curl_G\Sigma_{ij}}\numberthis\label{eq:REEqConstrB}\,,
\end{align*}
\end{subequations}
the \textbf{rescaled lapse equation}
\begin{subequations}
\begin{equation}\label{eq:REEqLapse1}
\Lap N=\left(12\pi C^2a^{-4}+\change{\frac13}\right)N+(N+1)a^{-4}\left[\langle\Sigma,\Sigma\rangle_G+8\pi\Psi^2+16\pi C\Psi\right]
\end{equation}
or equivalently
\begin{equation}\label{eq:REEqLapse2}
\Lap N=\left(12\pi C^2a^{-4}+\change{\frac13}\right)N+(N+1)\left[R[G]+\frac23-8\pi\lvert\nabla\phi\rvert_G^2\right]\,,
\end{equation}
\end{subequations}
the \textbf{rescaled evolution equations for the Bel-Robinson \changefinal{variables}}
\begin{subequations}
\begin{align*}
\numberthis\label{eq:REEqE}\del_t\RE_{ij}=&\,\left(3-N\right)\frac{\dot{a}}a\RE_{ij}\change{-a^{-1}(\nabla N\wedge_G\RB)_{ij}+(N+1)a^{-1}\curl_G\RB_{ij}}\\
&\,-(N+1)\left[\change{\frac52a^{-3}(\RE\times_G\Sigma)_{ij}}+\frac23a^{-3}\langle\RE,\Sigma\rangle_GG_{ij}\right]\\
&\,+4\pi(N+1) a^{-3}(\Psi+C)^2\Sigma_{ij}-4\pi(N+1) \dot{a}a^3\nabla_i\phi\nabla_j\phi\change{+4\pi a\nabla_{(i}N\nabla_{j)}\phi(\Psi+C)}\\
&\,-4\pi a(N+1)\left[\nabla_i\Psi\nabla_j\phi+\nabla_j\Psi\nabla_i\phi+\change{(\Sigma^\sharp)^l_{(i}\nabla_{j)}\phi}\nabla_l\phi-(\Psi+C)\nabla_i\nabla_j\phi\right]\\
&\,+(N+1)\left[\frac{2\pi}3a^6\del_0\left(a^{-6}(\Psi+C)^2+a^{-2}\lvert\nabla\phi\rvert_G^2\right)+4\pi\frac{\dot{a}}a(\Psi+C)^2\right]G_{ij}\\[0.5em]
\numberthis\label{eq:REEqB}\del_t\RB_{ij}=&\,\frac{\dot{a}}a\left(3-N\right)\RB_{ij}\change{+a^{-1}(\nabla N\wedge_G\RE)_{ij}-(N+1)a^{-1}\curl_G\RE_{ij}}\\
&\,-(N+1)\left[\change{\frac52a^{-3}(\RB\times_G\Sigma)_{ij}+\frac23a^{-3}\langle\RB,\Sigma\rangle_GG_{ij}}\deletemath{+\curl_G\RE_{ij}}\right]\,\\
&\,-4\pi(N+1)\epsilonLC[G]_{lmj}\left(a^3\nabla^{\sharp l}\nabla_{\change{j}}\phi\nabla^{\sharp m}\phi+a^{-1}{(\Sigma^\sharp)^l}_i\nabla^{\sharp m}\phi(\Psi+C)\right)
\end{align*}
\end{subequations}
and the \textbf{rescaled wave equation}
\begin{subequations}
\begin{equation}\label{eq:REEqWave}
\del_t\Psi=a\langle\nabla N,\nabla\phi\rangle_G+a(N+1)\Lap\phi-3\frac{\dot{a}}{a}N(\Psi+C)
\end{equation}
along with the evolution equation
\begin{equation}\label{eq:REEqNablaPhi}
\del_t\nabla\phi=a^{-3}(N+1)\nabla\Psi+a^{-3}(\Psi+C)\nabla N\,.
\end{equation}
\end{subequations}
Finally, we collect the \textbf{rescaled Ricci evolution equation}
\begin{align*}
\numberthis\label{eq:REEqRic}\del_t\Ric[G]_{ab}
=&\,a^{-3}(N+1)(\Lap_G\Sigma_{ab}-\nabla^{\sharp d}\nabla_a\Sigma_{db}-\nabla^{\sharp d}\nabla_b\Sigma_{da})\\
&\,+a^{-3}\nabla^{\sharp d}N(2\nabla_d\Sigma_{ab}-\nabla_a\Sigma_{db}-\nabla_b\Sigma_{da})\\
&\,-a^{-3}\left(\nabla_a N(\div_G\Sigma)_b+\nabla_b(\div_G\Sigma)_a\right)
+\Lap_GN(a^{-3}\Sigma_{ab}+\frac{\tau}3G_{ab})\\
&\,\change{-}a^{-3}\Bigr(\nabla^{\sharp d}\nabla_a N\cdot \Sigma_{db}+\nabla^{\sharp d}\nabla_b N\cdot\Sigma_{da}\Bigr)+\frac\tau3\nabla_a\nabla_bN
\end{align*}
as well as (in a coordinate neighbourhood) the Christoffel evolution equation
\begin{align*}
\del_t\Gamma_{ij}^k[G]=&\,\frac12 (G^{-1})^{kl}\left(\nabla_i(\del_tG_{jl})+\nabla_j(\del_tG_{il})-\nabla_l(\del_tG_{ij})\right) \numberthis\label{eq:REEqChr}\\
=&\,-(N+1)a^{-3}\left[\nabla_i{(\Sigma^\sharp)^k}_j+\nabla_j{(\Sigma^\sharp)^k}_i-\nabla^{\sharp k}\Sigma_{ij}\right]\\
&\,-a^{-3}\left[\nabla_iN{(\Sigma^\sharp)^k}_j+\nabla_jN{(\Sigma^\sharp)^k}_i-\nabla^{\sharp k}N\Sigma_{ij}\right]\\
&\,+\frac{\dot{a}}a\left[\nabla_i N\cdot\I^k_j+\nabla_jN\cdot\I^k_i-\nabla^{\sharp k}N\cdot G_{ij}\right]
\end{align*}
\end{prop}

\begin{proof}
For the first identity in \eqref{eq:REEqChr}, we refer to \cite[Lemma 2.27]{Chow06} and insert the evolution equation \eqref{eq:REEqG}. Otherwise, all equations simply follow by computing the effects of rescaling on the equations from Proposition \ref{prop:eq} (respectively the Ricci evolution equation as in \cite[Chapter 2.3, (2.32)]{Rendall08}) as well as the Bel-Robinson evolution equations \eqref{eq:EEqE}-\eqref{eq:EEqB} and constraint equations \eqref{eq:constr-E}-\eqref{eq:constr-B}. Notice that one already finds a solution to the system in Proposition \eqref{prop:eq} with the rescaled variables excluding \eqref{eq:REEqE}, \eqref{eq:REEqB}, \eqref{eq:REEqConstrE} and \eqref{eq:REEqConstrB}. Conversely, all of the rescaled equations are satisfied by a solution to Proposition \ref{prop:eq} at sufficiently high regularity. Hence, solving the full rescaled system is always sufficient to solve the Einstein system in Proposition \ref{prop:eq} and they are equivalent if the initial data is regular enough to ensure sufficiently high regularity of solutions.\\
\end{proof}

\subsection{Commuted equations}\label{subsec:comm-eq}

We collect Laplace-commuted versions of the equations for the rescaled variables in Proposition \ref{prop:REEq} in this subsection. For the sake of brevity, we will not state all possible commutations for every equation, but restrict ourselves to the ones we actually need within the bootstrap argument. We also refer to Appendix \ref{subsec:commutators} for expressions for commutators of spatial differential operators with each other and with $\del_t$.

The terms written down explicity in Lemma \ref{lem:laplace-commuted-eq} are ones that dominate the evolution behaviour or that are the largest higher order terms, both of which require careful treatment within the bootstrap argument. The error terms are broadly categorized into three groups:
\begin{itemize}
\item \enquote{Borderline} terms are terms that critically contribute to the fact that the energies diverge toward the Big Bang singularity. This almost always takes the form of adding energy terms at the same order as the evolved variable scaled by factors of the type $\epsilon a^{-3}$ or $\epsilon a^{-3-c\sqrt{\epsilon}}$, which causes the energies to slightly diverge since $a^{-3}$ is barely not integrable (see \eqref{eq:a-integrals}).
\item \enquote{Junk} terms are terms that are subcritical in the sense that they lead to integrable error terms, or terms that only contain lower order derivatives of the solution variables.
\item \enquote{Top} order terms (which only appear in \eqref{eq:comeq-RE} and \eqref{eq:comeq-RB}) are terms that are junk terms for low order energies, but become borderline terms at top order.
\end{itemize}
All of these error terms are tracked schematically in Section \ref{subsec:error-terms}. Since we will only need $L^2_G$-bounds on these error terms, which are given in Section \ref{subsec:L2-error-est}, we will treat them as notational \enquote{black boxes} outside of the appendix.

\begin{lemma}[Laplace-commuted rescaled equations]\label{lem:laplace-commuted-eq}
Let $L\in2\N,\,L\geq 2$. With error terms as defined in Appendix \ref{subsec:error-terms}, the system in Proposition \ref{prop:REEq} leads to the following Laplace-commuted equations:\\
The \textbf{Laplace-commuted rescaled evolution equations for the second fundamental form}
\begin{equation}\label{eq:comeq-Sigma}
\del_t\Lap^\frac{L}2\Sigma=-a\nabla^2\Lap^\frac{L}2N+a(N+1)\Lap^\frac{L}2\Ric[G]+\mathfrak{S}_{L,Border}+\mathfrak{S}_{L,Junk}\,,
\end{equation}
the \textbf{Laplace-commuted rescaled momentum constraint equations}
\begin{subequations}
\begin{equation}
\div_G\Lap^\frac{L}2\Sigma=-8\pi (\Psi+C)\left[\nabla\Lap^{\frac{L}2}\phi+\Lap^{\frac{L}2-1}\Ric[G]\ast\nabla\phi\right]+\nabla\Lap^{\frac{L}2-1}\Ric[G]\ast\Sigma+\mathfrak{M}_{L,Junk}\label{eq:comeq-mom-div}
\end{equation}
and
\begin{equation}
\change{\curl_G\Lap^{\frac{L}2}\Sigma=-a^{-2}\Lap^{\frac{L}2}\RB+\epsilonLC[G]\ast\nabla\Lap^{\frac{L}2-1}\Ric[G]\ast\Sigma+\tilde{\mathfrak{M}}_{L,Junk}}\numberthis\,,\label{eq:comeq-mom-curl}
\end{equation}
the \textbf{Laplace-commuted rescaled Hamiltonian constraint equations}
\begin{align*}
\numberthis\label{eq:comeq-Ham}&\,\Lap^{\frac{L}2}R[G]+a^{-4}\sum_{I_1+I_2=L}\nabla^{I_1}\Sigma\ast\nabla^{I_2}\Sigma\\
=&\,16\pi Ca^{-4}\Lap^{\frac{L}2}\Psi+a^{-4}\sum_{I_1+I_2=L}\left[\nabla^{I_1}\Psi\ast\nabla^{I_2}\Psi+\nabla^{I_1+1}\phi\ast\nabla^{I_2+1}\phi\right]
\end{align*}
and
\begin{equation}
\Lap^{\frac{L}2}\Ric[G]=a^{-4}\Lap^\frac{L}2\RE-\frac{\tau}3a^{-1}\Lap^{\frac{L}2}\Sigma+\mathfrak{H}_{L,Border}+\mathfrak{H}_{L,Junk}\,,\label{eq:comeq-Ham-BR}
\end{equation}
\end{subequations}
\begin{subequations}
the \textbf{Laplace-commuted rescaled lapse equations}
\begin{align}
\label{eq:comeq-lapse}\Lap^{\frac{L}2+1}N=&\,\left(12\pi C^2a^{-4}+\frac13\right)\Lap^{\frac{L}2}N+16\pi Ca^{-4}\cdot\Lap^{\frac{L}2}\Psi+\mathfrak{N}_{L,Border}+\mathfrak{N}_{L,Junk}\,,\\
\label{eq:comeq-lapse-odd}\nabla\Lap^{\frac{L}2+1}N=&\,\left(12\pi C^2a^{-4}+\frac13\right)\nabla\Lap^{\frac{L}2}N+16\pi Ca^{-4}\cdot\nabla\Lap^{\frac{L}2}\Psi+\mathfrak{N}_{L+1,Border}+\mathfrak{N}_{L+1,Junk}\,,
\end{align}
\end{subequations}
the \textbf{Laplace-commuted rescaled Bel-Robinson evolution equations}
\begin{subequations}
\begin{align*}
\del_t\Lap^{\frac{L}2}\RE=&\,\frac{\dot{a}}a\left(3-N\right)\Lap^{\frac{L}2}\RE\change{+(N+1)a^{-1}\curl_G\Lap^{\frac{L}2}\RB-a^{-1}\nabla\Lap^\frac{L}2 N\wedge_G \RB}\numberthis\label{eq:comeq-RE}\\
&\,+4\pi C^2a^{-3}(N+1)\Lap^{\frac{L}2}\Sigma\change{+4\pi a(\Psi+C)\nabla\Lap^{\frac{L}2}N\otimes\nabla\phi}\\
&\,+4\pi a(\Psi+C)(N+1)\nabla^2\Lap^\frac{L}2\phi-8\pi a(N+1)\left(\nabla\phi\otimes\nabla\Lap^\frac{L}2\Psi\right)\\
&\,+\mathfrak{E}_{L,Border}+\mathfrak{E}_{L,top}+\mathfrak{E}_{L,Junk}\\
\del_t\Lap^{\frac{L}2}\RB=&\,\frac{\dot{a}}a(3-N)\Lap^{\frac{L}2}\RB\change{-(N+1)a^{-1}\curl_G\Lap^{\frac{L}2}\RE+a^{-1}\nabla\Lap^{\frac{L}2}N\wedge_G\RE} \numberthis\label{eq:comeq-RB}\\
&\,+a^3\epsilonLC[G]\ast\nabla^2\Lap^{\frac{L}2}\phi\ast\nabla\phi+\mathfrak{B}_{L,Border}+\mathfrak{B}_{L,top}+\mathfrak{B}_{L,Junk}\,,
\end{align*}
\end{subequations}
the \textbf{Laplace-commuted rescaled matter evolution equations}
\begin{subequations}
\begin{align}
\del_t\Lap^\frac{L}2\Psi=&\,a\langle\nabla\Lap^\frac{L}2N,\nabla\phi\rangle_G+a(N+1)\Lap^{\frac{L}2+1}\phi-3C\frac{\dot{a}}a\Lap^{\frac{L}2}N+\mathfrak{P}_{L,Border}+\mathfrak{P}_{L,Junk} \label{eq:comeq-Psi-even}\\
\del_t\nabla\Lap^\frac{L}2\phi=&\,a^{-3}(N+1)\nabla\Lap^{\frac{L}2}\Psi+Ca^{-3}\nabla\Lap^\frac{L}2 N+\mathfrak{Q}_{L,Border}+\mathfrak{Q}_{L,Junk} \label{eq:comeq-nablaphi-even}
\end{align}
as well as (also allowing $L=0$ for \eqref{eq:comeq-nablaphi-odd})
\begin{align*}
\del_t\nabla_l\Lap^{\frac{L}2}\Psi=&\,a\nabla_l\nabla^{\sharp j}\Lap^{\frac{L}2}N\nabla_j\phi+a(N+1)\nabla_l\Lap^{\frac{L}2+1}\phi-3C\frac{\dot{a}}a\nabla_l\Lap^{\frac{L}2}N \numberthis\label{eq:comeq-Psi-odd}\\
&\,+\left(\mathfrak{P}_{L+1,Border}\right)_l+\left(\mathfrak{P}_{L+1,Junk}\right)_l \\
\numberthis\label{eq:comeq-nablaphi-odd}\del_t\Lap^{\frac{L}2+1}\phi=&\,a^{-3}(N+1)\Lap^{\frac{L}2+1}\Psi+Ca^{-3}\Lap^{\frac{L}2+1}N +\mathfrak{Q}_{L+1,Border}+\mathfrak{Q}_{L+1,Junk}
\end{align*}
\end{subequations}
and the \textbf{Laplace-commuted rescaled Ricci evolution equations}
\begin{subequations}
\begin{align*}
\numberthis\label{eq:comeq-Ric-even}\del_t\Lap^\frac{L}2\Ric[G]_{ij}=&\,a^{-3}\left(\Lap^{\frac{L}2+1}\Sigma_{ij}-2\nabla^{\sharp m}\nabla_{(i}\Lap^{\frac{L}2}\Sigma_{j)m}\right)\\
&\,-\frac{\dot{a}}a\left(\nabla_{i}\nabla_j\Lap^\frac{L}2N+\Lap^{\frac{L}2+1}N\cdot G_{ij}\right)+\left(\mathfrak{R}_{L,Border}\right)_{ij}+\left(\mathfrak{R}_{L,Junk}\right)_{ij}\\
\numberthis\label{eq:comeq-Ric-odd}\del_t\nabla_k\Lap^\frac{L}2\Ric[G]_{ij}=&\,a^{-3}\left(\nabla_k\Lap^{\frac{L}2+1}\Sigma_{ij}-2\nabla_k\nabla^{\sharp m}\nabla_{(i}\Lap^{\frac{L}2}\Sigma_{j)m}\right)\\
&\,-\frac{\dot{a}}a\left(\nabla_k\nabla_i\nabla_j\Lap^\frac{L}2N+\nabla_k\Lap^{\frac{L}2+1}N\cdot G_{ij}\right)\\
&\,+\left(\mathfrak{R}_{L+1,Border}\right)_{ijk}+\left(\mathfrak{R}_{L+1,Junk}\right)_{ijk}\,.
\end{align*}
\end{subequations}
\end{lemma}
\begin{proof}
The equations \eqref{eq:comeq-Ham},\eqref{eq:comeq-Ham-BR} and \eqref{eq:comeq-lapse} are obtained by simply applying $\Lap^\frac{L}2$ on both sides of \eqref{eq:REEqHam},\eqref{eq:REEqConstrE} and \eqref{eq:REEqLapse1} respectively. For \eqref{eq:comeq-mom-div} and \eqref{eq:comeq-mom-curl}, we additionally use the commutator formulas \eqref{eq:[Lap-l,div]T} and \eqref{eq:[Lap-l,curl]}, while for the evolution equations, we apply the respective commutator of $\del_t$ and spatial derivatives as collected in Lemma \ref{lem:com-time} and commute Laplacians past $\nabla$ and $\curl$ where needed (see the commutators in Lemma \ref{lem:comm-space}). The commutators with $\del_t$ only cause additional borderline and junk terms that do not substantially influence the behaviour, while the spatial commutators often lead to high order curvature terms, for example the Ricci terms in \eqref{eq:comeq-mom-div}, that need to be more carefully tracked.
\end{proof}

\begin{remark}[Simplified junk term notation]\label{rem:notation-parallel}
For junk terms that occur in an inner product with a tracefree symmetric tensor, any terms that are \change{pure trace }will immediately cancel and thus do not need to be taken into consideration for the following estimates, even if they have to be written down in the junk terms. Hence, we will denote with a superscript \enquote{$\parallel$} (for example $\mathfrak{H}_{L,Junk}^\parallel$) on a schematic error term the expressions that arise when dropping all terms of the form $\zeta\cdot G$ for some scalar function $\zeta$ \change{that }occur in this term's definition (see, for example, \eqref{eq:comeq-Ham-junk}).
\end{remark}

\section{Big Bang stability: Norms, energies and bootstrap assumptions}\label{sec:norm-en-bs}

Herein, we state the norms and energies we use to control the solution variables. These will allow us to state our initial data and bootstrap assumptions, and we then provide which improvement we aim to achieve for the latter. Note that we do not provide the coerciveness of our energies immediately (and actually cannot, at least not in a manner useful to our analysis), but will establish Sobolev norm control in the proof of Corollary \ref{cor:H-imp}, the key ingredient being Lemma \ref{lem:Sobolev-norm-equivalence-improved}. Furthermore, we collect a local well-posedness statement from previous work in Section \ref{subsec:lwp}.

\subsection{Norms}\label{subsec:norm}

Recall that $\gamma$ is the hyperbolic spatial reference metric on $M$ introduced in Definition \ref{def:spatial-mf}, which we view as a metric on any foliation hypersurface $\Sigma_t$ (see Section \ref{subsubsec:initial-data}), and $G$ is the rescaled spatial metric arising from the evolution (see Definition \ref{def:rescaled}).

\begin{definition}[Pointwise norms and volume forms] \label{def:pw-stuff}
We denote by $\lvert\cdot\rvert_\gamma$ (resp. $\lvert\cdot\rvert_{G(t,\cdot)}$) the pointwise norm with regard to $\gamma$ (resp. $G(t,\cdot)$). For the sake of simplicity, we define $\lvert\zeta\rvert_\gamma=\lvert\zeta\rvert_{G(t,\cdot)}=\lvert\zeta(t,\cdot)\rvert$ for any scalar function $\zeta$ on $\Sigma_t$. The volume forms on $\Sigma_t$ with respect to $\gamma$ and $G(t,\cdot)$ are written as $\vol{\gamma}$ and $\vol{G(t,\cdot)}$ (or just $\vol{G}$).
\end{definition}

\begin{definition}[$L^2$-norms] Let $\mathfrak{T}$ be a $\Sigma_t$-tangent $(r,s)$-tensor field (for $r,s\geq 0$). Then, we define:
\begin{subequations}
\begin{align}
\|\mathfrak{T}\|_{L^2_\gamma(\Sigma_t)}^2=\|\mathfrak{T}(t,\cdot)\|_{L^2_\gamma(\Sigma_t)}^2&\,:=\int_M \lvert\mathfrak{T}(t,\cdot)\rvert_\gamma^2\,\vol{\gamma},\\
\|\mathfrak{T}\|_{L^2_G(\Sigma_t)}^2=\|\mathfrak{T}(t,\cdot)\|_{L^2_G(\Sigma_t)}^2&\,:=\int_M\lvert\mathfrak{T}(t,\cdot)\rvert_{G(t,\cdot)}^2\,\vol{G(t,\cdot)}\,
\end{align} 
\end{subequations}
\end{definition}

\begin{definition}[Sobolev norms]\label{def:sob-norms} Let $\mathfrak{T}$ be as above and $J\in\N_0$. We define:
\begin{subequations}
\begin{align}
\|\mathfrak{T}\|_{\dot{H}^J_\gamma(\Sigma_t)}^2=\|\mathfrak{T}(t,\cdot)\|_{\dot{H}^J_\gamma}^2=&\,\int_{\Sigma_t}\left\lvert\nabhat^J\mathfrak{T}(t,\cdot)\right\rvert_\gamma^2\,\vol{\gamma}\\
\|\mathfrak{T}\|_{\dot{H}^J_G(\Sigma_t)}^2=\|\mathfrak{T}(t,\cdot)\|_{\dot{H}^J_G}^2=&\,\int_{\Sigma_t}\left\lvert\nabla^J\mathfrak{T}(t,\cdot)\right\rvert_{G(t,\cdot)}^2\,\vol{G(t,\cdot)}\\
\|\mathfrak{T}\|_{H^J_\gamma(\Sigma_t)}^2=\|\mathfrak{T}(t,\cdot)\|_{H^J_\gamma}^2=&\,\sum_{k=0}^J\|\mathfrak{T}\|_{\dot{H}^k_\gamma(\Sigma_t)}^2\\
\|\mathfrak{T}\|_{H^J_G(\Sigma_t)}^2=\|\mathfrak{T}(t,\cdot)\|_{H^J_G}^2=&\,\sum_{k=0}^J\|\mathfrak{T}\|_{\dot{H}^k_G(\Sigma_t)}^2
\end{align}
\end{subequations}
\end{definition}

\begin{definition}[Supremum norms]\label{def:sup-norms} For $\mathfrak{T}$ as above and $J\in\N_0$, we set:
\begin{subequations}
\begin{align}
\|\mathfrak{T}\|_{\dot{C}^J_\gamma(\Sigma_t)}=\|\mathfrak{T}(t,\cdot)\|_{\dot{C}^J_\gamma}=&\,\sup_{p\in \Sigma_t}\left\lvert\nabhat^J\mathfrak{T}(t,\cdot)\right\rvert_\gamma,& \|\mathfrak{T}\|_{C^J_\gamma(\Sigma_t)}=\sum_{k=0}^J\|\mathfrak{T}\|_{\dot{C}^k_\gamma(\Sigma_t)}\\
\|\mathfrak{T}\|_{\dot{C}^J_G(\Sigma_t)}=\|\mathfrak{T}(t,\cdot)\|_{\dot{C}^J_G}=&\,\sup_{p\in \Sigma_t}\left\lvert\nabla^J\mathfrak{T}(t,\cdot)\right\rvert_{G(t,\cdot)},& \|\mathfrak{T}\|_{C^J_G(\Sigma_t)}=\sum_{k=0}^J\|\mathfrak{T}\|_{\dot{C}^k_G(\Sigma_t)}
\end{align}
\end{subequations}
\end{definition}

\begin{remark}[Time dependence is \change{suppressed }in notation]
When the choice of $t$ and $\Sigma_t$ is clear from context, we will often drop time dependences of $G$, $\lvert\cdot\rvert_G$, $\vol{G}$ and $\mathfrak{T}$, suppress the hypersurface $\Sigma_t$ in the Sobolev and supremum norms, and simply write $\int_M$ instead of $\int_{\Sigma_t}$. For example, we write
\[\|\mathfrak{T}\|_{L^2_G}^2=\int_M\lvert\mathfrak{T}\rvert_G^2\,\vol{G}\,.\]
\end{remark}

\begin{definition}[Solution norms]\label{def:sol-norm} We define the following norms to measure the size of near-FLRW solutions:
\begin{subequations}
\begin{align*}
\numberthis\label{eq:def-H}\mathcal{H}=&\,\|\Psi\|_{\change{H^{18}_G}}+\|\nabla\phi\|_{\change{H^{17}_G}}+a^2\|\nabla\phi\|_{\change{\dot{H}^{18}_G}}\\
&\,+\|\Sigma\|_{H^{18}_G}+\|\RE\|_{H^{18}_G}+\|\RB\|_{H^{18}_G}\\
&\,+\|G-\gamma\|_{H^{18}_G}+\|\Ric[G]+\frac29G\|_{H^{16}_G}+a^{-4}\|N\|_{H^{16}_G}+a^{-2}\|N\|_{\dot{H}^{17}_G}+\|N\|_{\dot{H}^{18}_G}\\
\numberthis\label{eq:def-H-top}\change{\mathcal{H}_{top}=}&\,\change{a^2\|\Psi\|_{{\dot{H}^{19}_G}}}\change{+a^4\|\nabla\phi\|_{\change{\dot{H}^{19}_G}}}\change{+a^2\|\Sigma\|_{\dot{H}^{19}_G}}\change{+a^2\|\Ric[G]+\frac29G\|_{\dot{H}^{17}_G}}\,\\
\mathcal{C}=&\,\|\Psi\|_{\change{C^{16}_G}}+\|\nabla\phi\|_{\change{C^{15}_G}}+\|\Sigma\|_{C^{16}_G}+\|\RE\|_{C^{16}_G}+\|\RB\|_{C^{16}_G} \numberthis\label{eq:def-C}\\
&\,+\|G-\gamma\|_{C^{16}_G}
+\|\Ric[G]+\frac29G\|_{C^{14}_G}+a^{-4}\|N\|_{C^{14}_G}+a^{-2}\|N\|_{\dot{C}^{15}_G}+\|N\|_{\dot{C}^{16}_G}\\
\mathcal{C}_\gamma=&\,\|\Psi\|_{\change{C^{16}_\gamma}}+\|\nabla\phi\|_{\change{C^{15}_\gamma}}+\|\Sigma\|_{C^{16}_\gamma}+\|\RE\|_{C^{16}_\gamma}+\|\RB\|_{C^{16}_\gamma} \numberthis\label{eq:def-C-gamma}\\
&\,+\|G-\gamma\|_{C^{16}_\gamma}
+\|\Ric[G]+\frac29G\|_{C^{14}_\gamma}+a^{-4}\|N\|_{C^{14}_\gamma}+a^{-2}\|N\|_{\dot{C}^{15}_\gamma}+\|N\|_{\dot{C}^{16}_\gamma}
\end{align*}
\end{subequations}
\end{definition}

\begin{remark}[Choice of metric and controlling Christoffel symbols]\label{rem:relation-metric-Chr}
We could equivalently also phrase $\mathcal{H}$ in terms of $\gamma$-norms, or predominately use $\mathcal{C}_\gamma$ instead of $\mathcal{C}$, since we include the norms on $G-\gamma$ and $\Ric[G]+\frac29G=(\Ric[G]+\frac29\gamma)+\frac29(G-\gamma)$. We will demonstrate in Lemma \ref{lem:G-gamma-norm-switch} how $H_G$ and $C_\gamma$ norms can be used to control $H_\gamma$ and $C_G$ norms. We will also indicate how the initial data and bootstrap assumptions for $\mathcal{C}_\gamma$ and $\mathcal{C}$ are equivalent in Remarks \ref{rem:init-Cgamma} and \ref{rem:Bs-Cgamma}. The main reason for this is that, by successively replacing local coordinates in the expressions of $\Gamma-\Gamhat$ by $\nabhat$, one has
\begin{equation}\label{eq:Christoffel-norm-handwaving}
\|\Gamma-\Gamhat\|_{C^{l-1}_\gamma\change{(M)}}\lesssim P_l(\|G-\gamma\|_{C^l_\gamma(M)},\|G^{-1}-\gamma^{-1}\|_{C^l_\gamma(M)})\,.
\end{equation}

We choose to work predominately with norms in terms of the rescaled metric since quantities appearing in the Einstein system are naturally contracted by $G$, not $\gamma$, and we commute with differential operators associated with $G$.

\end{remark}

\begin{remark}[Redundancies in the solution norms]\label{rem:redundancy}
The solution norms $\mathcal{H}$, $\mathcal{C}$ and $\mathcal{C}_\gamma$ aren't \enquote{optimal} in the sense that controlling the norms of $\Psi,\,\nabla\phi,\Sigma$ and $G-\gamma$ is entirely sufficient to gain the claimed control (up to constant) on $N$ via the lapse equation, $\RE$ and $\RB$ via to the constraint equations and $\Ric[G]+\frac29G$ via local coordinates. We choose to include all variables in the norms and subsequent assumptions mainly for the sake of convenience.
\end{remark}

\subsection{Energies}\label{subsec:en}

The fundamental objects used to control the solution variables are the energies that take the following form:

\begin{definition}[Energies]\label{def:energies}
Let $l\in\N_0$. We define:
\begin{subequations}
\begin{align*}
\numberthis\changefinal{\E^{(l)}(\phi,t)=}&\,(-1)^l\int_M\Psi\Lap^l\Psi-a^4\phi\Lap^{l+1}\phi\,\vol{G}\\
\label{eq:energydef-ibp}=&\,\begin{cases}
\int_M\lvert\Lap^\frac{l}2\Psi\rvert^2+a^4\lvert\nabla\Lap^\frac{l}2\phi\rvert_G^2\,\vol{G} & l\text{ even}\\
\int_M\lvert\nabla\Lap^\frac{l-1}2\Psi\rvert_G^2+a^4\lvert\Lap^\frac{l+1}2\phi\rvert_G^2\,\vol{G} & l\text{ odd}
\end{cases}\\
\numberthis\E^{(l)}(W,t)=&\,(-1)^l\int_M \langle\RE,\Lap^l \RE\rangle_G +\langle\RB,\Lap^l\RB\rangle_G\,\vol{G}\\
\numberthis\E^{(l)}(\Sigma,t)=&\,(-1)^l\int_M\langle\Sigma,\Lap^l\Sigma\rangle_G\,\vol{G}\\
\numberthis\E^{(l)}(\Ric,\cdot)=&\,(-1)^l\int_M\left\langle \Ric[G]+\frac29G,\Lap^l\left(\Ric[G]+\frac29G\right)\right\rangle_G\,\vol{G}\\
\numberthis\E^{(l)}(N,\cdot)=&\,(-1)^l\int_M\langle N,\Lap^lN\rangle_G\,\vol{G}
\end{align*}
\end{subequations}
Usually, we will use integration by parts to distribute derivatives within the energies as in \eqref{eq:energydef-ibp}. Further, we introduce the notation
\begin{align}
\E^{(\leq l)}&=\sum_{i=0}^l\E^{(i)}
\end{align}
for any of the energies above.
\end{definition}
\noindent For any $l\in\N_0$ and any smooth functions $f_1,f_2\in C^\infty(\R_+)$, we have
\begin{equation}\label{eq:ibp-trick}
f_1\cdot f_2\cdot\E^{(2l+1)}\leq \frac{f_1^2}2\E^{(2l)}+\frac{f_2^2}2\E^{(2l+2)}\,.
\end{equation}
Performing the calculation for $\Sigma$ as an example, we have:
\begin{align*}
\E^{(2l+1)}(\Sigma,\cdot)&\,=-\int_M \langle\Lap^l\Sigma,\Lap^{l+1}\Sigma\rangle_G\,\vol{G}\leq \int_M\lvert\Lap^l\Sigma\rvert_G\lvert\Lap^{l+1}\Sigma\rvert_G\,\vol{G}\leq \sqrt{\E^{(2l)}(\Sigma,\cdot)}\sqrt{\E^{(2l+2)}(\Sigma,\cdot)}
\end{align*}
Now, \eqref{eq:ibp-trick} then follows from the Young inequality. As a consequence, we also have
\begin{equation}\label{eq:drop-odd-en}
\E^{(\leq 2l)}\lesssim \sum_{m=0}^l\E^{(2m)}\,,\,\E^{(\leq 2l+1)}\lesssim\sum_{m=0}^{l+1}\E^{(2m)}\,.
\end{equation}
This allows us to largely restrict our analysis to even order energies, outside of how we close the bootstrap argument at top order.

\subsection{Assumptions on the initial data}\label{subsec:init}

With the necessary solution norms and energies now defined, we can now state what we assume near-FLRW initial data to satisfy:

\begin{assumption}[Near-FLRW initial data]\label{ass:init}
For some small enough $\epsilon\in(0,1)$ and the solution norms $\mathcal{H}\change{, \mathcal{H}_{top}}$ and $\mathcal{C}$ as in Definition \ref{def:sol-norm}, we assume the rescaled initial data to be close to that of the FLRW solution in Lemma \ref{lem:FLRW} in the following sense:
\begin{equation}\label{eq:init-ass}
\mathcal{H}(t_0)+\mathcal{H}_{top}(t_0)+\mathcal{C}(t_0)\lesssim\epsilon^2
\end{equation}
\delete{Additionally [...]} 
\end{assumption}
\change{The assumptions }on $\mathcal{H}\change{+\mathcal{H}_{top}}$ also imply the following:
\begin{align*}
\numberthis\label{eq:init-ass-en}
\E^{(\leq \change{18})}(\phi,t_0)+\E^{(\leq 18)}(\Sigma,t_0)+\E^{(\leq 18)}(W,t_0)+\E^{(\leq 16)}(\Ric,t_0)&\\
+\|\nabla\phi\|_{H^{18}_G}^2+\E^{(\change{18})}(N,t_0)+a(t_0)^{-4}\E^{(\change{17})}(N,t_0)+a(t_0)^{-8}\E^{(\change{\leq 16})}(N,t_0)\\
\change{+a(t_0)^4\E^{(19)}(\phi,t_0)+a(t_0)^4\E^{(19)}(\Sigma,t_0)+a(t_0)^4\E^{(17)}(\Ric,t_0)}
&\,\lesssim\epsilon^4
\end{align*}

\begin{remark}[Initial data size in $\mathcal{C}_\gamma(t_0)$]\label{rem:init-Cgamma}
Notice that by \eqref{eq:Christoffel-norm-handwaving}, \eqref{eq:init-ass} implies \change{that}
\begin{subequations}
\begin{equation}\label{eq:init-ass-Chr-C}
\|\Gamma-\Gamhat\|_{C^{15}_G(\change{\Sigma_{t_0}})}\lesssim\epsilon^4\,,
\end{equation}
and arguing along similar lines and using $L^2-L^\infty$-estimates, also
\begin{equation}
\|\Gamma-\Gamhat\|_{H^{17}_G(\change{\Sigma_{t_0}})}\lesssim\epsilon^4\,.
\end{equation}
\end{subequations}
In particular, since moving from $C_\gamma^l$ to $C_G^l$ only requires control on Christoffel symbols to order $l-1$ for general tensors and $l-2$ for scalar functions, as well as zero order control on $G-\gamma$, it follows from \eqref{eq:init-ass} that
\begin{equation}\label{eq:init-ass-Cgamma}
\mathcal{C}_\gamma(t_0)\lesssim\epsilon^2
\end{equation}
We refer to the proof of Lemma \ref{lem:G-gamma-norm-switch} for a more detailed term analysis and how a similar argument also applies to the Sobolev norms.
\end{remark}

\begin{remark}[Redundancies in the initial data assumptions]
Similar to Remark \ref{rem:redundancy}, one could also reduce the \change{initial data }assumptions in \eqref{eq:init-ass}\change{, especially at top order}. In particular, we highlight that the Bel-Robinson energy can be entirely controlled by the other terms that occur due to the additional scaled $\Sigma$-energy at order \change{$19$}, or vice versa we could drop the latter in favour of the former. This will be reflected in Lemma \ref{lem:en-est-Sigma-top}.
\end{remark}

\begin{remark}[The volume form]\label{eq:rem-vol-form}
Let $\mu_G$ and $\mu_\gamma$ denote the volume elements of $G$ and $\gamma$ respectively. Since the determinant is a smooth map on invertible matrices, the initial data assumptions on $G-\gamma$ also imply
\begin{equation}\label{eq:init-vol-el}
\|\mu_{G}-\mu_{\gamma}\|_{C^0_G(\Sigma_{t_0})}=\|\mu_{G}-\mu_{\gamma}\|_{C^0_\gamma(\Sigma_{t_0})}\lesssim \epsilon^2\,.
\end{equation}
Consequently, we have
\[\|\vol{G}-\vol{\gamma}\|_{C^0_\gamma(\Sigma_{t_0})}=\mu_{\gamma}^{-1}\|\vol{\gamma}\|_{C^0_\gamma(\Sigma_{t_0})}\|\mu_{G(t_0,\cdot)}-\mu_{\gamma}\|_{C^0_\gamma(\Sigma_{t_0})}\lesssim\epsilon^2\]
and, since $\|G^{-1}-\gamma^{-1}\|_{C^0_\gamma(\Sigma_{t_0})}\lesssim\epsilon^2$ also follows by a von Neumann series argument from the initial data assumption on $G-\gamma$,
\begin{equation}\label{eq:init-vol-form}
\|\vol{G}-\vol{\gamma}\|_{C^0_G(\Sigma_{t_0})}\lesssim \epsilon^2\,.
\end{equation}
\end{remark}

\subsection{Local well-posedness and continuation criteria}\label{subsec:lwp}

For everything that follows, we need to establish that the initial data assumptions above also ensure local well-posedness. For the core system, we translate the local well-posedness result for stiff fluids in \cite{Rodnianski2014} to the sub-case of the scalar field system. While statement and proof there are for vanishing spatial sectional curvature and what corresponds to choosing $C=\sqrt{\nicefrac23}$, the arguments for our setting are completely analogous.

\begin{lemma}[Local well-posedness and continuation criteria for the Einstein scalar-field system (Big Bang version), see {\cite[Theorem 14.1]{Rodnianski2014}}]\label{lem:lwp}
Let $N\geq 4$ be an integer and $(M,\mathring{g},\mathring{k},\mathring{\pi},\mathring{\psi})$ be geometric initial data to the Einstein scalar-field system (see \change{Section \ref{subsubsec:initial-data}}) and assume that one has
\[\|\mathring{g}-a(t_0)^2\gamma\|_{H^{N+1}_\gamma(M)}+\|\mathring{k}+\frac{\tau(t_0)}3\cdot a(t_0)^2\gamma\|_{H^N_\gamma(M)}+\|\mathring{\pi}\|_{H^{N}_{\gamma}(M)}+\|\mathring{\psi}-Ca(t_0)^{-3}\|_{H^N_\gamma(M)}<\infty\,\]
as well as, for some sufficiently small $\eta^\prime>0$,
\[\|\mathring{\psi}-Ca(t_0)^{-3}\|_{C^0_\gamma(M)}\leq\eta^\prime\,.\]
Then, the CMC-transported Einstein scalar-field system (respectively the rescaled system) is locally well-posed in the following sense: The initial data $(\mathring{g},\mathring{k},\mathring{\pi},\mathring{\psi})$ launches a unique classical solution $(g,k,n,\nabla\phi,\del_t\phi)$ to \eqref{eq:Hamilton}-\eqref{eq:Momentum}, \eqref{eq:EEqg}-\eqref{eq:EEqk}, \eqref{eq:wave} and \eqref{eq:EEqLapse} on $[t_1,t_0]\times M$ for some $t_1\in(0,t_0)$ that satisfies ${k^l}_l=-3\dot{a}a^{-1}$ and $n>0$, launches a solution to the Einstein scalar-field system and such that the variables enjoy the following regularity:
\begin{align*}
g\in&C^{N-1}_{dt^2+\gamma}([t_1,t_0]\times M)\cap C^0([t_1,t_0],H_{\gamma}^{N+1}(M))\\
k\in&C^{N-2}_{dt^2+\gamma}([t_1,t_0]\times M)\cap C^0([t_1,t_0],H_{\gamma}^{N}(M))\\
\nabla\phi\in&C^{N-2}_{dt^2+\gamma}([t_1,t_0]\times M)\cap C^0([t_1,t_0],H_{\gamma}^{N}(M))\\
\del_t\phi\in&C^{N-2}_{dt^2+\gamma}([t_1,t_0]\times M)\cap C^0([t_1,t_0],H_{\gamma}^{N}(M))\\
n\in&\,C^N_{dt^2+\gamma}([t_1,t_0]\times M)\cap C^0([t_1,t_0],H_{\gamma}^{N+2}(M))
\end{align*}
The rescaled variables $(G,\Sigma,N,\nabla\phi,\Psi)$ enjoy the analogous regularity. If $(\mathfrak{t},t_0]$ is the maximal interval on which the above statements hold, then one either has $\mathfrak{t}=0$ or one of the following blow-up criteria are satisfied:
\begin{enumerate}
\item The smallest eigenvalue of $g(t_m,\cdot)$ converges to 0 for some sequence $(t_m,x_m)\subseteq (\mathfrak{t},t_0]\times M$ with $t_m\downarrow \mathfrak{t}$.
\item $n(t_m,x_m)$ converges to $0$ for some sequence $(t_m,x_m)\subseteq (\mathfrak{t},t_0]\times M$ with $t_m\downarrow \mathfrak{t}$.
\item $\left(\lvert\del_0\phi\rvert^2+\lvert\nabla\phi\rvert_g^2\right)(t_m,x_m)$ converges to $0$ for some sequence $(t_m,x_m)\subseteq (\mathfrak{t},t_0]\times M$ with $t_m\downarrow \mathfrak{t}$.
\item $s\in(\mathfrak{t},t_0]\mapsto\|g\|_{C^2_\gamma(\Sigma_s)}+\|k\|_{C^1_\gamma(\Sigma_s)}+\|n\|_{C^2_\gamma(\Sigma_s)}+\|\del_t\phi\|_{C^1_\gamma(\Sigma_s)}+\|\nabla\phi\|_{C^1_\gamma(\Sigma_s)}$ is unbounded.
\end{enumerate}
\end{lemma}
\begin{proof}[A note on the proof]
Note that the additional initial data requirement in the stiff-fluid setting that the pressure is strictly positive is covered by the smallness assumption on $\mathring{\psi}-Ca(t_0)^2$, since the pressure corresponds to $\lvert\mathring{\psi}\rvert^2+\lvert\mathring{\pi}\rvert_{\mathring{g}}^2$ and the assumptions on $\del_t\phi$ and $\nabla\phi$ ensure that (after embedding) this quantity behaves like $C^2a(t_0)^{-6}+\O{\eta^\prime}$ at $\Sigma_{t_0}$.
\end{proof}
\begin{corollary}[Local well-posedness for the Bel-Robinson variables]\label{cor:lwp-BR}
Under the assumptions of Lemma \ref{lem:lwp}, the Bel-Robinson variables $E$ and $B$ corresponding to the Lorentzian metric $\g=-n^2dt^2+g$ satisfy the equations \eqref{eq:constr-E}-\eqref{eq:constr-B}, are the unique classical solutions to the evolution equations \eqref{eq:EEqE}-\eqref{eq:EEqB}, and satisfy
\[E,B\in C^{N-3}([t_1,t_0]\times M)\cap C([t_1,t_0], H^{N-1}_\gamma(M))\]
\end{corollary}
\begin{proof}
That $E$ and $B$ satisfy the constraint equations, solve the evolution equations and have the stated regularity on the interval of existence is a direct consequence of Lemma \ref{lem:lwp} and the computations in Section \ref{subsec:BR}. Furthermore, with initial data derived from the constraint equations as in Remark \ref{rem:init-BR}, the hyperbolic system \eqref{eq:constr-E}-\eqref{eq:constr-B} launches a unique solution satisfying the regularity above that must then be $(E,B)$.
\end{proof}
\change{For sufficiently regular initial data ($N\geq 21$), it follows that
\[\E^{(\leq 19)}(\phi,\cdot),\E^{(\leq 18)}(W,\cdot),\E^{(\leq 19)}(\Sigma,\cdot),\E^{(\leq 17)}(\Ric,\cdot),\|G-\gamma\|_{H^{18}_G}\in C^{1}([t_1,t_0])\,,\]
and similarly the square of any supremum norm occurring in $\mathcal{C}$ is continuously differentiable on $[t_1,t_0]$. Strictly speaking, we would need to assume this additional regularity on our initial data for the computations in the following sections (especially Section \ref{sec:en-est}) to hold. However, since smooth functions are dense in $H^l(M)$ for any $l\in\N_0$, any bounds on $\mathcal{H}(t)$ and $\mathcal{C}(t)$ that we prove assuming sufficient regularity at $\Sigma_{t_0}$ then immediately extend to data only satisfying the regularity implied by \eqref{eq:init-ass}. }

\noindent \change{Thus, }from here on out, we will assume \change{without loss of \changefinal{generality }that }all energies \change{and squared norms are }continuously differentiable on the domain of existence, and similarly all variables \change{are }continuously differentiable for the lower order $C_G$-norm improvements in Section \ref{subsec:AP}.

\subsection{Bootstrap assumption}\label{subsec:bs}

To keep an overview of the entire bootstrap argument, we state all of the assumptions and comprehensively list how we intend to improve them.

\begin{assumption}[Bootstrap assumption]\label{ass:bootstrap}
Fix some $t_{Boot}\in[0,t_0)$. Further, let $c_0>0$, let $\sigma\in(\epsilon^\frac18,1]$ be suitably small such that $c_0\sigma<1$, and $K_0>0$ a suitable constant. For any $t\in(t_{Boot},t_0]$, we assume 
\begin{equation}\label{eq:BsC}
\mathcal{C}(t)\leq K_0\epsilon a(t)^{-c_0\sigma}\,.
\end{equation}
\delete{as well as the energy inequalities [...]}
\end{assumption}

\begin{remark}
More explicitly, \eqref{eq:BsC} means
\begin{subequations}
\begin{align}
\|\Psi\|_{\change{C^{16}_G}}\leq&\,K_0\epsilon a^{-c_0\sigma} \label{eq:BsPsi}\\
\|\nabla\phi\|_{\change{C^{15}_G}}\leq&\,K_0\epsilon a^{-c_0\sigma} \label{eq:Bsnablaphi}\\
\|\Sigma\|_{C^{16}_G}\leq&\,K_0\epsilon a^{-c_0\sigma} \label{eq:BsSigma}\\
\|\RE\|_{C^{16}_G}\leq&\,K_0\epsilon a^{-c_0\sigma} \label{eq:BsE}\\
\|\RB\|_{C^{16}_G}\leq&\,K_0\epsilon a^{-c_0\sigma} \label{eq:BsB}\\
\|\Ric[G]+\frac29G\|_{C^{14}_G}\leq&\,K_0\epsilon a^{-c_0\sigma} \label{eq:BsRic}\\
\|G-\gamma\|_{C^{16}_G}\leq&\,K_0\epsilon a^{-c_0\sigma} \label{eq:BsG}\\
\|N\|_{C^{14}_G}+a^2\|N\|_{\dot{C}^{15}_G}+a^4\|N\|_{\dot{C}^{16}_G}\leq&\,K_0\epsilon a^{4-c_0\sigma} \label{eq:BsN}\\
\change{\|\Gamma-\Gamhat\|_{C^{15}_G}\leq}&\change{\,K_0\epsilon a^{-c_0\sigma}} \label{eq:BsChr}
\end{align}
\end{subequations}
\end{remark}

\begin{remark}[Bootstrap assumptions with respect to $\gamma$]\label{rem:Bs-Cgamma}
Note again that we could equivalently make the above bootstrap assumptions with respect to $H_\gamma$- and $C_\gamma$-norms: For example, the assumptions \eqref{eq:BsChr} and \eqref{eq:BsG} imply
\begin{align*}
\|\zeta\|_{C^l_\gamma}\lesssim&\,a^{-c\sigma}\|\zeta\|_{C^l_G}+\|\zeta\|_{C^{\lceil\frac{l-1}2\rceil}_\gamma}\epsilon a^{-c\sigma},\quad
\|\mathfrak{T}\|_{C^l_\gamma}\lesssim a^{-c\sigma}\|\mathfrak{T}\|_{C^l_G}+\|\mathfrak{T}\|_{C^{\lceil\frac{l}2\rceil}_\gamma}\epsilon a^{-c\sigma}
\end{align*}
for any smooth function $\zeta\in C^\infty(\Sigma_t)$, any $\Sigma_t$-tangent tensor $\mathfrak{T}$ and a constant $c>0$. This is essentially a direct consequence of \eqref{eq:Christoffel-norm-handwaving}, and we will prove an improved version of this rigorously in Lemma \ref{lem:G-gamma-norm-switch}. Applying this to each norm in $\mathcal{C}$, we get
\begin{equation}\label{eq:BsCgamma}
\mathcal{C}_\gamma\lesssim\epsilon a^{-c\sigma}
\end{equation} 
for some updated constant $c\geq c_0$.
\end{remark}

\begin{remark}[Strategy for the bootstrap improvement]\label{rem:bs-strategy}

Our goal is to improve the $C$-norm estimate to
\begin{equation*}
\mathcal{C}\leq K_1\epsilon^\frac98a^{-c_1\epsilon^\frac18}\,,
\end{equation*}
where $c_1,K_1>0$ are positive constants independent of $\sigma$ and $\epsilon$. Notice how this is actually an improvement if we choose $\sigma$ suitably and then choose $\epsilon$ sufficiently small: Any update between $K_0$ and $K_1$ can be balanced out since we gain at least the additional prefactor $\epsilon^\frac18$ in each estimate, which we can then choose to have been suitably small. Similarly, we improve the power of $a$ if we have $\epsilon^{\frac18}\cdot\sigma^{-1}<\frac{c_0}{c_1}$. If we then retroactively choose $\sigma$ large enough compared to $\epsilon$ but small overall -- for example $\sigma=\epsilon^\frac{1}{16}$ -- and then ensure that $\max\{c_0,c_1\}\epsilon^\frac1{16}<1$ as well as $c_1\epsilon^\frac1{16}<c_0$ are satisfied by choosing $\epsilon$ to have been small enough, we have strictly improved the bootstrap assumptions.
\end{remark}

\begin{remark}[Conventions within the bootstrap argument]
Throughout the rest of the argument, we tacitly assume $t\in(t_{Boot},t_0]$ if not stated otherwise, and we assume $\epsilon$ and $\sigma$ to be sufficiently small. In the proof of Theorem \ref{thm:main}, we will choose $\sigma=\epsilon^\frac1{16}$, but this explicit choice will not be used or needed up to that point. Finally, 
we allow $c\geq c_0$ be a constant that we may update from line to line, and will similarly deal with prefactors by \enquote{$\lesssim$}-notation where the constant may change in each line. These updates will always be independent of $\sigma$ and $\epsilon$, but may depend on $t_0$, and the quantities arising from the FLRW reference solution. Hence, we not only assume $c_0\sigma<1$, but $c\sigma<1$ throughout the argument.
\end{remark}

\section{Big Bang stability: A priori estimates}\label{sec:ap}

In this section, we collect strong low order $C_G$-norm estimates that follow as an immediate consequence from the bootstrap assumptions, starting with key estimates at the base level and followed by weaker, but still improved estimates at higher levels. Finally, we collect a differentiation formula for integrals with respect to $\vol{G}$ as well as a Sobolev estimate that lays the groundwork for energy coercivity. In particular, using the strong $C_G$-norm estimates, said estimate proves that moving between energies and norms at most incurs an error involving lower order energies of the controlled variable and curvature energies, scaled by $a^{-c\sqrt{\epsilon}}$.

\subsection{Strong $C^0_G$-estimates}\label{subsec:APlow}

First, we establish a pointwise bound on the lapse that actually holds irrespective of the bootstrap assumptions:

\begin{lemma}[Maximum principle for the lapse]\label{lem:lapse-maxmin} The lapse remains positive and bounded throughout the evolution:
\begin{equation}
n=N+1\in(0,3]
\end{equation}
\end{lemma}
\begin{proof}
Let $t\in\R_+$ be arbitrary and let $n_{min}$ be the minimum of $n$ over $\Sigma_t$ at $(t,x_{min})$. Then, $(\Lap_g n)(t,x_{min})> 0$ holds. If $n_{min}$ were nonpositive, \eqref{eq:EEqLapse} would lead to the following contradiction:
\[0\geq -12\pi C^2a^{-6}-\frac13a^{-2}+n_{min}\left[\frac13a^{-2}+4\pi C^2a^{-6}+\langle\hat{k},\hat{k}\rangle_g+8\pi\lvert\del_0\phi\rvert^2\right]=\Lap_gn(t,x_{min})>0\]
This shows $n>0$, and the upper bound follows analogously.
\end{proof}

The following estimate will be essential in dealing with borderline terms throughout the bootstrap argument:

\begin{lemma}[Strong $C^0_G$ estimates]\label{lem:APzero}The following estimates hold:
\begin{subequations}
\begin{align}
\|\Psi\|_{C^0_G}\lesssim&\,\epsilon\label{eq:APPsi}\\
\|\Sigma\|_{C^0_G}\lesssim&\,\epsilon\label{eq:APSigma}\\
\|\RE\|_{C^0_G}\lesssim&\,\epsilon\label{eq:APE}
\end{align}
\end{subequations}
\end{lemma}
\begin{proof}
\underline{\eqref{eq:APPsi}:} From \eqref{eq:REEqWave}, we obtain the following using Lemma \ref{lem:lapse-maxmin} for $n$, the bootstrap assumptions \eqref{eq:BsN} and \eqref{eq:Bsnablaphi} and that $\dot{a}\simeq a^{-2}$ by \eqref{eq:Friedman}:
\begin{align*}
\lvert\del_t\Psi\rvert
\lesssim&\,\epsilon a^{5-c\sigma}+\epsilon a^{1-c\sigma}+\epsilon a^{1-c\sigma}\lvert\Psi\rvert+{\epsilon} a^{1-c\sigma}
\end{align*}
After integration, we thus obtain using the initial data assumption \eqref{eq:init-ass}:
\begin{align*}
\lvert\Psi(t)\rvert\lesssim&\,\lvert\Psi(t_0)\rvert+\int_t^{t_0}\epsilon a(s)^{1-c\sigma}\,ds+\int_t^{t_0}{\epsilon} a(s)^{1-c\sigma}\lvert\Psi(s)\rvert\,ds\\
\lesssim&\,\epsilon\left(1+\int_t^{t_0}a(s)^{1-c\sigma}\,ds\right)+\int_t^{t_0}{\epsilon}a(s)^{1-c\sigma}\lvert\Psi(s)\rvert\,ds
\end{align*}
By \eqref{eq:a-integrals}, the integral over $a^{1-c\sigma}$ is bounded since $c\sigma<1$, so the Gronwall lemma now yields \eqref{eq:APPsi}.\\

\underline{\eqref{eq:APSigma}:} Notice that
\begin{equation}\label{eq:deltSigma2}
\del_t\lvert\Sigma\rvert_G^2={(\del_t\Sigma^\sharp)^l}_m{(\Sigma^\sharp)^m}_l+{(\Sigma^\sharp)^l}_m{(\del_t\Sigma^\sharp)^m}_l=2{(\del_t\Sigma^\sharp)^l}_m{(\Sigma^\sharp)^m}_l\leq 2\lvert\del_t\Sigma^\sharp\rvert_G\lvert\Sigma\rvert_G\,.
\end{equation}
Now, we consider \eqref{eq:REEqSigmaSharp} and, using the bootstrap assumptions \eqref{eq:BsN}, \eqref{eq:BsRic} and \eqref{eq:Bsnablaphi}, get:
\begin{align*}
\lvert\del_t\Sigma^\sharp\rvert_G\lesssim&\,\tau\lvert N\rvert\lvert\Sigma^\sharp\rvert_G+\lvert\nabla^\sharp\nabla N\rvert_G a+\lvert N+1\rvert a\left\lvert\Ric[G]^\sharp+\frac29G^\sharp\right\rvert_G+\lvert N+1\rvert a\lvert\nabla^\sharp\phi\nabla\phi\rvert_G\\
&+\sqrt{3}\lvert N\rvert\cdot\left(4\pi C^2a^{-3}+\frac19a\right)\\
\lesssim&\,{\epsilon} a^{1-c\sigma}\lvert\Sigma\rvert_G+{\epsilon}a^{1-c\sigma}
\end{align*}
We can now apply Lemma \ref{lem:weak-ftoc} with $f=\lvert\Sigma\rvert_G^2$, and thus have along with \eqref{eq:init-ass} and \eqref{eq:deltSigma2}:
\[\lvert\Sigma\rvert_G(t)\leq\lvert\Sigma\rvert_G(t_0)+\int_t^{t_0}\lvert\del_t\Sigma^\sharp\rvert_G(s)\,ds\lesssim\epsilon\]

\underline{\eqref{eq:APE}:} Using the constraint equation \eqref{eq:REEqConstrE} and that $\langle G,\RE\rangle_G=\text{tr}_G\RE=0$, one sees
\begin{equation*}
\lvert\RE\rvert_G^2=\left\langle a^4\left(\Ric[G]+\frac29G\right)-\dot{a}a^2\Sigma-\Sigma\odot_G\Sigma-4\pi a^4\nabla\phi\nabla\phi,\RE\right\rangle_G\,.
\end{equation*}
Then, applying the bootstrap assumptions \eqref{eq:BsRic} and \eqref{eq:Bsnablaphi} shows the Ricci and matter terms are bounded by $\epsilon a^{4-c\sigma}\lvert\RE\rvert_G$, and the a priori estimate \eqref{eq:APSigma} along with $\dot{a}a^2\simeq 1$ by \eqref{eq:Friedman} bounds the remaining terms by $\epsilon\lvert\RE\rvert_G$. The statement then follows by dividing by $\lvert\RE\rvert_G$ and taking the supremum.
\end{proof}
\change{Note that, in the proof of \eqref{eq:APSigma}, it was essential that we used
\eqref{eq:REEqSigmaSharp} instead of \eqref{eq:REEqSigma}, since using the latter would incur terms of the type $\lvert\del_tG\rvert_G\lvert\Sigma\rvert_G^2$and $a^{-3}\lvert \Sigma\rvert_G^3$ when computing the time derivative of $\lvert\Sigma\rvert_G^2$, which, at this point, behave like $\epsilon a^{-3-c\sigma}\lvert \Sigma\rvert_G^2$, and thus not yield the sharp estimate (or even an improved estimate) that we will need to control borderline terms.}

\subsection{Strong low order $C_G$-norm estimates}\label{subsec:AP}

Now, we can prove the main supremum norm estimates in this section:

\begin{lemma}[Strong low order $C_G$-norm estimates]\label{lem:AP} The following estimates hold:
\begin{subequations}
\begin{align}
\|\Psi\|_{C^{13}_G}\lesssim&\,\epsilon a^{-c\sqrt{\epsilon}}\,\label{eq:APmidPsi}\\
\|\Sigma\|_{C^{12}_G}\lesssim&\,\epsilon a^{-c\sqrt{\epsilon}}\,\label{eq:APmidSigma}\\
\|G-\gamma\|_{C^{12}_G}\lesssim&\sqrt{\epsilon}a^{-c\sqrt{\epsilon}}\,\label{eq:APmidG}\\
\|G^{-1}-\gamma^{-1}\|_{C^{12}_G}\lesssim&\sqrt{\epsilon}a^{-c\sqrt{\epsilon}}\,\label{eq:APmidG-1}\\
\|\nabla\phi\|_{C^{12}_G}\lesssim&\,\sqrt{\epsilon}a^{-c\sqrt{\epsilon}}\label{eq:APmidphi}\\
\|\Ric[G]+\frac29G\|_{C^{10}_G}\lesssim&\,\sqrt{\epsilon}a^{-c\sqrt{\epsilon}}\label{eq:APmidRic}\\
\|\RB\|_{C^{11}_G}\lesssim&\,\epsilon \change{a^{2-c\sqrt{\epsilon}}}\label{eq:APmidB}\\
\|\RE\|_{C^{12}_G}\lesssim&\,\epsilon a^{-c\sqrt{\epsilon}}\label{eq:APmidE}
\end{align}
\end{subequations}
\end{lemma}
\begin{proof}
Before going into the individual estimates, we collect the following commutator term estimates from the expressions in \eqref{eq:commutator-aux-scalar}-\eqref{eq:commutator-aux-tensor}:
\begin{align}
\|[\del_t,\nabla^J]\zeta\|_{C^0_G}\lesssim&\,a^{-3}\|N+1\|_{C^{J-1}_G}\|\Sigma\|_{C^{J-1}_G}\|\zeta\|_{C^{J-1}_G}+\frac{\dot{a}}a\|N\|_{C^{J-1}_G}\|\zeta\|_{C^{J-1}_G} \label{eq:aux-comm-est-zeta}\\
\|[\del_t,\nabla^J]\mathfrak{T}\|_{C^0_G}\lesssim&\,a^{-3}\|N+1\|_{C^{J}_G}\left(\|\nabla^J\Sigma\|_{C^0_G}\|\mathfrak{T}\|_{C^0_G}+\|\Sigma\|_{C^{J-1}_G}\|\mathfrak{T}\|_{C^{J-1}_G}\right)+\frac{\dot{a}}a\|N\|_{C^J_G}\|\mathfrak{T}\|_{C^{J-1}_G} \label{eq:aux-comm-est-T}
\end{align}
With this in hand, we will prove each estimate by iterating over the derivative order as long as the bootstrap assumptions can be applied. In each step, we use the previously obtained estimates at lower order to control the commutator term (with some additional care for $\mathfrak{T}=\Sigma$ which we need to consider first), while we can use similar arguments to those at order $0$ to control the \enquote{core} of the evolution equations. 

To start out, we apply \eqref{eq:APSigma} on $\Sigma$ and the bootstrap assumption \eqref{eq:BsN} on $N$ to the rescaled evolution equations \eqref{eq:REEqG}-\eqref{eq:REEqG-1} and deduce
\begin{equation}\label{eq:AP-deltG}
\lvert \del_tG^{\pm1}\rvert_G=\lvert \del_t(G^{\pm 1}-\gamma^{\pm1})\rvert_G\lesssim \epsilon a^{-3}+\epsilon a^{1-c\sigma}\lesssim \epsilon a^{-3}\,.
\end{equation}

\underline{\eqref{eq:APmidSigma}:} We assume
\begin{equation}\label{eq:APmidSigma-indhyp}
\|\Sigma\|_{{C}^{J-1}_G}\lesssim \epsilon a^{-c\sqrt{\epsilon}}
\end{equation}
to be satisfied for some $J\in\{1,\dots,12\}$ (For $J=1$, this is true by \eqref{eq:APSigma}). Observe the following:
\[\del_t\lvert\nabla^J\Sigma\rvert_G^2=2\langle\del_t\nabla^J\Sigma,\nabla^J\Sigma\rangle_G+\del_tG^{-1}\ast\nabla^J\Sigma\ast\nabla^J\Sigma\]
Now, we commute \change{\eqref{eq:REEqSigmaSharp} }with $\nabla^J$: As before, $\nabla^J\del_t\Sigma$ is bounded by $\epsilon a^{-c\sigma}$ for any admissible $J$. Hence and using \eqref{eq:AP-deltG},
\[\del_t\lvert\nabla^J\Sigma\rvert_G^2\lesssim \epsilon a^{-3}\lvert\nabla^J\Sigma\rvert_G^2+\left(\epsilon a^{1-c\sigma}+\|[\del_t,\nabla^J]\Sigma\|_{C^0_G}\right)\lvert\nabla^J\Sigma\rvert_G\]
is satisfied. Looking at the commutator term using \eqref{eq:aux-comm-est-T}, we have with \eqref{eq:APSigma} that
\[\|[\del_t,\nabla^J]\Sigma\|_{C^0_G}\lesssim \epsilon a^{-3}\|\Sigma\|_{\dot{C}^J_G}+a^{-3}\cdot\|\Sigma\|_{C^{J-1}_G}^2+\epsilon a^{1-c\sigma}\|\Sigma\|_{C^{J-1}_G}\,.\]
Altogether, we obtain
\[\del_t\lvert\nabla^J\Sigma\rvert_G^2\lesssim\left(\epsilon a^{-3}\|\Sigma\|_{\dot{C}^J_G}+\epsilon a^{-c\sigma}+\epsilon^2a^{-3-c\sqrt{\epsilon}}\right)\lvert\nabla^J\Sigma\rvert_G\,.\]
With Lemma \ref{lem:weak-ftoc} as well as the initial data assumption \eqref{eq:init-ass} and the integral formula \eqref{eq:a-integrals} with $p=c\sqrt{\epsilon}$, this implies
\[\lvert\nabla^J\Sigma\rvert_G(t)\lesssim\,\int_t^{t_0}\epsilon a^{-3}\|\Sigma\|_{\dot{C}^J_G(\Sigma_s)}\,ds+\epsilon \left(1+\sqrt{\epsilon}a^{-c\sqrt{\epsilon}}\right)\]
and consequently, after taking the supremum on the left and applying the Gronwall lemma,
\[\|\Sigma\|_{\dot{C}^J_G(\Sigma_s)}\lesssim \epsilon a^{-c\sqrt{\epsilon}}\,.\]
Combining this with \eqref{eq:APmidSigma-indhyp} proves the statement up to order $J$, and hence shows \eqref{eq:APmidSigma} by iterating the argument up to $J=12$.\\

\underline{\eqref{eq:APmidPsi}:} We again assume that
\begin{equation}\label{eq:APmidPsi-indhyp}
\|\Psi\|_{C^{J-1}_G}\lesssim \epsilon a^{-c\sqrt{\epsilon}}
\end{equation}
holds for $J\in\{1,2,\dots,13\}$. Observe that
\[\lvert\del_t\nabla^J\Psi\rvert_G\lesssim a\|N+1\|_{C^{J+1}_G}\|\nabla\phi\|_{C^{J+1}_G}+\frac{\dot{a}}a\|\nabla N\|_{C^{J}_G}(1+\|\Psi\|_{C^J_G})+\|[\del_t,\nabla^J]\Psi\|_{C^0_G}\]
By \eqref{eq:BsN}, \eqref{eq:Bsnablaphi} and \eqref{eq:BsPsi}, the first two summands can be bounded (up to constant) by $\epsilon a^{1-c\sigma}$. By \eqref{eq:APmidPsi-indhyp}, \eqref{eq:APmidSigma} and \eqref{eq:BsN} and using \eqref{eq:aux-comm-est-zeta}, the commutator term is bounded (up to constant) by $\epsilon^2a^{-3-c\sqrt{\epsilon}}$. Altogether, 
\[\lvert\del_t\nabla^J\Psi\rvert_G\lesssim \epsilon a^{1-c\sigma}+\epsilon^2a^{-3-c\sqrt{\epsilon}}\]
follows. Inserting this and \eqref{eq:AP-deltG} into
\begin{equation*}
\lvert\del_t\left(\lvert\nabla^J\Psi\rvert_G^2\right)\rvert\leq\lvert\del_tG^{-1}\rvert_G\lvert\nabla^J\Psi\rvert_G^2+2\lvert\del_t\nabla^J\Psi\rvert_G\cdot\lvert\nabla^J\Psi\rvert_G
\end{equation*}
implies, with Lemma \ref{lem:weak-ftoc},
\begin{align*}
\lvert\nabla^J\Psi\rvert_G(t)\leq&\,\lvert\nabla^J\Psi\rvert(t_0)+\int_t^{t_0}\left(\frac12\lvert\del_t G^{-1}\rvert\lvert\nabla^J\Psi\rvert_G+\lvert\del_t\nabla^J\Psi\rvert_G\right)(s)\,ds\\
\lesssim&\,\epsilon^2+\int_t^{t_0}\left(\epsilon a(s)^{-3}\lvert\nabla^J\Psi(s,\cdot)\rvert_G+\epsilon a(s)^{1-c\sigma}+\epsilon^2a(s)^{-3-c\sqrt{\epsilon}}\right)\,ds\,.
\end{align*}
We obtain using \eqref{eq:a-integrals}:
\begin{equation*}
\lvert\nabla^J\Psi\rvert_G(t)\lesssim \epsilon a(t)^{-c\sqrt{\epsilon}}+\int_t^{t_0}\epsilon a(s)^{-3}\lvert\nabla^J\Psi(s,\cdot)\rvert_G\,ds
\end{equation*}
The Gronwall lemma, applying \eqref{eq:a-exp-est} and taking the supremum over $\Sigma_t$ then implies $\lvert\nabla^J\Psi\rvert_{\dot{C}^J_G}\lesssim \epsilon a^{-c\sqrt{\epsilon}}$. This proves \eqref{eq:APPsi} by iterating over $J$ and adding up the individual seminorms.\\

\underline{\eqref{eq:APmidG}-\eqref{eq:APmidG-1}:} Note that \eqref{eq:AP-deltG} implies \eqref{eq:APmidG} at order $0$ since one has
\[\lvert\del_t(\lvert G-\gamma\rvert_G)^2\rvert\lesssim\lvert\del_tG^{-1}\rvert_G\lvert G-\gamma\rvert_G^2+\lvert\del_t(G-\gamma)\rvert_G\lvert G-\gamma\rvert_G\lesssim \epsilon a^{-3}(1+\lvert G-\gamma\rvert_G)\lvert G-\gamma\rvert_G\]
which we can apply the Gronwall lemma to after integrating, along with \eqref{eq:log-est} for the error term, as in the proof of \eqref{eq:APmidSigma}.\\
For higher orders, commuting \eqref{eq:REEqG} with $\nabla^J$ and inserting \eqref{eq:APmidSigma} and \eqref{eq:BsN} implies
\[\|\del_t\nabla^J(G-\gamma)\|_{C^0_G}\lesssim \epsilon a^{-3-c\sqrt{\epsilon}}+\epsilon a^{1-c\sigma}+\|[\del_t,\nabla^J](G-\gamma)\|_{C^0_G}\]
with
\[\|[\del_t,\nabla^J](G-\gamma)\|_{C^0_G}\lesssim \left(\epsilon a^{-3-c\sqrt{\epsilon}}+\epsilon a^{1-c\sigma}\right)\|G-\gamma\|_{C^{J-1}_G}\,.\]
Once again doing the same iterative argument over $J\leq 12$ and assuming the estimate to hold up to $J-1$, this altogether becomes
\[\|\del_t\nabla^J(G-\gamma)\|_{C^0_G}\lesssim \epsilon a^{-3-c\sqrt{\epsilon}},\]
implying with \eqref{eq:a-integrals}
\[\|\nabla^J(G-\gamma)\|_{C^0_G}\lesssim \epsilon^2+\epsilon\int_t^{t_0}a(s)^{-3-c\sqrt{\epsilon}}\,ds\lesssim\sqrt{\epsilon}a^{-c\sqrt{\epsilon}}\,.\]
The argument for $G^{-1}-\gamma^{-1}$ is completely analogous.\\

\underline{\eqref{eq:APmidphi}:} We only prove the statement for $C^0_G$, the full estimate extends from there by the same iterative arguments as above. Considering \eqref{eq:REEqNablaPhi}, Lemma \ref{lem:lapse-maxmin}, \eqref{eq:APPsi} and \eqref{eq:Friedman}, we have
\begin{equation*}
\lvert\del_t\nabla\phi\rvert_G\lesssim a^{-3}\left(\lvert\nabla\Psi\rvert_G+\lvert\nabla N\rvert_G\right)
\end{equation*}
and thus, with \eqref{eq:APmidPsi} and the bootstrap assumption \eqref{eq:BsN},
\begin{equation*}
\lvert\del_t\nabla\phi\rvert_G\lesssim \epsilon a^{-3-c\sqrt{\epsilon}}\,.
\end{equation*}
With \eqref{eq:AP-deltG}, this implies
\begin{equation*}
\lvert\del_t\lvert\nabla\phi\rvert_G^2\rvert\lesssim \epsilon a^{-3}\lvert\nabla\phi\rvert_G^2+\epsilon a^{-3-c\sqrt{\epsilon}}\lvert\nabla\phi\rvert_G
\end{equation*}
and the statement follows as usual by applying Lemma \ref{lem:weak-ftoc}, \eqref{eq:log-est} and the Gronwall lemma.\\

\change{\underline{\eqref{eq:APmidRic}:} This follows as in the proof of \eqref{eq:APmidG} using \eqref{eq:REEqRic} and \eqref{eq:REEqG} and their commuted analogues.}

Once again, for $C^0_G$, we have
%
\underline{\eqref{eq:APmidB}:} This is obtained immediately from commuting \eqref{eq:REEqConstrB} with $\nabla^J$ and applying \eqref{eq:APmidSigma}. Notice that the Levi-Civita tensor can be absorbed into the implicit constants since $\lvert\epsilon[G]\rvert_G=\sqrt{6}$ holds (see \eqref{eq:LCS-contr3}).\\

\underline{\eqref{eq:APmidE}:} This follows like in the proof of \eqref{eq:APE} from applying \eqref{eq:BsRic}, \eqref{eq:Bsnablaphi} and \eqref{eq:APmidSigma} to the constraint equation \eqref{eq:REEqConstrE} commuted with $\nabla^J$.

\end{proof}

\subsection{Other useful a priori observations}\label{subsec:AP-misc}

Before moving on to the energy estimates, we collect a differentiation identity and lay the groundwork for energy coercivity:

\begin{lemma}[The volume form and differentiation of integrals]\label{lem:delt-int}
Let $\mu_G=\sqrt{\det G}$ denote the volume element with regard to $G$. It satisfies
\begin{equation}\label{eq:delt-muG}
\del_t\mu_G=\frac12\mu_G(G^{-1})^{ij}\del_tG_{ij}=-N\tau\mu_G\,\,,
\end{equation}
and hence one has
\begin{equation}\label{eq:APvol}
\|\mu_{G}-\mu_{\gamma}\|_{C^0_G}\lesssim\epsilon \,.
\end{equation}
on $(\Sigma_t)_{t\in(t_{Boot},t_0]}$. Further, for any differentiable function $\zeta$, one has
\begin{equation}\label{eq:delt-int}
\del_t\int_M \zeta\vol{G} = \int_M \del_t\zeta\vol{G}-\int_MN\tau \cdot \zeta\vol{G}
\end{equation}
\end{lemma}
\begin{proof}
From \eqref{eq:REEqG}, we obtain $(G^{-1})^{ij}\del_tG_{ij}=-2N\tau$,
and \eqref{eq:delt-muG} follows by
\[\del_t\mu_G
=\frac12\sqrt{\det G}(G^{-1})^{ij}\del_tG_{ij}=-N\tau\mu_G\,.\]
Hence, we have using \eqref{eq:BsN} and the initial data estimate \eqref{eq:init-vol-el} that
\begin{align*}
\lvert\mu_G-\mu_\gamma\rvert(t,\cdot)
\lesssim&\, \epsilon+\int_t^{t_0}\epsilon a(s)^{1-c\sigma}\lvert\mu_G-\mu_\gamma\rvert(s.\cdot)\,ds
\end{align*}
holds, and thus \eqref{eq:APvol} after applying the Gronwall lemma.\\
Finally, we obtain \eqref{eq:delt-int} by writing $\vol{G}=\frac{\mu_G}{\mu_\gamma}\vol{\gamma}$ and inserting \eqref{eq:delt-muG}
\end{proof}

\begin{lemma}[Preliminary Sobolev norm estimates]\label{lem:Sobolev-norm-equivalence-improved} Let $\zeta$ be a scalar function and $\mathfrak{T}$ be a symmetric $\Sigma_t$-tangent $(0,2)$-tensor, and let $l\in\{1,\dots,\change{9}\}$. Then, on $(t_{Boot},t_0]$, the following estimates are satisfied: For $l>5$, one has:
\begin{subequations}
\begin{align}
\|\nabla^2\zeta\|_{L^2_G}^2\lesssim&\,\|\Lap\zeta\|_{L^2_G}^2+a^{-c\sqrt{\epsilon}}\|\nabla\zeta\|_{L^2_G}^2\label{eq:Sobolev-norm-equiv-zetalow}\\
\|\zeta\|_{H^{2l}_G}^2\lesssim&\,\|\Lap^{l} \zeta\|_{L^2_G}^2+a^{-c\sqrt{\epsilon}}\left(\sum_{m=0}^{l-1}\|\Lap^m\zeta\|_{L^2_G}^2+\change{\|\zeta\|_{C_G^{\change{2l-12}}}^2\E^{(\leq 2l-3)}(\Ric,\cdot)}\right) \label{eq:Sobolev-norm-equiv-zeta2l}\\
\sum_{m=1}^{2l+1}\|\zeta\|_{\dot{H}^{m}_G}^2\lesssim&\,\|\nabla\Lap^{l}\zeta\|_{L^2_G}^2+a^{-c\sqrt{\epsilon}}\left(\sum_{m=0}^{l-1}\|\nabla\Lap^m\zeta\|_{L^2_G}^2+\|\nabla\zeta\|_{C_G^{\change{2l-12}}}^2\E^{(\leq2l-2)}(\Ric,\cdot)\right)\label{eq:Sobolev-norm-equiv-zeta2l+1}\\
\changefinal{\sum_{m=1}^{2l}\|\nabla\zeta\|_{\dot{H}^{m}_G}^2\lesssim}&\changefinal{\,\|\nabla\Lap^l\zeta\|_{L^2_G}^2+a^{-c\sqrt{\epsilon}}\left(\sum_{m=0}^{l-1}\|\nabla\Lap^m\zeta\|_{L^2_G}^2+\|\nabla\zeta\|_{C_G^{\change{2l-11}}}^2\E^{(\leq 2l-2)}(\Ric,\cdot)\right)\label{eq:Sobolev-norm-equiv-nablazeta2l}}
\end{align}
\end{subequations}
and
\begin{subequations}
\begin{align}
\|\mathfrak{T}\|_{H^{2l}_G}^2\lesssim&\,\|\Lap^{l} \mathfrak{T}\|_{L^2_G}^2+a^{-c\sqrt{\epsilon}}\left(\sum_{m=0}^{l-1}\|\Lap^m\mathfrak{T}\|_{L^2_G}^2+\|\mathfrak{T}\|_{C_G^{\change{2l-11}}}^2\E^{(\leq 2l-2)}(\Ric,\cdot)\right)\label{eq:Sobolev-norm-equiv-T2l}\\
\sum_{m=1}^{2l+1}\|\mathfrak{T}\|_{\dot{H}^{m}_G}^2\lesssim&\,\|\nabla\Lap^{l}\mathfrak{T}\|_{L^2_G}^2+a^{-c\sqrt{\epsilon}}\left(\sum_{m=0}^{l-1}\|\nabla\Lap^m\mathfrak{T}\|_{L^2_G}^2+\|\mathfrak{T}\|_{C_G^{\change{2l-10}}}^2\E^{(\leq2l-1)}(\Ric,\cdot)\right)\label{eq:Sobolev-norm-equiv-T2l+1}
\end{align}
\end{subequations}
More precisely, the Ricci energy terms can be dropped in all of the above estimates for $l\leq 5$.
\end{lemma}
\begin{remark}\label{rem:Sobolev-norm-equivalence-improved}
We stress that Lemma \ref{lem:Sobolev-norm-equivalence-improved} is crucial for everything that follows in multiple ways:\\

Firstly, the $L^2_G$-norms containing $\zeta$ and $\mathfrak{T}$ on the right hand sides \delete{of all inequalities }above except \eqref{eq:Sobolev-norm-equiv-nablazeta2l} are in precisely the form the energies in Definition \ref{def:energies} take. Hence, this is what will actually yield near-coercivity of \change{our energies since the $C_G$-norms can be controlled by $a^{-c\sqrt{\epsilon}}$ or better using the a priori estimates from Lemma \ref{lem:AP}, as well as \eqref{eq:BsN} for the lapse.} This will be shown more \change{explicitly }as an intermediary step in improving the bootstrap assumptions for $\mathcal{C}$ (see proof of Corollary \ref{cor:H-imp}).\\

Secondly, a downside of using $\Lap$ as the main differential operator to commute with the Einstein scalar-field system is that it creates error terms that we can only bound by Sobolev norms and not directly express as energies. Thus, we need a way to translate this information back to energies to formulate energy inequalities. A lot of this is done \enquote{under the hood} in the error term estimates in subsection \ref{subsec:L2-error-est}.\\

Finally, some top order terms also do not appear in a way that their $L^2$-norm is directly the square root of an energy (see, for example, the term $a\nabla^2\Lap^{\frac{L}2}N$ in \eqref{eq:comeq-Sigma}), and some borderline terms would lead to nonintegrable divergences \delete{at that level of derivatives }if we were to incur additional divergences in estimation (see, for example, the first term in \eqref{eq:comeq-Sigma-border}). Lemma \ref{lem:Sobolev-norm-equivalence-improved} precisely provides a way to relate these terms to energies \delete{without losing precision at top order in terms of powers of $a$}. Additionally, by applying these estimates for terms of the form $\Lap^\frac{L}2\zeta$ and $\Lap^\frac{L}2\mathfrak{T}$, one can avoid high order curvature energies \change{that run the risk of breaking the energy hierarchy.}
\end{remark}
\begin{proof} Since the arguments for all of the inequalities above are very similar, we only prove \eqref{eq:Sobolev-norm-equiv-T2l} in full and then briefly \changefinal{address }the other estimates.\\
Letting $\tilde{\mathfrak{T}}_{i_1\dots i_{2l}k_1k_2}=\nabla_{i_1}\dots\nabla_{i_{2l}}\mathfrak{T}_{k_1k_2}\,$, we compute with the commutator formula \eqref{eq:[Lap,nabla]T} and strong $C_G$-norm estimate \eqref{eq:APmidRic}:
\begin{align*}
\int_M\lvert\nabla^2\tilde{\mathfrak{T}}\rvert_G^2=&\,-\int_M\langle\nabla\tilde{\mathfrak{T}},\Lap\nabla\tilde{\mathfrak{T}}\rangle_G\,\vol{G}\\
=&-\int_M\langle\nabla\tilde{\mathfrak{T}},\nabla\Lap\tilde{\mathfrak{T}}\rangle_G\,\vol{G}+\int_M\nabla\tilde{\mathfrak{T}}\ast\left[\nabla\Ric[G]\ast \tilde{\mathfrak{T}}+\Ric[G]\ast\nabla\tilde{\mathfrak{T}}\right]\,\vol{G}\\
\lesssim& \int_M\lvert\Lap\tilde{\mathfrak{T}}\rvert_G^2\,\vol{G}+\changefinal{\left(1+\left\|\Ric[G]+\frac29G\right\|_{{C}^1_G}\right)}\cdot\left[\int_M\lvert\tilde{\mathfrak{T}}\rvert_G^2\,\vol{G}+\int_M\lvert\nabla\tilde{\mathfrak{T}}\rvert_G^2\,\vol{G}\right]\\
\lesssim&\int_M \lvert\Lap\tilde{\mathfrak{T}}\rvert_G^2+a^{-c\sqrt{\epsilon}}\int_M\lvert\tilde{\mathfrak{T}}\rvert_G^2\,\vol{G}
\end{align*}
In the final step, we used integration by parts to obtain
\[a^{-c\sqrt{\epsilon}}\int_M\lvert\nabla\tilde{\mathfrak{T}}\rvert_G^2\,\vol{G}\leq \int_M\lvert\Lap\tilde{\mathfrak{T}}\rvert_G\cdot a^{-c\sqrt{\epsilon}}\lvert\tilde{\mathfrak{T}}\rvert_G\,\vol{G}\lesssim \int_M\left(\lvert\Lap\tilde{\mathfrak{T}}\rvert_G^2+a^{-2c\sqrt{\epsilon}}\lvert\tilde{\mathfrak{T}}\rvert_G^2\right)\,\vol{G}\]
and updated $c$. This already shows \eqref{eq:Sobolev-norm-equiv-T2l} for $l=1$. Assume now that \eqref{eq:Sobolev-norm-equiv-T2l} holds up to some $l\in\N,\, l\leq 9$ and \textit{any} symmetric $\Sigma_t$-tangent $(0,2)$-tensor field. By applying \change{\eqref{eq:[Lap,nabla2]T}}, we have
\change{\begin{align*}
\Lap\tilde{\mathfrak{T}}
=&\,\nabla^{2l}\Lap\mathfrak{T}+[\Lap,\nabla^2]\nabla^{2l-2}\mathfrak{T}+\dots+\nabla^{2l-2}[\Lap,\nabla^2]\mathfrak{T}\\
=&\,\nabla^{2l}\Lap\mathfrak{T}+\sum_{I_\Ric+I_\mathfrak{T}=2l}\nabla^{I_\Ric}(\Ric[G]+\frac29G)\ast\nabla^{I_\mathfrak{T}}\mathfrak{T}+G\ast\nabla^{2l}\mathfrak{T}\numberthis\label{eq:lap-sob-equiv-crucial}
\end{align*}}
\change{Subsequently}, we have \change{for $l>5$ }using the strong $C_G$-norm estimate \eqref{eq:APmidRic} \change{for any Ricci term of order 10 or lower }that
\begin{align*}
\|\mathfrak{T}\|_{\dot{H}^{2(l+1)}_G}^2=\int_M \lvert\nabla^2\tilde{\mathfrak{T}}\rvert_G^2\,\vol{G}\lesssim \|\Lap\mathfrak{T}\|_{\dot{H}^{2l}_G}^2+\left(1+\sqrt{\epsilon}a^{-c\sqrt{\epsilon}}\right)\|\mathfrak{T}\|_{H^{2l}_G}^2+\|\mathfrak{T}\|_{C_G^{\change{2l-11}}}^2\|\Ric[G]+\frac29G\|_{H^{2l}_G}^2\,,
\end{align*}
\change{and get the same estimate without the final term for $l\leq 5.$ }By assumption, we can estimate $\|\Lap\mathfrak{T}\|_{H^{2l}_G}^2$, $\|\mathfrak{T}\|_{H^{2l}_G}^2$ and \changefinal{$\|\Ric[G]+\frac29G\|_{H^{2l}_G}^2$ }as in \eqref{eq:Sobolev-norm-equiv-T2l}, and get the following \change{for $l>5$} :
\begin{align*}
\int_M\lvert\nabla^{2l+2}\mathfrak{T}\rvert_G^2\,\vol{G}\lesssim&\,\left[\|\Lap^{l}\Lap\mathfrak{T}\|^2_{L^2_G}+a^{-c\sqrt{\epsilon}}\sum_{m=0}^{l-1}\|\Lap^{m+1}\mathfrak{T}\|_{L^2_G}^2+a^{-c\sqrt{\epsilon}}\|\Lap\mathfrak{T}\|_{C_G^{\change{2l-12}}}^2\E^{(\leq 2l-2)}(\Ric,\cdot)\right]\\
&\,+\left[a^{-c\sqrt{\epsilon}}\|\Lap^l\mathfrak{T}\|_{L^2_G}^2+a^{-c\sqrt{\epsilon}}\sum_{m=0}^{l-1}\|\Lap^{m}\mathfrak{T}\|^2_{L^2_G}+a^{-c\sqrt{\epsilon}}\|\mathfrak{T}\|_{C^{\change{2l-12}}_G}^2\E^{(\leq 2l-2)}(\Ric,\cdot)\right]\\
&\,+a^{-c\sqrt{\epsilon}}\|\mathfrak{T}\|_{C_G^{\change{2l-11}}}^2\left[\E^{(2l)}(\Ric,\cdot)+\left(\E^{(\leq 2l-2)}(\Ric,\cdot)+\epsilon a^{-c\sqrt{\epsilon}}\E^{(\leq 2l-2)}(\Ric,\cdot)\right)\right]\\
\lesssim&\,\|\Lap^{l+1}\mathfrak{T}\|_{L^2_G}^2+a^{-c\sqrt{\epsilon}}\left(\sum_{m=0}^l\|\Lap^m\mathfrak{T}\|_{L^2_G}^2+\|\mathfrak{T}\|_{C^{\change{2l-10}_G}}^2\E^{(\leq 2l)}(\Ric,\cdot)\right)
\end{align*}
\change{For $l=5$, we get analogous estimates dropping the Ricci energies in the first two lines, and for $l=4$, the same with all curvature terms dropped.\\}
To prove the statement for $l+1$, it now remains to be shown that $\|\mathfrak{T}\|_{\dot{H}^{2l+1}_G}^2$ can be bounded by the same right hand side (up to constant) as above. By integration by parts, one has
\[\|\nabla^{2l+1}\mathfrak{T}\|_{L^2_G}^2\lesssim \|\nabla^{2l}\mathfrak{T}\|_{L^2_G}^2+\|\Lap\nabla^{2l}\mathfrak{T}\|_{L^2_G}^2,\]
where the latter tensor is precisely $\Lap\tilde{\mathfrak{T}}$ which we just treated, and the former is covered by the induction assumption at order $2l$. So, \eqref{eq:Sobolev-norm-equiv-T2l} now follows for $l+1$, and thus by iteration up to $l=10$.\\

\noindent The proof of \eqref{eq:Sobolev-norm-equiv-T2l+1} is analogous -- we note that since we actually only needed a strong estimate on $\|\Ric[G]+2G\|_{C^9_G}$ for the previous inequality, but \eqref{eq:APmidRic} holds at $C^{10}_G$, this gives enough room to extend the argument in full despite the extra derivative order.\\
For both, note that we only need to estimate the Ricci terms in the $L^2_G$-norm if one cannot apply the a priori estimate \eqref{eq:APmidRic} to all $\nabla^{I_\Ric}\Ric[G]$ that occur in \eqref{eq:lap-sob-equiv-crucial}, and thus we could easily adjust the proof such that the Ricci energy does not occur in any of the proofs as long as $2l-1\leq 10$ is satisfied, so for $l\leq 5$.\\

The estimates \eqref{eq:Sobolev-norm-equiv-zeta2l}-\eqref{eq:Sobolev-norm-equiv-nablazeta2l} are proved identically, the only difference being that one order of curvature less enters in the commutator terms in \eqref{eq:lap-sob-equiv-crucial}, leading to one order less in curvature in total. For \eqref{eq:Sobolev-norm-equiv-zetalow}, we note that we can avoid incurring any $L^2$-norm by carefully repeating the argument we made for $\tilde{\mathfrak{T}}$ using \eqref{eq:[Lap,nabla]SF}:
\begin{align*}
\int_M\lvert\nabla^2\zeta\rvert_G^2\,\vol{G}
=&\,\int_M-\langle\nabla\zeta,\nabla\Lap\zeta\rangle_G\,\vol{G}+\int_M\Ric[G]\ast\nabla\zeta\ast\nabla\zeta\,\vol{G}\\
\lesssim&\,\int_M\lvert\Lap\zeta\rvert_G^2\,\vol{G}+a^{-c\sqrt{\epsilon}}\int_M\lvert\nabla\zeta\rvert_G^2\,\vol{G}
\end{align*}
\end{proof}

\section{Big Bang stability: Elliptic lapse estimates}\label{sec:lapse}

In this section, we study the elliptic structure of the equations \eqref{eq:REEqLapse1}-\eqref{eq:REEqLapse2}, which admit estimates controlling (time-scaled) lapse energies by other energy quantities. To this end, we recast these equations as follows:

\begin{definition}[Elliptic operators] For any (sufficiently regular) scalar function $\zeta$ on $\Sigma_t$, we define the differential operators
\begin{subequations}
\begin{align}
\L \zeta=&\,a^4\Lap \zeta-f\cdot \zeta,&f=\frac13a^4+12\pi C^2+\underbrace{\langle\Sigma,\Sigma\rangle_G+8\pi\Psi^2+16\pi C\Psi}_{=:F}\,, \\
\Ltilde \zeta=&\,a^4\Lap\zeta-\tilde{f}\cdot\zeta,&\quad \tilde{f}=\frac13a^4+12\pi C^2+\underbrace{a^4\left[R[G]+\frac23-8\pi\lvert\nabla\phi\rvert^2_G\right]}_{=\tilde{F}}\,.
\end{align}
\end{subequations}
\end{definition}
\noindent Note that the lapse equations \eqref{eq:REEqLapse1}, respectively \eqref{eq:REEqLapse2}, now read
\begin{equation}\label{eq:lapse-with-op}
\L N=F,\ \text{respectively}\ \Ltilde N=\tilde{F}.\,
\end{equation}
Furthermore, observe that
\begin{subequations}
\begin{align}
\label{eq:[L,Lap]}[\L,\Lap]\zeta=&\,\Lap f\cdot\zeta+2\langle\nabla f,\nabla\zeta\rangle_G=\change{\Lap F}\cdot\zeta+2\langle\nabla F,\nabla\zeta\rangle_G\,,\\
\label{eq:[Ltilde,Lap]}[\Ltilde,\Lap]\zeta=&\,\change{\Lap \tilde{F}}\cdot\zeta+2\langle\nabla \tilde{F},\nabla\zeta\rangle_G\,.
\end{align}
\end{subequations}

\subsection{Elliptic lapse estimates with $\L$}\label{subsec:lapse-L}

We first study the elliptic operator $\L$, which will admit weak lapse energy estimates in terms of scalar field quantities and $\Sigma$, up to curvature errors, that can in particular be utilized at high orders without having to resort to higher derivative levels. Before moving on to the estimates themselves, we collect a couple of inequalities we can deduce from the bootstrap assumptions and strong $C_G$-norm estimates.

\begin{remark}\label{ass:lapse}
There exists a constant $K>0$ such that, for $\epsilon>0$ small enough, the following estimates hold:
\begin{itemize}
\item $F\geq -K\epsilon$, and equivalently $f\geq 12\pi C^2-K\epsilon$. This is ensured by \eqref{eq:APSigma} and \eqref{eq:APPsi}. In particular, we can assume $\epsilon$ to have been small enough such that $f-6\pi C^2$ can be bounded from below by a positive constant that is independent of $\epsilon$ (for example $3\pi C^2$).
\item $\lvert\nabla f\rvert_{G}=\lvert\nabla F\rvert_{G}\leq  K\epsilon a^{-c\sqrt{\epsilon}}$. This is given by \eqref{eq:APmidSigma} and \eqref{eq:APmidPsi}.
\end{itemize}
\end{remark}

\begin{lemma}[Elliptic estimates with $\L$]\label{lem:ell-lapse-1}
Consider scalar functions $\zeta,Z$ on $\Sigma_t$ that satisfy
\begin{equation}\label{eq:ell-L-eq}\L \zeta=Z\,.\end{equation} Then,
\begin{equation}\label{eq:ell-L-est}
a^4\|\Lap \zeta\|_{L^2_G}+a^2\|\nabla \zeta\|_{L^2_G}+\|\zeta\|_{L^2_G}\lesssim \|Z\|_{L^2_G}
\end{equation}
\end{lemma}

\begin{proof} The proof follows along the same lines as that of \cite[Lemma 16.5]{Speck2018}: First, we obtain the following by \change{multiplying \eqref{eq:ell-L-eq} with $-\zeta$ and integrating}:
 and using that $f-6\pi C^2$ is bounded from below by a positive constant (see the first point in Remark \ref{ass:lapse}):
\[\int_M (a^4\lvert\nabla \zeta\rvert_G^2+\lvert\zeta\rvert^2)\vol{G}\lesssim \|Z\|^2_{L^2_G}
\]
Next, we \change{multiply }\eqref{eq:ell-L-eq} with $a^4\Lap \zeta$ and obtain
\begin{align*}
\int_{M}\left(a^8\lvert\Lap \zeta\rvert^2+a^4\lvert\nabla \zeta\rvert_G^2f\right)\,\vol{G}
&\leq\int_M\frac12\lvert Z\rvert^2+\frac{a^8}2\lvert\Lap \zeta\rvert^2+\left(\frac12\lvert\zeta\rvert^2+\frac{a^4}2\lvert\nabla \zeta\rvert_G^2\right)a^2\|\nabla f\|_{L^\infty_G}\,\vol{G}
\end{align*}
Using the second point in Remark \ref{ass:lapse} as well as the previous step, we can now conclude
\begin{align*}
\int_{M} \left(a^8\lvert\Lap \zeta\rvert^2+a^4\lvert\nabla \zeta\rvert_G^2\right)\,\vol{G}
&\lesssim (1+\epsilon)\|Z\|_{L^2_G}^2
\end{align*}
and thus the statement after rearranging.
\end{proof}

\begin{corollary}[Intermediary elliptic lapse estimate with $\L$]\label{cor:ell-lapse-L}
The following estimates hold for any $l\in\{0,\dots,\change{10}\}$:
\begin{align*}
&a^4\|\Lap^{l+1} N\|_{L^2_G}+a^2\|\nabla\Lap^l N\|_{L^2_G}+\|\Lap^lN\|_{L^2_G}\\
\lesssim&\,\|\Lap^lF\|_{L^2_G}+\underbrace{\epsilon a^{-c\sqrt{\epsilon}}\|F\|_{H^{2(l-1)}_G}}_{\text{not present for }l=0}+\underbrace{\epsilon^2a^{4-c\sigma}\sqrt{\E^{(\leq 2l-4)}(\Ric,\cdot)}}_{\text{not present for }l\leq 1}\\
\end{align*}
\end{corollary}
\begin{proof}
We prove the statement by induction over $l\in\N$: For $l=0$, the estimates immediately follow from \eqref{eq:lapse-with-op} and Lemma \ref{lem:ell-lapse-1}. 
Assume the statement to be satisfied up to $l-1$ for some $l\in\N_0,\,l\leq 11$. We get, applying \eqref{eq:[L,Lap]} iteratively,
\begin{align*}
\L\Lap^lN
=&
\sum_{I=1}^{2l-1}\nabla^IF\ast\nabla^{2l-I}N+(N+1)\Lap^lF\,.
\end{align*}
Applying Lemma \ref{lem:ell-lapse-1} to $\zeta=\Lap^l N$ as well as Lemma \ref{lem:lapse-maxmin} yields
\[a^4\|\Lap^{l+1}N\|_{L^2_G}+a^2\|\nabla\Lap^{l}N\|_{L^2_G}+\|\Lap^l N\|_{L^2_G}\lesssim \sum_{I=1}^{2l-1}\|\nabla^IF\ast\nabla^{2l-I}N\|_{L^2_G}+\|\Lap^lF\|_{L^2_G}\]
Hence, using \eqref{eq:Sobolev-norm-equiv-zeta2l+1} (replacing $l$ with $l-1$) and  \eqref{eq:BsN}, \eqref{eq:APmidPsi} and \eqref{eq:APmidSigma} to estimate low order terms, we get
\begin{align*}
&a^4\|\Lap^{l+1}N\|_{L^2_G}+a^2\|\nabla\Lap^{l}N\|_{L^2_G}+\|\Lap^l N\|_{L^2_G}\\ 
\lesssim&\,\|\Lap^lF\|_{L^2_G}+\epsilon a^{-c\sqrt{\epsilon}}\left(\sum_{m=0}^{l-1}\|\nabla\Lap^mN\|_{L^2_G}+\epsilon a^{4-c\sigma}\sqrt{\E^{(\leq 2l-4)}(\Ric,\cdot)}\right)\\
&\,+\epsilon a^{4-c\sigma}\left(\|F\|_{H^{2(l-1)}_G}+\|\nabla\Lap^{l-1}F\|_{L^2_G}+\epsilon a^{-c\sqrt{\epsilon}}\sqrt{\E^{(\leq 2l-4)}(\Ric,\cdot)}\right)
\end{align*}
For the top order lapse term, we can redistribute the divergent prefactor as follows:
\begin{equation*}
\epsilon a^{-c\sqrt{\epsilon}}\|\nabla\Lap^{l-1}N\|_{L^2_G}\lesssim \epsilon \|\Lap^lN\|_{L^2_G}+\epsilon a^{-2c\sqrt{\epsilon}}\|\Lap^{l-1}N\|_{L^2_G}
\end{equation*}
The lower order lapse terms as well as $\|\nabla\Lap^{l-1}F\|_{L^2_G}$ can be estimated similarly, just without having to redistribute the prefactor. Updating $c>0$ and rearranging then yields the statement at order $l$ for suitably small $\epsilon>0$, and thus the entire statement after iteration.
\end{proof}

\begin{corollary}[Lapse energy estimates with $\L$]\label{cor:en-est-lapse}
For any $l\in\{0,\dots,\change{9}\}$, one has
\begin{align*}
&\,a^8\E^{(2(l+1))}(N,\cdot)+a^4\E^{(2l+1)}(N,\cdot)+\E^{(2l)}(N,\cdot)\\
\lesssim&\,\epsilon^2\E^{(2l)}(\Sigma,\cdot)+\E^{(2l)}(\phi,\cdot)+\underbrace{\epsilon^2a^{-c\sqrt{\epsilon}}\left[\E^{(\leq 2(l-1))}(\Sigma,\cdot)+\E^{(\leq 2(l-1))}(\phi,\cdot)\right]}_{\text{not present for }l=0}\\
&\,+\underbrace{\left(\epsilon^4a^{-c\sqrt{\epsilon}}+\epsilon^2a^{8-c\sigma}\right)\E^{(\leq 2l-3)}(\Ric,\cdot)}_{\text{not present for }l\leq 1}
\end{align*}
\end{corollary}
\begin{proof}
Note that, by Corollary \ref{cor:ell-lapse-L}, all that needs to be done is to relate all Sobolev norms of $F$ that occur to the respective energies. Schematically, we have
\[\Lap^{l}F=16\pi(\Lap^l\Psi)(\Psi+C)+2\langle\Lap^l\Sigma,\Sigma\rangle_G+\sum_{I=1}^{2l-1}\left(\nabla^I\Psi\ast\nabla^{2l-I}\Psi+\nabla^I\Sigma\ast\nabla^{2l-I}\Sigma\right)\,.\]
For the first two terms, we can use \eqref{eq:APPsi} and \eqref{eq:APSigma} to bound $\lvert\Sigma\rvert_G$ and $\lvert\Psi+C\rvert$ by $\epsilon$ and $1$ up to constant, respectively. For the remaining terms, we similarly always bound the lower order in $L^\infty_G$ with \eqref{eq:APmidPsi}-\eqref{eq:APmidSigma} and bound the higher order with the energy estimates in Lemma \ref{lem:Sobolev-norm-equivalence-improved}. Further, we can use \eqref{eq:ibp-trick} to redistribute divergent prefactors onto energies of order $l-2$ and lower. 
This already incurs the terms on the right hand side of the claimed estimate, and the lower order norms of $F$ only incur at equivalent or weaker error terms.
\end{proof}

\subsection{Elliptic lapse estimates with $\Ltilde$}\label{subsec:lapse-Ltilde}

While the estimates in the previous subsection are useful at high orders, they are not enough to close the bootstrap assumptions for $N$. This can be achieved by deriving estimates in terms of $\Ltilde$ -- however, due to the explicit presence of Ricci terms in this version of the lapse equation, we use this to bound $N$ at \change{lower }orders. Since the arguments are largely identical to the ones above, we only sketch the proofs.

\begin{remark}\label{rem:L-Ltilde-bridge}
Note that, when replacing $f$ by $\tilde{f}$ and $F$ by $\tilde{F}$ in Remark \ref{ass:lapse}, the same statements hold for a suitable constant $K$. In fact, the bootstrap assumptions on $\Ric[G]$ and $\nabla\phi$ even imply $\|\tilde{F}\|_{C^1_G}\lesssim \epsilon a^{4-c\sigma}$, noting
\[\left\lvert R[G]+\frac23\right\rvert_G\leq\lvert G^{-1}\rvert_G\left\lvert\Ric[G]+\frac29G\right\rvert_G\lesssim \left\lvert\Ric[G]+\frac29G\right\rvert_G\]
and
\[\lvert\nabla R[G]\rvert_G=\left\lvert\nabla\left(R[G]+\frac23\right)\right\rvert_G\lesssim \left\lvert\nabla\left(\Ric[G]+\frac29G\right)\right\rvert_G\,.\]
\end{remark}

\begin{lemma}\label{lem:ell-lapse-Ltilde}
Any scalar functions $\zeta$ and $Z$ such that
\[\Ltilde \zeta=Z\]
holds satisfy the estimate
\begin{align*}
a^4\|\Lap\zeta\|_{L^2_G}+a^2\|\nabla \zeta\|_{L^2_G}+\|\zeta\|_{L^2_G}\lesssim&\,\|Z\|_{L^2_G}\,.
\end{align*}
\end{lemma}
\begin{proof}
The proof follows identically to Lemma \ref{lem:ell-lapse-1} 
 since all tools relating to $f$ and $F$ used in proving these statements were collected in Remark \ref{ass:lapse}, and these extend to $\Ltilde$ by Remark \ref{rem:L-Ltilde-bridge}.
\end{proof}

\begin{corollary}For $l\in\{0,\dots,\change{8}\}$
\label{cor:en-est-lapse-tilde}
\begin{align*}
a^8\E^{(2(l+1))}(N,\cdot)+a^4\E^{(2l+1)}(N,\cdot)+\E^{(2l)}(N,\cdot)\lesssim&\,a^8\E^{(\leq 2l)}(\Ric,\cdot)+{\epsilon}a^{8-c\sqrt{\epsilon}}\|\nabla\phi\|_{H^{2l}_G}^2
\end{align*}
\end{corollary}
\begin{proof} As in the proof of Corollary \ref{cor:ell-lapse-L}, this follows by commuting $\Ltilde$ with $\Lap^l$ iteratively and applying \eqref{eq:APmidphi} and \eqref{eq:APmidRic} to bound lower order terms within the nonlinearities.
\end{proof}

\section{Big Bang stability: Energy and norm estimates}\label{sec:en-est}

In this section, we derive energy estimates for matter variables and the geometric quantities as well as Sobolev norm estimates for spatial derivatives of $\phi$ and for metric quantities. To derive all of the inequalities in this section beside the elliptic inequality in Lemma \ref{lem:en-est-Sigma-top} \change{and the bound on $\nabla\phi$ in Lemma \ref{lem:norm-est-nablaphi}}, we will use the same basic strategy. Hence, we give a brief overview on the form our integral inequalities are going to take and how we intend to obtain \change{improved energy bounds }from there:

\begin{remark}[Integral inequalities and the Gronwall argument]\label{rem:en-est-strat}
Let $\mathcal{F}_L$ denote an energy or a squared Sobolev(-type) norm at derivative level $L\in2\N$, for example, $\E^{(L)}(\phi,\cdot)$. To derive an integral inequality for $\mathcal{F}_L$, we will take its time derivative, apply the respective commuted evolution equations in the integrand, estimate the resulting terms and integrate that inequality. Schematically, the resulting integral inequalities for $\mathcal{F}_L$ then take the following form:
\begin{align*}
&\,\mathcal{F}_L(t)+\int_t^{t_0}\langle\text{ultimately nonnegative contributions}\rangle\,ds\\
\lesssim&\,\mathcal{F}_L(t_0)+\int_t^{t_0}\left(\epsilon^\frac18a(s)^{-3}+a(s)^{-1-c\sqrt{\epsilon}}\right)\mathcal{F}_L(s)\,ds\\
&\,+\int_t^{t_0}a(s)^{-3}\langle\text{other energies/squared Sobolev norms at same derivative level}\rangle\,ds\\
&\,+\int_t^{t_0}a(s)^{-3-c\sqrt{\epsilon}}\langle\text{energies/squared Sobolev norms at derivative levels up to }L-2\rangle\,ds
\end{align*}

For some inequalities, we will not be able to derive any \changefinal{beneficial }$\epsilon$-prefactors in the penultimate line. \change{For example, for $\E^{(L)}(\Sigma,\cdot)$, linear lapse terms in the evolution of $\Sigma$ incur a term of the form
\[\int_t^{t_0}a(s)^{-3}\cdot a(s)^4\|\Lap^\frac{L}2N\|_{\dot{H}^2_G}\cdot\sqrt{\E^{(L)}(\Sigma,s)}\,ds\]
on the right hand side, which after applying lapse energy estimates creates $\epsilon^{-\frac18}\E^{(L)}(\phi,\cdot)$ on the right}. However, combining the respective inequalities for the core energy mechanism at each derivative level with appropriate \changefinal{$\epsilon$-weights}, this will then combine to an inequality of the following form for a total energy, which we informally denote by $\mathcal{F}_{\text{total},L}$:

\change{\begin{align*}
\numberthis\label{eq:toy-energy}&\,\mathcal{F}_{\text{total},L}(t)+\int_t^{t_0}\langle\text{nonnegative quantity}\rangle\,ds\\
\lesssim&\,\mathcal{F}_{\text{total},L}(t_0)+\int_t^{t_0}\left(\epsilon^\frac18a(s)^{-3}+a(s)^{-1-c\sqrt{\epsilon}}\right)\mathcal{F}_{\text{total},L}(s)\,ds\\
&\,+\underbrace{\int_t^{t_0}\epsilon^\frac18 a(s)^{-3-c\epsilon^\frac18}\langle\text{already improved terms}\rangle\,ds}_{\text{not present for }L=0}\\
&\,+\epsilon^\frac14\mathcal{F}_{\text{total},L}(t)+\underbrace{\sqrt{\epsilon}\cdot\langle\text{small lower order terms}\rangle(t)}_{\text{not present for }L=0}
\end{align*}
In the mentioned example, \changefinal{multiplying $\E^{(L)}(\Sigma,\cdot)$ with the weight $\epsilon^\frac14$, in turn, mitigates the otherwise offending term to $\epsilon^\frac18a(s)^{-3}\E^{(L)}(\phi,s)$, which can be absorbed into the first line.\\}}

\changefinal{Furthermore, $\epsilon^\frac14\mathcal{F}_{\text{total},L}(t)$ in the penultimate line of (6-1) can be absorbed into the left-hand side after updating the implicit constant in \enquote{$\lesssim$}. Applying the Gronwall lemma (see Lemma \ref{lem:gronwall}) and the initial data assumption (which implies $\mathcal{F}_{\text{total},L}(t_0)\lesssim \epsilon^4$) then yields
\begin{align*}
\mathcal{F}_{\text{total},L}(t)\lesssim&\,\biggr(\epsilon^4+\underbrace{\int_t^{t_0}\epsilon^\frac18 a(s)^{-3-c\epsilon^\frac18}\langle\text{already improved terms}\rangle\,ds+\sqrt{\epsilon}\cdot\langle\text{lower order terms}\rangle(t)}_{\text{not present for }L=0}\biggr)\\
&\,\cdot\exp\left(K\cdot \int_t^{t_0}\epsilon^\frac18a(s)^{-3}+a(s)^{-1-c\sqrt{\epsilon}}\,ds\right)
\end{align*}
for some constant $K>0$. \eqref{eq:log-est} and \eqref{eq:a-integrals}, the exponential factor can be bounded by $a^{-c\epsilon^\frac18}$, up to constant and updating $c>0$.
 Hence}, for $L=0$, this implies $\mathcal{F}_{\text{total},0}\lesssim\change{\epsilon^4}a^{-c\epsilon^\frac18}$, and thus leads to improved bounds for base level energy quantities (see Remark \ref{rem:bs-strategy} for the precise scaling hierarchy that will achieve). 
By iterating this argument for $L>0$, the already improved terms will then be bounded (at worst) by $\change{\epsilon^4}a^{-c\epsilon^\frac18}$ , and \eqref{eq:a-integrals} shows that the first line can be bounded by $\change{\epsilon^4}a^{-c\epsilon^\frac18}$ after updating $c$. \change{This allows us }to bound $\mathcal{F}_{\text{total},L}$ by $\change{\epsilon^4}a^{-c\epsilon^\frac18}$ for any $L$ up to and including top order. \\

\change{Finally, we mention that, to control energies at order $L$, we need to consider scaled energies at order $L+1$ within $\mathcal{F}_{total,L}$ -- this arises since the scalar field occurs at first order in the evolution equations for $\RE$ and $\RB$. We avoid losing derivatives by employing the div-curl-estimate in Lemma \ref{lem:en-est-Sigma-top} at order $L+1$, which allows us to control $a^4\E^{(L+1)}(\Sigma,\cdot)$ by quantities at order $L$. This is precisely what generates the non-integral terms in the schematics above. We note that it is crucial that the scalar field occurs at no worse scaling than $a^{-1}$ in \eqref{eq:comeq-RE}-\eqref{eq:comeq-RB} -- else, moving to these time-scaled estimates at order $L+1$ would lose too many powers of $a$ and lead to exponentially divergent terms after applying the Gronwall argument.}
\end{remark}
Recall that $L^2_G$-norm estimates for error terms arising in the Laplace-commuted equations in Lemma \ref{lem:laplace-commuted-eq} are collected in Section \ref{subsec:L2-error-est}. Low order estimates (in particular estimates for $L=2$) could often be improved if needed by more carefully avoiding curvature error terms, but we refrain from doing so where it is not necessary to keep estimates as unified as possible. \delete{Moreover [...]}

\subsection{Integral and energy estimates for the scalar field}\label{subsec:en-SF}

\subsubsection{Scalar field energy estimates}\label{subsubsec:en-SF}

Over the following two lemmas, we prove the core energy estimates to control the matter variables, which are immediately prepared differently at base, intermediate and top order for the total energy estimates in Section \ref{sec:bs-imp}.

\begin{lemma}\label{lem:en-est-SF}
[\change{Even order scalar field energy estimates}] Let $t\in(t_{Boot},t_0]$. Then, one has
\begin{align*}
\numberthis\label{eq:en-est-SF0}&\,\E^{(0)}(\phi,t)+\int_t^{t_0}\dot{a}(s)a(s)^{3}\E^{(1)}(N,s)+\frac{\dot{a}(s)}{a(s)}\E^{(0)}(N,s)\,ds\\
\lesssim&\,\E^{(0)}(\phi,t_0)+\int_t^{t_0}\epsilon a(s)^{-3}\E^{(0)}(\phi,s)+\epsilon a^{-3}\E^{(\change{0})}(\Sigma,s)\,ds\,.
\end{align*}
Further, for any $L\in 2\N,\,2\leq L\leq \change{18}$, the following estimate is satisfied: 
\begin{align*}
\numberthis\label{eq:en-est-SF}&\,\E^{(L)}(\phi,t)+\int_t^{t_0}\dot{a}(s)a(s)^{3}\E^{(L+1)}(N,s)+\frac{\dot{a}(s)}{a(s)}\E^{(L)}(N,s)\,ds\\
\lesssim&\,\E^{(L)}(\phi,t_0)+\int_t^{t_0}\left(\change{\epsilon}a(s)^{-3}+a(s)^{-1-c\sqrt{\epsilon}}\right)\E^{(L)}(\phi,s)\,ds\\
&\,+\int_t^{t_0}\epsilon a(s)^{-3}\E^{(L)}(\Sigma,s)+\epsilon^\frac32 a(s)^{-3}\E^{(L-2)}(\Ric,s)\,ds\\
&\,+\int_t^{t_0}\sqrt{\epsilon}a(s)^{-3-c\sqrt{\epsilon}}\E^{(\leq L-2)}(\phi,s)+\epsilon a(s)^{-3-c\sqrt{\epsilon}}\E^{(\leq L-2)}(\Sigma,s)\,ds\\
&\underbrace{+\int_t^{t_0}\epsilon^\frac32a(s)^{-3-c\sqrt{\epsilon}}\E^{(\leq L-4)}(\Ric,s)\,ds}_{\text{if }L\geq4}
\end{align*}
\end{lemma}
\begin{remark}
This proof relies on two mechanisms: Firstly, we use the structure of the wave equation and integration by parts to cancel the highest order scalar field derivative terms. Getting this cancellation is what necessitates scaling the potential term in the scalar field energy by $a^4$. Secondly, we deal with the highest order lapse terms using the elliptic structure of the (Laplace-commuted) lapse equation -- both in an indirect way by invoking the elliptic energy estimate in Corollary \ref{cor:en-est-lapse} as well as by directly inserting \eqref{eq:comeq-lapse} to cancel some ill-behaved terms. While the framework significantly differs from the scalar field energy estimates \cite{Speck2018}, these two core mechanisms also appear there and play similarly crucial roles.
\end{remark}
\begin{proof}
Since the arguments are essentially the same, we will only write down the proof for $L\geq 2$ in full and make short comments throughout the argument which terms do not occur for $L=0$.\\
We use the evolution equations \eqref{eq:comeq-Psi-even} and \eqref{eq:comeq-nablaphi-even} and Lemma \ref{lem:delt-int} to compute, for $L\geq 2$,
\begin{subequations}
\begin{align*}
-\del_t\E^{(L)}(\phi,\cdot)=\int_M&\,-2\del_t\Lap^{\frac{L}2}\Psi\cdot\Lap^{\frac{L}2}\Psi-2a^4\langle\del_t\nabla\Lap^{\frac{L}2}\phi,\nabla\Lap^{\frac{L}2}\phi\rangle_G-a^4(\del_tG^{-1})^{ij}\nabla_i\Lap^{\frac{L}2}\phi\nabla_j\Lap^{\frac{L}2}\phi\\
&\,-3N\frac{\dot{a}}a\left[\lvert\Lap^{\frac{L}2}\Psi\rvert^2+a^4\lvert\nabla\Lap^\frac{L}2\phi\rvert_G^2\right]-4\frac{\dot{a}}a\cdot a^4\lvert\nabla\Lap^\frac{L}2\phi\rvert_G^2\,\vol{G}\\
\numberthis\label{eq:diff-eq-SF1}=\int_M&\, \left(-2a(N+1)\Lap^{\frac{L}2+1}\phi-2a\langle\nabla\Lap^{\frac{L}2}N,\nabla\phi\rangle_G+6C\frac{\dot{a}}a\Lap^{\frac{L}2}N\right)\cdot(\Lap^{\frac{L}2}\Psi)\\
\numberthis\label{eq:diff-eq-SF2}&\,-2a(N+1)\langle\nabla\Lap^{\frac{L}2}\Psi,\nabla\Lap^{\frac{L}2}\phi\rangle_G-2Ca\langle\nabla\Lap^{\frac{L}2}N,\nabla\Lap^{\frac{L}2}\phi\rangle_G\\
\numberthis\label{eq:diff-eq-SF3}&\,-2\left(\mathfrak{P}_{L,Border}+\mathfrak{P}_{L,Junk}\right)\cdot\Lap^{\frac{L}2}\Psi-2a^4\langle\mathfrak{Q}_{L,Border}+\mathfrak{Q}_{L,Junk},\nabla\Lap^{\frac{L}2}\phi\rangle_G\\
\numberthis\label{eq:diff-eq-SF4}&\,+2(N+1)a\cdot(\Sigma^\sharp)^{ij}\nabla_i\Lap^{\frac{L}2}\phi\nabla_j\Lap^{\frac{L}2}\phi-2N\frac{\dot{a}}a\cdot a^4\lvert\nabla\Lap^{\frac{L}2}\phi\rvert_G^2\\
\numberthis\label{eq:diff-eq-SF5}&\,-3N\frac{\dot{a}}a\left[\lvert\Lap^{\frac{L}2}\Psi\rvert^2+a^4\lvert\nabla\Lap^\frac{L}2\phi\rvert_G^2\right]-4\frac{\dot{a}}a\cdot a^4\lvert\nabla\Lap^\frac{L}2\phi\rvert_G^2\,\vol{G}\,.
\end{align*}
\end{subequations}
Note that, for $L=0$, the equivalent equality holds where the borderline and junk terms are replaced by $-2a^4\Psi\langle\nabla N,\nabla\phi\rangle_G$ (to verify this, insert \eqref{eq:REEqWave} and \eqref{eq:REEqNablaPhi} instead of \eqref{eq:comeq-Psi-even} and \eqref{eq:comeq-nablaphi-even}). We now go through \eqref{eq:diff-eq-SF1}-\eqref{eq:diff-eq-SF5} term by term:\\
After integrating by parts, the first term in \eqref{eq:diff-eq-SF1} reads
\begin{equation}\label{eq:SF-cancellation-trick}
\int_M 2a(N+1)\langle\nabla\Lap^{\frac{L}2}\phi,\nabla\Lap^{\frac{L}2}\Psi\rangle_G+2a\langle\nabla N,\nabla\Lap^{\frac{L}2}\phi\rangle_G\cdot\Lap^{\frac{L}2}\Psi\,\vol{G}\,.
\end{equation}
The first term \textbf{precisely} cancels the first term in \eqref{eq:diff-eq-SF2}, while we can use the bootstrap assumption \eqref{eq:BsN} to estimate the other term in \eqref{eq:SF-cancellation-trick} up to constant by
\begin{equation*}
\epsilon a^{3-c\sigma}\cdot a^2\|\nabla\Lap^{\frac{L}2}\phi\|_{L^2_G}\|\Lap^{\frac{L}2}\Psi\|_{L^2_G}\lesssim\epsilon a^{3-c\sigma}\E^{(L)}(\phi,\cdot)\,.
\end{equation*}
For the second term in \eqref{eq:diff-eq-SF1}, we use \eqref{eq:APmidphi} to estimate $\nabla\phi$ and Corollary \ref{cor:en-est-lapse} at order $L$ to deal with the lapse, getting
\begin{align*}
\left\lvert\int_M 2a\langle\nabla\Lap^{\frac{L}2}N,\nabla\phi\rangle_G\cdot \Lap^{\frac{L}2}\Psi\,\vol{G}\right\rvert\lesssim&\,\sqrt{\epsilon}a^{-1-c\sqrt{\epsilon}}\sqrt{a^4\E^{(L+1)}(N,\cdot)}\sqrt{\E^{(L)}(\phi,\cdot)}\\
\lesssim&\,\epsilon a^{-1}\cdot a^4\E^{(L+1)}(N,\cdot)+a^{-1-c\sqrt{\epsilon}}\E^{(L)}(\phi,\cdot)\\
\lesssim&\,\epsilon^3a^{-1}\E^{(L)}(\Sigma,\cdot)+a^{-1-c\sqrt{\epsilon}}\E^{(L)}(\phi,\cdot)\\
&\,+\epsilon^2a^{-1-c\sqrt{\epsilon}}\E^{(\leq L-2)}(\Sigma,\cdot)+\epsilon^2a^{-1-c\sqrt{\epsilon}}\E^{(\leq L-2)}(\phi,\cdot)\\
&\,+\underbrace{\epsilon^2a^{-1-c\sqrt{\epsilon}}\E^{(\leq L-3)}(\Ric,\cdot)}_{\text{not present for }L=2}\,.
\end{align*}
Repeating this argument for $L=0$, the last two lines do not appear.\\

To deal with the remaining term in \eqref{eq:diff-eq-SF1}, we can insert the following zero on the right hand side of the differential equality, where the equality \eqref{eq:SF-en-lapse-trick} holds due to \eqref{eq:comeq-lapse}:
\begin{align*}
0=&\,-\frac3{8\pi}\dot{a}a^3\int_M\div_G\left(\nabla\Lap^{\frac{L}2} N\cdot \Lap^{\frac{L}2}N\right)\,\vol{G}\\
=&\,-\frac3{8\pi}\dot{a}a^3\int_M \Lap^{{\frac{L}2}+1}N\cdot\Lap^\frac{L}2N+\lvert\nabla\Lap^{\frac{L}2}N\rvert_G^2\,\vol{G}\\
=&\numberthis\label{eq:SF-en-lapse-trick}\,\int_M -\frac3{8\pi}\dot{a}a^3\lvert\nabla\Lap^{\frac{L}2}N\rvert_G^2-\frac{3}{8\pi}\left(\frac13\dot{a}a^3+12\pi C^2\frac{\dot{a}}a\right)\lvert\Lap^{\frac{L}2}N\rvert^2\\
&\quad -6 C\frac{\dot{a}}a\Lap^{\frac{L}2}N\cdot \Lap^{\frac{L}2}\Psi-\frac3{8\pi}\dot{a}a^3\left[\mathfrak{N}_{L,Border}+\mathfrak{N}_{L,Junk}\right]\cdot\Lap^{\frac{L}2}N\,\vol{G}
\end{align*}
Note that the first line has a negative sign, so (after absorbing a few terms into it without changing the sign, see namely lapse quantities in \eqref{eq:en-est-SF-Friedman-ineq} and \eqref{eq:diff-ineq-SF1}), we pull it to the left hand side of the differential inequality. Further, the first term in the second line of \eqref{eq:SF-en-lapse-trick} precisely cancels the third term in \eqref{eq:diff-eq-SF1}. That leaves the borderline and junk terms in \eqref{eq:SF-en-lapse-trick}, for which we use \eqref{eq:L2-Border-N} and \eqref{eq:L2-junk-N} (along with $\dot{a}\simeq a^{-2}$ due to \eqref{eq:Friedman}) to get, for $L\geq 4$,
\begin{align*}
&\,\frac{3}{8\pi}\dot{a}a^3\left\lvert\int_M\left[\mathfrak{N}_{L,Border}+\mathfrak{N}_{L,Junk}\right]\cdot\Lap^{\frac{L}2}N\,\vol{G}\right\rvert\\
\lesssim&\,\epsilon a^{-3}\left[\E^{(L)}(\phi,\cdot)+\E^{(L)}(\Sigma,\cdot)+\E^{(L)}(N,\cdot)\right]\\
&+\,\epsilon a^{-3-c\sqrt{\epsilon}}\left[\E^{(\leq L-2)}(\phi,\cdot)+\E^{(\leq L-2)}(\Sigma,\cdot)+\E^{(\leq L-2)}(N,\cdot)\right]\\
&\,+\underbrace{\epsilon^3 a^{-3}\E^{(\leq L-2)}(\Ric,\cdot)+\epsilon^3 a^{-3-c\sqrt{\epsilon}}\E^{(\leq L-3)}(\Ric,\cdot)}_{\text{not present for }L=2}\,.
\end{align*}
Again, the same estimate holds for $L=0$ with the last two lines dropped.

From \eqref{eq:diff-eq-SF1}-\eqref{eq:diff-eq-SF2}, only the term $-2Ca\langle\nabla\Lap^{\frac{L}2}N,\nabla\Lap^\frac{L}2\phi\rangle_G$ still needs to be handled: Using the inequality \eqref{eq:diff-ineq-Friedman} arising from the Friedman equation, we can estimate this by
\begin{equation}\label{eq:en-est-SF-Friedman-ineq}
\int_M2\sqrt{\frac{3}{4\pi}}\dot{a}a^3\lvert\nabla\Lap^{\frac{L}2}N\rvert_G\lvert\nabla\Lap^{\frac{L}2}\phi\rvert_G\,\vol{G}\leq \int_M4\dot{a}a^3\lvert\nabla\Lap^\frac{L}2\phi\rvert_G^2+\frac{3}{16\pi}{\dot{a}}a^3\lvert\nabla\Lap^{\frac{L}2} N\rvert_G^2\,\vol{G}\,.
\end{equation}
Note that the first term precisely cancels the final term in \eqref{eq:diff-eq-SF5}, while the second term can be absorbed into the first term in \eqref{eq:SF-en-lapse-trick} while preserving that term's sign.\\

To bound the error terms in \eqref{eq:diff-eq-SF3}, we insert the borderline term estimates \eqref{eq:L2-Border-P-even} and \eqref{eq:L2-Border-Q-even} as well as the junk term estimates \eqref{eq:L2-junk-P-even} and \eqref{eq:L2-junk-Q-even}, where \eqref{eq:ibp-trick} is used to estimate odd order by even order energies where needed. Furthermore, observe that we can estimate the $\mathfrak{Q}_L$-terms as
\[\left(a^2\|\mathfrak{Q}_L\|_{L^2_G}\right)\cdot \sqrt{\E^{(L)}(\phi,\cdot)},\]
so all borderline and junk terms arising from it, beside the scalar field \change{energies}, are dominated by terms occurring elsewhere.\\
Finally, all terms that remain, namely \eqref{eq:diff-eq-SF4} and the first term in \eqref{eq:diff-eq-SF5}, can be bounded by $\epsilon a^{-3}\E^{(L)}(\phi,\cdot)$ due to the strong base level estimate \eqref{eq:APSigma} and \eqref{eq:BsN}. In summary, and always only keeping the worst terms for each energy and squared norm, this yields for $L\geq 4$:
\begin{subequations}
\begin{align*}
&-\del_t\E^{(L)}(\phi,\cdot)+\dot{a}a^{3}\E^{(L+1)}(N,\cdot)+\frac{\dot{a}}a\E^{(L)}(N,\cdot)\\
\numberthis\label{eq:diff-ineq-SF1}\lesssim&\,\left(\change{\epsilon}a^{-3}+a^{-1-c\sqrt{\epsilon}}\right)\E^{(L)}(\phi,\cdot)+\left(\epsilon a^{-3}+\sqrt{\epsilon}a^{-1-c\sqrt{\epsilon}}\right)\left(a^4\E^{(L+1)}(N,\cdot)+\E^{(L)}(N,\cdot)\right)\\
\numberthis\label{eq:diff-ineq-SF2}&\,+\epsilon a^{-3}\E^{(L)}(\Sigma,\cdot)+\epsilon^\frac32 a^{-3}\E^{(L-2)}(\Ric,\cdot)+
\change{\epsilon}a^{-3-c\sqrt{\epsilon}}\E^{(\leq L-2)}(\phi,\cdot)
\\
\numberthis\label{eq:diff-ineq-SF3}&\,\change{
+\epsilon a^{-3-c\sqrt{\epsilon}}\E^{(\leq L-2)}(\Sigma,\cdot)+\left[\epsilon a^{-3-c\sqrt{\epsilon}}+\sqrt{\epsilon}a^{-1-c\sqrt{\epsilon}}\right]\E^{(\leq L-2)}(N,\cdot)}\\
\numberthis\label{eq:diff-ineq-SF4}&\,\underbrace{
\change{+\epsilon^\frac32a^{-3-c\sqrt{\epsilon}}\E^{(\leq L-4)}(\Ric,\cdot)}}_{\text{not present for }L=2}
\end{align*}
\end{subequations}
The lapse energies in \eqref{eq:diff-ineq-SF1} can now also be absorbed into those on the left hand side of the inequality by updating the implicit constant in \enquote{$\lesssim$}. We can treat the lower order lapse energies in \change{\eqref{eq:diff-ineq-SF3} }with Corollary \ref{cor:en-est-lapse} and see that the resulting terms are all dominated by terms we already have on the right hand side of the inequality above.\\

\change{Inserting }these estimates and integrating over $(t,t_0]$ then yields \eqref{eq:en-est-SF} for $L\geq 4$, and the statement for $L=2$ is obtained completely analogously.

\noindent As mentioned earlier, \eqref{eq:diff-eq-SF3} is replaced by the following term for $L=0$:
\[\int_M -2\change{a}\Psi\langle\nabla N,\nabla\phi\rangle_G\,\vol{G}\lesssim\epsilon a^{-3}\int_Ma^2\lvert\nabla N\rvert_G\cdot a^2\lvert\nabla\phi\rvert_G\,\vol{G}\lesssim \epsilon \dot{a}a^3 \E^{(1)}(N,\cdot)+\epsilon a^{-3}\E^{(0)}(\phi,\cdot)\,,\]
Here, we applied \eqref{eq:APPsi} and \eqref{eq:Friedman}. Both of these terms can be absorbed into terms that are already present, and \eqref{eq:en-est-SF0} then follows by dealing with terms in $\del_t\E^{(0)}(\phi,\cdot)$ as described and integrating. 
\end{proof}

To close the \change{argument}, we will need a scaled scalar field energy estimate at the \change{odd orders $L+1$}, which is not covered by the previous lemma and we hence establish separately:


\begin{lemma}[Odd order scalar field energy estimate]\label{lem:en-est-SF-top}
For $L\in2\N$, $\change{2\leq L\leq \change{18}$}, we have:
\begin{align*}
\numberthis\label{eq:en-est-SF-top}&\,a(t)^4\E^{(L+1)}(\phi,t)+\int_t^{t_0}\left\{\dot{a}(s)a(s)^7\E^{(L+2)}(N,s)+\dot{a}(s)a(s)^3\E^{(L+1)}(N,s)\right\}\,ds\\
\lesssim&\,a(t_0)^4\E^{(L+1)}(\phi,t_0)+\int_t^{t_0}\left(\epsilon a(s)^{-3}+a(s)^{-1-c\sqrt{\epsilon}}\right)\cdot a(s)^4\E^{(L+1)}(\phi,s)\,ds\\
&\,+\int_t^{t_0}\left\{\epsilon a(s)^{-3}\cdot a(s)^4\E^{(L+1)}(\Sigma,s)+\left(\epsilon a(s)^{-3}+a(s)^{-1-c\sqrt{\epsilon}}\right)\E^{(L)}(\phi,s)+\epsilon a(s)^{-3}\E^{(L)}(\Sigma,s)\right.\\
&\,\phantom{+\int_t^{t_0}}+\epsilon a(s)^{-1-c\sqrt{\epsilon}}\cdot a(s)^4\E^{(L-1)}(\Ric,s)+\left(\epsilon^3a^{-3}+\epsilon a^{-1-c\sqrt{\epsilon}}\right)\E^{(L-2)}(\Ric,s)\\
&\,\phantom{+\int_t^{t_0}}+{\epsilon}a(s)^{-3-c\sqrt{\epsilon}}\left(\E^{(\leq L-2)}(\phi,s)+\E^{(\leq L-2)}(\Sigma,s)\right)\\
&\change{\,\phantom{+\int_t^{t_0}}\underbrace{\left.+\left(\epsilon^3a(s)^{-3-c\sqrt{\epsilon}}+\epsilon^2a(s)^{-1-c\sqrt{\epsilon}}\right)\E^{(\leq L-4)}(\Ric,s)\right\}\,ds}_{\text{not present for }L=2}}
\end{align*}
\change{At order $1$, the analogous estimate holds where the last three lines of \eqref{eq:en-est-SF-top} are dropped.}
\end{lemma}

\begin{proof}
\change{These estimates follow }completely analogously to Lemma \ref{lem:en-est-SF}, with the exception that high order lapse terms can now be estimated at order $L+2$ due to the scalar field energy being scaled by $a^4$. In particular, we note that to deal with the analogous term to \eqref{eq:diff-eq-SF1}, one now inserts the following zero on the right and applies the commuted lapse equation \eqref{eq:comeq-lapse-odd}:
\begin{align*}
0=&\,-\frac{3}{8\pi}\dot{a}a^7\int_M\div_G\left(\nabla\Lap^{\frac{L}2}N\cdot\Lap^{\frac{L}2+1}N\right)\,\vol{G}\\
=&\,\int_M\left\{-\frac{3}{8\pi}\dot{a}a^7\lvert\Lap^{\frac{L}2+1}N\rvert^2-\frac{3}{8\pi}\left(\frac13\dot{a}a^3+12\pi C^2\frac{\dot{a}}a\right)\cdot a^4\lvert\nabla\Lap^{\frac{L}2}N\rvert_G^2\right.\\
&\,\,\phantom{\int_M}\left.-6C\frac{\dot{a}}a\cdot a^4\langle\nabla\Lap^{\frac{L}2}N,\nabla\Lap^{\frac{L}2}\Psi\rangle_G-\frac{3}{8\pi}\dot{a}a^7\langle\mathfrak{N}_{L+1,Border}+\mathfrak{N}_{L+1,Junk},\nabla\Lap^{\frac{L}2}N\rangle_G\right\}\,\vol{G}
\end{align*}
\change{For $L=0$, the argument is again the same as at higher orders with less complicated junk terms. We briefly highlight some specific junk terms: The term analogous to \eqref{eq:diff-eq-SF3} is now estimated as follows using \eqref{eq:Sobolev-norm-equiv-zetalow}:
\begin{align*}
a^4\cdot \int_M-2a\langle\nabla\Psi\nabla N,\nabla^2\phi\rangle_G\lesssim&\,\epsilon \int_M a^\frac12\lvert \nabla N\rvert_G\cdot a^{\frac12-c\sqrt{\epsilon}}\cdot a^4\lvert \nabla\phi\rvert_G\\
\lesssim&\,\epsilon \dot{a}a^3\E^{(1)}(N,\cdot)+\epsilon a^{1-c\sqrt{\epsilon}}\cdot a^4\E^{(\leq 1)}(\phi,\cdot) 
\end{align*}
Further, note that, by the commutator formula \eqref{eq:[del-t,Lap]zeta} and applying \eqref{eq:APmidphi}, one has
\begin{align*}
\left\lvert\int_Ma^8[\del_t,\Lap]\phi\cdot\Lap\phi\,\vol{G}\right\rvert\lesssim&\,\epsilon a^{5-c\sqrt{\epsilon}}\left(\|\nabla\Sigma\|_{L^2_G}+\|\nabla N\|_{L^2_G}\right)\|\Lap\phi\|_{L^2_G}\\
\lesssim&\,\epsilon a^{-1-c\sqrt{\epsilon}}\left(a^4\E^{(1)}(\phi,\cdot)+a^4\E^{(1)}(\Sigma,\cdot)\right)+\epsilon a^{6-c\sigma}\cdot \dot{a}a^3\E^{(1)}(N,\cdot)\,.
\end{align*}
}
\end{proof}

\subsubsection{\change{Sobolev norm estimate for $\nabla\phi$}}\label{subsubsec:int-nabla-phi}

\change{To improve the bootstrap assumptions on $\nabla\phi$, we will need sharper bounds than those on $a^4\|\nabla\phi\|_{H^L}^2$ that will follow from bounds on $\E^{(L)}(\phi,\cdot)$:

\begin{lemma}\label{lem:norm-est-nablaphi}
Let $l\in(t_{Boot},t_0]$. Then, for $l\in\Z_+$, $l\leq 17$, the following estimate holds:
\begin{align*}
\|\nabla\phi\|_{H^l_G(\Sigma_t)}\lesssim (1+\epsilon a(t)^{-c\sqrt{\epsilon}})\|\Sigma\|_{H^{l+1}_G(\Sigma_t)}+\epsilon a(t)^{-c\sqrt{\epsilon}}\|\Psi\|_{H^{l}_G(\Sigma_t)}
\end{align*}
\end{lemma}
\begin{proof}
By \eqref{eq:APPsi}, $\Psi+C>\frac{C}2$ holds if $\epsilon$ is chosen small enough. Consequently, we can rearrange \eqref{eq:REEqMom} and apply the product rule to obtain
\[\lvert\nabla^l\nabla\phi\rvert_G=\frac{1}{8\pi}\left\lvert\nabla^l\left(\frac{\div_G\Sigma}{\Psi+C}\right)\right\rvert_G\lesssim \sum_{I_\Sigma+I_\Psi=l}\lvert\nabla^{I_\Sigma+1}\Sigma\rvert_G\lvert\nabla^{I_\Psi}\changefinal{(\Psi+C)}\rvert_G\,.\]
The statement then follows by integrating over $\Sigma_t$ and applying \eqref{eq:APPsi} and \eqref{eq:APSigma}.
\end{proof}
}

\subsection{Energy estimates for the Bel-Robinson variables}\label{subsec:en-BR}

In this subsection, we collect the energy estimates for the Bel-Robinson variables:

\begin{lemma}[\change{Bel-Robinson energy estimates}]\label{lem:en-est-BR} Let $t\in(t_{Boot},t_0]$. Then one has
\begin{align*}
\numberthis\label{eq:en-est-BR0}&\,\E^{(0)}(W,t)+\int_t^{t_0}\int_M\left[8\pi C^2a(s)^{-3}(N+1)\langle\Sigma,\RE\rangle_G+6\frac{\dot{a}(s)}{a(s)}\changefinal{(N+1)}\lvert\RE\rvert_G^2\right]\,\vol{G}\,ds\\
\lesssim&\,\E^{(0)}(W,t_0)+\int_t^{t_0}\left({\epsilon}a(s)^{-3}+a(s)^{-1-c\sqrt{\epsilon}}\right)\E^{(0)}(W,s)+\change{\epsilon^{-\frac18}a(s)^{-3}\cdot a(s)^4\E^{(1)}(\phi,s)}\,ds\\
&\,\phantom{\E^{(0)}(W,t_0)}+\int_t^{t_0}a(s)^{-1-c\sqrt{\epsilon}}\E^{(0)}(\phi,s)+\epsilon a(s)^{-3}\E^{(0)}(\Sigma,s)\,ds
\end{align*}
as well as, for $L\in 2\N,\,2\leq L\leq \change{18}$,
\change{\begin{align*}
\numberthis\label{eq:en-est-BR}&\,\E^{(L)}(W,t)+\int_t^{t_0}\int_M\left[8\pi C^2a(s)^{-3}(N+1)\langle\Lap^\frac{L}2\Sigma,\Lap^\frac{L}2\RE\rangle_G+6(N+1)\frac{\dot{a}(s)}{a(s)}\lvert\Lap^\frac{L}2\RE\rvert_G^2\right]\,\vol{G}\,ds\\
&\,\lesssim\E^{(L)}(W,t_0)+\int_t^{t_0}\left(\epsilon^\frac18 a(s)^{-3}+a(s)^{-1-c\sqrt{\epsilon}}\right)\E^{(L)}(W,s)\,ds\\
&\,+\int_t^{t_0}\left\{\epsilon^{-\frac18}a(s)^{-3}\cdot a(s)^4\E^{(L+1)}(\phi,s)+\left(\epsilon^\frac18 a(s)^{-3}+a(s)^{-1}\right)\E^{(L)}(\phi,s)\right.\\
&\,\phantom{\int_t^{t_0}}+\epsilon a(s)^{-3}\E^{(L)}(\Sigma,s)+\epsilon^\frac{7}8a(s)^{-3}\cdot a(s)^4\E^{(L-1)}(\Ric,s)+\epsilon^\frac{31}8 a(s)^{-3}\E^{(\leq L-2)}(\Ric,s)\\
&\,\phantom{\int_t^{t_0}}+\left(\epsilon^\frac{15}8 a(s)^{-3-c\sqrt{\epsilon}}+a(s)^{-1-c\sqrt{\epsilon}}\right)\E^{(\leq L-2)}(\phi,s)\\
&\,\left.\phantom{\int_t^{t_0}}+\epsilon^\frac{15}8 a(s)^{-3-c\sqrt{\epsilon}}\left(\E^{(\leq L-2)}(\Sigma,s)+\E^{(\leq L-2)}(W,s)\right)+\underbrace{\epsilon^\frac{15}8 a(s)^{-3-c\sqrt{\epsilon}}\E^{(\leq L-4)}(\Ric,s)\,ds}_{\text{not present for }L=2}\right\}\,.
\end{align*}}
\end{lemma}
\begin{remark}
We preemptively note that the error terms on the left hand side, once combined with the similar terms on the left hand side in Lemma \ref{lem:en-est-Sigma} and given suitable weights, will turn out to have positive sign, even if they do not have definite sign in isolation.\\

The main idea in deriving this inequality is that we can use the algebraic identity \eqref{eq:div-to-curl} and integration by parts to exploit the Maxwell system that lies at the core of the Bel-Robinson evolution equations. As a result, we avoid having higher order energies of the Bel-Robinson variables on the right hand side of the integral energy inequalities (which would break the bootstrap argument), then only having to deal with scalar field and Ricci energies at the next derivative \change{level. }
\end{remark}
\begin{proof}
We first prove \eqref{eq:en-est-BR}, and then explain how the same ideas lead to the simpler estimate \eqref{eq:en-est-BR0}. To this end, we start out by taking the time derivative of the energy as usual:
\begin{align*}
-\del_t\E^{(L)}(W,\cdot)=&\,\int_M-3N\frac{\dot{a}}a\left[\lvert\Lap^{\frac{L}2}\RE\rvert_G^2+\lvert\Lap^{\frac{L}2}\RB\rvert_G^2\right]-2\left(\langle\del_t\Lap^{\frac{L}2}\RE,\Lap^{\frac{L}2}\RE\rangle_G+\langle\del_t\Lap^{\frac{L}2}\RB,\Lap^{\frac{L}2}\RB\rangle_G\right)\\
&\quad-2(\del_tG^{-1})^{i_1j_1}(G^{-1})^{i_2j_2}\left[\Lap^{\frac{L}2}\RE_{i_1i_2}\Lap^{\frac{L}2}\RE_{j_1j_2}+\Lap^{\frac{L}2}\RB_{i_1i_2}\Lap^{\frac{L}2}\RB_{j_1j_2}\right]\,\vol{G}
\end{align*}
$\RE$ and $\RB$ are symmetric and tracefree, thus symmetrizations become redundant, and any scalar product with a tensor that is pure trace or with an antisymmetric tensor can be dropped.\footnote{Recall the superscript \enquote{$\parallel$} notation for error terms, see Remark \ref{rem:notation-parallel}.} With this in hand, we get, inserting \eqref{eq:comeq-RE} and \eqref{eq:comeq-RB}:

\begin{subequations}
\begin{align}
-\del_t\E^{(L)}(W,\cdot)=&\,\int_M\biggr\{\frac{\dot{a}}a(-6(N+1)+9N)\left(\lvert\Lap^{\frac{L}2}\RE\rvert_G^2+\lvert\Lap^{\frac{L}2}\RB\rvert_G^2\right)\label{eq:en-eq-BR1}\\
&\,\phantom{\int_M}\change{+2(N+1)a^{-1}\left(\langle\curl_G\Lap^{\frac{L}2}\RE,\Lap^{\frac{L}2}\RB\rangle_G-\langle\curl_G\Lap^{\frac{L}2}\RB,\Lap^{\frac{L}2}\RE\rangle_G}\right)\label{eq:en-eq-BR2}\\
&\,\phantom{\int_M}+\change{2a^{-1}\left(\langle\nabla\Lap^\frac{L}2N\wedge_G\RB,\Lap^\frac{L}2\RE\rangle_G-\langle\nabla\Lap^{\frac{L}2}N\wedge_G\RE,\Lap^{\frac{L}2}\RB\rangle_G\right)}\label{eq:en-eq-BR3}\\
&\,\phantom{\int_M}-8\pi C^2a^{-3}(N+1)\langle\Lap^\frac{L}2\Sigma,\Lap^\frac{L}2\RE\rangle_G-8\pi a(\Psi+C)\langle\nabla\Lap^{\frac{L}2}N\nabla\phi,\Lap^\frac{L}2\RE\rangle_G\label{eq:en-eq-BR4}\\
&\,\phantom{\int_M}-8\pi a(\Psi+C)(N+1)\langle\nabla^2\Lap^{\frac{L}2}\phi,\Lap^{\frac{L}2}\RE\rangle_G\label{eq:en-eq-BR5}\\
&\,\phantom{\int_M}+16\pi a(N+1)\langle \nabla\phi\nabla\Lap^{\frac{L}2}\Psi,\Lap^\frac{L}2\RE\rangle_G\label{eq:en-eq-BR6}\\
&\,\phantom{\int_M}+a^3\epsilonLC[G]\ast\nabla\phi\ast\nabla^2\Lap^{\frac{L}2}\phi\ast\Lap^{\frac{L}2}\RB\label{eq:en-eq-BR7}\\
&\,\phantom{\int_M} +(N+1)a^{-3}\Sigma\ast\left(\Lap^{\frac{L}2}\RE\ast\Lap^{\frac{L}2}\RE+\Lap^{\frac{L}2}\RB\ast\Lap^{\frac{L}2}\RB\right)\label{eq:en-eq-BR8}\\
&\,\phantom{\int_M} -2\langle\mathfrak{E}_{L,Border}+\mathfrak{E}_{L,top}+\mathfrak{E}_{L,Junk}^\parallel,\Lap^{\frac{L}2}\RE\rangle_G\label{eq:en-eq-BR9}\\
&\,\phantom{\int_M} -2\langle\mathfrak{B}_{L,Border}+\mathfrak{B}_{L,top}+\mathfrak{B}_{L,Junk}^\parallel,\Lap^{\frac{L}2}\RB\rangle_G\biggr\}\,\vol{G}\label{eq:en-eq-BR10}
\end{align}
\end{subequations}
For \eqref{eq:en-eq-BR1}, we pull $6(N+1){\dot{a}}a^{-1}\lvert\Lap^{\frac{L}2}\RE\rvert_G^2$ to the left. This leaves
\[\int_M-6\frac{\dot{a}}a\lvert\Lap^\frac{L}2\RB\rvert_G^2+3N\frac{\dot{a}}a\lvert\Lap^\frac{L}2\RB\rvert_G^2+9N\frac{\dot{a}}a\lvert\Lap^\frac{L}2\RE\rvert_G^2\,\vol{G},\]
where we can estimate the last two terms up to constant by $\epsilon a^{1-c\sigma}\E^{(L)}(W,\cdot)$ by \eqref{eq:BsN} and can drop the first term since it is nonpositive.\\
Regarding \eqref{eq:en-eq-BR2}, note that \delete{by \eqref{eq:div-to-curl}, }we have
\begin{align*}
\change{a^{-1}\left(\langle\curl_G\Lap^{\frac{L}2}\RE,\Lap^{\frac{L}2}\RB\rangle_G-\right.}&\change{\,\left.\langle\curl_G\Lap^{\frac{L}2}\RB,\Lap^{\frac{L}2}\RE\rangle_G\right)}
\change{=-a^{-1}\div_G\left(\Lap^{\frac{L}2}\RE\wedge_G\Lap^{\frac{L}2}\RB\right).}
\end{align*}
Hence, the absolute value of \eqref{eq:en-eq-BR2}, using \eqref{eq:wedge} for the wedge product and \eqref{eq:BsN}, can be bounded by:
\begin{align*}
\left\lvert\int_M \change{2a^{-1}(N+1)\div_G\left(\Lap^{\frac{L}2}\RE\wedge_G\Lap^{\frac{L}2}\RB}\right)\,\vol{G}\right\rvert=&\left\lvert\int_M\change{2a^{-1}\langle\nabla N,\Lap^{\frac{L}2}\RE\wedge_G\Lap^{\frac{L}2}\RB\rangle_G}\,\vol{G}\right\rvert\\
\lesssim&\int_M a^{-1}\lvert\nabla N\rvert_G\lvert\Lap^{\frac{L}2}\RE\rvert_G\lvert\Lap^{\frac{L}2}\RB\rvert_G\,\vol{G}\\
\lesssim&\,\epsilon a^{3-c\sigma}\E^{(L)}(W,\cdot)
\end{align*}
For \eqref{eq:en-eq-BR3}, we use the pointwise wedge product estimate \eqref{eq:wedge2} and a priori estimates \eqref{eq:APE} and \eqref{eq:APmidB} to bound it as follows:
\begin{align*}
\lvert\eqref{eq:en-eq-BR3}\rvert\leq&\, 2a^{-1}\lvert\nabla\Lap^\frac{L}2N\rvert_G\left(\lvert\RB\rvert_G\cdot\lvert\Lap^\frac{L}2\RE\rvert_G+\lvert\RE\rvert_G\cdot\lvert\Lap^\frac{L}2\RB\rvert_G\right)\\
\lesssim&\,\epsilon a^{-3}\sqrt{a^4\E^{(L+1)}(N,\cdot)}\sqrt{\E^{(L)}(W,\cdot)}\\
\lesssim&\,\epsilon a^{-3}\left(\E^{(L)}(W,\cdot)+a^4\E^{(L+1)}(N,\cdot)\right)
\end{align*}
We pull the first term of \eqref{eq:en-eq-BR4} to the left as well, and estimate the second using the strong $C_G$-norm estimates \eqref{eq:APPsi} and \eqref{eq:APmidphi} by
\[\sqrt{\epsilon}a^{-1-c\sqrt{\epsilon}}\sqrt{a^4\E^{(L+1)}(N,\cdot)}\sqrt{\E^{(L)}(W,\cdot)}\lesssim a^{-1-c\sqrt{\epsilon}}\E^{(L)}(W,\cdot)+\epsilon a^{-1}\cdot a^4\E^{(L+1)}(N,\cdot)\,.\]
Moving on to \eqref{eq:en-eq-BR5}-\eqref{eq:en-eq-BR7}, we see [using \eqref{eq:APPsi}, \eqref{eq:APmidphi}, \eqref{eq:Sobolev-norm-equiv-zetalow} with $\zeta=\Lap^\frac{L}2\phi$ and \eqref{eq:ibp-trick}]:
\begin{align*}
\change{\lvert\eqref{eq:en-eq-BR5}\rvert\lesssim}&\change{\,\left(a^{-3}\sqrt{a^4\E^{(L+1)}(\phi,\cdot)}+a^{-1}\sqrt{\E^{(L)}(\phi,\cdot)}+a^{-1-c\sqrt{\epsilon}}\sqrt{\E^{(L-2)}(\phi,\cdot)}\right)\sqrt{\E^{(L)}(W,\cdot)}\\
\lesssim&\,\left(\epsilon^\frac18a^{-3}+a^{-1-c\sqrt{\epsilon}}\right)\E^{(L)}(W,\cdot)+\epsilon^{-\frac18}a^{-3}\cdot a^4\E^{(L+1)}(\phi,\cdot)+a^{-1}\E^{(\leq L)}(\phi,\cdot)\\[0.5em]}
\left\lvert\eqref{eq:en-eq-BR6}\right\rvert\lesssim&\,\sqrt{\epsilon}a^{1-c\sqrt{\epsilon}}\sqrt{\E^{(L+1)}(\phi,\cdot)}\sqrt{\E^{(L)}(W,\cdot)}\\
\lesssim&\,a^{1-c\sqrt{\epsilon}}\E^{(L)}(W,\cdot)+\epsilon a^{1-c\sqrt{\epsilon}}\E^{(L+1)}(\phi,\cdot)\,,\\[0.5em]
\left\lvert\eqref{eq:en-eq-BR7}\right\rvert\lesssim&\,\sqrt{\epsilon}a^{1-c\sqrt{\epsilon}}\cdot a^2\|\nabla^2\Lap^{\frac{L}2}\phi\|_{L^2_G}\cdot\sqrt{\E^{(L)}(W,\cdot)}\\
\lesssim&\,\sqrt{\epsilon}a^{1-c\sqrt{\epsilon}}\left(\sqrt{\E^{(L+1)}(\phi,\cdot)}+a^{-c\sqrt{\epsilon}}\sqrt{\E^{(L-1)}(\phi,\cdot)}\right)\cdot\sqrt{\E^{(L)}(W,\cdot)}\,\\
\lesssim&\,a^{1-c\sqrt{\epsilon}}\E^{(L)}(W,\cdot)+\epsilon a^{1-c\sqrt{\epsilon}}\left[\E^{(L+1)}(\phi,\cdot)+\E^{(L)}(\phi,\cdot)+\E^{(\leq L-2)}(\phi,\cdot)\right]
\end{align*}
We can estimate \eqref{eq:en-eq-BR8} by $\epsilon a^{-3}\E^{(L)}(W,\cdot)$ as usual, and obtain the following in summary:
\begin{align*}
&\,-\del_t\E^{(L)}(W,\cdot)+8\pi C^2a^{-3}\int_M(N+1)\langle \change{\Lap^\frac{L}2\Sigma,\Lap^\frac{L}2\RE}\rangle_G\,\vol{G}+6\frac{\dot{a}}a\int_M(N+1)\lvert\Lap^\frac{L}2\RE\rvert_G^2\,\vol{G}\\
\lesssim&\,\left(\epsilon a^{-3}+a^{-1-c\sqrt{\epsilon}}\right)\E^{(L)}(W,\cdot)+a^{-1}\E^{(L+1)}(\phi,\cdot)+a^{-1}\E^{(L)}(\phi,\cdot)\\
&\,+\epsilon a^{-3}\cdot a^{4}\E^{(L+1)}(N,\cdot)+a^{-1}\E^{(\leq L-2)}(\phi,\cdot)\\
&\,\left[\|\mathfrak{E}_{L,Border}\|_{L^2_G}+\|\mathfrak{E}_{L,top}\|_{L^2_G}+\|\mathfrak{E}_{L,Junk}^\parallel\|_{L^2_G}\right.\\
&\,\left.+\|\mathfrak{B}_{L,Border}\|_{L^2_G}+\|\mathfrak{B}_{L,top}\|_{L^2_G}+\|\mathfrak{B}_{L,Junk}^\parallel\|_{L^2_G}\right]\sqrt{\E^{(L)}(W,\cdot)}
\end{align*}
We can now apply Corollary \ref{cor:en-est-lapse} for $2l=L$ to estimate the lapse energy in the second line (leading to borderline scalar field energy and $\Sigma$-energy contributions as well as junk terms), and insert the borderline (see \eqref{eq:L2-Border-BR}), top (see \eqref{eq:L2-top-E} and \eqref{eq:L2-top-B}) and junk estimates (see \eqref{eq:L2-junk-BR-par}), dealing with the lapse energies there analogously. \change{In particular, the top order curvature terms arise as follows:
\begin{align*}
\|\mathfrak{E}_{L,top}\|_{L^2_G}\sqrt{\E^{(L)}(W,\cdot)}\lesssim&\,\sqrt{\epsilon} a^{-1-c\sqrt{\epsilon}}\sqrt{a^4\E^{(L-1)}(\Ric,\cdot)}\sqrt{\E^{(L)}(W,\cdot)}\\
\lesssim&\,\epsilon^\frac18a^{-1-c\sqrt{\epsilon}}\E^{(L)}(W,\cdot)+\epsilon^\frac{7}8a^{-1}\cdot a^4\E^{(L-1)}(\Ric,\cdot)\\
\|\mathfrak{B}_{L,top}\|\sqrt{\E^{(L)}(W,\cdot)}\lesssim&\,\epsilon a^{-3}\sqrt{a^4\E^{(L-1)}(\Ric,\cdot)}\sqrt{\E^{(L)}(W,\cdot)}\\
\lesssim&\,\epsilon^\frac18 a^{-3}\E^{(L)}(W,\cdot)+\epsilon^\frac{15}8 a^{-3}\cdot a^4\E^{(L-1)}(\Ric,\cdot)
\end{align*}
Hence, both top order curvature terms can be bounded by $\epsilon^\frac78a^{-3}\cdot a^4\E^{(L-1)}(\Ric,\cdot)$.\\
Integrating the inequality yields \eqref{eq:en-est-BR}.\\}

For \eqref{eq:en-est-BR0}, we get applying \eqref{eq:REEqE} and \eqref{eq:REEqB} and again using that $\RE$ and $\RB$ are symmetric and tracefree:
\begin{align*}
-\del_t\E^{(0)}(W,\cdot)=&\,\int_M\biggr\{\frac{\dot{a}}a(-6(N+1)+9N)\left(\lvert\RE\rvert_G^2+\lvert\RB\rvert_G^2\right)\\
&\,\phantom{\int_M}+2(N+1)\left(\langle\curl\RE,\RB\rangle_G-\langle\curl\RB,\RE\rangle_G\right)\\
&\,\phantom{\int_M}+2\left(\langle\nabla N\wedge \RB,\RE\rangle_G-\langle\nabla N\wedge\RE,\RB\rangle_G\right)\\
&\,\phantom{\int_M}+(N+1)a^{-3}\Sigma\ast\left(\RE\ast\RE+\RB\ast\RB\right)\\
&\,\phantom{\int_M}-8\pi a^{-3}(N+1)(\Psi+C)^2\langle\Sigma,\RE\rangle_G\\
&\,\phantom{\int_M}+\left[\dot{a}a^3\nabla\phi\ast\nabla\phi+a(\Psi+C)\cdot\nabla N\ast\nabla\phi\right]\ast\RE\\
&\,\phantom{\int_M}+a(N+1)\left[\nabla\phi\ast\nabla\Psi+\Sigma\ast\nabla\phi\ast\nabla\phi+(\Psi+C)\nabla^2\phi\right]\ast\RE\\
&\,\phantom{\int_M}+(N+1)\epsilonLC[G]\ast\left(a^3\nabla^2\phi\ast\nabla\phi+a^{-1}(\Psi+C)\Sigma\ast\nabla\phi\right)\ast\RB\biggr\}\,\vol{G}
\end{align*}
The first two lines are treated as in the general case. For the third line, we get $\epsilon a^{3-c\sigma}\E^{(0)}(W,\cdot)$ with \eqref{eq:wedge2} and \eqref{eq:BsN}, while the fourth term is bounded by $\epsilon a^{-3}\E^{(0)}(W,\cdot)$ with \eqref{eq:APSigma}. This leaves the surviving matter terms in the final four lines.\\
We pull $\int_M8\pi a^{-3}(N+1)C^2\langle\Sigma,\RE\rangle_G\vol{G}$ to the left as before. For the remaining terms, we can apply a priori estimates \eqref{eq:APPsi}, \eqref{eq:APmidPsi} and \eqref{eq:APmidphi}, the bootstrap assumption \eqref{eq:BsN} and Lemma \ref{lem:lapse-maxmin} for $N$, which yields the following bound up to constant remaining terms in the last three lines:
\begin{align*}
\sqrt{\E^{(0)}(W,\cdot)}\cdot&\left[a^{-1}\cdot a^2\|\nabla^2\phi\|_{L^2_G}+\sqrt{\epsilon}a^{-1-c\sqrt{\epsilon}}\sqrt{\E^{(0)}(\phi,\cdot)}\right.\\
&\left.+\left(\epsilon a^{-3}+\sqrt{\epsilon}a^{-1-c\sqrt{\epsilon}}\right)\sqrt{\E^{(0)}(\Sigma,\cdot)}\right]
\end{align*}
Applying \eqref{eq:Sobolev-norm-equiv-zetalow} to the scalar field norm and then \eqref{eq:ibp-trick}, this leads to \eqref{eq:en-est-BR0} along with the previous observations.
\end{proof}

\subsection{Energy estimates for the second fundamental form}\label{subsec:en-Sigma}

For the energy estimates for $\Sigma$, we again first derive \change{even order }integral estimates:

\begin{lemma}[Energy estimates for the second fundamental form \change{for even orders}]\label{lem:en-est-Sigma} Let $t\in(t_{Boot},t_0]$. Then, one has:
\begin{align*}
\numberthis\label{eq:en-est-Sigma0}&\,\E^{(0)}(\Sigma,t)+2\int_t^{t_0}\int_M\left[a(s)^{-3}(N+1)\langle\RE,\Sigma\rangle_G+\frac{\dot{a}(s)}{a(s)}(N+1)\lvert\Lap^{\frac{L}2}\Sigma\rvert_G^2\right]\,\vol{G}\,ds\\
\lesssim&\,\E^{(0)}(\Sigma,t_0)+\int_t^{t_0}\epsilon^\frac18a(s)^{-3}\E^{(0)}(\Sigma,s)\,ds+\int_t^{t_0}\epsilon^{-\frac18}a(s)^{-3}\E^{(0)}(\phi,s)\,ds\\
\end{align*}
For $L\in 2\N, L\leq \change{18}$, the following holds:
\begin{align*}
\numberthis\label{eq:en-est-Sigma}&\,\E^{(L)}(\Sigma,t)+2\int_t^{t_0}\int_M\left[a(s)^{-3}(N+1)\langle\Lap^\frac{L}2\RE,\Lap^{\frac{L}2}\Sigma\rangle_G+\frac{\dot{a}(s)}{a(s)}(N+1)\lvert\Lap^{\frac{L}2}\Sigma\rvert_G^2\right]\,\vol{G}\,ds\\
&\,\lesssim\E^{(L)}(\Sigma,t_0)+\int_t^{t_0}\epsilon^\frac18a(s)^{-3}\E^{(L)}(\Sigma,s)\,ds\\
&\,+\int_t^{t_0}\Bigr\{\epsilon^{-\frac18}a(s)^{-3}\E^{(L)}(\phi,s)+\epsilon^\frac{15}8a(s)^{5-c\sigma}\E^{(L-1)}(\Ric,s)+\epsilon^2a(s)^{-3}\E^{(L-2)}(\Ric,s)\\
&\,\phantom{\int_t^{t_0}}+\epsilon^{\frac{15}8}a(s)^{-3-c\sqrt{\epsilon}}\E^{(\leq L-2)}(\Sigma,s)+\left(\epsilon^\frac{15}8a(s)^{-3-c\sqrt{\epsilon}}+\epsilon a(s)^{-1-c\sqrt{\epsilon}}\right)\E^{(\leq L-2)}(\phi,s)\\
&\,\phantom{\int_t^{t_0}}+\underbrace{\epsilon^{\frac{15}8}a(s)^{-3-c\sqrt{\epsilon}}\E^{(\leq L-4)}(\Ric,s)}_{\text{not present for }L=2}\Bigr\}\,ds
\end{align*}
\end{lemma}
\begin{remark}
The main hurdle of dealing with the second fundamental form is that a high order curvature term occurs in the evolution equation. It is to precisely this end that the Bel-Robinson variables needed to be introduced, since \eqref{eq:comeq-Ham-BR} is what facilitates controlling said term without having to use $\E^{(L)}(\Ric,\cdot)$ or similar high order metric energies. Again, the resulting leading terms will turn out to have definite sign when combined with the Bel-Robinson energy estimates above.
\end{remark}
\begin{proof}
Here, we omit the proof for the inequality at order zero since is completely analogous in structure to the one for orders 2 and higher and the only differences that arise are that lower order error terms do not occur. \\
Once again, we start out by differentiating $-\E^{(L)}(\Sigma,\cdot)$ and insert \eqref{eq:comeq-Sigma}:
\begin{align*}
-\del_t\E^{(L)}(\Sigma,\cdot)=&\,\int_M -2\left\langle\del_t\Lap^{\frac{L}2}\Sigma,\Lap^{\frac{L}2}\Sigma\right\rangle_G+(\del_tG^{-1})\ast G^{-1}\ast\Lap^{\frac{L}2}\Sigma\ast\Lap^{\frac{L}2}\Sigma-3N\frac{\dot{a}}a\lvert\Lap^{\frac{L}2}\Sigma\rvert_G^2\,\vol{G}\\
=&\,\int_M \left\{2a\langle\nabla^2\Lap^{\frac{L}2} N,\Lap^{\frac{L}2}\Sigma\rangle_G-2a(N+1)\langle\Lap^{\frac{L}2}\Ric[G],\Lap^{\frac{L}2}\Sigma\rangle_G+\right.\\
&\,\phantom{\int_M}+(\del_tG^{-1})\ast G^{-1}\ast\Lap^{\frac{L}2}\Sigma\ast\Lap^{\frac{L}2}\Sigma-3N\frac{\dot{a}}a\lvert\Lap^{\frac{L}2}\Sigma\rvert_G^2\\
&\,\phantom{\int_M}\left.-2\langle\mathfrak{S}_{L,Border},\Lap^{\frac{L}2}\Sigma\rangle_G-2\langle\mathfrak{S}_{L,Junk}^\parallel,\Lap^{\frac{L}2}\Sigma\rangle_G\right\}\,\vol{G}
\end{align*}

For the first term, one can use \eqref{eq:Sobolev-norm-equiv-zetalow}, Corollary \ref{cor:en-est-lapse} at order $L$ and \eqref{eq:ibp-trick} to bound its absolute value by the following:
\begin{align*}
\lesssim&\,a\|\Lap^{\frac{L}2}N\|_{\dot{H}^2_G}\sqrt{\E^{(L)}(\Sigma,\cdot)}\\
\lesssim&\,\left[a^{-3}\sqrt{a^8\E^{(L+2)}(N,\cdot)}+a^{1-c\sqrt{\epsilon}}\sqrt{\E^{(L)}(N,\cdot)}\right]\sqrt{\E^{(L)}(\Sigma,\cdot)}\\
\lesssim&\,\biggr[\epsilon a^{-3}\sqrt{\E^{(L)}(\Sigma,\cdot)}+a^{-3}\E^{(L)}(\phi,\cdot)+\epsilon a^{-3-c\sqrt{\epsilon}}\left(\sqrt{\E^{(\leq L-2)}(\Sigma,\cdot)}+\sqrt{\E^{(\leq L-2)}(\phi,\cdot)}\right)\\
&\,+\underbrace{\left(\epsilon^2a^{-3-c\sqrt{\epsilon}}+\epsilon a^{1-c\sigma}\right)\sqrt{\E^{(\leq L-3)}(\Ric,\cdot)}}_{\text{not present for }L=2}\biggr]\sqrt{\E^{(L)}(\Sigma,\cdot)}\\
\lesssim&\,\left(\epsilon^\frac18 a^{-3}+a^{1-c\sigma}\right)\E^{(L)}(\Sigma,\cdot)+\epsilon^{-\frac18}a^{-3}\E^{(L)}(\phi,\cdot)+\epsilon a^{-3-c\sqrt{\epsilon}}\left[\E^{(\leq L-2)}(\Sigma,\cdot)+\E^{(\leq L-2)}(\phi,\cdot)\right]\\
&\,+\underbrace{\left(\epsilon^{\frac{31}8}a^{-3}+\epsilon^2 a^{1-c\sigma}\right)\E^{(L-2)}(\Ric,\cdot)+\left(\epsilon^\frac{31}8a^{-3-c\sqrt{\epsilon}}+\epsilon^2a^{1-c\sigma}\right)\E^{(\leq L-4)}(\Ric,\cdot)}_{\text{not present for }L=2}
\end{align*}
Next, we replace the high order curvature term as follows, using the commuted rescaled Hamiltonian constraint equation \eqref{eq:comeq-Ham-BR} that $\Lap^\frac{L}2\Sigma$ is tracefree and symmetric:
\begin{align*}
&\,\int_M-2a(N+1)\langle\Lap^{\frac{L}2}\Ric[G],\Lap^{\frac{L}2}\Sigma\rangle_G\,\vol{G}\\
=&\,\int_M-2(N+1)a^{-3}\langle\Lap^{\frac{L}2}\RE,\Lap^{\frac{L}2}\Sigma\rangle_G-2(N+1)\frac{\dot{a}}{a}|\Lap^{\frac{L}2}\Sigma|_G^2+\langle\mathfrak{H}_{L,Border}+\mathfrak{H}_{L,Junk}^\parallel,\Lap^{\frac{L}2}\Sigma\rangle_G\,\vol{G}
\end{align*}
We pull the first two terms to left, only keeping the error terms on the right. After inserting the borderline and junk term estimates for the Hamiltonian constraint equations (\eqref{eq:L2-Border-H} and \eqref{eq:L2-junk-H-par}) and the evolution equation itself (\eqref{eq:L2-Border-S} and \eqref{eq:L2-junk-S}), as well as bounding $\lvert\del_tG^{-1}\rvert\lesssim \epsilon a^{-3}$ and inserting \eqref{eq:BsN} as usual, we obtain \eqref{eq:en-est-Sigma} by integrating.
\end{proof}

\noindent Additionally, we can exploit the structure of the momentum constraint equations to gain an elliptic estimate for $\E^{(\change{L}+1)}(\Sigma,\cdot)$. Crucially, the upper bound only depends on $\Sigma$-, scalar field and Bel-Robinson energies up to order $\change{L}$, and appropriately small and time-scaled curvature contributions up to order $\change{L}-1$. This will allow us to close the argument since we do not need to consider the Bel-Robinson energy at order $\change{L}+1$ to control $\Sigma$ at that order
, which would require higher order scalar field and curvature energies\change{.}

\begin{lemma}[\change{Odd } order energy estimate for the second fundamental form]\label{lem:en-est-Sigma-top} For any $L\in 2\Z_+$, $2\leq L\leq \change{18}$, we have
\begin{align*}
\numberthis\label{eq:en-est-Sigma-top}a^4\E^{(\change{L}+1)}(\Sigma,\cdot)\lesssim&\,\left(a^{4-c\sqrt{\epsilon}}+\epsilon a^{2-c\sqrt{\epsilon}}\right)\E^{(\change{L})}(\Sigma,\cdot)+\E^{(\change{L})}(\phi,\cdot)+\E^{(\change{L})}(W,\cdot)+\epsilon^2a^4\E^{(\change{L}-1)}(\Ric,\cdot)\\
&\,+\change{\epsilon a^{-c\sqrt{\epsilon}}}\E^{(\leq \change{L}-2)}(\phi,\cdot)+a^{2-c\sqrt{\epsilon}}\E^{(\leq \change{L}-2)}(\Sigma,\cdot)+\epsilon a^{2-c\sqrt{\epsilon}}\E^{(\leq \change{L}-2)}(\Ric,\cdot)\,.
\end{align*}
\change{For $L=0$, one analogously has}
\begin{equation}\label{eq:en-est-Sigma-1}
\change{a^4\E^{(1)}(\Sigma,\cdot)\lesssim\,\left(a^{4-c\sqrt{\epsilon}}+\epsilon a^{2-c\sqrt{\epsilon}}\right)\E^{(0)}(\Sigma,\cdot)+\E^{(0)}(\phi,\cdot)+\E^{(0)}(W,\cdot)\,.}
\end{equation}
\end{lemma}
\begin{proof}\change{We prove the statement for \changefinal{$L\geq2$}, since the proof of \eqref{eq:en-est-Sigma-1} is entirely analogous.}\\
By \cite[(A.22)]{AM03}, since $(\Sigma_t,g)$ is a three-dimensional compact Riemannian manifold for any $t\in(t_{Boot},t_0]$, any tracefree $(0,2)$ tensor $U_{ij}$ on $(\Sigma_t,g)$ satisfies
\begin{equation}\label{eq:curl-div-elliptic}
\int_{\Sigma_t} \lvert\nabla U\rvert_g^2+3\Ric[g]\cdot U\cdot U-\frac{R[g]}2\lvert U\rvert_g^2\,\vol{g}=\int_{\Sigma_t} \lvert\curl U\rvert_g^2+\frac32\lvert\div_g U\rvert_g^2\,\vol{g}\,.
\end{equation}
In particular, for $U=\Lap^{\frac{\change{L}}2}\Sigma$ and after rescaling, this reads:
\begin{align*}
&\,\int_M \left\lvert\nabla\Lap^{\frac{\change{L}}2}\Sigma\right\rvert_G^2+3{\left(\Ric[G]^\sharp\right)^i}_j{\left(\Lap^{\frac{\change{L}}2}\Sigma^\sharp\right)^j}_l{\left(\Lap^{\frac{\change{L}}2}\Sigma^\sharp\right)^l}_i-\frac{R[G]}2\lvert\Lap^\frac{\change{L}}2\Sigma\rvert_G^2\,\vol{G}\\
=&\,\int_M\frac32\lvert\div_G\Lap^\frac{\change{L}}2\Sigma\rvert_G^2+a^2\lvert\curl\Lap^\frac{\change{L}}2\Sigma\rvert_G^2\,\vol{G}
\end{align*}
The last two terms on the left hand side can be estimated by $(1+\sqrt{\epsilon}a^{-c\sqrt{\epsilon}})\E^{(\change{L})}(\Sigma,\cdot)$ in absolute value using the strong $C_G$-norm estimate \eqref{eq:APmidRic}. Thus, inserting the Laplace-commuted rescaled momentum constraint equations \eqref{eq:comeq-mom-div} and \eqref{eq:comeq-mom-curl}, we obtain for a suitable constant $K>0$:
\begin{align*}
&\E^{(\change{L}+1)}(\Sigma,\cdot)-K\left(1+\sqrt{\epsilon}a^{-c\sqrt{\epsilon}}\right)\E^{(\change{L})}(\Sigma,\cdot)\\
\lesssim&\,\int_M\left\{\lvert\Psi+C\rvert^2\left\lvert\nabla\Lap^{\frac{\change{L}}2}\phi\right\rvert_G^2+\lvert\nabla\phi\rvert_G^2\left\lvert\Lap^{\frac{\change{L}}2-1}\Ric[G]\right\rvert_G^2+\lvert\Sigma\rvert_G^2\left\lvert\nabla\Lap^{\frac{\change{L}}2-1}\Ric[G]\right\rvert_G^2+\lvert\mathfrak{M}_{\change{L},Junk}\rvert_G^2\right.\\
&\,\phantom{\int_M}\left.+a^{-4}\lvert\Lap^\frac{\change{L}}2\RB\rvert_G^2+\lvert\Sigma\rvert_G^2\left\lvert\nabla\Lap^{\frac{\change{L}}2-1}\Ric[G]\right\rvert_G^2+\lvert\nabla\Sigma\rvert_G^2\left\lvert\nabla^2\Lap^{\frac{\change{L}}2-2}\Ric[G]\right\rvert_G^2+\left\lvert\tilde{\mathfrak{M}}_{\change{L},Junk}\right\rvert_G^2\,\right\}\vol{G}
\end{align*}
\noindent After rearranging, using the strong $C_G$-norm estimates \eqref{eq:APPsi}, \eqref{eq:APmidphi}, \eqref{eq:APSigma} and \eqref{eq:APmidSigma} and multiplying by $a^4$ on both sides, we get
\begin{align*}
a^4\E^{(\change{L}+1)}(\Sigma,\cdot)\lesssim&\,\left(1+\sqrt{\epsilon}a^{-c\sqrt{\epsilon}}\right)a^4\E^{(\change{L})}(\Sigma,\cdot)+\E^{(\change{L})}(\phi,\cdot)+\E^{(\change{L})}(W,\cdot)+\epsilon^2a^4\E^{(\change{L}-1)}(\Ric,\cdot)\\
&\,+\epsilon^2a^{4-c\sqrt{\epsilon}}\E^{(\change{L}-2)}(\Ric,\cdot)+a^4\|\mathfrak{M}_{\change{L},Junk}\|_{L^2_G}^2+a^4\|\tilde{\mathfrak{M}}_{\change{L},Junk}\|_{L^2_G}^2\,.
\end{align*}
The statement follows inserting the estimates \eqref{eq:L2-junk-M} and \eqref{eq:L2-junk-Mtilde}.
\end{proof}

\subsection{Energy estimates for the curvature}\label{subsec:en-Ric}

To control commutator errors, we will also need some additional estimates on curvature energies. Unlike the other energies, these inequalities do not rely on any delicate structure within the equations and instead just rely on pointwise estimates, the Young inequality and near-coercivity of energies in the sense of Lemma \ref{lem:Sobolev-norm-equivalence-improved}. For the sake of convenience, we phrase these estimates for $\E^{(L-2)}(\Ric,\cdot)$ since this is the order needed when improving behaviour of the total energy at order $L$.

\begin{lemma}[Curvature energy estimates \change{at even orders}]\label{lem:en-est-Ric}
Let $L\in 2\Z$, $4\leq L\leq \change{16}$ and $t\in(t_{Boot},t_0]$. Then, one has
\begin{align*}
\numberthis\label{eq:en-est-Ric}\E^{(L-2)}(\Ric,t)\lesssim&\,\E^{(L-2)}(\Ric,t_0)+\int_t^{t_0}\left(\epsilon^\frac18a(s)^{-3}+a(s)^{8-c\sigma}\right)\E^{(L-2)}(\Ric,s)\,ds\\
&\,+\int_t^{t_0}\Bigr\{\epsilon^{-\frac18}a(s)^{-3}\left(\E^{(L)}(\phi,s)+\E^{(L)}(\Sigma,s)\right)\\
&\,\phantom{\int_t^{t_0}}+\epsilon^{-\frac18}a(s)^{-3-c\sqrt{\epsilon}}\left(\E^{(\leq L-2)}(\phi,s)+\E^{(\leq L-2)}(\Sigma,s)\right)\\
&\,\phantom{+\int_t^{t_0}}+\epsilon^\frac78a(s)^{-3-c\sqrt{\epsilon}}\E^{(\leq L-4)}(\Ric,s)\Bigr\}\,ds\,.
\end{align*}
Additionally,
\begin{align*}
\numberthis\label{eq:en-est-Ric0}\E^{(0)}(\Ric,t)\lesssim&\,\E^{(0)}(\Ric,t_0)+\int_t^{t_0}\epsilon^\frac18a(s)^{-3}\E^{(0)}(\Ric,s)\,ds\\
&\,+\int_t^{t_0}\epsilon^{-\frac18}a(s)^{-3}\left(\E^{(0)}(\phi,s)+\E^{(0)}(\Sigma,s)\right)ds\,.
\end{align*}
\end{lemma}
\begin{proof}
First, we note that
\[\|\div_G^\sharp\nabla\Lap^{\frac{L}2-1}\Sigma\|_{L^2_G}\lesssim \|\nabla^2\Lap^{\frac{L}2-1}\Sigma\|_{L^2_G}\lesssim \|\Lap^{\frac{L}2}\Sigma\|_{L^2_G}+a^{-c\sqrt{\epsilon}}\sqrt{\E^{(L-2)}(\Sigma,\cdot)}\]
holds using the low order version of \eqref{eq:Sobolev-norm-equiv-T2l} with $\mathfrak{T}=\Lap^{\frac{L}2-1}\Sigma$ for $l=2$, and similarly
\[\|\nabla^2\Lap^{\frac{L}2-1}N\|_{L^2_G}\lesssim \|\Lap^{\frac{L}2}N\|_{L^2_G}+a^{-c\sqrt{\epsilon}}\sqrt{\E^{(L-2)}(N,\cdot)}\]
using \eqref{eq:Sobolev-norm-equiv-zeta2l} at order $2$. Now, using $\Lap^{\frac{L}2-1} G=0$ for $L\geq 4$, we continue as usual by applying \eqref{eq:comeq-Ric-even} to the expression below:
\begin{align*}
-\del_t\E^{(L-2)}(\Ric,\cdot)
\lesssim&\,\int_M \Bigr\{a^{-3}\left(\lvert\Lap^{\frac{L}2}\Sigma\rvert_G+\lvert\nabla^2\Lap^{\frac{L}2-1}\Sigma\rvert_G\right)\lvert\Lap^{\frac{L}2-1}\Ric[G]\rvert_G\\
&\,+\frac{\dot{a}}a\left(\lvert\nabla^2\Lap^{\frac{L}2-1}N\rvert_G+\lvert\Lap^{\frac{L}2}N\rvert_G\right)\lvert\Lap^{\frac{L}2-1}\Ric[G]\rvert_G\\
&\,+\left(\lvert\mathfrak{R}_{L-2,Border}\rvert_G+\lvert\mathfrak{R}_{L-2,Junk}\rvert_G\right)\cdot\lvert\Lap^{\frac{L}2-1}\Ric[G]\rvert_G\\
&\,+a^{-3}\Sigma\ast\Lap^{\frac{L}2-1}\Ric[G]\ast\Lap^{\frac{L}2-1}\Ric[G]+N\frac{\dot{a}}a\lvert\Lap^{\frac{L}2-1}\Ric[G]\rvert_G^2\Bigr\}\,\vol{G}
\end{align*}
Due to the estimates above as well as \eqref{eq:APSigma} and \eqref{eq:BsN}, this implies
\begin{align*}
-\del_t\E^{(L-2)}(\Ric,\cdot)\lesssim&\,a^{-3}\left[\sqrt{\E^{(L)}(\Sigma,\cdot)}+\sqrt{\E^{(L)}(N,\cdot)}\right]\sqrt{\E^{(L-2)}(\Ric,\cdot)}\\
&\,+a^{-3-c\sqrt{\epsilon}}\left[\sqrt{\E^{(\leq L-2)}(\Sigma,\cdot)}+\sqrt{\E^{(\leq L-2)}(N,\cdot)}\right]\sqrt{\E^{(L-2)}(\Ric,\cdot)}\\
&\,+\left(\|\mathfrak{R}_{L-2,Border}\|_{L^2_G}+\|\mathfrak{R}_{L-2,Junk}\|_{L^2_G}\right)\sqrt{\E^{(L-2)}(\Ric,\cdot)}\\
&\,+\epsilon a^{-3}\E^{(L-2)}(\Ric,\cdot)\,.
\end{align*}
Using Corollary \ref{cor:en-est-lapse} at order $L$ and distributing terms containing $\E^{(L-3)}(\Ric,\cdot)$ with \eqref{eq:ibp-trick} as usual, we get
\begin{align*}
-\del_t\E^{(L-2)}(\Ric,\cdot)\lesssim&\,\left[\epsilon^{\frac18} a^{-3}+a^{8-c\sigma}\right]\E^{(L-2)}(\Ric,\cdot)+\epsilon^{-\frac18}a^{-3}\left[\E^{(L)}(\Sigma,\cdot)+\E^{(L)}(\phi,\cdot)\right]\\
&\,+\left[\epsilon^\frac{31}8a^{-3-c\sqrt{\epsilon}}+\epsilon^2a^{8-c\sigma}\right]\E^{(\leq L-4)}(\Ric,\cdot)\\
&\,+\epsilon^{-\frac18}a^{-3-c\sqrt{\epsilon}}\left[\E^{(\leq L-2)}(\Sigma,\cdot)+\E^{(\leq L-2)}(\phi,\cdot)\right]\\
&\,+\left(\|\mathfrak{R}_{L-2,Border}\|_{L^2_G}+\|\mathfrak{R}_{L-2,Junk}\|_{L^2_G}\right)\sqrt{\E^{(L-2)}(\Ric,\cdot)}\,.
\end{align*}
Equation \eqref{eq:en-est-Ric} now follows inserting the borderline and junk term estimates \eqref{eq:L2-Border-R-even} and \eqref{eq:L2-junk-R-even} and applying the lapse energy estimates from Corollary \ref{cor:en-est-lapse}.\\
Equation \eqref{eq:en-est-Ric0} follows almost identically by inserting \eqref{eq:REEqRic} instead of \eqref{eq:comeq-Ric-even} as well as \eqref{eq:REEqG} for the additional $\del_tG\ast(\Ric[G]+\nicefrac29G)$-terms. These can be estimated as
\begin{align*}
\lesssim&\,\int_Ma^{-3}\Sigma\ast\left(\Ric[G]+\frac29G\right)+\frac{\dot{a}}a N\cdot G\ast\left(\Ric[G]+\frac29G\right)\,\vol{G}\\
\lesssim&\,a^{-3}\left(\sqrt{\E^{(0)}(\Sigma,\cdot)}+\sqrt{\E^{(0)}(N,\cdot)}\right)\sqrt{\E^{(0)}(\Ric,\cdot)}\,,
\end{align*}
which can be treated as at higher orders.
\end{proof}

\begin{lemma}[\change{Odd }order curvature energy estimate]\label{lem:en-est-Ric-top} For $\change{L}\in2\N$, $4\leq \change{L}\leq \change{18}$ and $t\in(t_{Boot},t_0]$,
\begin{align*}
\numberthis\label{eq:en-est-Ric-top}a(t)^4\E^{(\change{L}-1)}(\Ric,t)\lesssim&\,a(t_0)^4\E^{(\change{L}-1)}(\Ric,t_0)\\
&\,+\int_t^{t_0}\left(\epsilon^\frac18a(s)^{-3}+a(s)^{-1-c\sqrt{\epsilon}}\right)\left(a(s)^4\E^{(\change{L}-1)}(\Ric,s)\right)\,ds\\
&\,+\int_t^{t_0}\epsilon^{-\frac18}a(s)^{-3}\cdot a(s)^4\E^{(\change{L}+1)}(\Sigma,s)\,ds\\
&\,+\int_t^{t_0}\Bigr\{\epsilon^{-\frac18}a(s)^{-3}\E^{(\change{L})}(\phi,s)+\left(\epsilon^\frac{15}8a(s)^{-3}+a(s)^{-1-c\sqrt{\epsilon}}\right)\E^{(\change{L})}(\Sigma,s)\\
&\,\phantom{+\int_t^{t_0}}+\left(\epsilon^\frac{15}8a(s)^{-3-c\sqrt{\epsilon}}+a(s)^{-1-c\sqrt{\epsilon}}\right)\left(\E^{(\leq \change{L}-2)}(\phi,s)+\E^{(\leq \change{L}-2)}(\Sigma,s)\right)\\
&\,\phantom{+\int_t^{t_0}}+\epsilon^\frac{15}8a(s)^{-3}\E^{(\change{L}-2)}(\Ric,\cdot)+\epsilon^\frac{15}8a^{-3-c\sqrt{\epsilon}}\E^{(\leq \change{L}-4)}(\Ric,s)\Bigr\}ds
\end{align*}
\end{lemma}
\begin{proof}
The proof is very similar to that of Lemma \ref{lem:en-est-Ric} since we did not exploit any structure within \eqref{eq:comeq-Ric-even} that does not equally occur in \eqref{eq:comeq-Ric-odd}, and thus we omit the details. As in the proof of Lemma \ref{lem:en-est-SF-top}, we note that the differences within the estimate come from how top order lapse terms are treated: The scaling of the top order energy allows one to estimate $a^4\E^{(\change{L}+1)}(N,\cdot)$ by scalar field energies and $\Sigma$-energies of up to order $\change{L}$ and curvature energies up to order $\change{L}-3$.
\end{proof}

\subsection{Sobolev norm estimates for metric objects}\label{subsec:int-metric}

To close the bootstrap argument, we need to improve the behaviour of metric quantities in addition to the energy formalism, both to capture the intrinsic behaviour of the metric and to relate energies to supremum norms.

\begin{lemma}[Sobolev norm estimates for Christoffel symbols]\label{lem:int-est-Chr} Let $U$ be a coordinate neighbourhood on $M$, viewed as a coordinate neighbourhood on $\Sigma_t$ for $t\in(t_{Boot},t_0]$. For any $l\in\N,\,l\leq \changefinal{17}$, the following Sobolev estimate then holds:
\begin{equation}\label{eq:norm-est-Chr}
\|\Gamma-\Gamhat\|^2_{H^l_G(U)}\lesssim a^{-c\epsilon^\frac18}\left(\epsilon^4+\epsilon^{-\frac14}\sup_{s\in(t,t_0)}\left(\|N\|_{H^{l+1}_G(\Sigma_s)}^2+\|\Sigma\|_{H^{l+1}_G(\Sigma_s)}^2\right)\right)
\end{equation}
\end{lemma}
\begin{proof}
Commuting the evolution equation \eqref{eq:REEqChr} with $\nabla^J$, we get for $J\in\N$, $J\leq 17$:
\begin{align*}
-\del_t\|\Gamma-\Gamhat\|_{\dot{H}^J_G}^2=&\,\int_U\Bigr[(\del_tG^{-1})\ast G^{-1}\ast\dots\ast G^{-1}\ast G\ast\nabla^J(\Gamma-\Gamhat)\ast\nabla^J(\Gamma-\Gamhat)\\
&\,\phantom{\int_M}+(G^{-1})\ast\dots\ast(G^{-1})\ast\del_tG\ast\nabla^J(\Gamma-\Gamhat)\ast\nabla^J(\Gamma-\Gamhat)\\
&\,\phantom{\int_M}+\left(a^{-3}\sum_{I_N+I_\Sigma=J+1}\nabla^{I_N}(N+1)\ast\nabla^{I_\Sigma}\Sigma+\frac{\dot{a}}a\nabla^{J+1}N\right)\ast\nabla^J(\Gamma-\Gamhat)\\
&\,\phantom{\int_M}+2\langle[\del_t,\nabla^J](\Gamma-\Gamhat),\nabla^J(\Gamma-\Gamhat)\rangle_G-3N\frac{\dot{a}}a\left\lvert\nabla^J(\Gamma-\Gamhat)\right\rvert_G^2\Bigr]\,\vol{G}
\end{align*}
We recall that \eqref{eq:APmidG} implies
\begin{equation}\label{eq:APChr}
\|\Gamma-\Gamhat\|_{C^{11}_G}\lesssim\sqrt{\epsilon}a^{-c\sqrt{\epsilon}}
\end{equation}
by \eqref{eq:Christoffel-norm-handwaving}. It follows from inserting this in \eqref{eq:commutator-aux-tensor} along with \eqref{eq:APSigma}, \eqref{eq:APmidSigma} and \eqref{eq:BsN} that
\begin{align*}
\|[\del_t,\nabla^J](\Gamma-\Gamhat)\|_{L^2_G}\lesssim&\, \sqrt{\epsilon}a^{-3-c\sqrt{\epsilon}}\|\Sigma\|_{H^J_G}+\sqrt{\epsilon} a^{-3-c\sqrt{\epsilon}}\|N\|_{H^J_G}+\epsilon a^{-3}\|\Gamma-\Gamhat\|_{H^{J-1}_G}
\end{align*}
is satisfied. Consequently and using the same strong $C_G$-norm bounds along with Lemma \ref{lem:lapse-maxmin}, the differential inequality becomes
\begin{align*}
-\del_t\|\Gamma-\Gamhat\|_{\dot{H}^J_G}^2
\lesssim&\,\epsilon^\frac18 a^{-3}\|\Gamma-\Gamhat\|_{\dot{H}^{J}_G}^2+\epsilon^{-\frac18}a^{-3}\left(\|N\|_{H^{J+1}_G}^2+\|\Sigma\|_{H^{J+1}_G}^2\right)\\
&\,+\epsilon^\frac{7}8a^{-3-c\sqrt{\epsilon}}\|\Sigma\|_{H^{J}_G}^2+\epsilon^\frac78a^{-3-c\sqrt{\epsilon}}\|N\|_{H^J_G}^2\\
&\,+\epsilon^\frac{15}8a^{-3-c\sqrt{\epsilon}}\|\Gamma-\Gamhat\|_{H^{J-1}_G}^2\,.
\end{align*}
Further, we analogously get
\begin{equation*}
-\del_t\|\Gamma-\Gamhat\|_{L^2_G}^2\lesssim\epsilon^\frac18a^{-3}\|\Gamma-\Gamhat\|_{L^2_G}^2+\epsilon^{-\frac18}a^{-3}\left(\|\Sigma\|_{H^1_G}^2+\|N\|_{H^1_G}^2\right)\,
\end{equation*}
and thus, with the Gronwall lemma and \eqref{eq:log-est}, 
\begin{align*}
\|\Gamma-\Gamhat\|_{L^2_G(\Sigma_t)}^2\lesssim&\,a^{-c\epsilon^\frac18}\left(\epsilon^4+\int_t^{t_0}\epsilon^{-\frac18}a(s)^{-3}\left(\|\Sigma\|_{H^1_G(\Sigma_s)}^2+\|N\|^2_{H^1_G(\Sigma_s)}\right)\,ds\right)\\
\lesssim&\,a^{-c\epsilon^{\frac18}}\left(\epsilon^4+\epsilon^{-\frac14}\sup_{s\in(t,t_0)}\left(\|\Sigma\|_{H^1_G(\Sigma_s)}^2+\|N\|^2_{H^1_G(\Sigma_s)}\right)\right)\,.
\end{align*}
This proves \eqref{eq:norm-est-Chr} for $l=0$, and we assume for an iterative argument that the statement has been proved for $l=J-1$. Then, we obtain (estimating the error terms in $\Sigma$ and $N$ by their supremum immediately):
\begin{align*}
-\del_t\|\Gamma-\Gamhat\|_{\dot{H}^J_G}^2\lesssim&\,\epsilon^\frac18a^{-3}\|\Gamma-\Gamhat\|_{\dot{H}^J_G}^2+\epsilon^{-\frac18}a^{-3}\left(\|N\|_{H^J_G}^2+\|\Sigma\|_{H^J_G}^2\right)\\
&\,+\epsilon^{\frac{7}8+4}a^{-3-c\epsilon^\frac18}+\epsilon^\frac{7}8a^{-3-c\epsilon^\frac18}\sup_{s\in(\cdot,t_0)}\left(\|N\|_{H^{J-1}(\Sigma_s)}^2+\|\Sigma\|_{H^{J-1}(\Sigma_s)}^2\right)\,.
\end{align*}
After integrating, applying the Gronwall lemma and dealing with the first line as before, we get
\begin{align*}
\|\Gamma-\Gamhat\|_{\dot{H}^J_G}^2\lesssim&\,\epsilon^4a^{-c\epsilon^\frac18}+\epsilon^{-\frac14}a^{-c\epsilon^\frac18}\sup_{s\in(\cdot,t_0]}\left(\|N\|_{H^J_G(\Sigma_s)}^2+\|\Sigma\|_{H^J_G(\Sigma_s)}^2\right)\\
&\,+\epsilon^{\frac{6}8+4}a^{-c\epsilon^\frac18}+\epsilon^\frac{6}8a^{-c\epsilon^\frac18}\sup_{s\in(\cdot,t_0)}\left(\|N\|_{H^{J-1}(\Sigma_s)}^2+\|\Sigma\|_{H^{J-1}(\Sigma_s)}^2\right)\,,
\end{align*}
where the second line can obviously be absorbed into the first up to constant. Combining this with the assumption yields \eqref{eq:norm-est-Chr} for $l=J$ and thus iteratively for all $l\leq 17$.
\end{proof}

\begin{lemma}[Sobolev norm estimates for the metric]\label{lem:norm-est-G} For any $t\in(t_{Boot},t_0]$ and any $l\in\N,\,l\leq 18$, we have:
\begin{equation}\label{eq:norm-est-G}
\|G-\gamma\|_{H^l_G(\Sigma_t)}^2\lesssim a^{-c\epsilon^\frac18}\left(\epsilon^4+\epsilon^{-\frac14}\sup_{s\in(t,t_0)}\left(\|N\|_{H^{l}_G(\Sigma_s)}^2+\|\Sigma\|_{H^{l}_G(\Sigma_s)}^2\right)\right)
\end{equation}
\end{lemma}
\begin{proof}
For $l=0$, we compute the following using \eqref{eq:REEqG} and \eqref{eq:REEqG-1}:
\begin{align*}
-\del_t\|G-\gamma\|_{L^2_G}^2=&\,\int_M\Bigr\{-2(\del_tG^{-1})^{i_1j_1}(G^{-1})^{i_2j_2}(G-\gamma)_{i_1i_2}(G-\gamma)_{j_1j_2}-2\langle\del_tG,G-\gamma\rangle_G\\
&\,\phantom{\int_M}-3N\frac{\dot{a}}a\lvert G-\gamma\rvert_G^2\,\Bigr\}\vol{G}\\
=&\,\int_M \Bigr\{(N+1)a^{-3} \left[\Sigma\ast(G-\gamma)+\Sigma\right]\ast(G-\gamma)+N\frac{\dot{a}}a\lvert G-\gamma\rvert_G^2\\
&\,\phantom{\int_M}-4N\frac{\dot{a}}a\langle G,G-\gamma\rangle_G\Bigr\}\vol{G}
\end{align*}
We apply \eqref{eq:APSigma} and \eqref{eq:BsN} and get
\begin{align*}
-\del_t\|G-\gamma\|_{L^2_G}^2\lesssim&\,\epsilon a^{-3}\|G-\gamma\|_{L^2_G}^2+a^{-3}\left(\|\Sigma\|_{L^2_G}+\|N\|_{L^2_G}\right)\|G-\gamma\|_{L^2_G}\\
\lesssim&\,\epsilon^\frac18 a^{-3}\|G-\gamma\|_{L^2_G}^2+\epsilon^{-\frac18}a^{-3}\left(\|\Sigma\|_{L^2_G}^2+\|N\|_{L^2_G}^2\right)\,.
\end{align*}
After integrating and applying the Gronwall lemma (as well as the initial data assumption), we obtain
\begin{align*}
\|G-\gamma\|_{L^2_G(\Sigma_t)}^2\lesssim&\, a^{-c\epsilon^\frac18}\left(\epsilon^4+\epsilon^{-\frac18}\int_t^{t_0}a(s)^{-3}\left(\|\Sigma\|_{L^2_G(\Sigma_s)}^2+\|N\|_{L^2_G(\Sigma_s)}^2\right)\,ds\right)\,.
\end{align*}
The statement for $l=0$ now follows taking the supremum over the norms under the integral and applying \eqref{eq:log-est}. This extends to higher orders via the same iteration argument as in Lemma \ref{lem:int-est-Chr}.\\

\end{proof}

\section{Big Bang stability: Improving the bootstrap assumptions}\label{sec:bs-imp}

In this section, we combine the energy estimates obtained in the last two sections to improve the boostrap assumptions for the energies themselves, and then show how this improves the behaviour of the solution norms. For an outline of the energy improvement arguments that we perform in \change{Section \ref{subsec:bs-imp-core}}, we refer back to Remark \ref{rem:en-est-strat}.

Before carrying out the improvements themselves, we quickly collect an estimate that shows that combining Lemmas \ref{lem:en-est-BR} and \ref{lem:en-est-Sigma} yields sufficient control on the energies themselves:

\begin{lemma}\label{lem:en-error-cancellation} Let $L\in 2\N$. Then, the following estimate is satisfied:
\begin{equation}
16\pi C^2a^{-3}(N+1)\langle\Lap^\frac{L}2\RE,\Lap^{\frac{L}2}\Sigma\rangle_G+8\pi C^2\frac{\dot{a}}a(N+1)\lvert\Lap^\frac{L}2\Sigma\rvert_G^2+6\frac{\dot{a}}a(N+1)\lvert\Lap^\frac{L}2\RE\rvert_G^2\geq 0
\end{equation}
\end{lemma}
\begin{proof}
First, we recall that $N+1>0$ holds by Lemma \ref{lem:lapse-maxmin}. Additionally, we can apply \eqref{eq:diff-ineq-Friedman} and the Young inequality and get
\begin{align*}
\left\lvert 16\pi C^2a^{-3}(N+1)\langle\Lap^\frac{L}2\RE,\Lap^{\frac{L}2}\Sigma\rangle_G\right\rvert\leq&\,16\pi C^2\cdot\sqrt{\frac3{4\pi C^2}}\frac{\dot{a}}a\cdot (N+1)\lvert\Lap^{\frac{L}2}\RE\rvert_G\lvert\Lap^{\frac{L}2}\Sigma\rvert_G\\
\leq&\,4(N+1)\frac{\dot{a}}a \cdot\left(\sqrt{3}\cdot\lvert\Lap^\frac{L}2\RE\rvert_G\right)\cdot\left(\sqrt{4\pi C^2}\lvert\Lap^{\frac{L}2}\Sigma\rvert_G\right)\\
\leq
&\,6\frac{\dot{a}}a(N+1)\lvert\Lap^\frac{L}2\RE\rvert_G^2+8\pi C^2\frac{\dot{a}}a(N+1)\lvert\Lap^\frac{L}2\Sigma\rvert_G^2\,.
\end{align*}
\end{proof}
This shows that $\E^{(L)}(W,\cdot)+4\pi C^2\E^{(L)}(\Sigma,\cdot)$ is controlled by the sum of the left hand sides of the inequalities in Lemmas \ref{lem:en-est-BR} and \ref{lem:en-est-Sigma} for $L\in 2\N, 0\leq L\leq \change{18}$. 

\subsection{\change{Improving energy bounds}}\label{subsec:bs-imp-core}

\change{\begin{prop}[Improved energy bounds]\label{prop:en-bs-imp} Under the bootstrap assumptions (see Assumption \ref{ass:bootstrap}) and the initial data assumptions in Assumption \eqref{ass:init}, the following improved estimates hold on $(t_{Boot},t_0]$:
\begin{subequations}
\begin{align}
\E^{(\leq \change{18})}(\phi,\cdot)\lesssim&\,\epsilon^4a^{-c\epsilon^\frac18}\label{eq:en-bs-imp-SF}\\
\E^{(\leq \change{18})}(\Sigma,\cdot)\lesssim&\,\epsilon^\frac{15}4a^{-c\epsilon^\frac18}\label{eq:en-bs-imp-Sigma}\\
\E^{(\leq \change{18})}(W,\cdot)\lesssim&\,\epsilon^\frac{15}4a^{-c\epsilon^\frac18}\label{eq:en-bs-imp-BR}\\
\E^{(\leq 16)}(\Ric,\cdot)\lesssim&\,\epsilon^\frac72a^{-c\epsilon^\frac18}\label{eq:en-bs-imp-Ric}\\
\E^{(\leq 16)}(N,\cdot)+a^4\E^{(17)}(N,\cdot)+a^8\E^{(18)}(N,\cdot)\lesssim&\,\epsilon^\frac72a^{8-c\epsilon^\frac18}\label{eq:en-bs-imp-N}
\end{align}
\end{subequations}
\end{prop}}
\change{
\begin{proof}
We prove this estimate by performing an induction over even energy orders. Starting at order $0$, we first observe that by Lemma \ref{lem:en-error-cancellation}, we can bound the (base level) total energy
\[\E^{(0)}_{total}:=\E^{(0)}(\phi,\cdot)+\epsilon^\frac14\left(\E^{(0)}(W,\cdot)+4\pi C^2\E^{(0)}(\Sigma,\cdot)\right)+a^4\E^{(1)}(\phi,\cdot)+\epsilon^\frac12\E^{(1)}(\Sigma,\cdot)\]
by the sum of the left hand side of \eqref{eq:en-est-SF0}, the left hand side of \eqref{eq:en-est-BR0} weighted by $\epsilon^\frac14$ and the left hand side of \eqref{eq:en-est-Sigma0} weighted by $4\pi C^2\cdot\epsilon^\frac14$, and the left hand sides of \eqref{eq:en-est-SF-top} and $\epsilon^\frac12\cdot$\eqref{eq:en-est-Sigma-top} extended to $L=0$.\footnote{We need to weight $\E^{(0)}(\Sigma,\cdot)$ in the total energy by $K\cdot\epsilon^\frac14$ for some $K>0$ to balance out the $\epsilon^{-\frac18}$-weight from the scalar field energy on the right hand side of \eqref{eq:en-est-Sigma0}. The weight on the Bel-Robinson energy is then needed to obtain the cancellation in Lemma \ref{lem:en-error-cancellation}\change{. The additional weight on $a^4\E^{(1)}(\Sigma,\cdot)$ is needed so that the div-curl-estimates only generates a term of size $\epsilon^\frac14\E^{(L)}_{total}$ that can be absorbed later in the argument.}} Combining said estimates and inserting the initial data assumption from \eqref{eq:init-ass-en}, the following holds in total:
\begin{equation}\label{eq:en-eq-total0}
\E^{(0)}_{total}(t)\lesssim\epsilon^4+\int_t^{t_0}\left(\epsilon^\frac18a(s)^{-3}+a(s)^{-1-c\sqrt{\epsilon}}\right)\E^{(0)}_{total}(s)\,ds
\end{equation}
Applying the Gronwall lemma (see Lemma \ref{lem:gronwall}) to \eqref{eq:en-eq-total0}, we get for some suitable constant $c^\prime>0$:
\begin{equation*}\label{eq:en-bs-imp-base}
\E^{(0)}_{total}(t)\lesssim \epsilon^4\exp\left(c^\prime\int_t^{t_0}\epsilon^\frac18 a(s)^{-3}+a(s)^{-1-c\sqrt{\epsilon}}\,ds\right)\lesssim \epsilon^4a^{-c^\prime\epsilon^\frac18}
\end{equation*}
Now assume that, for $L\in2\N, 2\leq L\leq 18$, we have already shown
\begin{subequations}
\begin{align}
\E^{(\leq L-2)}(\phi,\cdot)+\epsilon^\frac14\left(\E^{(\leq L-2)}(\Sigma,\cdot)+\E^{(\leq L-2)}(W,\cdot)\right)\lesssim&\,\epsilon^4 a^{-c\epsilon^\frac18}\label{eq:ind-ass-core-en}
\end{align}
on $(t_{Boot},t_0]$. Note that \eqref{eq:en-bs-imp-base} means this holds true for $L=2$ after updating $c>0$. Further, if $L\geq 4$ holds, we assume
\begin{align}
\E^{(\leq L-4)}(\Ric,\cdot)\lesssim&\,\epsilon^{\frac72}a^{-c\epsilon^\frac18}\,.\label{eq:ind-ass-Ric}
\end{align}
\end{subequations}
We will show that these assumptions hold at $L=4$ after having shown the induction step for $L=2$.
We define, for $2\leq L\leq 18$,
\begin{align*}
\E^{(L)}_{total}:=&\,\E^{(L)}(\phi,\cdot)+\epsilon^\frac14\left(\E^{(L)}(W,\cdot)+4\pi C^2\E^{(L)}(\Sigma,\cdot)\right)+a^4\E^{(L+1)}(\phi,\cdot)+\epsilon^\frac12a^4\E^{(\changefinal{L+1})}(\Sigma,\cdot)\\
&\,+\epsilon^\frac12\E^{(L-2)}(\Ric,\cdot)+\epsilon^\frac34a^4\E^{(L-1)}(\Ric,\cdot)\,.
\end{align*}
We combine the respective energy estimates with the appropriate scalings\footnote{The weights on all terms beside the curvature energies are necessary for the same reasons as at order $0$. We need to scale the curvature energy at order $L$ by $\epsilon^\frac12$ to account for $\epsilon^{-\frac18}a^{-3}\E^{(L)}(\Sigma,\cdot)$ in \eqref{eq:en-est-Ric}, and the weight on the $L+1$-order curvature energy again needs to be chosen according to that on $\E^{(L+1)}(\Sigma,\cdot)$.}, namely (in the listed order) \eqref{eq:en-est-SF}, \eqref{eq:en-est-BR}, \eqref{eq:en-est-Sigma}, \eqref{eq:en-est-SF-top}, \eqref{eq:en-est-Sigma-top}, \eqref{eq:en-est-Ric} and \eqref{eq:en-est-Ric-top}. Observe that the sum of these scaled left hand sides controls $\E_{total}^{(L)}$ by Lemma \ref{lem:en-error-cancellation}. Combining all of these estimates and inserting the initial data assumption \eqref{eq:init-ass-en}, we get the following estimate:

\begin{subequations}
\begin{align}
\E_{total}^{(L)}(t)\lesssim&\,\epsilon^4+\int_t^{t_0}\left(\epsilon^\frac18a(s)^{-3}+a(s)^{-1-c\sqrt{\epsilon}}\right)\E_{total}^{(L)}(s)\label{eq:en-est-top-total1}\,ds\\
&\,+\int_t^{t_0}\left\{\epsilon^\frac38a(s)^{-3-c\sqrt{\epsilon}}\left[\E^{(\leq L-2)}(\phi,s)+\E^{(\leq L-2)}(\Sigma,s)\right]\right.\label{eq:en-est-top-total2}\\
&\,\phantom{+\int_t^{t_0}\bigr\{}
+\epsilon^\frac{17}8a(s)^{-3-c\sqrt{\epsilon}}\E^{(\leq L-2)}(W,s)+\underbrace{\epsilon^\frac{11}8\E^{(\leq L-4)}(\Ric,s)}_{\text{if }L=4}\bigr\}\,ds
\label{eq:en-est-top-total4}\\
&\,+\epsilon^\frac14\left(a(t)^{4-c\sqrt{\epsilon}}+\epsilon a(t)^{2-c\sqrt{\epsilon}}\right)\cdot \epsilon^\frac12\E^{(L)}(\Sigma,t)+\epsilon^\frac12\E^{(L)}(\phi,t)+\epsilon^\frac14\cdot \epsilon^\frac14\E^{(L)}(W,t)\label{eq:en-est-top-total5}\\
&\,+\epsilon^\frac74\cdot \epsilon^\frac34a^4\E^{(L-1)}(\Ric,t)+\epsilon^\frac52 a^{-c\sqrt{\epsilon}}\E^{(\leq L-2)}(\phi,t)+\epsilon^\frac12 a^{2-c\sqrt{\epsilon}}\E^{(\leq L-2)}(\Sigma,t)\label{eq:en-est-top-total6}\\
&\,+\epsilon^\frac32 a^{2-c\sqrt{\epsilon}}\E^{(\leq L-2)}(\Ric,t)\label{eq:en-est-top-total7}
\end{align}
\end{subequations}

We briefly summarize which inequalities contain the listed error term bounds as explicit terms: The first two terms in \eqref{eq:en-est-top-total2} come from \eqref{eq:en-est-Ric} and the latter from \eqref{eq:en-est-BR}
, those in \eqref{eq:en-est-top-total4} from 
\eqref{eq:en-est-SF} and \eqref{eq:en-est-Ric}, and finally the last three lines are precisely the scaled right hand side of \eqref{eq:en-est-Sigma-top}. Regarding the curvature energies in the various individual energy estimates, any summand with $\E^{(L-3)}(\Ric,\cdot)$ can be split using \eqref{eq:ibp-trick}, the resulting summands containing $\E^{(L-2)}(\Ric,\cdot)$ can always be absorbed into the total energy term in the first line, and anything with $\E^{(\leq L-4)}(\Ric,\cdot)$ is tracked in $\langle\text{Err}\rangle_L$ for $L\geq 4$.\\

\noindent Inserting \eqref{eq:ind-ass-core-en}-\eqref{eq:ind-ass-Ric}, \eqref{eq:en-est-top-total2}-\eqref{eq:en-est-top-total4} can be bounded up to constant by $\epsilon^\frac{33}8a^{-3-c\epsilon^\frac18}$. Here, the error term dominating all others arises from 
\[\epsilon^\frac38a(s)^{-3-c\sqrt{\epsilon}}\E^{(\leq L-2)}(\Sigma,s)\,.\]

Regarding \eqref{eq:en-est-top-total5}-\eqref{eq:en-est-top-total7}, notice that the first four summands can be bounded from above by $\epsilon^\frac14\E_{total}^{(L)}(t)$ up to constant. For the remaining three terms, we can again insert the induction assumptions \eqref{eq:ind-ass-core-en}-\eqref{eq:ind-ass-Ric}, bounding them by $\epsilon^\frac{17}4a(t)^{-c\epsilon^\frac18}$. \\

\noindent In summary and after rearranging, \eqref{eq:en-est-top-total1}-\eqref{eq:en-est-top-total7} becomes for some constant $K>0$:
\begin{align*}
(1-K\epsilon^\frac14)\E_{total}^{(L)}(t)\lesssim&\,\epsilon^4+\int_t^{t_0}\left(\epsilon^\frac18a(s)^{-3}+a(s)^{-1-c\sqrt{\epsilon}}\right)\E_{total}^{(L)}(s)\,ds+\int_t^{t_0}\epsilon^\frac{33}8a(s)^{-3-c\epsilon^\frac18}\,ds\\
&\,+\epsilon^\frac{17}4a(t)^{-c\epsilon^\frac18}\\
\lesssim&\,\epsilon^4a(t)^{-c\epsilon^\frac18}+\int_t^{t_0}\left(\epsilon^\frac18a(s)^{-3}+a(s)^{-1-c\sqrt{\epsilon}}\right)\E_{total}^{(L)}(s)\,ds\,
\end{align*} 
The prefactor on the left hand side is positive for small enough $\epsilon>0$, and the Gronwall lemma then yields
\begin{equation}\label{eq:total-en-imp}
\E^{(L)}_{total}(t)\lesssim \epsilon^4a^{-c\epsilon^\frac18}\,.
\end{equation}
In particular, this directly implies that the induction assumptions \eqref{eq:ind-ass-core-en} and \eqref{eq:ind-ass-Ric}, using \eqref{eq:drop-odd-en} to cover the skipped odd order, hold at order $L$, completing the induction step, and clearly also that \eqref{eq:ind-ass-Ric} holds for $L-2=2$ using \eqref{eq:total-en-imp} at order $2$. This completes the induction argument, proving \eqref{eq:en-bs-imp-SF}-\eqref{eq:en-bs-imp-Ric}.
Finally, applying the obtained improved estimates for $\nabla\phi$ and $\Ric[G]$ to Corollary \ref{cor:en-est-lapse-tilde}, we also get \eqref{eq:en-bs-imp-N}. 
\end{proof}}

\subsection{Improving solution norm control}\label{subsec:bs-imp-norm}

To close the bootstrap argument, it now remains to show that the improved energy bounds also imply improved bounds for $\mathcal{H}$ and $\mathcal{C}$. The former follows almost directly using Lemma \ref{lem:Sobolev-norm-equivalence-improved}:

\begin{corollary}[Improved Sobolev norm bounds]\label{cor:H-imp} On $(t_{Boot},t_0]$, the following estimates hold:
\begin{subequations}
\change{\begin{align}
\mathcal{H}\lesssim&\,\epsilon^\frac74 a^{-c\epsilon^\frac18}\label{eq:H-norm-imp}\\
\|\Sigma\|_{H^{18}_G}^2\lesssim&\,\epsilon^\frac{15}4a^{-c\epsilon^\frac18}\label{eq:H-Sigma-imp}\\
\|N\|_{H^{18}_G}^2\lesssim&\,\epsilon^4a^{-c\epsilon^\frac18}\label{eq:H-lapse-imp}
\end{align}}
\end{subequations}
\end{corollary}
\begin{proof}
First, we apply the improved energy estimates from \change{Proposition \ref{prop:en-bs-imp} }as well as the strong $C_G$-norm bounds from Lemma \ref{lem:AP} to the near-coercivity estimates in Lemma \ref{lem:Sobolev-norm-equivalence-improved}. With this, we directly obtain the following Sobolev norm estimates (updating $c$):
\change{\begin{align*}
\|\Psi\|_{H_G^{\change{18}}}^2\lesssim&\,\epsilon^4a^{-c\epsilon^\frac18}+\epsilon a^{-c\epsilon^\frac18}\cdot\epsilon^\frac{15}4a^{-c\epsilon^\frac18}\lesssim\epsilon^4a^{-c\epsilon^\frac18}\\
\|\Sigma\|_{H^{18}_G}^2\lesssim&\,\epsilon^\frac{15}4a^{-c\epsilon^\frac18}\\
\|\Ric[G]+\frac29G\|_{H^{16}_G}^2\lesssim&\,\epsilon^\frac72 a^{-c\epsilon^\frac18}\\
\|\RE\|_{H^{18}_G}^2+\|\RB\|_{H^{18}_G}^2\lesssim&\,\epsilon^\frac{15}4a^{-c\epsilon^\frac18}
\end{align*}}
\change{By Lemma \ref{lem:norm-est-nablaphi}, we also have
\changefinal{\begin{equation}\label{eq:nabla-phi-norm-imp}
\|\nabla\phi\|_{H^{17}_G}\lesssim \left(1+\epsilon a^{-c\sqrt{\epsilon}}\right)\|\Sigma\|_{H^{18}_G}+\epsilon\|\Psi\|_{H^{18}_G}\lesssim \epsilon^\frac{15}4a^{-c\epsilon^\frac18}\,.
\end{equation}}}
We take particular care in showing that the improved bound holds for \changefinal{$a^2\|\nabla\phi\|_{\dot{H}^{18}_G}$}: First, note that \eqref{eq:en-bs-imp-Ric} \change{implies }$\E^{(\leq 17)}(\Ric,\cdot)\lesssim\epsilon^\frac72 a^{-c\epsilon^\frac18}$. Applying this along with \eqref{eq:APmidphi} to \changefinal{\eqref{eq:Sobolev-norm-equiv-nablazeta2l}}, as well as \eqref{eq:Sobolev-norm-equiv-zetalow} in the second line and \change{\eqref{eq:en-bs-imp-SF} as well as \eqref{eq:nabla-phi-norm-imp} }in the final step, we obtain:
\changefinal{\begin{align*}
a^4\|\nabla\phi\|_{H^{18}_G}^2\lesssim&\, a^4\|\nabla\Lap^9\phi\|_{L^2_G}^2+a^{4-c\sqrt{\epsilon}}\sum_{m=0}^8\|\nabla\Lap^m\phi\|_{L^2_G}^2+\epsilon a^{4-c\sqrt{\epsilon}}\cdot \E^{(\leq 16)}(\Ric,\cdot)\\
\lesssim&\,a^{4-c\sqrt{\epsilon}}\left(\|\nabla\Lap^{9}\phi\|_{L^2_G}^2+\|\nabla\phi\|_{H^{17}_G}^2\right)+\epsilon^\frac{9}2a^{-c\epsilon^\frac18}\\
\lesssim&\,a^{-c\sqrt{\epsilon}}\cdot \E^{(\leq 18)}(\phi,\cdot)+a^{4-c\sqrt{\epsilon}}\|\nabla\phi\|_{H^{17}_G}^2+\epsilon^\frac{9}2a^{-c\epsilon^\frac18}\\
\lesssim&\,\epsilon^\frac{15}4 a^{-c\epsilon^\frac18}
\end{align*}}
Further, inserting \eqref{eq:en-bs-imp-SF}, \eqref{eq:en-bs-imp-Sigma} and \eqref{eq:en-bs-imp-Ric} into Corollary \ref{cor:en-est-lapse} implies
\begin{align*}
\,a^8\|\Lap^{10} N\|^2_{L^2_G}+a^4\|\nabla\Lap^{9}N\|^2_{L^2_G}+\sum_{m=0}^{9}\|\Lap^mN\|_{L^2_G}^2
\lesssim&\,\epsilon^\frac{11}4a^{-c\epsilon^\frac18}
\end{align*}
and subsequently, applying Lemma \ref{lem:Sobolev-norm-equivalence-improved} as before,
\begin{equation*}
a^8\|N\|^2_{\dot{H}^{20}_G}+a^4\|N\|^2_{\dot{H}^{19}_G}+\|N\|^2_{H^{18}_G}\lesssim\epsilon^4a^{-c\epsilon^\frac18}\,.
\end{equation*}
Finally, having now shown \eqref{eq:H-Sigma-imp} and \eqref{eq:H-lapse-imp}, we can apply these to \eqref{eq:norm-est-G} to get
\begin{equation*}
\|G-\gamma\|^2_{H^{18}_G}\lesssim a^{-c\epsilon^\frac18}\left(\epsilon^4+\epsilon^{-\frac14+\frac{15}4}+\epsilon^{-\frac14+4}\right)\lesssim\epsilon^\frac{7}2a^{-c\epsilon^\frac18}\,,
\end{equation*}
proving \eqref{eq:H-norm-imp}.
\end{proof}

Intuitively, the bounds on $\mathcal{C}$ should now follow from $\mathcal{H}$ by the standard Sobolev embedding. However, since both of these norms are with respect to $G$, the embedding constant may be time dependent. To circumvent this issue, we need to switch between norms with \changefinal{respect }to $G$ and $\gamma$ and then apply the embedding with respect to $C_\gamma$ and $H_\gamma$. The following lemma ensures that we still obtain bootstrap improvements after performing these norm switches:

\begin{lemma}[Moving between norms]\label{lem:G-gamma-norm-switch} Let $l\in\N,\,l\leq \change{18}$, $\zeta$ be a scalar field, $\mathfrak{T}$ be an arbitrary $\Sigma_t$-tangent tensor\delete{ and let $U$ be either a coordinate neighbourhood on $\Sigma_t$ or $\Sigma_t$ itself}. Then, for some multivariate polynomial $P_l$ with $P_l(0,0)=0$, we have
\begin{subequations}
\begin{align}
\|\zeta\|_{C^l_G(U)}\lesssim&\,a^{-c\sqrt{\epsilon}}\|\zeta\|_{C^l_\gamma(\change{M})}+a^{-c\sqrt{\epsilon}}\|\zeta\|_{C_\gamma^{\max\left\{0,\lfloor\frac{l-1}2\rfloor\right\}}(\change{M})}\cdot P_l\left(\|G-\gamma\|_{C^{l-1}_\gamma(\change{M})},\|G^{-1}-\gamma^{-1}\|_{C^{l-1}_\gamma(\change{M})}\right) \label{eq:C-norm-exch-zeta}\\
\|\mathfrak{T}\|_{C^l_G(\change{M})}\lesssim&\,a^{-c\sqrt{\epsilon}}\|\mathfrak{T}\|_{C^l_\gamma(\change{M})}+a^{-c\sqrt{\epsilon}}\|\mathfrak{T}\|_{C_\gamma^{\max\left\{0,\lfloor\frac{l-1}2\rfloor\right\}}(\change{M})}\cdot P_l\left(\|G-\gamma\|_{C^l_\gamma(\change{M})},\|G^{-1}-\gamma^{-1}\|_{C^{l}_\gamma(\change{M})}\right)\label{eq:C-norm-exch-T}
\end{align}
\end{subequations}
as well as the same inequalities with the roles of $G$ and $\gamma$ reversed. For $l\leq 12$, this reduces to:
\begin{subequations}
\begin{align}
a^{c\sqrt{\epsilon}}\|\zeta\|_{C^l_\gamma(\change{M})}\lesssim\|\zeta\|_{C^l_G(\change{M})}\lesssim\,a^{-c\sqrt{\epsilon}}\|\zeta\|_{C^l_\gamma(\change{M})} \label{eq:AP-exch-zeta}\\
a^{c\sqrt{\epsilon}}\|\mathfrak{T}\|_{C^l_\gamma(\change{M})}\lesssim\|\mathfrak{T}\|_{C^l_G(\change{M})}\lesssim\,a^{-c\sqrt{\epsilon}}\|\mathfrak{T}\|_{C^l_\gamma(\change{M})}\,\label{eq:AP-exch-T}
\end{align}
\end{subequations}
Further, one has
\begin{subequations}
\begin{align}
\|\zeta\|^2_{H^l_\gamma(\change{M})}\lesssim&\,a^{-c\sqrt{\epsilon}}\|\zeta\|^2_{H^l_G(\change{M})}+a^{-c\epsilon^\frac18}\|\zeta\|_{C_G^{\lceil \frac{l-1}2\rceil}(M)}^2\left(\epsilon^4+\epsilon^{-\frac14}\sup_{s\in(\cdot,t_0)}\left(\|N\|_{H^{l-1}_G(\Sigma_s)}^2+\|\Sigma\|_{H^{l-1}_G(\Sigma_s)}^2\right)\right)\label{eq:H-norm-exch-zeta}\,,\\
\|\mathfrak{T}\|^2_{H^l_\gamma(\change{M})}\lesssim&\,a^{-c\sqrt{\epsilon}}\|\mathfrak{T}\|^2_{H^l_G(\change{M})}+a^{-c\epsilon^\frac18}\|\mathfrak{T}\|_{C_G^{\lceil \frac{l-1}2\rceil}(M)}^2\left(\epsilon^4+\epsilon^{-\frac14}\sup_{s\in(\cdot,t_0)}\left(\|N\|_{H^{l}_G(\Sigma_s)}^2+\|\Sigma\|_{H^{l}_G(\Sigma_s)}^2\right)\right) \label{eq:H-norm-exch-T}
\end{align}
\end{subequations}
\end{lemma}
\begin{remark}
While we only need the tensorial inequalities for gradient vector fields and $(0,2)$-tensors when applied to norms in $\mathcal{H}$ and $\mathcal{C}$, the proof is simpler when considering tensors of arbitrary rank.
\end{remark}
\begin{proof}
We restrict ourselves to proving the tensorial statements; the scalar field analogues follow analagously except for the fact that, since $\nabla_i\zeta=\nabhat_i\zeta=\del_i\zeta$, error terms caused by Christoffel symbols always enter at one order less. Thus, it remains to show \eqref{eq:C-norm-exch-T}, \eqref{eq:AP-exch-T} and \eqref{eq:H-norm-exch-T} by iterating over derivative order. \\
Starting with the base level estimates, we have if $\mathfrak{T}$ is of rank $(r,s)$:
\begin{align*}
\lvert\mathfrak{T}\rvert_G^2-\lvert\mathfrak{T}\rvert_\gamma^2=&\,\Bigr[G_{i_1j_1}\dots G_{i_rj_r} (G^{-1})^{p_1q_1}\dots(G^{-1})^{p_sq_s}-\gamma_{i_1j_1}\dots \gamma_{i_rj_r} (\gamma^{-1})^{p_1q_1}\dots(\gamma^{-1})^{p_sq_s}\Bigr]\cdot\\
&\,\cdot{\mathfrak{T}^{i_1\dots i_r}}_{p_1\dots p_s}{\mathfrak{T}^{j_1\dots j_r}}_{q_1\dots q_s}
\end{align*}
We successively replace $G^{\pm 1}$ by $(G^{\pm 1}-\gamma^{\pm 1})+\gamma^{\pm 1}$, 
take the $\lvert\cdot\rvert_\gamma$-norm of each factor and use \eqref{eq:APmidG}-\eqref{eq:APmidG-1}. This yields
\begin{equation*}
\left\lvert\lvert\mathfrak{T}\rvert_{G}^2-\lvert\mathfrak{T}\rvert_\gamma^2\right\rvert\lesssim \sqrt{\epsilon}a^{-c\sqrt{\epsilon}}\lvert\mathfrak{T}\rvert_\gamma^2\,,
\end{equation*}
implying 
\eqref{eq:C-norm-exch-T} (and \eqref{eq:AP-exch-T}) for $l=0$ after rearranging and taking supremums suitably.\\
To show \eqref{eq:H-norm-exch-T} at base level, consider
\begin{align*}
&\,\int_{\change{M}} \lvert\mathfrak{T}\rvert_G^2\,\vol{G}-\int_{\change{M}}\lvert\mathfrak{T}\rvert_{\gamma}^2\,\vol{\gamma}\\
=&\,\int_{\change{M}} \left(\lvert\mathfrak{T}\rvert_G^2-\lvert\mathfrak{T}\rvert_\gamma^2\right)\,\vol{G}+\int_{\change{M}}\lvert\mathfrak{T}\rvert_\gamma^2\frac{\mu_G-\mu_\gamma}{\mu_\gamma}\,\vol{\gamma}\,.
\end{align*}
We can control the first summand on the right hand side 
as before, while we have $\lvert\mu_G-\mu_\gamma\rvert\lesssim \epsilon$ by \eqref{eq:APvol}. Hence,
\[(1-K\epsilon)\|\mathfrak{T}\|_{L^2_\gamma}^2\lesssim (1+\sqrt{\epsilon}a^{-c\sqrt{\epsilon}})\|\mathfrak{T}\|_{L^2_G}^2\]
follows for a suitable constant $K>0$, implying the statement for small enough $\epsilon>0$.

Next, we perform the iteration for \eqref{eq:C-norm-exch-T}, assuming the statement and the analogue with $\gamma$ and $G$ reversed to hold \change{up to order $l-1$}. As above, note that
\begin{equation*}
\left\lvert\,\left\lvert\nabla^J\mathfrak{T}\right\rvert_G^2-\left\lvert\nabhat^J\mathfrak{T}\right\rvert_\gamma^2\right\rvert\lesssim\sqrt{\epsilon}a^{-c\sqrt{\epsilon}}\left\lvert\nabhat^J\mathfrak{T}\right\rvert_\gamma^2+\change{(1+\sqrt{\epsilon}a^{-c\sqrt{\epsilon}})}\left\lvert\left\lvert\nabla^J\mathfrak{T}\right\rvert_\gamma^2-\left\lvert\nabhat^J\mathfrak{T}\right\rvert_\gamma^2\right\rvert\,
\end{equation*}
where we can rewrite the second term as
\begin{equation*}
\left\lvert 2\langle\nabhat^J\mathfrak{T}-\nabla^J\mathfrak{T},\nabhat\mathfrak{T}\rangle_\gamma-\lvert\nabla^J\mathfrak{T}-\nabhat^J\mathfrak{T}\rvert_\gamma^2\right\rvert
\end{equation*}
and hence obtain (moving between pointwise norms as before)
\begin{align*}
\left\lvert\,\left\lvert\nabla^J\mathfrak{T}\right\rvert_G^2-\left\lvert\nabhat^J\mathfrak{T}\right\rvert_\gamma^2\right\rvert
\lesssim&\,a^{-c\sqrt{\epsilon}}\left\lvert\nabhat^J\mathfrak{T}\right\rvert_\gamma^2+a^{-c\sqrt{\epsilon}}\left\lvert \nabla^J\mathfrak{T}-\nabhat^J\mathfrak{T}\right\rvert_\gamma^2\,.
\end{align*}
Regarding $\nabla^J\mathfrak{T}-\nabhat^J\mathfrak{T}$, we have the following schematic decomposition \delete{locally on some coordinate neighbourhood $V\subseteq\Sigma_t$}:
\begin{align*}
\numberthis\label{eq:nabla-nabhat}\nabla^J\mathfrak{T}-\nabhat^J\mathfrak{T}=&\,\sum_{I=0}^{J-1}\nabhat^{J-I-1}(\Gamma-\Gamhat)\ast_\gamma\left(\nabla^{I}\mathfrak{T}+\nabhat^{I}\mathfrak{T}\right)\\
&\,+\langle\text{at least cubic nonlinear terms}\rangle\,,
\end{align*} 
Here, $\ast_\gamma$ encodes the analogous schematic product notation with regard to $\gamma$ (see subsection \ref{subsubsec:schematic-notation}). Regarding the Christoffel symbols, notice \eqref{eq:C-norm-exch-T} with roles of $\gamma$ and $G$ reversed holding up to $l-1$ implies that, for any $m\in\{0,\dots,l-1\}$ and some multivariate polynomial $\tilde{P}_m$, we have
\[\|\Gamma-\Gamhat\|_{C_\gamma^{m}(\change{M})}\lesssim a^{-c\sqrt{\epsilon}}\tilde{P}_m(\|\Gamma-\Gamhat\|_{C_G^{m}(\change{M})},\|G-\gamma\|_{C^{\change{m}}_G(\change{M})},\|G^{-1}-\gamma^{-1}\|_{C^{\change{m}}_G(\change{M})})\,.\]
As explained in Remark \ref{rem:relation-metric-Chr}, we can bound $\|\Gamma-\Gamhat\|_{C_G^{m}(\change{M})}$ by a polynomial in $\|G-\gamma\|_{C_G^{m+1}(\change{M})}$. Hence, we can apply \eqref{eq:APmidG} to obtain 
\begin{equation}\label{eq:APmidChr}
\|\Gamma-\Gamhat\|_{C_\gamma^{l-1}(\change{M})}\lesssim \sqrt{\epsilon}a^{-c\sqrt{\epsilon}}\,.
\end{equation}
Moving back to \eqref{eq:nabla-nabhat} and just considering the first line for now, this implies 
\begin{align*}
\numberthis\label{eq:C-norm-exch-T-last-step}\left\lvert\|\mathfrak{T}\|^2_{\dot{C}^l_G(\change{M})}-\|\mathfrak{T}\|^2_{\dot{C}_\gamma^l(\change{M})}\right\rvert\lesssim&\,a^{-c\sqrt{\epsilon}}\left(\|\mathfrak{T}\|^2_{C^l_\gamma(\change{M})}+\sum_{m=0}^{l-1}\|\nabla^{m}\mathfrak{T}\|^2_{C^0_\gamma(\change{M})}\right)\\
&\,+\left(\|\mathfrak{T}\|^2_{C^{\lceil\frac{l-1}2\rceil}_\gamma(\change{M})}+\sum_{m=0}^{\lceil\frac{l-1}2\rceil}\|\nabla^{m}\mathfrak{T}\|^2_{C^0_\gamma(\change{M})}\right)\|\Gamma-\Gamhat\|^2_{C^{l-1}_\gamma(\change{M})}\,\\
&\,+\langle\text{at least cubic nonlinear terms}\rangle\,.
\end{align*}
We can rewrite $\nabla^m\mathfrak{T}$-norms in $C_\gamma$ as ones in $C_G$ up to $a^{-c\sqrt{\epsilon}}$ as before. Then, we can apply the already obtained estimates up to order $l-1$ show that the first two lines of the right hand side can be estimated by the right hand side of \eqref{eq:C-norm-exch-T}. The highly nonlinear terms can be dealt with similarly, closing the induction over admissible $l$. \delete{ on local coordinate [...]} 
The inequality in \eqref{eq:AP-exch-T} immediately follows by applying \eqref{eq:APmidG}-\eqref{eq:APmidG-1} and \eqref{eq:APmidChr}.



Now, assume \eqref{eq:H-norm-exch-T} to be proven up to order $J-1$. By analogous arguments as at order zero, we get, \change{after rearranging},
\change{\[\int_M\lvert\nabhat^J\mathfrak{T}\rvert_\gamma^2\,\vol{\gamma}\lesssim \left\lvert\int_M\left(\lvert\nabla^J\mathfrak{T}\rvert_G^2-\lvert\nabhat^J\mathfrak{T}\rvert_{\gamma}^2\right)\,\vol{G}\right\rvert+\int_M\lvert\nabla^J\mathfrak{T}\rvert_G^2\,\vol{G}\,,\]}
so we only need to concern ourselves with the first summand. Reversing roles of $G$ and $\gamma$ compared to the proof of \eqref{eq:C-norm-exch-T}, we get
\begin{equation*}
\left\lvert\,\left\lvert\nabla^J\mathfrak{T}\right\rvert_G^2-\left\lvert\nabhat^J\mathfrak{T}\right\rvert_\gamma^2\right\rvert\lesssim\sqrt{\epsilon}a^{-c\sqrt{\epsilon}}\left\lvert\nabla^J\mathfrak{T}\right\rvert_G^2+\change{a^{-c\sqrt{\epsilon}}}\left\lvert 2\langle\nabla^J\mathfrak{T}-\nabhat^J\mathfrak{T},\nabla\mathfrak{T}\rangle_G-\lvert\nabla^J\mathfrak{T}-\nabhat^J\mathfrak{T}\rvert_G^2\right\rvert
,
\end{equation*}
and have the following, applying Lemma \ref{lem:int-est-Chr} immediately to estimate $\|\Gamma-\Gamhat\|_{H^{l-1}_G}$:
\begin{align*}
&\,\left\rvert\int_V\left\{2\langle\nabla^l\mathfrak{T}-\nabhat^l\mathfrak{T},\nabla\mathfrak{T}\rangle_G-\lvert\nabla^l\mathfrak{T}-\nabhat^l\mathfrak{T}\rvert_G^2\right\}\vol{G}\right\rvert\\
\lesssim&\,a^{-c\sqrt{\epsilon}}\left(\|\mathfrak{T}\|^2_{H^l_G(\change{M})}+\|\nabhat^{\leq l-1}\mathfrak{T}\|^2_{H^{0}_G(\change{M})}\right)\\
&\,+\left(\|\mathfrak{T}\|^2_{C^{\lceil\frac{l-1}2\rceil}_G(M)}+\|\nabhat^{\leq \lceil\frac{l-1}2\rceil}\mathfrak{T}\|^2_{C^{0}_G(M)}\right)\cdot \\
&\,\quad \cdot a^{-c\epsilon^\frac18}\left(\epsilon^4+\epsilon^{-\frac14}\sup_{s\in(\cdot,t_0)}\left(\|N\|_{H^{l}_G(\Sigma_s)}^2+\|\Sigma\|_{H^{l}_G(\Sigma_s)}^2\right)\right)
\end{align*}
By the same arguments as earlier, we have $\|\nabhat^{\leq l-1}\mathfrak{T}\|_{H^{l-1}_G(\change{M})}\lesssim a^{-c\sqrt{\epsilon}}\|\mathfrak{T}\|_{H^{l-1}_\gamma(\change{M})}$ and can then apply the induction hypothesis. \change{This proves \eqref{eq:H-norm-exch-T}.}
\end{proof}

\begin{corollary}[Improved $C$-norm bounds]\label{cor:C-impr} On $(t_{Boot},t_0]$, the following estimate is satisfied:
\begin{equation}\label{eq:C-impr}
\change{\mathcal{C}+\mathcal{C}_\gamma\lesssim\changefinal{\epsilon^\frac74}a^{-c\epsilon^\frac18}}
\end{equation}
\end{corollary}
\begin{proof}
We first apply the Sobolev norm estimates in Lemma \ref{lem:G-gamma-norm-switch} \change{to \eqref{eq:H-norm-imp}}, to then control $\mathcal{C}_\gamma$ via the standard Sobolev embedding $H^{l+2}_\gamma(M)\hookrightarrow C^l(M)$, and finally control $\mathcal{C}$ with \eqref{eq:C-norm-exch-zeta}-\eqref{eq:C-norm-exch-T}.\\
Note that by Lemma \ref{lem:AP}, we can control the $C_G$-norm up to order $10$ of every quantity occurring in $\mathcal{H}$ beside the lapse by at worst $\sqrt{\epsilon}a^{-c\sqrt{\epsilon}}$, while the bootstrap assumption already implies better behaviour for the lapse. Thus, we can apply \eqref{eq:H-norm-exch-zeta}-\eqref{eq:H-norm-exch-T} to every norm appearing in $\mathcal{H}$, and obtain by applying \eqref{eq:H-Sigma-imp} and \eqref{eq:H-lapse-imp} in the second line:
\begin{align*}
\mathcal{C}_\gamma^2\lesssim&\,a^{-c\sqrt{\epsilon}}\cdot \mathcal{H}^2+\epsilon a^{-c\sqrt{\epsilon}}\cdot a^{-c\epsilon^\frac18}\left(\epsilon^4+\epsilon^{-\frac14}\sup_{s\in(t,t_0)}\left(\|N\|_{H^{18}_G(\Sigma_s)}^2+\|\Sigma\|_{H^{18}_G(\Sigma_s)}^2\right)\right)\\
\lesssim&\,\change{\epsilon^\frac72a^{-c\epsilon^\frac18}}+\epsilon a^{-c\epsilon^\frac18}\left(\epsilon^4+\change{\epsilon^\frac72}\right)\\
\lesssim&\,\change{\epsilon^\frac72}\cdot a^{-c\epsilon^\frac18}
\end{align*}
In particular, we can update $c$ such that
\[\lvert P(\|G-\gamma\|_{C^{16}_\gamma(\Sigma_t)},\|G-\gamma\|_{C^{16}_\gamma(\Sigma_t)})\rvert\lesssim\change{\epsilon^\frac72} a^{-c\epsilon^\frac18}\]
holds for any multivariate polynomial $P$ that appears when applying \eqref{eq:C-norm-exch-zeta}-\eqref{eq:C-norm-exch-T}. Again using the strong $C_G$-norm estimates from Lemma \ref{lem:AP}, this then implies $\mathcal{C}\lesssim \changefinal{\epsilon^\frac74} a^{-c\epsilon^\frac18}$.
\end{proof}

\section{Big Bang stability: The main theorem}\label{sec:main-thm}

In this section, we provide the proof of the first main result, Theorem \ref{thm:main-past}, which we state in more detail in Theorem \ref{thm:main} below. As in \cite{Rodnianski2014,Speck2018}, most of the work has already been done by establishing the necessary bounds on solution norms.

\begin{remark}[Existence of a CMC hypersurface]\label{rem:CMC-hypersurface}
As mentioned in \change{Section \ref{subsubsec:initial-data}}, it may seem that the generality of the results in Theorem \ref{thm:main} is restricted by taking the initial data on $\Sigma_{t_0}$ to be CMC. 
However, as long as one remains close enough to a constant time hypersurface of the FLRW reference metric (which is CMC), one can locally evolve the perturbed data in harmonic gauge to a nearby hypersurface that is CMC and remains close to the FLRW reference solution. 
\change{To make this a bit more precise, and also since this is a little less involved than the arguments in \cite{Rodnianski2014}, we will briefly sketch how the arguments from \cite[Section 2.5]{FajKr20} extend to our setting. \\%
First, we once again assume without loss of generality that our initial data is sufficiently regular. Note that we can locally evolve our data within harmonic gauge to get a $C^{17}$-regular family of metrics with near-FLRW initial data (for well-posedness, consider the analogue of \cite[Proposition 14.1]{Rodnianski2014}). Consider the Banach manifold $\mathcal{M}^{17}$ formed by the set of $C^{17}$ Lorentz metrics on $I\times\M$ for an open interval $I$ around $t_0$ such that the surfaces of constant time are Riemannian, endowed with the norm
\[\|\tilde{g}\|=\|\tilde{n}^2\|_{C^{17}_{dt^2+\gamma}(I\times M)}+\|\tilde{X}\|_{C^{17}_{dt^2+\gamma}(I\times M)}+\|\tilde{g}_t\|_{C^{17}_{dt^2+\gamma}(I\times M)}\,,\]
where $\tilde{g}\in\mathcal{M}^{17}$ has lapse $\tilde{n}$, shift $\tilde{X}$ and spatial metrics $(\tilde{g}_t)_{t\in I}$. Further, for any $f\in C^{17}(M,I)$, we define the embedding $\iota_f: M\hookrightarrow \M$ by $x\mapsto (f(x),x)$, and subsequently define the smooth map
\begin{gather*}
H_0:\mathcal{D}:=\{(\tilde{g},f)\in \mathcal{M}^{17}\times C^{17}(M,I)\vert \iota_f^\ast\tilde{g}\text{ is Riemannian}\}\longrightarrow C^{16}(M)\\
(\tilde{g},f)\mapsto \text{mean curvature of }(M,\tilde{g}_t)\text{ embedded along }\iota_f\,.
\end{gather*}
One easily checks that $(\g_{FLRW},t_0)$ is a regular point of $H_0$. By the implicit function theorem for Banach manifolds, this means there is a (unique) smooth function $F$ that maps an open neighbourhood of $\g_{FLRW}$ in $\mathcal{M}^{17}$ to an open neighbourbood of the constant function $x\mapsto t_0$ in $C^{17}(M,I)$ such that $H_0(\cdot,F(\cdot))=\tau(t_0)$ holds in that neighbourhood.\\
Thus, we can choose a surface $\Sigma^\prime$ with mean curvature $\tau(t_0)$ near the original $\Sigma_{t_0}$. Furthermore, for small enough $\epsilon>0$, the initial data on $\Sigma^\prime$ remains close to the FLRW initial data in the sense of Assumption \ref{ass:init}, using similar arguments to control Sobolev norms. Thus, we can replace $\Sigma_{t_0}$ by $\Sigma^\prime$ without loss of generality, proving that the CMC assumption \eqref{eq:CMC} is not a true restriction.}
\end{remark}

\begin{theorem}[Stability of Big Bang formation]\label{thm:main}
Let $(M,\mathring{g},\mathring{k},\mathring{\pi},\mathring{\psi})$ be initial data to the Einstein scalar-field system as discussed in Section \ref{subsubsec:initial-data}. Further, let the data be embedded into a time-oriented 4-manifold such that it induces initial data for the rescaled solution variables (see Definition \ref{def:rescaled}) at the initial hypersurface $\Sigma_{t_0}$. We also assume this rescaled initial data is close to that of the FLRW reference solution (see \eqref{eq:FLRW-metric} and \eqref{eq:FLRW-wave}) in the sense that
\begin{equation}\label{eq:ass-init-main-thm}
\mathcal{H}(t_0)\change{+\mathcal{H}_{top}(t_0)}+\mathcal{C}(t_0)\change{\leq}\epsilon^2
\end{equation}
is satisfied (with $\mathcal{H}$ and $\mathcal{C}$ as in Definition \ref{def:sol-norm}).  \delete{along with [...] }
\footnote{Essentially, this translates to smallness in $H_\gamma^{\change{19}}$ and $C_\gamma^{17}$ \delete{and [...]}. 
For $\epsilon=0$, the solution is the FLRW reference solution.}\\

Then, the past maximal globally hyperbolic development $((0,t_0]\times M,\g,\phi)$ of this data within the Einstein scalar-field system \eqref{eq:ESF1}-\eqref{eq:ESF2} in CMC gauge \eqref{eq:CMC} with zero shift is foliated by the CMC hypersurfaces $\Sigma_s=t^{-1}(\{s\})$, and one has
\begin{equation}\label{eq:bs-imp-main-thm}
\mathcal{H}(t)+\mathcal{C}(t)+\mathcal{C}_{\gamma}(t)\lesssim \change{\epsilon^\frac74} a(t)^{-c\epsilon^\frac18}
\end{equation}
for some $c>0$ and any $t\in(0,t_0]$. In particular, this implies the following statements:\\[1em]

\textbf{Asymptotic behaviour of solution variables:} We denote the solution metric as $\g=-n^2dt^2+g$, the second fundamental form (viewed as a $(1,1)$-tensor) with respect to $\Sigma_t$ as $k$ and the volume form with regard to $g$ on $\Sigma_t$ by $\vol{g}$. There exist a smooth function $\Psi_{Bang}\in C_\gamma^{15}(M)$, a $(1,1)$-tensor field $K_{Bang}\in C^{15}_\gamma(M)$ and a volume form $\vol{Bang}\in C^{15}_\gamma(M)$ such that the following estimates hold for any $t\in(0,t_0]$:
\begin{subequations}\label{eq:asymp}
\begin{align}
\label{eq:asymp-lapse}\|n-1\|_{C^l_\gamma(\Sigma_t)}\lesssim&\,\begin{cases} 
\epsilon a(t)^{4-c\epsilon^\frac18} & l\leq 14\\
\epsilon a(t)^{2-c\epsilon^\frac18} & l=15
\end{cases}\\
\label{eq:asymp-vol}\left\|a^{-3}\vol{g}-\vol{Bang}\right\|_{C^{l}_\gamma(\Sigma_t)}\lesssim&\,\begin{cases} 
\epsilon a(t)^{4-c\epsilon^\frac18} & l\leq 14\\
\epsilon a(t)^{2-c\epsilon^\frac18} & l=15
\end{cases}\\
\label{eq:asymp-Psi}\left\|a^3\del_t\phi-(\Psi_{Bang}+C)\right\|_{C^{l}_\gamma(\Sigma_t)}\lesssim&\,\begin{cases} 
\epsilon a(t)^{4-c\epsilon^\frac18} & l\leq 14\\
\epsilon a(t)^{2-c\epsilon^\frac18} & l=15
\end{cases}\\
\label{eq:asymp-phi}\changefinal{\left\|\phi(t,\cdot)-\phi(t_0,\cdot)+\int_t^{t_0}a(s)^{-3}\,ds\cdot(\Psi_{Bang}+C)\right\|_{\dot{C}^l_\gamma(\Sigma_t)}}\lesssim&\,\begin{cases}
\epsilon a(t)^{4-c\epsilon^\frac18} & 1\leq l\leq 14\\
\epsilon a(t)^{2-c\epsilon^\frac18} & l=15
\end{cases}\\
\label{eq:asymp-K}\left\|a^{3}k-K_{Bang}\right\|_{C^{l}_\gamma(\Sigma_t)}\lesssim&\,\begin{cases} 
\epsilon a(t)^{4-c\epsilon^\frac18} & l\leq 14\\
\epsilon a(t)^{2-c\epsilon^\frac18} & l=15\,
\end{cases}
\end{align}
\end{subequations}

\noindent Further, these footprint states satisfy the equations\footnote{These are precisely the (generalized) Kasner relations, see Section \ref{subsec:FLRW-Kasner}.}
\begin{subequations}
\begin{align}
\label{eq:Bang-CMC}{(K_{Bang})^a}_a=&\,-\sqrt{12\pi}C\,,\\
8\pi(\Psi_{Bang}+C)^2+{(K_{Bang})^a}_b{(K_{Bang})^b}_a=&\,12\pi C^2
\label{eq:Bang-Hamil}
\end{align}
\end{subequations}
and remain close to the data of the reference solution in the following sense, where $\I$ denotes the Kronecker symbol:
\begin{subequations}\label{eq:footprint}
\begin{align}
\left\|\vol{\gamma}-\vol{Bang}\right\|_{C^{15}_\gamma(M)}\lesssim&\,\epsilon \label{eq:footprint-vol}\\
\left\|\Psi_{Bang}\right\|_{C^{15}_\gamma(M)}\lesssim&\,\epsilon \label{eq:footprint-Psi}\\
\left\|K_{Bang}+\sqrt{\frac{4\pi}3}C\I\right\|_{C^{15}_\gamma(M)}\lesssim&\,\epsilon\,\label{eq:footprint-K}
\end{align}
\end{subequations}
Additionally, there exists a $(0,2)$-tensor field $M_{Bang}\in C^{15}_\gamma(M)$ satisfying
\begin{equation}\label{eq:footprint-G}
\left\|M_{Bang}-\gamma\right\|_{C^{15}_\gamma(M)}\lesssim\epsilon
\end{equation}
and, with $\odot$ and $\exp$ meant in the matrix product and exponential sense respectively, one has
\begin{equation}
\label{eq:asymp-G}\left\|g\odot\exp\left[\left(\change{-2}\int_t^{t_0}a(s)^{-3}\,ds\right)\cdot K_{Bang}\right]-M_{Bang}\right\|_{C^l_\gamma(\Sigma_t)}\lesssim\begin{cases}
 \epsilon a(t)^{4-c\epsilon^\frac18} & l\leq 14\\
\epsilon a(t)^{2-c\epsilon^\frac18} & l=15\,.
\end{cases}
\end{equation}
Moreover, the Bel-Robinson variables $E$ and $B$ satisfy the estimates
\begin{subequations}
\begin{align}
\label{eq:asymp-E}\|E\|_{C^{16}_\gamma(\Sigma_t)}\lesssim&\,
\epsilon a^{-4-c\epsilon^\frac18}
\\
\label{eq:asymp-B}\|B\|_{C^l_\gamma(\Sigma_t)}\lesssim&\,\begin{cases}
\epsilon a^{-2-c{\epsilon}^\frac18} & l\leq 15\\
\epsilon a^{-4-c\epsilon^\frac18} & l\leq 16\,.
\end{cases}
\end{align}
\end{subequations}

\textbf{Causal disconnectedness:} Let $\curve$ be a past directed causal curve on $(\change{(0,t]}\times M,\g)$ \change{ for $t\leq t_0$ }with domain $[s_1,s_{max})$ such that $\curve(s_1)\in \change{\Sigma_{t}}$ and $s_{max}$ is maximal. Then, there exists a constant $\mathcal{K}>0$ that does not depend on $\curve$ such that one has
\begin{equation}\label{eq:causal-disconn}
L[\curve]=\int_{s_1}^{s_{max}}\sqrt{(\gamma_{ab})_{\curve(s)}\dot{\curve}^a(s)\dot{\curve}^b(s)}\,ds\leq \mathcal{K} a(t)^{2-c\epsilon^\frac18}\,,
\end{equation}
where $\gamma$ is the negative Einstein spatial reference metric on $M$ (see Definition \ref{def:spatial-mf}). \change{Hence, for points $p,q\in\Sigma_t$ with $\text{dist}_\gamma(p,q)>2\mathcal{K}a(t)^{2-c\epsilon^\frac18}$, the causal pasts of $p$ and $q$ cannot intersect.}\\

\textbf{Geodesic incompleteness:} \changefinal{Let $\curve(\mathcal{A})$ be a past directed, affinely parametrized causal geodesic emanating from $\Sigma_{t_0}$, where $\mathcal{A}:(0,t_0]\rightarrow [0,\infty)$ denotes the parameter time that is normalized to $\mathcal{A}(t_0)=0$. }Then,
\begin{equation}\label{eq:geod-incomp}
\mathcal{A}(0)\leq \mathcal{K}_1\cdot \lvert \mathcal{A}^\prime(t_0)\rvert \cdot a(t_0)^{1+K_2\epsilon}\int_0^{t_0}a(s)^{-1-\mathcal{K}_2\epsilon}\,ds<\infty\,,
\end{equation}
holds for suitable constants $\mathcal{K}_1,\mathcal{K}_2>0$ that are independent of $\curve$, and thus any such geodesic crashes into the Big Bang hypersurface in finite affine parameter time.\\

\textbf{Blow-up:} \begin{subequations}
The norm $\lvert k\rvert_g$ behaves toward the Big Bang hypersurface as follows:
\begin{equation}\label{eq:blowup-k}
\left\|a^6\lvert k\rvert_g^2-{(K_{Bang})^i}_j{(K_{Bang})^j}_i\right\|_{C^0_\gamma(\Sigma_t)}\lesssim \epsilon a^{4-c{\epsilon}^\frac18}
\end{equation}
Further, with $W[\g]$ denoting the Weyl curvature and $P[\g]=\Riem[\g]-W[\g]$,
\begin{equation}\label{eq:blowup-P}
\left\|a^{12}P_{\alpha\beta\gamma\delta}P^{\alpha\beta\gamma\delta}-\frac{5}3\cdot(8\pi)^2(\Psi_{Bang}+C)^4\right\|_{C^0_\gamma(M)}\lesssim \epsilon a^{4-c{\epsilon}^\frac18}\,
\end{equation}
is satisfied, whereas there exists a scalar footprint $W_{Bang}\in C^{15}_\gamma(M)$ such that one has
\begin{equation}\label{eq:blowup-W}
\left\|a^{12}W_{\alpha\beta\gamma\delta}W^{\alpha\beta\gamma\delta}-W_{Bang}\right\|_{C^0(M)}\lesssim \epsilon a^{2-c{\epsilon}^\frac18}\,.
\end{equation}
Here, $W_{Bang}$ is a fourth order polynomial in $\hat{K}_{Bang}=K_{Bang}+\sqrt{\frac{4\pi}3}C\I$ and $\Psi_{Bang}$ and satisfies
\begin{equation}\label{eq:footprint-W}
\|W_{Bang}\|_{C^{15}_\gamma(M)}\lesssim\epsilon\,.
\end{equation}
Finally, \change{the scalar curvature $R[\g]$ and the Ricci curvature invariant $\Ric[\g]_{\alpha\beta}\Ric[\g]^{\alpha\beta}$ blow up with the asymptotics
\begin{align}
\label{eq:blowup-scalar} \|a^6R[\g]-8\pi(\Psi_{Bang}+C)^2\|_{C^0(M)}\lesssim&\,\epsilon a^{4-c\epsilon^\frac18}\,,\\
\label{eq:blowup-Ricci}\|a^{12}\Ric[\g]_{\alpha\beta}\Ric[\g]^{\alpha\beta}-(8\pi)^2(\Psi_{Bang}+C)^4\|_{C^0(M)}\lesssim&\,\epsilon a^{4-c\epsilon^\frac18}\,,
\end{align}
and }the Kretschmann scalar $\mathcal{K}=\Riem[\g]_{\alpha\beta\gamma\delta}\Riem[\g]^{\alpha\beta\gamma\delta}$ exhibits stable blow-up in the following sense:
\begin{equation}\label{eq:blowup-Kretschmann}
\change{\left\|a^{12}\mathcal{K}-\frac{5}3\cdot(8\pi)^2(\Psi_{Bang}+C)^4-W_{Bang}\right\|_{C^0(M)}\lesssim \epsilon a^{2-c\epsilon^\frac18}}
\end{equation}
\end{subequations}
\end{theorem}
\begin{remark}[The solution variables exhibit AVTD behaviour]\label{rem:AVTD} The estimates \eqref{eq:asymp-lapse}-\eqref{eq:asymp-K} and \eqref{eq:asymp-G} imply that the solution is asymptotically velocity term dominated (AVTD) in the sense that, toward the Big Bang singularity, they behave at leading order like solutions to the (formal) velocity term dominated equations. These arise by dropping any terms containing spatial derivatives in the decomposed Einstein system, i.e.\,in \eqref{eq:EEqg}, \eqref{eq:EEqk}, \eqref{eq:EEqLapse} and \eqref{eq:wave}.
\end{remark}
\begin{proof}
\change{As argued at the end of Section \ref{subsec:lwp}, we can assume without loss of generality that our initial data is sufficiently regular. Hence, the local existence statement in Lemma \ref{lem:lwp} and the initial data requirements \eqref{eq:ass-init-main-thm} ensure that there exists a local solution to the Einstein scalar-field system on $[t_1,t_0]\times M$ and that the bootstrap assumption (see Assumption \ref{ass:bootstrap}) holds }on $[t_1,t_0]\times M$ with $t_1\in(0,t_0)$ and $\sigma=\epsilon^\frac1{16}.$ 
Let $\mathfrak{t}\in(0,t_0)$ be such that $(\mathfrak{t},t_0]\times M$ is the maximal domain on which the solution variables exist and satisfy the bootstrap assumptions. For contradiction, we now assume that $\mathfrak{t}>0$ were to hold.\\

Due to Corollary \ref{cor:C-impr}, there exist (summarizing all updates) constants $c_1,K_1>0$ such that, for any $t\in(\mathfrak{t},t_0]$,
\begin{equation}\label{eq:bs-imp-total}
\mathcal{C}(t)\leq K_1\changefinal{\epsilon^\frac74}a(t)^{-c_1\epsilon^\frac18}\,
\end{equation}
If $\epsilon$ is small enough such that $K_1\epsilon^\frac18<K_0$ and $c_1\epsilon^\frac1{8}<c_0\sigma$ hold, this is a strict improvement of the bootstrap assumption. Furthermore, argued exactly as in the proof of \cite[Theorem 15.1]{Rodnianski2014}, above improvement ensures none of the blow-up criteria of Lemma \ref{lem:lwp} are satisfied if $\mathfrak{t}>0$ were to hold, essentially as a direct consequence of \eqref{eq:bs-imp-total}. Hence, the solution could be classically extended to a CMC hypersurface $\Sigma_{\mathfrak{t}}$ diffeomorphic to $M$ while satisfying the improved estimates by continuity, and further to an interval $(\mathfrak{t}^\prime,t_0]$ for some $0<\mathfrak{t}^\prime<\mathfrak{t}$ on which the bootstrap assumptions must then be satisfied, also by continuity. This contradicts the maximality of $(\mathfrak{t},t_0]$.\\

Thus, the rescaled solution variables induce a unique solution to the Einstein scalar-field system on $(0,t_0]\times M$ such that \eqref{eq:bs-imp-total} \change{is }satisfied for any $t\in(0,t_0]$. The core estimate \eqref{eq:bs-imp-main-thm} follows since Corollaries \ref{cor:H-imp} and \ref{cor:C-impr} now hold on $(0,t_0]$.\\

From \eqref{eq:bs-imp-main-thm}, the asymptotic behaviour in \eqref{eq:asymp-lapse}-\eqref{eq:asymp-K} and \eqref{eq:asymp-G} is established as in \cite[Theorem 15.1]{Rodnianski2014}, which we briefly outline: First, we note that \eqref{eq:asymp-lapse} follows directly from \eqref{eq:bs-imp-main-thm}. For the remaining estimates, the arguments are similar, so consider for example $\del_t\phi$: By the rescaled wave equation \eqref{eq:REEqWave} and \eqref{eq:asymp-lapse}, we have that
\begin{equation*}
\left\|\del_t\Psi\right\|_{C^l_\gamma(\Sigma_t)}\lesssim\begin{cases}
\epsilon a^{1-c\epsilon^\frac18} & l\leq 14\\
\epsilon a^{-1-c\epsilon^\frac18} & l=15
\end{cases}
\end{equation*}
Hence, for an arbitrary decreasing sequence $(t_m)_{m\in\N},$ on $(0,t_0]$ that converges to zero, we have
\[\|\Psi(t_{m_1},\cdot)-\Psi(t_{m_2},\cdot)\|_{C^l_\gamma(M)}\lesssim\begin{cases}
\epsilon a(t_{m_1})^{4-c\epsilon^\frac18} & l\leq 14\\
\epsilon a(t_{m_1})^{2-c\epsilon^\frac18} & l=15
\end{cases}\]
for any $m_1,m_2\in\N,\,m_1<m_2$ by \eqref{eq:a-integrals}. This shows that $\Psi(t_{m_1},\cdot)$ is a Cauchy sequence in $C_\gamma^{15}(M)$ and hence there exists a limit function ${\Psi}_{Bang}\in C_\gamma^{15}(M)$ that satisfies
\[\|\Psi(t,\cdot)-{\Psi}_{Bang}\|_{C^l_\gamma(M)}\lesssim\begin{cases}
\epsilon a(t)^{4-c\epsilon^\frac18} & l\leq 14\\
\epsilon a(t)^{2-c\epsilon^\frac18} & l=15
\end{cases}\]
for any $t\in(0,t_0]$. Since $\Psi=a^3n^{-1}\del_t\phi-C$ holds by definition, \eqref{eq:asymp-Psi} now follows by examining the Taylor expansion of $n^{-1}-1$ at $0$ using \eqref{eq:asymp-lapse}.\\

\noindent The identity \eqref{eq:Bang-CMC} follows directly from the CMC condition \eqref{eq:CMC}, the asymptotic behaviour \eqref{eq:asymp-K} of $a^3k$ and the Friedman equation \eqref{eq:Friedman}, while \eqref{eq:Bang-Hamil} follows from the asymptotic limit of the Hamiltonian constraint \eqref{eq:Hamilton} with \eqref{eq:Friedman}, \eqref{eq:asymp-lapse}, \eqref{eq:asymp-Psi} and \eqref{eq:asymp-K} as well as \eqref{eq:bs-imp-main-thm} for lower order terms. The asymptotics in \eqref{eq:asymp-G} follows exactly as in \cite[Theorem 15.1]{Rodnianski2014}, and \eqref{eq:footprint-vol}-\eqref{eq:footprint-G} are a direct result of the initial data assumptions and applying the respective asymptotic estimates to $t=t_0$.

\noindent For the first estimate in \eqref{eq:asymp-B}, we apply the momentum constraint \eqref{eq:REEqConstrB} to get
\[\lvert\nabla^J B\rvert_G=\change{a^{-4}\lvert\nabla^J \RB\rvert_G}= \change{a^{-2}\lvert\nabla^J\curl_G\Sigma\rvert_G}\lesssim a^{-2}\lvert\nabla^{J+1}\Sigma\rvert_G\]
and consequently, with Lemma \ref{lem:G-gamma-norm-switch} as well as \eqref{eq:APmidB} and \eqref{eq:bs-imp-main-thm},
\begin{align*}
\|B\|_{C^{15}_\gamma(\Sigma_t)}\lesssim&\,a^{-c\sqrt{\epsilon}}\|B\|_{C^{15}_G(\Sigma_t)}+\epsilon a^{-2-c\sqrt{\epsilon}}\cdot P_{{15}}(\|G-\gamma\|_{C^{15}_\gamma(\Sigma_t)})\\
\lesssim&\,a^{-2-c\sqrt{\epsilon}}\|\Sigma\|_{C^{16}_G(\Sigma_t)}+\epsilon a^{-2-c\sqrt{\epsilon}}\cdot P_{15}(\|G-\gamma\|_{C^{15}_\gamma(\Sigma_t)})\\
\lesssim&\,\epsilon a^{-2-c\epsilon^\frac18}\,.
\end{align*}
The remaining estimates in \eqref{eq:asymp-E} and \eqref{eq:asymp-B} are contained in \eqref{eq:bs-imp-main-thm}. The results \eqref{eq:causal-disconn} and \eqref{eq:geod-incomp} follow as in the proofs of (15.6) and (15.7) in \cite[Theorem 15.1]{Rodnianski2014} from the asymptotic behaviour of the solution variables in \eqref{eq:asymp-lapse}-\eqref{eq:asymp-K} and \eqref{eq:asymp-G}. \change{We briefly sketch the proof of \eqref{eq:geod-incomp}: Consider a geodesic $\curve$ affinely parametrized by $\mathcal{A}$ as in the statement. The geodesic equations then lead to the following estimate for some suitable $\mathcal{K}>0$:
\[\lvert \mathcal{A}^{\prime\prime}\rvert\leq \frac{\dot{a}}a\lvert\mathcal{A}^\prime\rvert+\mathcal{K}\left[\frac{\dot{a}}a\lvert N\rvert+n^{-1}\lvert\del_tN\rvert+n^{-1}\lvert \nabla N\rvert_g+n\lvert\hat{k}\rvert_g\right]\lvert\mathcal{A}^\prime\rvert\,.\]
The leading term is hereby arises from the mean curvature condition. Arguing as with the elliptic estimates in Section \ref{sec:lapse}, one can show that $\lvert\del_tN\rvert\lesssim \epsilon a^{-1-c\epsilon^\frac18}$. Thus, along with the other pointwise bounds on $n$, $g$ and $\hat{k}$, one obtains
\[\lvert \mathcal{A}^{\prime\prime}\rvert\leq \frac{\dot{a}}a\left(1+c\epsilon\right)\lvert\mathcal{A}^\prime\rvert\]
and consequently
\[\lvert\mathcal{A}^\prime(t)\rvert\leq \lvert \mathcal{A}^\prime(t_0)\rvert a(t)^{-1-c\epsilon}\]
by the Gronwall lemma. \eqref{eq:geod-incomp} follows by integrating.\\}

Turning to the blow-up behaviour of geometric invariants, observe \eqref{eq:blowup-k} is a direct consequence of \eqref{eq:asymp-K}. Regarding \eqref{eq:blowup-W}, we first compute
using \eqref{eq:Weyl-reconstruct} and standard algebraic manipulations that
\[a^{12}W_{\alpha\beta\gamma\delta}W^{\alpha\beta\gamma\delta}=a^{12}\left(\change{8}\lvert E\rvert_g^2+8\lvert B\rvert_g^2\right)=\change{8}\lvert\RE\rvert_G^2+8\lvert\RB\rvert_G^2\,.\]
By the rescaled constraint equation \eqref{eq:REEqConstrE}, we have
\begin{equation*}
\RE_{ij}=-\dot{a}a^2\Sigma_{ij}+(\Sigma\odot\Sigma)_{ij}-\left[\frac{8\pi}3\Psi^2+\frac{16\pi}3C\Psi\right]G_{ij}+\O{\epsilon a^{4-c\epsilon^\frac18}}\,,
\end{equation*}
for $t\downarrow 0$. Further, by expanding \eqref{eq:Friedman} around $a=0$, we have 
$\dot{a}a^2=\sqrt{\frac{4\pi}3}C+\O{a^2}$.
Since $\Sigma^\sharp$ and $\Psi$ converge to footprint states $\hat{K}_{Bang}=K_{Bang}+\sqrt{\frac{4\pi C}3}\I$ and $\Psi_{Bang}$ in $C^{15}_\gamma(M)$ respectively, this shows that $8\lvert\RE\rvert_G^2$ converges to some $W_{\text{Bang}}\in C^{15}_\gamma(M)$ that can be expressed as a fourth-order polynomial in $\hat{K}_{\mathrm{Bang}}$ and $\Psi_{Bang}$ and satsifies
\[\left\|\lvert\RE\rvert_G^2-\change{\frac1{8}}W_{Bang}\right\|_{C^0(M)}\lesssim \epsilon a^{2-c\epsilon^\frac18}\]
as well as \eqref{eq:footprint-W}. Due to \eqref{eq:asymp-B}, the $\lvert\RB\rvert_G^2$-term in the Weyl curvature scalar is negligible in comparison, and thus \eqref{eq:blowup-W} immediately follows.\\
Furthermore, 
one has 
\begin{align*}
P_{\alpha\beta\gamma\delta}P^{\alpha\beta\gamma\delta}
=&\,2\Ric[\g]_{\alpha\beta}\Ric[\g]^{\alpha\beta}-\frac29 R[\g]^2\deletemath{-\frac16R[\g]},\,
\end{align*}
and \eqref{eq:blowup-P} is a direct consequence of \eqref{eq:blowup-Ricci} and \eqref{eq:blowup-scalar}, which follow once more with  \eqref{eq:asymp-Psi} and \eqref{eq:asymp-lapse} as well as \eqref{eq:bs-imp-main-thm} for error terms. Finally, \eqref{eq:blowup-Kretschmann} is obtained from \eqref{eq:blowup-P}-\eqref{eq:blowup-W}.
\end{proof}

\section{Future stability}\label{sec:fut}

\noindent The goal of this section is to show the following theorem:

\begin{theorem}[Future stability of Milne spacetime]\label{thm:fut-stab-simple} Let the rescaled initial data $(\fg,\bm{k},\nabla\phi,\phi^\prime)$ on $M$ be sufficiently close to $(\gamma,\frac13\gamma,0,0)$ in $H^5\times H^4\times H^4\times H^4$ on some initial hypersurface $\Sigma_{\tau=\tau_0}$ (see Definition \ref{def:fut-rescaled} and Assumption \ref{ass:fut-init}). Then, its maximal globally hyperbolic development $(\M,\g,\phi)$ within the Einstein scalar-field system in CMCSH gauge is foliated by the CMC Cauchy hypersurfaces $(\Sigma_{\tau})_{\tau\in[\tau_0,0)}$, is future (causally) complete and exhibits the following asymptotic behaviour:
\[(\fg,\bm{k},\phi^\prime,\nabla\phi)(\tau)\longrightarrow(\gamma,\frac13\gamma,0,0)\text{ as }\tau\uparrow 0\]
\end{theorem}

\noindent Since the control of geometric perturbations uses the same arguments as in \cite{AndFaj20}, the focus in this section will lie on dealing with the the scalar field. The key idea is controlling decay of the scalar field using an indefinite corrective term on top of the canonical energy (see Definition \ref{def:fut-stab}).

\subsection{Preliminaries}\label{subsec:fut-prelim}

\subsubsection{Notation, gauge and spatial reference geometry}

Within this section, we will decompose the Lorentzian metric as follows:
\begin{subequations}
\begin{equation}\label{eq:fut-metric}
\g=-n^2dt^2+g_{ab}(dx^a+X^a)(dx^b+X^bdt)
\end{equation}
We impose CMCSH gauge (see \cite{AM03}) via
\begin{equation}\label{eq:CMCSH}
t=\tau,\,g^{ij}(\Gamma^{a}_{ij}-\Gamhat^{a}_{ij})=0\,,
\end{equation}
\end{subequations}
where $\Gamhat$ refers to the Christoffel symbols with regards to the spatial reference metric $\gamma$.\\
We extend the notation from the Big Bang stability analysis regarding foliations, derivatives, indices and schematic term notation to this setting (see Section \ref{subsec:notation}). In particular, $\Sigma_{T}$ and $\Sigma_{\tau}$ will refer to spatial hypersurfaces along which the logarithmic time $T$ (see \eqref{eq:fut-time}) and the mean curvature $\tau$ are constant (see \eqref{eq:fut-time} on why these are interchangeable), and we will write for example $\Sigma_{T=0}$ when inserting a specific value to avoid potential ambiguity. We use similar notation for scalar functions and tensors that depend on $T$ or, respectively, $\tau$.\\

For the extent of the future stability analysis, we have to introduce an additional condition for the spatial geometry beyond Definition \ref{def:spatial-mf}:

\begin{definition}[Spectral condition for the Laplacian of the spatial reference manifold]\label{def:spatial-mf-spectral}
Let $\mu_0(\gamma)$ to be the smallest positive eigenvalue of the Laplace operator $-\Lap_\gamma=\changefinal{-(\gamma^{-1})^{ab}\nabhat_a\nabhat_b}$ acting on scalar functions, where $(M,\gamma)$ is as in Definition \ref{def:spatial-mf}. $(M,\gamma)$ additionally is assumed to satisfy
\[\mu_0(\gamma)>\frac19\,.\]
\end{definition}

\begin{remark}[\change{Manifolds that satisfy Definition \ref{def:spatial-mf-spectral}}]\label{rem:weeks-and-friends}

\changefinal{The available literature on spectra of $-\Lap_\gamma$ usually focuses on hyperbolic manifolds with sectional curvature $\kappa=-1$. Thus, one needs to check that $\mu_0$ is strictly greater than $1$ to verify the analogue of Definition \ref{def:spatial-mf-spectral} after rescaling.\\

\noindent Numerical works, e.g., \cite{Cornish99, Ino01}, provide evidence for over 250 compact hyperbolic $3$-manifolds to satisfy this spectral bound, many of which are closed. In particular, both \cite{Cornish99}\footnote{These results have to be interpreted cautiously since the numerical method cannot detect eigenvalues below $1$.} and \cite{Ino01} consider the smallest closed orientable hyperbolic $3$-manifold, the Weeks space m003(-3,1), and compute that it falls under Definition \ref{def:spatial-mf-spectral} with $\mu_0\approx 27,8$ in \cite[Table IV]{Cornish99} and $26\lessapprox\mu_0\lessapprox 27,8$ in \cite[Table 2]{Ino01}. Moreover, as demonstrated in \cite[Figure 6]{Ino01}, many manifolds with small enough diameter $d$ satisfy this condition. In fact, the analytical bound
\[\mu_0\geq\max\left\{\frac{\pi^2}{2d^2}-\frac12,\sqrt{\frac{\pi^4}{d^4}+\frac14}-\frac32,\frac{\pi^2}{d^2}e^{-d}\right\}\]
(see \cite[Theorem 1.1-1.2]{ChZ95} with $L=2$) implies that $\mu_0>10$ holds for Weeks space, which has diameter $d\approx 0,843$ (see \cite[Table V]{Cornish99}). Furthermore, \cite{Ino01} finds no closed hyperbolic manifolds that violate this bound. More recently, the Selberg trace formula has been used in \cite{LinLip22,LinLip24,BoMaPa25} to compute candidates for eigenvalues of $-\Lap_\gamma$ and related operators, based on an optimization approach originating in \cite{BooSt07}. In particular, the calculations visualized in \cite[Figure 3]{BoMaPa25} demonstrate that one must have $\mu_0\geq 27,6$ on the Weeks manifold.\\

We also note that it is conjectured that one at least has $\mu_0\geq 1$ for any arithmetic hyperbolic 3-manifold (see \cite[Conjecture 2.3]{Ber03}). In fact, this is tied to the Ramanujan conjecture for automorphic forms. Finally, one can construct compact manifolds with boundary and with constant sectional curvature $-1$ where $\mu_0$ becomes arbitrarily small, see \cite[Corollary 4.4]{Cal94}.}
\end{remark}

\subsubsection{Rescaled variables and Einstein equations}

We will use the standard rescaling of the solution variables by $\tau$:

\begin{definition}[Rescaled variables for future stability]\label{def:fut-rescaled}
\begin{subequations}
\begin{gather}
\fg_{ij}=\tau^2g_{ij}\,,(\fg^{-1})^{ij}=\tau^{-2}g^{ij}\,,\ \fk_{ij}=\tau \hat{k}_{ij}\label{eq:fut-resc-metric}\\
\fn=\tau^2 n\,,\ \fN=\frac{\fn}3-1\,,\ \fX^a=\tau X^a\label{eq:fut-resc-gauge}
\end{gather}
Furthermore, we introduce the logarithmic time 
\begin{equation}\label{eq:fut-time}
T=-\log\left(\frac{\tau}{\tau_0}\right)\,\Leftrightarrow\,\tau=\tau_0e^{-T}
\end{equation}
which satisfies $\del_T=-\tau\del_\tau$. Toward the future, $\tau$ increases from $\tau_0$ to $0$, and thus $T$ increases from $0$ to $\infty$. We additionally introduce:
\begin{gather}
\fdel=\del_T+\Lie_{\fX}=-\tau(\del_\tau-\Lie_X)\label{eq:fut-del0}\\
\phi^\prime=\fn^{-1}\fdel\phi=n^{-1}(-\tau)^{-1}(\del_\tau-\Lie_X)\phi\label{eq:fut-delphi}
\end{gather}
Moreover, for any scalar function $\zeta$, we denote by $\overline{\zeta}$ the mean integral with respect to $(\Sigma_{T},\fg_T)$.
\end{subequations}
\end{definition}
\noindent For symmetric $(0,2)$-tensors $h$, we define the \change{perturbed Lichnerowicz Laplacian}
\begin{equation}\label{eq:LG}
\LG h_{ab}=-\frac1{\mu_{\fg}}\nabhat_k\left((\fg^{-1})^{kl}\mu_{\fg}\nabhat_lh_{ab}\right)-2\Riem[\gamma]_{akbl}(\fg^{-1})^{kk^\prime}(\fg^{-1})^{ll^\prime}h_{k^\prime l^\prime}\,.
\end{equation}
\change{This operator satisfies
\begin{equation}\label{eq:fut-Ric-ell}
\left(\Ric[\fg]-\Ric[\gamma]\right)_{ij}=\frac12\LG(\fg-\gamma)_{ij}+J_{ij},\quad \|J\|_{H^{l-1}}\lesssim \|\fg-\gamma\|_{H^l}\,,
\end{equation}
see \cite[Pf. of Theorem 3.1]{AM03B}. }Under our conditions for the reference geometry, \cite{Kroen15} implies that the smallest positive eigenvalue of $\mathcal{L}_{\gamma,\gamma}$, denoted by $\lambda_0$, satisfies $\lambda_0\geq\frac19$, and that $\mathcal{L}_{\gamma,\gamma}$ has trivial kernel. The spectral condition in Definition \ref{def:spatial-mf-spectral} is not necessary for this to hold true.\\

\change{We now collect the $(3+1)$-decomposition of the Einstein scalar-field equations in CMCSH gauge with the help of \cite[(2.13)-(2.18)]{AndFaj20}:

%

\begin{lemma}[Rescaled CMCSH equations] The rescaled CMCSH Einstein scalar-field equations take the following form:
\begin{subequations}
The constraint equations
\begin{equation}\label{eq:fut-constr}
R[\fg]-\lvert\fk\rvert_{\fg}^2-\frac23=
8\pi\left[\lvert\phi^\prime\rvert^2+\lvert\nabla\phi\rvert_{\fg}^2\right],\ \div_{\fg}\fk_b=
8\pi\tau^{3}\phi^\prime\nabla_{b}\phi,
\end{equation}
the elliptic lapse and shift equations
\begin{align*}
\left(\fLap-\frac13\right)\fn=&\,\fn\left(\lvert\fk\rvert_{\fg}^2+4\pi\left[\lvert\phi^\prime\rvert^2+\lvert\nabla\phi\rvert_{\fg}^2\right]\right)-1,\label{eq:fut-lapse-eq}\numberthis\\
\fLap \fX^a+(\fg^{-1})^{ab}\Ric[\fg]_{bm}\fX^m=&\,2(\fg^{-1})^{am}(\fg^{-1})^{bn}\nabla_b\fn\cdot\fk_{mn}-(\fg^{-1})^{ab}\nabla_b\fN+8\pi\fn\tau^3\phi^\prime\nabla_b\phi\numberthis\label{eq:fut-shift-eq}\\
&\,-2(\fg^{-1})^{bk}((\fg^{-1})^{cl}\fn\cdot\fk_{bc}-\nabla_b \fX^l)(\Gamma^a_{kl}-\Gamhat^a_{kl})\,,
\end{align*}
the geometric evolution equations
\begin{align*}
\fdel\fg_{ab}=&\,2\fn\fk_{ab}+2\fN\fg_{ab}\,\numberthis\label{eq:fut-eq-g}\,,\\
\fdel(\fg^{-1})^{ab}=&\,-2\fn(\fg^{-1})^{ac}(\fg^{-1})^{bd}\fk_{cd}-2\fN(\fg^{-1})^{ab}\numberthis\label{eq:fut-eq-g-1}\,,\\
\fdel\fk_{ab}=&\,-2\fk_{ab}-\fn\left(\Ric[\fg]_{ab}+\frac29\fg_{ab}\right)+\nabla_a\nabla_b\fn\numberthis\label{eq:fut-eq-Sigma}\\
&\,+2\fn\cdot(\fg^{-1})^{mn}\fk_{am}\fk_{bn}-\frac13\fN\fg_{ab}-\fN\Sigma_{ab}-8\pi \fn\nabla_a\phi\nabla_b\phi
\end{align*}
and the wave equation
\begin{equation}\label{eq:fut-wave}
\fdel\phi^\prime=\langle \nabla\fn,\nabla\phi\rangle_{\fg}+\fn\fLap\phi+(1-\fn)\phi^\prime\,.
\end{equation}
\end{subequations}
\end{lemma}}

\subsubsection{Energies and data assumptions}

The proof will rely on the following corrected energy quantities:

\begin{definition}[Energies for future stability]\label{def:fut-stab}
\begin{subequations}
\begin{align*}
\fE^{(l)}=&\,(-1)^l\int_{M}\left[\phi^\prime\fLap^l\phi^\prime-\phi\fLap^{l+1}\phi\right]\vol{\fg}, \qquad
\fC^{(l)}=(-1)^l\int_{M}(\phi-\phim)\fLap^l\phi^\prime\,\vol{\fg}\numberthis\\
E_{SF}=&\,\sum_{m=0}^4 \left(\fE^{(m)}+\frac23\fC^{(m)}\right)\numberthis\\
\fEg=&\,\sum_{m=1}^5\biggr(\frac92\int_{M}\langle\fg-\gamma,\LG^m(\fg-\gamma)\rangle_{\fg}\vol{\fg}+\frac12\int_{M}\langle6\fk,\LG^{m-1}(6\fk)\rangle_{\fg}\vol{\fg}\numberthis\\
&\,\phantom{\sum_{m=1}^l}+c_E\int_{M}\langle 6\fk,\LG^{m-1}(\fg-\gamma)\rangle_{\fg}\vol{\fg}\biggr)
\end{align*}
\end{subequations}
The constant $c_E$ is given by
\begin{equation}\label{eq:fut-corr-const}
c_E=\begin{cases}
1 & \lambda_0>\frac19 \\
9(\lambda_0-\epsilonnew^\prime) & \lambda_0=\frac19\,,
\end{cases}
\end{equation}
where $\epsilonnew^\prime>0$ is chosen to be small enough within the argument.
\end{definition}

The Sobolev norms $H_{\fg}^l$ and $C_{\fg}^l$ are defined analogously to Definitions \ref{def:sob-norms} and \ref{def:sup-norms}, with similar conventions on suppressing time dependence in notation whereever possible. Since norms with respect to $\fg$ and $\gamma$ are equivalent under the bootstrap assumption (and consequently throughout the entire argument), we will simply denote the norms by $H^l$ and $C^l$ throughout unless the specific metric is crucial.

\begin{assumption}[Initial data assumption]\label{ass:fut-init} The initial data on the spatial hypersurface $\Sigma_{T=0}$ is assumed to be small in the following sense:
\begin{align*}
\numberthis\label{eq:fut-init}\change{\|\fg-\gamma\|_{C^3}+\|\fk\|_{C^2}+\|\fN\|_{C^4}+\|\fX\|_{C^4}+\|\phi^\prime\|_{C^2}+\|\nabla\phi\|_{C^2}&\\
\change{+\|\fg-\gamma\|_{H^5}+\|\fk\|_{H^4}+\|\fN\|_{H^6}+\|\fX\|_{H^6}+\|\phi^\prime\|_{H^4}+\|\nabla\phi\|_{H^4}}}&\,\leq\epsilonnew^2
\end{align*}
\end{assumption}

\begin{remark}[Local well-posedness toward the future]\label{rem:fut-lwp}
Under the above initial data assumption, local well-posedness is satisfied by analogizing the arguments for local well-posedness in the vacuum setting (see \change{\cite[Theorem 3.1]{AM03B}}) with the matter coupling added. Since this only consists of adding another wave equation to the hyperbolic system, the argument is structurally unchanged given appropriate smallness assumptions on $\phi^\prime$ and $\nabla\phi$ (where $\phi$ itself does not enter into the Einstein system). \change{As before, we can without loss of generality assume that the initial is sufficiently regular to ensure that }$E_{geom},\,\E^{(l)}_{SF}$ and $\fC^{(l)}$ initially are continuously differentiable (in time) for any $l\leq 4$. \delete{Beside obtaining this existence result, we will only need smallness of initial data of up to one order less to prove our stability result.}
\end{remark}

\begin{assumption}[Bootstrap assumption]\label{ass:fut-bootstrap} On the bootstrap interval $T\in[0,T_{Boot})$, \change{we assume one has}
\begin{align*}
\change{\label{eq:fut-bootstrap}\numberthis\|\fg-\gamma\|_{C^3}+\|\fk\|_{C^2}+\|\fN\|_{C^4}+\|\fX\|_{C^4}+\|\phi^\prime\|_{\change{C^2}}+\|\nabla\phi\|_{C^2}}&\\
\change{+\|\fg-\gamma\|_{H^5}+\|\fk\|_{H^4}+\|\fN\|_{H^6}+\|\fX\|_{H^6}+\|\phi^\prime\|_{H^4}+\|\nabla\phi\|_{H^4}}&\change{\,\leq\epsilonnew e^{-\frac{T}2}\,.}
\end{align*}
\end{assumption}

\noindent We only choose not to use \enquote{$\lesssim$}-notation in the above assumptions for notational convenience in some technical computations. As before, $\epsilonnew$ can be chosen to have been sufficiently small for the following estimates to hold and for the decay estimates we derive from the bootstrap assumptions to be strict improvements. Moreover, note that \eqref{eq:fut-bootstrap} is satisfied since all of the norms are continuous in time (see Remark \ref{rem:fut-lwp})\delete{, and \eqref{eq:fut-bootstrap-phi-mean} is satisfied local-in-time since the spatial hypersurfaces are compact and $\phi$ is continuous.}\\

Before moving on to the energy estimates, we quickly collect the following immediate consequence of the bootstrap assumptions:

\begin{lemma}[Sobolev estimate for the curvature]\label{lem:fut-Ric-est}
The following estimate holds for any $l\in\N_0$:
\begin{subequations}
\begin{equation}\label{eq:fut-Ric-est}
\left\|\Ric[\fg]+\frac29\fg\right\|_{H^l}\lesssim \|\fg-\gamma\|_{H^{l+2}}+\|\fg-\gamma\|_{H^{l+1}}^2
\end{equation}
Under the bootstrap assumptions, this implies
\begin{equation}\label{eq:fut-Ric-bs}
\left\|\Ric[\fg]+\frac29\fg\right\|_{C^1}+\left\|\Ric[\fg]+\frac29\fg\right\|_{H^3}\lesssim \epsilonnew e^{-\frac{T}2}
\end{equation}
\end{subequations}
\end{lemma}

\begin{proof}
By \change{\eqref{eq:fut-Ric-ell}}, one has
\begin{align*}
\left\|\Ric[\fg]+\frac29\fg\right\|_{H^l}\leq&\,\frac12\|\LG(\fg-\gamma)\|_{H^l}+K\|\fg-\gamma\|_{H^{l+1}}^2\,
\end{align*}
for some suitably large $K>0$ along with the fact that $\LG$ is elliptic. This implies the first inequality, while the latter follows from directly from the bootstrap assumption \eqref{eq:fut-bootstrap} and by applying the standard Sobolev embedding.
\end{proof}

\subsection{Elliptic estimates}\label{subsec:fut-ell-est} We briefly collect the elliptic estimates for lapse and shift:

\begin{lemma}[Elliptic estimates for lapse and shift]\label{lem:fut-ell-est}
Let $l\in\{3,4,5,6\}$. Then, one has $\fn\in(0,3)$ (thus $\fN\in(-1,0)$) and the following estimates hold:
\begin{subequations}
\begin{align*}
\numberthis\label{eq:ell-est-lapse}\|\fN\|_{H^l}\lesssim&\change{\,\epsilonnew e^{-\frac{T}2}\|\fk\|_{H^{l-2}}+\epsilonnew^2e^{-T}\|\fg-\gamma\|_{H^{l-2}}+\epsilonnew e^{-\frac{T}2}\left[\|\phi^\prime\|_{H^{l-2}}+\|\nabla\phi\|_{H^{l-2}}\right]}\\
\numberthis\label{eq:ell-est-shift}\|\fX\|_{H^l}\lesssim&\,\change{\epsilonnew e^{-\frac{T}2}\|\fk\|_{H^{l-2}}+\epsilonnew e^{-\frac{T}2}\|\fg-\gamma\|_{H^{l-1}}+\epsilonnew e^{-\frac{T}2}\left[\|\phi^\prime\|_{H^{l-2}}+\|\nabla\phi\|_{H^{l-2}}\right]}
\end{align*} 
\end{subequations}
\end{lemma}
\begin{proof}
The pointwise bounds on $\fn$ follow via \eqref{eq:fut-lapse-eq} and the maximum principle as in Lemma \ref{lem:lapse-maxmin}. \change{For the remaining estimates, applying elliptic regularity theory to \eqref{eq:fut-lapse-eq} and \eqref{eq:fut-shift-eq} implies:}
\begin{align*}
\|\fN\|_{H^l}\lesssim&\,\change{\|\fk\|_{C^{\lfloor\frac{l-2}2\rfloor}}\|\fk\|_{H^{l-2}}+\|\nabla\phi\|_{C^2}^2\|\fg-\gamma\|_{H^{l-2}}}\\
&\,\change{+\left[\|\nabla\phi\|_{C^2}\left(1+\|\fg-\gamma\|_{C^2}\right)+\|\phi^\prime\|_{C^2}\right]\left[\|\phi^\prime\|_{H^{l-2}}+\|\nabla\phi\|_{H^{l-2}}\right]}\\
\|\fX\|_{H^l}\lesssim&\,\change{\|\fk\|_{C^{\lfloor\frac{l-2}2\rfloor}}\|\fk\|_{H^{l-2}}+\|\fg-\gamma\|_{H^{l-1}}^2+\|\nabla\phi\|_{C^1}\|\fg-\gamma\|_{H^{l-3}}}\\
&\,\change{+\left[\|\nabla\phi\|_{C^2}^2\left(1+\|\fg-\gamma\|_{C^2}\right)+\|\phi^\prime\|_{C^2}\right]\left[1+\|\fN\|_{C^2}\right]\left[\|\phi^\prime\|_{H^{l-2}}+\|\nabla\phi\|_{H^{l-2}}\right]}
\end{align*}
\change{The statement then follows by inserting \eqref{eq:fut-bootstrap}.}
\end{proof}


\subsection{Scalar field energy estimates}\label{subsec:fut-ESF}

\subsubsection{Near-coercivity of $E_{SF}$}

We will be able to prove a decay estimate via a Gronwall argument only for the corrected energy $E_{SF}$ . Hence, we first need to verify that this energy controls the solution norms, for which we first show that it controls the \enquote{canonical} scalar field energies:

\begin{lemma}[Positivity of corrected scalar field energies]\label{lem:fut-ESF-coercivity}
Let
\[Q=\frac{\sqrt{1+9q}-1}{\sqrt{1+9q}}\ \text{with}\ q=\frac12\left(\mu_0(\gamma)-\frac19\right)\,.\]
Then, for any $l\in\{0,1,2,3,4\}$ and $\epsilonnew>0$ small enough, one has
\begin{equation}\label{eq:fut-ESF-coercivity}
Q\fE^{(l)}\leq \fE^{(l)}+\frac23\fC^{(l)},\quad \text{hence}\quad Q\sum_{m=0}^4\fE^{(l)}\leq E_{SF}
\end{equation}
\end{lemma}
\begin{proof}
We denote the smallest positive eigenvalue of \changefinal{$-\fLap$ }acting on scalar functions on $\Sigma_T$ by $\mu_0(\fg_T)$. By the bootstrap assumption \eqref{eq:fut-bootstrap} and since $\mu_0$ depends continuously on the metric, we obtain the following for small enough $\epsilonnew>0$:
\begin{equation*}
\mu_0(\fg_T)\geq \mu_0(\gamma)-\frac12\left(\mu_0(\gamma)-\frac19\right)\geq \frac19+q
\end{equation*}
By the Poincaré inequality applied on $(\Sigma_T,\fg_T)$ (see \cite[p.1037]{CBM01}), the above spectral bound implies the following for any $\zeta\in H^1(\Sigma_T)$:
\begin{equation}\label{eq:adapted-poincare}
\|\zeta-\overline{\zeta}\|_{L^2_{\fg}(\Sigma_T)}^2\leq \mu_0(\fg_T)^{-1}\|\nabla\zeta\|_{L^2_{\fg}(\Sigma_T)}^2\leq \left(\frac19+q\right)^{-1}\|\nabla\zeta\|_{L^2_{\fg}(\Sigma_T)}^2
\end{equation}
For $l=0$, this means
\begin{align*}
\fE^{(0)}+\frac23\fC^{(0)}\geq&\,\|\phi^\prime\|_{L^2_{\fg}}^2+\|\nabla\phi\|_{L^2_{\fg}}^2-\frac23\|\phi-\phim\|_{L^2_{\fg}}\|\phi^\prime\|_{L^2_{\fg}}\\
\geq&\,\|\phi^\prime\|_{L^2_{\fg}}^2+\|\nabla\phi\|_{L^2_{\fg}}^2-2\left(1+9q\right)^{-\frac12}\|\nabla\phi\|_{L^2_{\fg}}\|\phi^\prime\|_{L^2_{\fg}}\\
\geq&\, \frac{\sqrt{1+9q}-1}{\sqrt{1+9q}}\fE^{(0)}\,.
\end{align*}
For $l=1$, notice that we can rewrite $\fC^{(1)}$ as 
\[\fC^{(1)}=\int_{M}\langle\nabla\phi,\nabla\phi^\prime\rangle_{\fg}\,\vol{\fg}=\int_{M}\left\langle\nabla\phi,\nabla\left(\phi^\prime-\overline{\phi^\prime}\right)\right\rangle_{\fg}\,\vol{\fg}=-\int_{M}\left(\phi^\prime-\overline{\phi^\prime}\right)\Lap_{\fg}\phi\,\vol{\fg}\,.\]
Hence, applying \eqref{eq:adapted-poincare} to $\zeta=\phi^\prime$ yields
\[\fE^{(1)}+\frac23\fC^{(1)}\geq \fE^{(1)}-2\left(1+9q\right)^{-\frac12}\|\nabla\phi^\prime\|_{L^2_{\fg}}\|\Lap_{\fg}\phi\|_{L^2_{\fg}}\geq \frac{\sqrt{1+9q}-1}{\sqrt{1+9q}}\fE^{(1)}\,.\]
For $l=2,3,4$, notice $\overline{\fLap\phi}=\overline{\fLap\phi^\prime}=\overline{\fLap^2\phi}=0$ holds due to the divergence theorem, hence the argument proceeds as in $l=0,1$.
\end{proof}

\begin{lemma}[Near-coercivity of corrected scalar field energy]\label{lem:fut-Sob-est}
For any differentiable function $\zeta$ and $k\in\{1,2\}$, one has the following under the bootstrap assumptions:
\begin{align*}
\int_M\lvert\nabla^2\zeta\rvert_{\fg}^2\,\vol{\fg}\lesssim&\,\int_M\lvert\fLap\zeta\rvert_{\fg}^2+\lvert\nabla\zeta\rvert_{\fg}^2\,\vol{\fg}\\
\|\zeta\|_{\dot{H}^{2k}}^2\lesssim&\,\|\fLap^{k}\zeta\|_{L^2}^2+\left(\|\zeta\|_{\dot{H}^{2k-1}}^2+\|\zeta\|_{\dot{H}^{2k-2}}^2\right)+\|\nabla\zeta\|_{C^{1}}^2\left\|\Ric[\fg]+\frac29\fg\right\|_{H^{2k-2}}^2\\
\|\nabla\zeta\|_{\dot{H}^{2k}}^2\lesssim&\,\|\nabla\fLap^{k}\zeta\|_{L^2}^2+\left(\|\nabla\zeta\|_{\dot{H}^{2k-1}}^2+\|\nabla\zeta\|_{\dot{H}^{2k-2}}^2\right)+\|\nabla\zeta\|_{C^{2}}^2\left\|\Ric[\fg]+\frac29\fg\right\|_{H^{2k-2}}^2
\end{align*}
Consequently, the following estimate holds:
\begin{align*}
\numberthis\label{eq:fut-coerc}\|\phi^\prime\|_{H^4}^2+\|\nabla\phi\|_{H^4}^2
\lesssim&\,E_{SF}^{(4)}+\left(\|\phi^\prime\|_{C^2}^2+\|\nabla\phi\|_{C^2}^2\right)\left\|\Ric[\fg]+\frac29\fg\right\|_{H^2}^2
\end{align*}
\end{lemma}
\begin{proof}
The inequalities for $\zeta$ follows from the same arguments as Lemma \ref{lem:Sobolev-norm-equivalence-improved}, except that we have $\|\Ric[\fg]\|_{C^1_{\fg}}\lesssim 1+\epsilonnew\lesssim 1$ by Lemma \ref{lem:fut-Ric-est}.  The final estimate then follows by applying these estimates to $\zeta=\phi^\prime$ and $\zeta=\phi$ and applying Lemma \ref{lem:fut-ESF-coercivity}.
\end{proof}

\subsubsection{Preparations for energy estimates}

Before proving the energy estimate, we need to establish two technical lemmas: First, we collect a formula to differentiate integrals, and then some estimates needed to deal with the mean value of $\phi$ in the base level correction term.
\begin{lemma}[Differentiation of integrals, future stability version] For any diffentiable function $\zeta$, one has
\begin{equation}\label{eq:fut-delt-int}
\del_T\int_M\zeta\vol{\fg}=\int_M\left(\fdel\zeta+3\fN\zeta\right)\,\vol{\fg}\,.
\end{equation}
\end{lemma}
\begin{proof}
\change{As in the proof of \eqref{eq:delt-int}, we obtain
\begin{align*}
\del_T\int_M\zeta\vol{\fg}=&\,\int_M\del_T\zeta+\frac{\del_T\mu_{\fg}}{\mu_{\fg}}\zeta\,\vol{\fg}\\
=&\,\int_M\del_T\zeta+3\fN\zeta-\frac12(\fg^{-1})^{ab}\Lie_{\fX}\fg_{ab}\zeta\,\vol{\fg}\\
=&\,\int_M\del_T\zeta+3\fN\zeta-\div_{\fg}\fX\cdot\zeta\,\vol{\fg}
\end{align*}
The statement now follows by applying Stokes' theorem to the final term and rearranging.}
\end{proof}

\begin{lemma}[Decay estimate for the integrated time derivative]\label{lem:fut-SF-technicality} For any $T>0$, we have
\begin{equation}\label{eq:fut-SF-technicality-1}
\int_{\Sigma_T}\phi^\prime\,\vol{\fg}=\left(\int_{\Sigma_{T=0}}\phi^\prime\,\vol{\fg}\right)\cdot e^{-2T}\,.
\end{equation}
Consequently, the bootstrap assumptions imply
\begin{equation}\label{eq:fut-SF-technicality-2}
\left\lvert\int_{\Sigma_T}\fdel\phim\cdot\phi^\prime\,\vol{\fg}\right\rvert\lesssim\,\epsilonnew^3e^{-\change{\frac52}T}
\end{equation}
for $\epsilonnew>0$ small enough.
\end{lemma}
\begin{proof}
Using that the integral of $\div_{\fg}(\fn\nabla\phi)$ vanishes, we compute:
\changediss{
\begin{align*}
\del_T\left(\int_M\phi^\prime\,\vol{\fg}\right)=&\,\int_{M}\left(\fdel\phi^\prime+3\fN\phi^\prime\right)\,\vol{\fg}=\int_M\left[(1-\fn)\phi^\prime+(\fn-3)\phi^\prime\right]\,\vol{\fg}\\
=&\,-2\left(\int_M\phi^\prime\,\vol{\fg}\right)
\end{align*}}
Hence, \eqref{eq:fut-SF-technicality-1} precisely describes the solution to this ODE ($f^\prime=-2f$) with prescribed initial value at $T=0$, and the initial data assumption \eqref{eq:fut-init} implies
\[\left\lvert\int_M\phi^\prime\,\vol{\fg}\right\rvert\leq \|\phi^\prime\|_{C^0(\Sigma_{T=0})}\vol{\fg}(\Sigma_{T=0})e^{-2T}\lesssim \epsilonnew^2e^{-2T}\,.\]
Furthermore, one has by \eqref{eq:fut-delt-int} \change{that }
\begin{equation}\label{eq:fut-volume-evol}
\change{\del_T\vol{\fg}\left(\Sigma_{T}\right)=\int_{\Sigma_T}3\fN\vol{\fg}}
\end{equation}
\change{Consequently, one has
\[\fdel\phim=\left[-\frac{\del_T\vol{\fg}\left(\Sigma_{T}\right)}{\vol{\fg}\left(\Sigma_{T}\right)}\cdot\phim+\frac1{\vol{\fg}\left(\Sigma_{T}\right)}\int_M\left(\fdel\phi+3\fN\phi\right)\,\vol{\fg}\right]=\int_M\left(\fn\phi^\prime+3\fN(\phi-\phim)\right)\,\vol{\fg}\,.\]
By applying $\lvert\fn\rvert<3$, the adapted Poincare inequality \eqref{eq:adapted-poincare} and the bootstrap assumptions \eqref{eq:fut-bootstrap}, this implies
\[\left\lvert\fdel\phim\right\rvert\lesssim\|\phi^\prime\|_{L^2}+\|\nabla\phi\|_{L^2}\|\fN\|_{L^2_{\fg}}\lesssim\epsilonnew e^{-\frac{T}2}\,.\]
The bound \eqref{eq:fut-SF-technicality-2} now follows by combining this with \eqref{eq:fut-SF-technicality-1}.}
\end{proof}

\subsubsection{Energy estimates}

Now, we can collect the following estimates for the corrected scalar field energies:

\begin{lemma}[Base level estimate for the corrected scalar field energy]\label{lem:fut-en-est-ESF0} Under the bootstrap assumptions, the following estimate holds \change{for some $K>0$}:
\begin{equation}
\del_T E^{(0)}_{SF}\leq-2E^{(0)}_{SF}+K\delta e^{-\frac{T}2}\sqrt{E^{(0)}_{SF}}\left(\sqrt{E^{(0)}_{SF}}+\|\bm{\Sigma}\|_{L^2}+\|\fg-\gamma\|_{L^2}\right)+K\delta^3e^{-\frac{5T}2}
\end{equation}
\end{lemma}
\begin{proof}
We compute, using $[\fdel,\nabla]\phi=0$, $\fdel\phi=\fn\phi^\prime$ and the rescaled wave equation \eqref{eq:fut-wave}:
\begin{align*}
\del_T\fE^{(0)}=&\,\int_M \left[2\fdel\phi^\prime\cdot\phi^\prime+2\left\langle\nabla\phi,\nabla\fdel\phi\right\rangle_{\fg}+\left(\fdel\fg^{-1}\right)^{ab}\nabla_a\phi\nabla_b\phi+3\fN\left(\lvert\phi^\prime\rvert^2+\lvert\nabla\phi\rvert_{\fg}^2\right)\right]\,\vol{\fg}\\
=&\,\int_M \biggr[2\left(\langle\nabla\fn,\nabla\phi\rangle_{\fg}+\fn\fLap\phi+(1-\fn)\phi^\prime\right)\phi^\prime-2(\fn\phi^\prime)\cdot\fLap\phi\\
&\,\,\phantom{\int_M}-2\fn\langle\fk,\nabla\phi\nabla\phi\rangle_{\fg}+3\fN\lvert\phi^\prime\rvert^2+\fN\lvert\nabla\phi\rvert_{\fg}^2\biggr]\,\vol{\fg}
\end{align*}
With $2(1-\fn)=-4-6\fN$, integration by parts and using the bootstrap assumption \eqref{eq:fut-bootstrap} on $C$-norms, we get for some constant $K>0$ that we update from line to line:
\begin{align*}
\del_T\fE^{(0)}\leq&\,\int_M -4\lvert\phi^\prime\rvert_{\fg}^2\,\vol{\fg}+K\left[\|\nabla\phi\|_{C^0}\|\fN\|_{H^1}\sqrt{\fE^{(0)}}+\left(\|\fk\|_{C^0}+\|\fN\|_{C^0}\right)\fE^{(0)}\right]\\
\leq&\int_M-4\lvert\phi^\prime\rvert^2\,\vol{\fg}+K\epsilonnew e^{-\frac{T}2}\left(\change{\sqrt{\fE^{(0)}}}{\|\fN\|_{H^1}}+\fE^{(0)}\right)
\end{align*}
Similarly and using the same evolution equations, we obtain:
\begin{align*}
\del_T\fC^{(0)}=&\,\int_M \left[\fdel\phi\cdot\phi^\prime-\fdel\phim\cdot\phi^\prime+(\phi-\phim)\fdel\phi^\prime+3\fN(\phi-\phim)\phi^\prime\right]\,\vol{\fg}\\
=&\,\int_M\left[3\lvert\phi^\prime\rvert^2+3\fN\lvert\phi^\prime\rvert^2+\left(\phi-\phim\right)\cdot \div_{\fg}\left(\fn\nabla\phi\right)-2\left(\phi-\phim\right)\phi^\prime-\fdel\phim\cdot\phi^\prime\right]\,\vol{\fg}\\
\leq&\,-2\fC^{(0)}+\int_M3\left[\lvert\phi^\prime\rvert^2-\lvert\nabla\phi\rvert_{\fg}^2\right]\,\vol{\fg}+3\|\fN\|_{C^0}\fE^{(0)}-\int_M\left(\fdel\phim\cdot\phi^\prime\right)\,\vol{\fg}
\end{align*}
Applying Lemma \ref{lem:fut-SF-technicality} to the last term, we get:
\begin{align*}
\del_T\fC^{(0)}\leq&-2\fC^{(0)}+K\epsilonnew e^{-\frac{T}2}\fE^{(0)}+\change{K\epsilonnew^3e^{-\frac52T}}
\end{align*}
Combining these two estimates, inserting \eqref{eq:ell-est-lapse} and \eqref{eq:fut-ESF-coercivity} as well as updating $K$ yields:
\begin{align*}
\del_TE_{SF}^{(0)}=&\,\del_T\fE^{(0)}+\frac23\del_T\fC^{(0)}\\
=&\,\int_M \left[\left(-4+\frac23\cdot 3\right)\lvert\phi^\prime\rvert^2-\frac23\cdot 3\lvert\nabla\phi\rvert_{\fg}^2\right]\,\vol{\fg}-2\cdot\frac23\change{\fC^{(0)}}\\
&\,+\change{K\delta e^{-\frac{T}2}\change{\left(\sqrt{\fE^{(0)}}{\|\hat{\bm{n}}\|_{H^1}}+\fE^{(0)}\right)}+K\delta^3 e^{-\frac52T}}\\
\leq&\,-2E_{SF}^{(0)}+K\delta \change{e^{-\frac{T}2}}\sqrt{E^{(0)}_{SF}}\left(\|\Sigma\|_{L^2}+\|\fg-\gamma\|_{L^2}+\sqrt{E^{(0)}_{SF}}\right)+K\delta^3 e^{-\frac52T}
\end{align*}
\end{proof}

\begin{lemma}[Higher order estimates for the corrected scalar field energy]\label{lem:fut-en-est-ESF}
For any $l\in\{1,\dots,4\}$, the following estimate holds:
\begin{align*}
\del_T\left(\fE^{(l)}+\frac23\fC^{(l)}\right)\leq&\,-2\left(\fE^{(l)}+\frac23\fC^{(l)}\right)+K\epsilonnew e^{-\frac{T}2}\left(\sum_{m=0}^l\sqrt{\fE^{(m)}}\right)\cdot\\
&\,\qquad\cdot\change{\left(\|\phi^\prime\|_{H^{l}}+\|\nabla\phi\|_{H^{l}}+\|\Sigma\|_{H^{l}}+\|\fg-\gamma\|_{H^l}
\right)}
\end{align*}
\end{lemma}
\begin{proof}
Starting with $l=2k,\,k\in\{1,2\}$, one calculates:
\begin{subequations}
\begin{align}
\del_T\fE^{(2k)}=&\,\int_M\biggr[2\fLap^k\fdel\phi^\prime\cdot\fLap^k\phi^\prime+2\langle\nabla\fLap^{k}\phi,\nabla\fLap^k\fdel\phi\rangle_{\fg}\label{eq:fESF-1}\\
&\,\phantom{\int_M}+(\fdel {\fg}^{-1})^{ab}\cdot\nabla_a\fLap^k\phi\cdot\nabla_b\fLap^k\phi+3\fN\left(\lvert\Lap^k\phi^\prime\rvert_{\fg}^2+\lvert\nabla\Lap^k\phi\rvert_{\fg}^2\right)\label{eq:fESF-2}\\
&\,\phantom{\int_M}+2[\fdel,\fLap^k]\phi^\prime\cdot\fLap^k\phi^\prime+2\left\langle[\fdel,\nabla\fLap^k]\phi,\nabla\fLap^k\phi\right\rangle_{\fg}\biggr]\,\vol{\fg}\label{eq:fESF-3}
\end{align}
We insert the rescaled wave equation \eqref{eq:fut-wave} and $\fdel\phi=\fn\phi^\prime$ into the right hand side of \eqref{eq:fESF-1} and obtain for some constant $K>0$:
\begin{align*}
\eqref{eq:fESF-1}\leq &\,\int_M \biggr[-4\lvert\fLap^k\phi^\prime\rvert^2 -6\fN\lvert\fLap^k\phi^\prime\rvert^2+\fn\fLap^{k+1}\phi\cdot\fLap^k\phi^\prime\biggr]\,\vol{\fg}\\
&\,+K\|\Lap^k\phi^\prime\|_{L^2}\left(\|\fN\|_{H^{2k+1}}\|\nabla\phi\|_{C^0}+\|\nabla\phi\|_{H^{2k}}\|\fN\|_{C^{2k}}\right)\\
&\,+\int_M\left[-\fn\fLap^{k}\phi^\prime\cdot\fLap^{k+1}\nabla\phi-3\langle\nabla\fN,\nabla\fLap^k\phi\rangle_ {\fg}\cdot\fLap^k\phi^\prime\right]\,\vol{\fg}\\
&\,+K\|\nabla\fLap^k\phi\|_{L^2}\left(\changefinal{\|\fN\|_{H^{2k+1}}\|\phi^\prime\|_{C^0}+\|\fN\|_{C^{2k}}\|\phi^\prime\|_{H^{2k}}}\right)\\
\leq&\,\int_M -4\lvert\fLap^k\phi^\prime\rvert^2\,\vol{\fg}\\
&\,+K\sqrt{\fE^{(2k)}}\cdot\Bigr[\left(\|\nabla\phi\|_{C^0}+\|\phi^\prime\|_{C^0}\right)\cdot\|\fN\|_{H^{2k+1}}+\left(\|\nabla\phi\|_{H^{2k}}+\|\phi^\prime\|_{H^{2k}}\right)\cdot\|\fN\|_{C^{2k}}\Bigr]
\end{align*}
\end{subequations}
For \eqref{eq:fESF-2}, we use \eqref{eq:fut-eq-g-1} and the bootstrap assumption \eqref{eq:fut-bootstrap} to bound it by $K\epsilonnew e^{-\frac{T}2}\fE^{(2k)}$. Regarding \eqref{eq:fESF-3}, the commutator formulas \eqref{eq:[fdel,Lapk]}-\eqref{eq:[fdel,nablaLapk]} imply
\begin{align*}
\|[\fdel,\fLap^k]\phi^\prime\|_{L^2}\lesssim&\,\change{\|\fn\|_{C^{2k-1}}\left(\|\phi^\prime\|_{C^1}\|\fk\|_{\dot{H}^{2k-1}}+\|\fk\|_{C^{2k-2}}\|\phi^\prime\|_{H^{2k}}\right)}\changefinal{+\|\fN\|_{C^{2k-1}}\|\phi^\prime\|_{H^{2k}}}\\
\|[\fdel,\nabla\fLap^k]\phi\|_{L^2}\lesssim&\,\change{\|\fn\|_{C^{2k}}\left(\|\nabla\phi\|_{C^1}\|\fk\|_{{H}^{2k}}+\|\fk\|_{C^{2k-2}}\|\nabla\phi\|_{H^{2k}}\right)}\changefinal{+\|\fN\|_{C^{2k}}\|\nabla\phi\|_{H^{2k}}}
\end{align*}
Summarizing, inserting the $C$-norm bounds from the bootstrap assumption \eqref{eq:fut-bootstrap} and updating $K$, this implies
\begin{align*}
\del_T\fE^{(2k)}\leq&\,\int_M-4\lvert\fLap^k\phi^\prime\rvert^2\,\vol{\fg}+K\epsilonnew e^{-\frac{T}2}\fE^{(2k)}\\
&\,+K\epsilonnew e^{-\frac{T}2}\sqrt{\fE^{(2k)}}\left(\|\phi^\prime\|_{H^{2k}}+\|\nabla\phi\|_{H^{2k}}\right)\\
&\,+K\epsilonnew e^{-\frac{T}2}\sqrt{\fE^{(2k)}}\left(\change{\|\fN\|_{H^{2k+1}}+\|\Sigma\|_{H^{2k}}}
\right)
\end{align*}
Moving on to the corrective term, we compute:
\begin{subequations}
\begin{align*}
\del_T\fC^{(2k)}
=&\,\int_M\biggr[\fLap^k\fdel\phi\cdot\fLap^k\phi^\prime+\fLap^k\phi\cdot\fLap^k\fdel\phi^\prime\numberthis\label{eq:fESFcorr-1}\\
&\,\phantom{\int_M}+3\fN\cdot\fLap^k\phi\cdot\fLap^k\phi^\prime+[\fdel,\fLap^k]\phi\cdot\fLap\phi^\prime+\fLap^k\phi\cdot[\fdel,\fLap^k]\phi^\prime\biggr]\,\vol{\fg}\numberthis\label{eq:fESFcorr-2}
\end{align*}
\end{subequations}
Inserting the evolution equations into the right hand side of \eqref{eq:fESFcorr-1}, we can bound that line by
\begin{align*}
\leq&\,\int_M\left[3\lvert\fLap^k\phi^\prime\rvert^2+3\fN\lvert\fLap^k\phi^\prime\rvert^2\right]\vol{\fg}+K\|\fN\|_{C^{2k}}\|\phi^\prime\|_{H^{2k-1}}\|\Lap^k\phi^\prime\|_{L^2}\\
&+\int_M\left[ -2\fLap^k\phi\cdot\fLap^k\phi^\prime+3\fN\fLap^k\phi\cdot\fLap^k\phi^\prime +3\fLap^k\phi\cdot\fLap^{k+1}\phi+3\fN\fLap^k\phi\cdot\fLap^{k+1}\phi\right]\,\vol{\fg}\\
&\,+K\left[\|\fN\|_{C^{2k}}\left(\|\nabla\phi\|_{H^{2k}}+\|\phi^\prime\|_{H^{2k-1}}\right)+\left(\|\nabla\phi\|_{C^0}+\|\phi^\prime\|_{C^0}\right)\|\fN\|_{H^{2k+1}}\right]\|\fLap^k\phi\|_{L^2}
\end{align*}
Note that, after integrating by parts, the last two terms in the second line can be bounded by
\[\int_M-3\lvert\nabla\fLap\phi\rvert_{\fg}^2\,\vol{\fg}+\|\fN\|_{C^1}(\|\nabla\Lap^k\phi\|_{L^2}+\|\Lap^k\phi\|_{L^2})\|\nabla\Lap^k\phi\|_{L^2}\,.\]
For the terms in \eqref{eq:fESFcorr-2}, notice that the first term can be bounded by $ \epsilonnew e^{-\frac{T}2}\|\nabla\phi\|_{H^{2k-1}}\sqrt{\fE^{(2k)}}$, 
while the commutator terms can be estimated as before, with
\begin{align*}
\|[\fdel,\fLap^k]\phi\|_{L^2}\lesssim&\,\change{\|\nabla\phi\|_{C^{0}}\|\fn\|_{C^0}\|\fk\|_{\dot{H}^{2k-1}}}
\change{+\|\fn\|_{C^{2k}}\|\fk\|_{C^{2k-2}}}
\change{\|\nabla\phi\|_{H^{2k-1}}}
\end{align*}

\noindent Combining all of the above, we get
\begin{align*}
\del_T\fC^{(2k)}\leq&\,-2\fC^{(2k)}+\int_M \left[3\lvert \fLap^k\phi^\prime\rvert-3\lvert\nabla\fLap^k\phi\rvert_{\fg}\right]\,\vol{\fg}
\\
&\,+K\epsilonnew e^{-\frac{T}2}\change{\left[\|\phi^\prime\|_{H^{2k}}+\|\nabla\phi\|_{H^{2k}}+\|\fN\|_{H^{2k+1}}+\|\fk\|_{H^{2k}}
\right]}\\
&\,\phantom{+K}\cdot\left(\sqrt{\fE^{(2k)}}+\sqrt{\fE^{(2k-1)}}\right)
\end{align*}

\noindent Finally, combining both differential estimates 
yields the statement for $l=2k$.
For $l=2k-1,\,k\in\{1,2\}$, the argument is completely analogous and hence omitted.
\end{proof}

\subsection{Geometric variables}\label{subsec:fut-geom-var} We can take the following results from prior literature, where \change{we additionally apply the elliptic estimates in Lemma \ref{lem:fut-ell-est}}:


\begin{lemma}[Coercivity of geometric energies, {\cite[Lemma 7.4]{AM11}}]\label{lem:fut-en-geom-coercivity}
For sufficiently small $\epsilonnew>0$, the following estimate holds:
\begin{equation}\label{eq:fut-en-geom-coerc}
\|\fg-\gamma\|_{H^5}^2+\|\Sigma\|_{H^4}^2\lesssim \fEg
\end{equation}
\end{lemma}

\begin{lemma}[Geometric energy estimate, {\cite[Lemma 20]{AndFaj20}}]\label{lem:fut-geom-est}
Let $\epsilonnew>0$ be chosen appropriately small, and let
\begin{equation}\label{eq:fut-alpha}
\alpha=\begin{cases}
1 & \lambda_0>\frac19\\
1-3\sqrt{\epsilonnew^\prime} & \lambda_0=\frac19\,,
\end{cases}
\end{equation}
where $\epsilonnew^\prime>0$ is the same as in \eqref{eq:fut-corr-const}, in particular suitably small. Then, there exists some constant $K>0$ such that the following estimate holds:
\begin{align*}
\numberthis\label{eq:fut-geom-est}\del_T\fEg\leq&\,-2\alpha\fEg+K\fEg^\frac32+\change{K\epsilonnew e^{-\frac{T}2}}\sqrt{E_{geom}}\left[\|\phi^\prime\|_{H^4}+\|\nabla\phi\|_{H^4}\right]\\
\end{align*}
\end{lemma}

\subsection{Closing the bootstrap}\label{subsec:fut-bs-imp}

Now, we can collect our estimates to improve the bootstrap assumptions:

\begin{prop}[Improved bounds for future stability]\label{prop:fut-bs-imp} Let the bootstrap assumption (see Assumption \ref{ass:fut-bootstrap}) be satisfied for $T\in[0,T_{Boot})$ and assume the initial data assumption holds at $T=0$ (see Assumption \ref{ass:fut-init}). For $\epsilonnew>0$ sufficiently small and $\alpha$ as in \eqref{eq:fut-alpha} with $\epsilonnew^\prime>0$ sufficiently small, the following estimates hold: 
\begin{subequations}
\begin{align}
\|\phi^\prime\|_{C^2}+\|\nabla\phi\|_{C^2}+\|\phi^\prime\|_{H^4}+\|\nabla\phi\|_{H^4}\lesssim&\,\epsilonnew^\frac32 \change{e^{-\alpha T}}\label{eq:fut-sf-imp}\\
\|\fg-\gamma\|_{C^3}+\|\Sigma\|_{C^2}+\|\fg-\gamma\|_{H^5}+\|\Sigma\|_{H^4}\lesssim&\,\epsilonnew^\frac32e^{-\alpha T}\label{eq:fut-geom-imp}\\
\|\fN\|_{C^4}+\|\fX\|_{C^4}+\|\fN\|_{H^6}+\|\fX\|_{H^6}\lesssim&\,\epsilonnew^3e^{-2\alpha T}\label{eq:fut-ell-imp}
\end{align}
\delete{Further[...]}
\end{subequations}
\end{prop}
\begin{proof}\change{In the following, the positive constant $K$ may be updated from line to line.\\
Combining the estimate from Lemma \ref{lem:fut-en-est-ESF0} as well as those from Lemma \ref{lem:fut-en-est-ESF} at each level with Lemma \ref{lem:fut-geom-est} and applying the (near)-coercivity estimates \eqref{eq:fut-coerc} and \eqref{eq:fut-en-geom-coerc} to the right hand sides, we obtain:
\begin{align*}
\del_T\left(E^{(4)}_{SF}+E_{geom}\right)\leq&\, -2E_{SF}+K\epsilonnew e^{-\frac{T}2} \sqrt{E_{SF}}\left(\sqrt{E_{SF}+\epsilonnew^2e^{-T}\left\|\Ric[\fg]+\frac29\fg\right\|_{H^2}^2}+\sqrt{E_{geom}}\right)\\
&\,+K\delta^3 e^{-\frac52T}\\
&\,-2\alpha E_{geom}+KE^\frac32_{geom}+K\epsilonnew e^{-\frac{T}2}\sqrt{E_{geom}}\sqrt{E_{SF}+\epsilonnew^2e^{-T}\left\|\Ric[\fg]+\frac29\fg\right\|_{H^2}^2}
\end{align*}
Applying \eqref{eq:fut-Ric-est} to the curvature norms, as well as \eqref{eq:fut-en-geom-coerc} to the resulting norms on $\fg-\gamma$ and \eqref{eq:fut-bootstrap} (which implies $\sqrt{E_{geom}}\lesssim\epsilonnew e^{-\frac{T}2}$), this becomes
\[\del_T\left(E_{SF}+E_{geom}\right)\leq -2\alpha (E_{SF}^{(4)}+E_{geom})+K\epsilonnew e^{-\frac{T}2}\left(E_{SF}+E_{geom}\right)+K\epsilonnew^3e^{-\frac52T}\,.\]
and consequently, since $\alpha\leq1$,
\[\del_T\left[e^{2\alpha T}\left(E_{SF}+E_{geom}\right)\right]\lesssim \epsilonnew e^{-\frac{T}2}\cdot e^{2\alpha T}\left(E_{SF}+E_{geom}\right)+\epsilonnew^3 e^{-\frac{T}2}\]
The Gronwall lemma, along with the initial data assumption \eqref{eq:fut-init}, now implies
\begin{equation}\label{eq:fut-en-imp}
E_{SF}+E_{geom}\lesssim \epsilonnew^3e^{-2\alpha T}\,.
\end{equation}
}
Lemma \ref{lem:fut-en-geom-coercivity} and the standard Sobolev embedding then imply \eqref{eq:fut-geom-imp}. In particular, this means
\begin{equation}\label{eq:fut-curv-imp}
\left\|\Ric[\fg]+\frac29\fg\right\|_{H^2}\change{\lesssim} \epsilonnew^\frac32e^{-\alpha T}
\end{equation}
due to Lemma \ref{lem:fut-Ric-est}, and for $\epsilonnew^\prime>0$ small enough, inserting \eqref{eq:fut-en-imp} and \eqref{eq:fut-curv-imp} into \eqref{eq:fut-coerc} shows \eqref{eq:fut-sf-imp}. Moreover, \eqref{eq:fut-ell-imp} follows directly from \changefinal{the proof of }Lemma \ref{lem:fut-ell-est} and the already obtained improvements. \delete{Finally [..]}
\end{proof}

\begin{proof}[Proof of Theorem \ref{thm:fut-stab-simple}]

The problem is locally well-posed as outlined in Remark \ref{rem:fut-lwp}. There then is some maximal interval $[0,T_{Boot})$ for the logarithmic time $T$ -- or, equivalently, some maximal time interval $[\tau_0,\tau_{Boot})$ -- on which the solution exists and the bootstrap assumptions (see Assumption \ref{ass:fut-bootstrap}) are satisfied. By the analogous argument to the proof of Theorem \ref{thm:main}, the decay estimates in Proposition \ref{prop:fut-bs-imp} are strictly stronger than the bootstrap assumptions for small enough $\epsilonnew,\epsilonnew^\prime>0$. This implies $T_{Boot}=\infty$ (resp. $\tau_{Boot}=0$) since we could else extend the solution strictly beyond $T_{Boot}$ while also satisfying the bootstrap assumptions. \delete{To be a bit more precise [...]}
This proves the convergence statement in Theorem \ref{thm:fut-stab-simple}.

Finally, the decay estimates imply that $\lvert\nabla n\rvert_{g}$, respectively $\lvert k\rvert_{g}$, are bounded by $\tau^{\alpha-1}$, respectively $\tau^{\alpha+1}$, up to constant on $[\tau_0,\tau)$. Since $\alpha$ is at worst slightly smaller than $1$, both functions are integrable on $[\tau_0,0)$ for suitably small $\epsilonnew^\prime>0$ . By \cite{CB02}, this means the spacetime is future complete.
\end{proof}

\section{Global stability}\label{sec:full-stab}

To prove Theorem \ref{thm:main-full}, what still needs to be shown is that initial data as in Theorem \ref{thm:main-past} develops from $\Sigma_{t_0}$ to some hypersurface $\Sigma_{t_1}\equiv \Sigma_{\tau(t_1)}$ in its future such that the data in $\Sigma_{t_1}$ is near-Milne in the sense of Assumption \ref{ass:fut-init} and in CMCSH gauge. From there, near-Milne stability yields the behaviour in the future of $\Sigma_{\tau(t_1)}$, and hence future stability of near-FLRW spacetimes as in Theorem \ref{thm:main-full}. 

\begin{proof}[Proof of Theorem \ref{thm:main-full}] Within this proof, $t$ will denote the \enquote{physical} time coordinate used throughout the Big Bang stability analysis, while $\tau$ denotes the mean curvature time used within CMCSH gauge.

Consider initial data $(g,k,\nabla\phi,\del_0\phi)$ induced on the CMC hypersurface $\Sigma_{t_0}$ within $\M$ such that the rescaled variables are close to FLRW reference data in the sense of Theorem \ref{thm:main}. Moreover, let $(\mathring{g},\mathring{k},\mathring{\pi},\mathring{\psi})$ be the geometric initial data on $M$ that induce it via the embedding $\iota:M\hookrightarrow \M$.

Notice that
\[P:H^{20}_\gamma(M)\rightarrow H^{18}_\gamma(M),\,Y^i\mapsto \Lap_\gamma Y^i+(\gamma^{-1})^{il}{\Ric[\gamma]}_{lj}Y^j=\Lap_{\gamma} Y^i-\frac29 Y^i\]
is an isomorphism since $\Lap_\gamma$ has no positive eigenvalues. Hence, using \cite[Theorem 2.5, Remark 2.6]{FajKr20}, there is a metric $\mathring{g}^\prime$ isometric to $\mathring{g}$ that remains \changefinal{close in $H^{18}_\gamma(M)$ }to $\gamma$ and satisfies
\[((\mathring{g}^\prime)^{-1})^{ij}\left(\Gamma[\mathring{g}^\prime]^k_{ij}-\Gamhat[\gamma]^{k}_{ij}\right)=0.\]
Let $\theta\in\text{Diff}(M)$ be the diffeomorphism such that $\theta^\ast \mathring{g}=\mathring{g}^\prime$, then the proof of \cite[Theorem 2.5]{FajKr20} implies that $\theta$ can be chosen close to the identity map within $H^{18}(\text{Diff}(M))$, and consequently that $\theta^\ast\mathring{k}=\mathring{k}^\prime,\, \theta^\ast\mathring{\pi}=\mathring{\pi}^\prime$ and $\theta^\ast\mathring{\psi}=\mathring{\psi}^\prime$ remain close to $-\dot{a}(t_0)a(t_0)\gamma,\,0$ and $Ca(t_0)^{-3}$ in $H^{18}_\gamma(M)$. By the same argument as in Remark \ref{rem:CMC-hypersurface}, we can now evolve this data locally and obtain a new initial hypersurface $\Sigma^\prime$ close to $\Sigma_{t_0}$ that is in CMCSH gauge and that $(g,k,\nabla\phi,\del_0\phi)$ is close to the reference data in the sense of Assumption \ref{ass:init}, exchanging the initial time $t_0$ by some close time $t_0^\prime$.\\

Since $\tau$ is strictly increasing, $t\equiv t(\tau)$ exists and we can interchangeably view $a$ as a function in $t$ or $\tau$ with some abuse of notation. The Friedman equation \eqref{eq:Friedman} implies $\del_ta\geq\frac19$ and thus $a(t)\geq\frac19t$ on $(0,\infty)$, as well as 
\[-\tau=3\frac{\dot{a}}a= \frac1a+\langle\text{lower order terms}\rangle\ \text{as}\ t\to \infty\ (\text{resp. }\tau\to 0)\,.\]

We choose $t_1>\max\{1,t_0^\prime\}$ large enough (resp. $\tau(t_1)\equiv\tau_0$ small enough) that the following estimates hold for some small $\chi\in(0,\frac12)$ that depends only on $\epsilonnew$:
\begin{align}
Ca(t_1)^{-3}\tau(t_1)^{-1}\leq&\, \chi \label{eq:connect-time-1}
\\
 -\tau(t_1)\cdot{a(t_1)}\in&\,[1-\chi,1+\chi] \label{eq:connect-time-3}
\end{align}

As the solution \change{is Cauchy stable, i.e.,~ it and }its maximal time of existence depend continuously upon the initial data,\footnote{For the argument for Einstein vacuum in CMCSH gauge, see \cite[Theorem 3.1]{AM03}. As with local existence, the argument in the Einstein scalar-field system is largely identical since the only difference amounts to coupling the hyperbolic parts of the system with a further hyperbolic one.} one can choose $\epsilon>0$ in the analogue of Assumption \ref{ass:init} small enough to ensure the following: The solution exists until $t_1>t_0^\prime$ and $(a^{-2}g,a\hat{k},\nabla\phi,a^{3}\fdel\phi)$ remain $K\epsilon$-close to $(\gamma,0,0,C)$ in $H_\gamma^6\times H_\gamma^5\times H_\gamma^5\times H_\gamma^5$ for some suitable $K>0$ along the slab $\cup_{s\in [t_0^\prime,t_1]}\Sigma_s$. What now remains to be shown is that this implies Assumption \ref{ass:fut-init} in the sense that, if $\epsilon$ is small enough, $\epsilonnew$ can be made as small as necessary for Theorem \ref{thm:fut-stab-simple} to apply.\\

Note that the scalings in Definition \ref{def:fut-rescaled} can be rewritten as
\begin{gather*}
\fg-\gamma=(\tau \cdot a)^2\cdot (a^{-2} g-\gamma)+(\tau^2\cdot a^2-1)\gamma,\quad \fk=\frac{\tau}a (a\hat{k})\,,\\
\phi^\prime=C\left(-\tau^{-1}\cdot a^{-3}\right)+\left(-\tau^{-1}\cdot a^{-3}\right)\cdot (a^3n^{-1}(\del_\tau-\Lie_X)\phi-C)\,.
\end{gather*}
Since \eqref{eq:connect-time-3} implies $\tau\cdot a$ is close to $-1$ at $t_1$, $\|(\tau\cdot a)^2(a^{-2}g-\gamma)\|_{H^6}$ can be bounded by $\frac{\epsilonnew^3}2$ for small enough $\epsilon$. Choosing $\chi<\frac{\epsilonnew^3}2$ then implies $\|\fg-\gamma\|_{H^6(\Sigma_{\tau_0})}<\epsilonnew^3$. That $\|\fk\|_{H^5}$ can be made smaller than $\epsilonnew^3$ for small enough $\epsilon>0$ follows since $\frac\tau{a}$ behaves like $\frac1{a^2}$ up to a constant by \eqref{eq:connect-time-3}.\\
For the normal derivative of the wave, notice that $\lvert C\left(-\tau^{-1}\cdot a^{-3}\right)\rvert $ is bounded by $\chi$ due to \eqref{eq:connect-time-1}, and that $-\tau^{-1}a^{-3}$ is equivalent to $a^{-2}$ by \eqref{eq:connect-time-3}. Hence, we can similarly ensure that $\phi^\prime$ is bounded in $H^5$ by $\epsilonnew^3$. Since $\nabla\phi$ is not changed in either rescaling, and bounds on lapse and shift (up to constant) follow from the elliptic estimates in Lemma \ref{lem:fut-ell-est}, it follows each individual norm in Assumption \ref{ass:fut-init} can be bounded by $\epsilonnew^3$ up to constants that depend only on $\gamma$, and hence the initial data assumption itself can be satisfied for suitably small $\epsilonnew>0$.\\
This proves that we can develop from initial data for the Big Bang stability proof to near-Milne initial data within a CMCSH foliation, and thus we obtain Theorem \ref{thm:main-full} from Theorems \ref{thm:main-past} and \ref{thm:fut-stab-simple}.
\end{proof}

\section{Appendix -- Big Bang Stability}\label{sec:appendix}

\subsection{Basic formulas and estimates}\label{subsec:basic}

\subsubsection{Tools from elementary calculus}

\begin{lemma}[A Gronwall lemma]\label{lem:gronwall} Let $f,\chi,\xi:[a,b]\longrightarrow \R$ be continuous functions such that $\chi\geq 0$, $\xi$ is decreasing and, for any $s\in[a,b]$,
\[f(s)\leq \int_s^b\chi(r)f(r)\,dr+\xi(s)\]
is satisfied. Then, for any $t\in[a,b]$, we have
\[f(t)\leq \xi(t)\exp\left(\int_t^b\chi(r)\,dr\right)\,.\]
\end{lemma}
\begin{proof}
This follows by standard arguments as in \cite[Corollary 2-3]{Dra03}.
\end{proof}

\begin{lemma}[A weak fundamental theorem of calculus for square roots]\label{lem:weak-ftoc}
Let $f:(0,t_0]\longrightarrow\R^+_0$ be a $C^1$-function. Then, we have for any $t\in(0,t_0]$:
\begin{equation}\label{eq:weak-ftoc}
\sqrt{f(t)}\leq \sqrt{f(t_0)}+\int_t^{t_0}\frac{\lvert f^\prime(s)\rvert}{2\sqrt{f(s)}}\,ds
\end{equation}
\end{lemma}
\begin{proof} This follows from a straightforward application of the monotone convergence theorem to $g_n=\sqrt{f+\frac1n}$.
\end{proof}

\subsubsection{Levi-Civita tensor identities}\label{subsubsec:basic-contract}
Herein, we collect some basic identities for the Levi-Civita tensor $\epsilonLC[g]$: Firstly, it satisfies the contraction identities, where $\I^a_b$ denotes the Kronecker-symbol:
\begin{subequations}
\begin{align}
\epsilonLC^{ai_1i_2}\epsilonLC_{aj_1j_2}=&\,\I^{i_1}_{j_1}\I^{i_2}_{j_2}-\I^{i_1}_{j_2}\I^{i_2}_{j_1}\label{eq:LCS-contr1}\\
\epsilonLC^{abi}\epsilonLC_{abj_2}=&\,2\I^{i}_j\label{eq:LCS-contr2}\\
\epsilonLC^{abc}\epsilonLC_{abc}=&\,6\label{eq:LCS-contr3}\\
\nabla\epsilonLC=&\,0\label{eq:LCS-cov}
\end{align}
\end{subequations}
The analogous formulas hold for $\epsilonLC[G]$ when raising indices with regard to $G$ instead of $g$.\\
\noindent For a tracefree and symmetric $\Sigma_t$-tangent $(0,2)$-tensor $\mathfrak{T}$ and a $\Sigma_t$-tangent $(0,2)$-tensor $\mathfrak{A}$, the following simplified identities hold:
\begin{subequations}
\begin{align}
(\mathfrak{T}\times \mathfrak{A})_{ij}&={\epsilonLC_i}^{ab}{\epsilonLC_j}^{pq}\mathfrak{T}_{ap}\mathfrak{A}_{bq}+\frac13(\mathfrak{T}\cdot\mathfrak{A})g_{ij} \label{eq:times-lemma-1}\\
(\mathfrak{T}\times g)_{ij}&=-\mathfrak{T}_{ij} \label{eq:times-lemma-2}\\
(\mathfrak{T}\times k)_{ij}&=-\frac{\tau}3\mathfrak{T}_{ij}+(\mathfrak{T}\times\hat{k})_{ij} \label{eq:times-lemma-3}
\end{align}
Further, note the following formulas (for $\tilde{\mathfrak{T}}$ as $\mathfrak{T}$,  $\tilde{\mathfrak{A}}$ as $\mathfrak{A}$ and any $\Sigma_t$-tangent $(0,1)$-tensor $\xi$) (see \cite[p.30]{AM03}):
\begin{align}
\div_g(\mathfrak{A}\wedge\tilde{\mathfrak{A}})&=-\curl \mathfrak{A}\cdot\tilde{\mathfrak{A}}+\mathfrak{A}\cdot\curl\tilde{\mathfrak{A}} \label{eq:div-to-curl}\\
\mathfrak{A}\cdot(\xi\wedge\tilde{\mathfrak{A}})&=-2\xi\cdot(\mathfrak{A}\wedge\tilde{\mathfrak{A}})\\
\mathfrak{T}\cdot(\mathfrak{A}\times\tilde{\mathfrak{T}})&=(\mathfrak{T}\times \mathfrak{A})\cdot\tilde{\mathfrak{T}}
\end{align}
\end{subequations}

\subsubsection{Estimates on contracted tensors}\label{subsubsec:basic-est}

\begin{lemma}\label{lem:BelRobinsonLemmas}
Let $\mathfrak{S},\mathfrak{T}$ be traceless and symmetric $\Sigma_t$-tangent $(0,2)$-tensors, $\mathfrak{M},\mathfrak{N}$ symmetric $\Sigma_t$-tangent $(0,2)$-tensors and $\xi$ a $\Sigma_t$-tangent $(0,1)$-tensor. We define $G, G^{-1}$ and $\lvert\cdot\rvert_G$ via \eqref{eq:rescalingGK}). Then:
\begin{subequations}
\change{\begin{align}
\lvert \mathfrak{M}\odot_G\mathfrak{N}\rvert_G\leq\lvert \mathfrak{M}\rvert_G\lvert \mathfrak{N}\rvert_G,&\quad \mathfrak{M}\odot_g\mathfrak{N}= a^{-2}\mathfrak{M}\odot_G\mathfrak{N}\label{eq:odot}\\
\lvert \mathfrak{S}\times_G \mathfrak{T}\rvert_G\lesssim\lvert \mathfrak{S}\rvert_G\lvert \mathfrak{T}\rvert_G,&\quad (\mathfrak{S}\times \mathfrak{T})_{ij}=a^{-3}(\mathfrak{S}\times_G \mathfrak{T})_{ij}\label{eq:times}\\
\lvert \mathfrak{S}\wedge_G \mathfrak{T}\rvert_G\leq \lvert \mathfrak{S}\rvert_G\lvert \mathfrak{T}\rvert_G,&\quad \left(\mathfrak{S}\wedge \mathfrak{T}\right)_l= a^{-3}\left(\mathfrak{S}\wedge_G \mathfrak{T}\right)\label{eq:wedge}\\
\lvert \xi\wedge_G \mathfrak{T}\rvert_G\leq\lvert \xi\rvert_G\lvert \mathfrak{T}\rvert_G,&\quad \left(\xi\wedge\mathfrak{T}\right)_{ij}= a^{-1}(\xi\wedge_G\mathfrak{T})_{ij}\label{eq:wedge2}\\
\lvert\curl_G\mathfrak{M}\rvert_G\lesssim \lvert\nabla \mathfrak{M}\rvert_G,&\quad \curl \mathfrak{M}_{ij}= a^{-1}\curl_G\mathfrak{M}_{ij}\label{eq:curl}
\end{align}}
\end{subequations}
\end{lemma}
\begin{proof} The estimates with respect to the unrescaled metric are direct consequences of the contraction identities \eqref{eq:LCS-contr1}-\eqref{eq:LCS-contr3} \change{replacing $g$ with $G$}, and \change{the scalings }follow simply by tracking the effects of the rescaling in Definition \ref{def:rescaled}. In particular, 
\begin{equation}\label{eq:epsilonLC-resc}
{\epsilonLC[g]_{i}}^{cd}=g^{cj}g^{dk}\epsilonLC[g]_{ijk}=\left(a^{-2} (G^{-1})^{cj}\right)\left(a^{-2} (G^{-1})^{dk}\right)a^3\epsilonLC[G]_{ijk}=a^{-1}{\epsilon[G]_{i}}^{\sharp cd}\
\end{equation}
determines the scaling of the Levi-Civita-Symbol.
\end{proof}

\subsection{Commutators}\label{subsec:commutators} Herein, we collect a variety of commutators of spatial derivative operators with each other as well as with time derivatives. While these mostly follow by standard computations, we use the fact that our spatial hypersurfaces are three-dimensional to significantly simplify the spatial commutator formulas, and need to apply the rescaled equations from Proposition \ref{prop:REEq} for the time derivative formulas.

For higher order commutators, we denote by $\mathfrak{J}$ terms within the commutator formula that contribute junk terms at any point where this commutator formula is used. Furthermore, in the following, $\zeta$ denotes a scalar function on $\M$ and $\mathfrak{T}$ denotes a $\Sigma_t$-tangent, symmetric $(0,2)$-tensor, always with sufficient regularity for the equations to make sense. Moreover, recall the schematic $\ast$-notation as introduced in subsection \ref{subsubsec:schematic-notation}.

\begin{corollary}[Schematic first order spatial commutators]\label{lem:comm-space-first}
For $\zeta$ and $\mathfrak{T}$ as above, the following identities hold:
\begin{subequations}
\begin{align}
[\Lap,\nabla]\zeta=&\,\Ric[G]\ast\change{\nabla\zeta}\label{eq:[Lap,nabla]SF}\\
[\Lap,\nabla^2]\zeta=&\,\Ric[G]\ast\nabla^2\zeta+\nabla\Ric[G]\ast\nabla\zeta\label{eq:[Lap,nabla2]SF}\\
[\Lap,\nabla]\mathfrak{T}=&\,\Ric[G]\ast\nabla\mathfrak{T}+\nabla \Ric[G]\ast\mathfrak{T}\label{eq:[Lap,nabla]T}\\
[\Lap,\nabla^2]\mathfrak{T}=&\,\Ric[G]\ast\nabla^2\mathfrak{T}+\nabla\Ric[G]\ast\nabla\mathfrak{T}+\nabla^2\Ric[G]\ast\mathfrak{T}\label{eq:[Lap,nabla2]T}\\
[\Lap,\div_G]\mathfrak{T}=&\,\Ric[G]\ast\nabla\mathfrak{T}+\nabla\Ric[G]\ast\mathfrak{T}\label{eq:[Lap,div]}\\
\change{[\Lap,\curl_G]}=&\,\change{\epsilonLC[G]\ast\left(\Ric[G]\ast\nabla\mathfrak{T}+\nabla\Ric[G]\ast\mathfrak{T}\right)}\label{eq:[Lap,curl]}
\end{align}
\end{subequations}
\end{corollary}
\begin{proof}
Since we are working in three spatial dimensions, the following identity holds:
\begin{align*}
\Riem[G]_{ijkl}=&\,G_{ik}\Ric[G]_{jl}-G_{il}\Ric[G]_{jk}+G_{jl}\Ric[G]_{ik}-G_{jk}\Ric[G]_{il}\\
&\,-\frac12 (G^{-1})^{mn}\Ric[G]_{mn}(G_{ik}G_{jl}-G_{il}G_{jk})
\end{align*}
Hence, for any $I\in\N_0$, any $\nabla^I\Riem[G]$-term reduces to a sum of products and contractions of $\nabla^I\Ric[G]$ with various metric tensors that are all suppressed in schematic notation. With this in mind, the above statements are simply direct consequences of standard commutation cormulas and \eqref{eq:epsilonLC-resc}. 
\end{proof}

\begin{lemma}[Higher order spatial commutators]\label{lem:comm-space}
For $l\in\N,\,l\geq 2$, the following formulas hold (and extend to $l=1$ when dropping any term involving $\Lap^{l-2}$):
\begin{subequations}
\begin{align*}
\numberthis\label{eq:[Lap-l,nabla]SF}[\Lap^l,\nabla]\zeta=&\,\Lap^{l-1}\Ric[G]\ast\nabla\zeta+\nabla\Lap^{l-2}\Ric[G]\ast\nabla^2\zeta+\mathfrak{J}([\Lap^l,\nabla]\zeta)\\
\numberthis\label{eq:[Lap-l,nabla2]SF}[\Lap^l,\nabla^2]\zeta=&\,\nabla\Lap^{l-1}\Ric[G]\ast\nabla\zeta+\nabla^2\Lap^{l-2}\Ric[G]\ast\nabla^2\zeta+\mathfrak{J}([\Lap^l,\nabla^2]\zeta)\\
\numberthis\label{eq:[Lap-l,nabla]T}[\Lap^l,\nabla]\mathfrak{T}=&\,\nabla\Lap^{l-1}\Ric[G]\ast \mathfrak{T}+\nabla^2\Lap^{l-2}\Ric[G]\ast\nabla \mathfrak{T}+\mathfrak{J}([\Lap^l,\nabla]\mathfrak{T})\,,\\
\numberthis\label{eq:[Lap-l,nabla2]T}[\Lap^l,\nabla^2]\mathfrak{T}=&\,\nabla^2\Lap^{l-1}\Ric[G]\ast \mathfrak{T}+\nabla^3\Lap^{l-2}\Ric[G]\ast \nabla \mathfrak{T}+\mathfrak{J}([\Lap^l,\nabla^2]\mathfrak{T})\,,\\
\numberthis\label{eq:[Lap-l,div]T}[\Lap^l,\div_G]\mathfrak{T}=&\,\nabla\Lap^{l-1}\Ric[G]\ast \mathfrak{T}+\nabla^2\Lap^{l-2}\Ric[G]\ast \nabla \mathfrak{T}+\mathfrak{J}([\Lap^l,\div_G]\mathfrak{T})\,,\\
\numberthis\label{eq:[Lap-l,curl]}\change{[\Lap^l,\curl_G]\mathfrak{T}=&\change{\,\epsilonLC[G]\ast\left(\nabla\Lap^{l-1}\Ric[G]\ast \mathfrak{T}+\nabla^2\Lap^{l-2}\Ric[G]\ast\nabla \mathfrak{T}\right)+\mathfrak{J}([\Lap^l,\curl_G]\mathfrak{T})}}
\end{align*}
with junk terms, where $\mathcal{I}=I_1+\dots+I_{l-m}$,
{\scriptsize 
\begin{align*}
\mathfrak{J}([\Lap^l,\nabla]\zeta)=&\,\sum_{\substack{I_1+I_\zeta=2(l-1),\\\,I_\zeta\geq 2}}\nabla^{I_1}\Ric[G]\ast\nabla^{I_\zeta+1}\zeta+\sum_{m=0}^{l-2}\sum_{\mathcal{I}+I_{\zeta}=2m}\nabla^{I_1}\Ric[G]\ast\dots\ast\nabla^{I_{l-m}}\Ric[G]\ast\nabla^{I_\zeta+1}\zeta\\
\mathfrak{J}([\Lap^l,\nabla^2]\zeta)=&\,\sum_{\substack{I_1+I_\zeta=2(l-1)+1,\\I_1,I_\zeta\geq 2}}\nabla^{I_1}\Ric[G]\ast\nabla^{I_\zeta+1}\zeta+\sum_{m=0}^{l-2}\sum_{\mathcal{I}+I_\zeta=2m+1}\nabla^{I_1}\Ric[G]\ast\dots\ast\nabla^{I_{l-m}}\Ric[G]\ast\nabla^{I_\zeta+1}\zeta\\
\mathfrak{J}([\Lap^l,\nabla]\mathfrak{T})=&\,\sum_{\substack{I_1+I_\mathfrak{T}=2(l-1)+1,\\I_\mathfrak{T}\geq2}}\nabla^{I_1}\Ric[G]\ast\nabla^{I_\mathfrak{T}}\mathfrak{T}+\sum_{m=0}^{l-2}\sum_{\mathcal{I}+I_\mathfrak{T}=2m+1}\nabla^{I_1}\Ric[G]\ast\dots\ast\nabla^{I_{l-m}}\Ric[G]\ast\nabla^{I_\mathfrak{T}}\mathfrak{T}\\
\mathfrak{J}([\Lap^l,\nabla^2]\mathfrak{T})=&\,\sum_{\substack{I_1+I_\mathfrak{T}=2l,\\I_\mathfrak{T}\geq2}}\nabla^{I_1}\Ric[G]\ast\nabla^{I_\mathfrak{T}}\mathfrak{T}+\sum_{m=0}^{l-2}\sum_{\mathcal{I}+I_\mathfrak{T}=2m+2}\nabla^{I_1}\Ric[G]\ast\dots\ast\nabla^{I_{l-m}}\Ric[G]\ast\nabla^{I_\mathfrak{T}}T\\
\mathfrak{J}([\Lap^l,\div_G]\mathfrak{T})=&\,\sum_{\substack{I_1+I_\mathfrak{T}=2(l-1)+1,\\I_\mathfrak{T}\geq2}}\nabla^{I_1}\Ric[G]\ast\nabla^{I_\mathfrak{T}}\mathfrak{T}+\sum_{m=0}^{l-2}\sum_{\mathcal{I}+I_\mathfrak{T}=2m+1}\nabla^{I_1}\Ric[G]\ast\dots\ast\nabla^{I_{l-m}}\Ric[G]\ast\nabla^{I_\mathfrak{T}}\mathfrak{T}\\
\change{\mathfrak{J}([\Lap^l,\curl_G]\mathfrak{T})}=&\change{\epsilonLC[G]\ast\left[\sum_{\substack{I_1+I_\mathfrak{T}=2(l-1)+1,\\I_\mathfrak{T}\geq2}}\nabla^{I_1}\Ric[G]\ast\nabla^{I_\mathfrak{T}}\mathfrak{T}+\sum_{m=0}^{l-2}\sum_{\substack{\mathcal{I}+\\+I_\mathfrak{T}=2m+1}}\nabla^{I_1}\Ric[G]\ast\dots\ast\nabla^{I_{l-m}}\Ric[G]\ast\nabla^{I_\mathfrak{T}}\mathfrak{T}\right]}\,.
\end{align*}}
\end{subequations}
\end{lemma}
\begin{proof}
The formulas follow by applying the formulas from Lemma \ref{lem:comm-space-first} inductively.
\end{proof}

\begin{lemma}[Time derivative commutators]\label{lem:com-time-first} With respect to a solution to the Einstein scalar-field system as in Proposition \ref{prop:eq}, the following commutator formulas hold:
\begin{subequations}
\begin{align*}
\numberthis\label{eq:[del-t,nabla]zeta}[\del_t,\nabla_i]\zeta=&\,0\\
\numberthis\label{eq:[del-t,nabla-sharp]zeta}[\del_t,\nabla^{\sharp i}]\zeta=&\,2(N+1)a^{-3}\Sigma^{\sharp ij}\nabla_j\zeta-2N\frac{\dot{a}}a\nabla^{\sharp i}\zeta\\
\numberthis\label{eq:[del-t,Lap]zeta}[\del_t,\Lap]\zeta=&\,2(N+1)a^{-3}\langle\Sigma,\nabla^2\zeta\rangle_G-2N\frac{\dot{a}}a\Lap\zeta\\
&\,-2(N+1)a^{-3}\langle\div_G\Sigma,\nabla\zeta\rangle_G-2a^{-3}\langle\Sigma,\nabla N\nabla\zeta\rangle_G+\frac{\dot{a}}a\langle\nabla N,\nabla\zeta\rangle_G\\
\numberthis\label{eq:[del-t,nabla]T}
[\del_t,\nabla]\mathfrak{T}=&\,a^{-3}\left((N+1)\,\nabla\Sigma+\Sigma\ast\nabla N\right)\ast\mathfrak{T}+\frac{\dot{a}}a\,\nabla N\ast\mathfrak{T}\\
\numberthis\label{eq:[del-t,Lap]T}[\del_t,\Lap]\mathfrak{T}
=&\,a^{-3}(N+1)\Sigma\ast\nabla^2\mathfrak{T}+\frac{\dot{a}}aN\change{\Lap\mathfrak{T}}+a^{-3}\nabla((N+1)\Sigma)\ast\nabla \mathfrak{T}+\frac{\dot{a}}a\nabla N\ast\nabla \mathfrak{T}\\
&\,+a^{-3}\nabla^2((N+1)\Sigma)\ast \mathfrak{T}-\frac{\dot{a}}a\nabla^2N\ast \mathfrak{T}
\end{align*}
\end{subequations}
\end{lemma}
\begin{proof}
Equation \eqref{eq:[del-t,nabla]zeta} is simply that coordinate derivatives commute, and \eqref{eq:[del-t,nabla-sharp]zeta} follows by applying \eqref{eq:REEqG-1} and the product rule.\\
For the \change{commutators \eqref{eq:[del-t,Lap]zeta}, }\eqref{eq:[del-t,nabla]T} and \eqref{eq:[del-t,Lap]T}, we write out the covariant derivatives in local coordinates, apply the product rule, and then the evolution equations \eqref{eq:REEqG-1} and \eqref{eq:REEqChr} for the inverse metric and Christoffel symbols. \delete{For \eqref{eq:[del-t,Lap]zeta} [...]}
\end{proof}

\begin{lemma}[High order time derivative commutators]\label{lem:com-time}
For $l\in\N,\,l\geq 2$, the time derivative commutators take the form
\begin{subequations}
\begin{align*}
\numberthis\label{eq:[del-t,Lap-l]zeta}[\del_t,\Lap^l]\zeta=&\,2a^{-3}(N+1)\langle\Sigma,\nabla^2\Lap^{l-1}\zeta\rangle_G\change{+a^{-3}\nabla\Sigma\ast\nabla^3\Lap^{l-2}\zeta}\\
&\change{\,-2(N+1)a^{-3}\langle\div_G\Lap^{l-1}\Sigma,\nabla\zeta\rangle_G+(N+1)a^{-3}\nabla^{2l-3}\Ric\ast\Sigma\ast\nabla\zeta+\mathfrak{J}([\del_t,\Lap^l]\zeta),}\\
\numberthis\label{eq:[del-t,nabla-Lap-l]zeta}[\del_t,\nabla\Lap^l]\zeta=&\,2a^{-3}(N+1)\langle\Sigma,\nabla^3\Lap^{l-1}\zeta\rangle_G+a^{-3}(N+1)\nabla\Sigma\ast\nabla^{2l}\zeta\\
&\,-2(N+1)a^{-3}\langle \nabla\div_G\Lap^{l-1}\Sigma,\nabla\zeta\rangle_G\\
&\,+\frac{\dot{a}}a\langle\nabla^2\Lap^{l-1}N,\nabla\zeta\rangle_G+(N+1)a^{-3}\nabla^{2l-2}\Ric[G]\ast\Sigma\ast\nabla\zeta\\
&\,+\mathfrak{J}([\del_t,\nabla\Lap^l]\zeta)\\
\numberthis\label{eq:[del-t,Lap-l]T}[\del_t,\Lap^l]\mathfrak{T}=&\,a^{-3}\left(\Sigma\ast\nabla^2\Lap^{l-1}\mathfrak{T}+\nabla\Sigma\ast\nabla^3\Lap^{l-2}\mathfrak{T}+\nabla \mathfrak{T}\ast\nabla\Lap^{l-1}\Sigma+\mathfrak{T}\ast\Lap^{l}\Sigma\right)\\
&\,+a^{-3}\left((N+1)\Sigma\ast \mathfrak{T}\ast\nabla^2\Lap^{l-2}\Ric[G]+\nabla((N+1)\Sigma\ast \mathfrak{T})\ast\nabla^{2l-3}\Ric[G]\right)\\
&\,+\frac{\dot{a}}a\Lap^lN\cdot \mathfrak{T}+\frac{\dot{a}}a\nabla\Lap^{l-1}N\ast\nabla \mathfrak{T}+\mathfrak{J}([\del_t,\Lap^l])\mathfrak{T}\,,\\
\numberthis\label{eq:[del-t,nabla-Lap-l]T}[\del_t,\nabla\Lap^l]\mathfrak{T}=&\,a^{-3}\nabla\Sigma\ast\Lap^l\mathfrak{T}+a^{-3}(N+1)\Sigma\ast\nabla^3\Lap^{l-1}\mathfrak{T}+a^{-3}(N+1)\mathfrak{T}\ast\nabla\Lap^l\Sigma\\
&+\frac{\dot{a}}a\nabla\Lap^lN\ast \mathfrak{T}+\frac{\dot{a}}a\nabla^2\Lap^{l-1}N\ast\nabla \mathfrak{T}+a^{-3}(N+1)\Sigma\ast\nabla^3\Lap^{l-2}\Ric[G]\ast \mathfrak{T}\\
&\,+\mathfrak{J}([\del_t,\nabla\Lap^l]\mathfrak{T})\,,
\end{align*}
where the junk terms are, where $\mathcal{I}=\sum_{i=1}^{l-m-1}I_i$,
{\scriptsize
\change{\begin{align*}
\numberthis\label{eq:[del-t,Lap]zeta-junk}\mathfrak{J}([\del_t,\Lap^l]\zeta)=&\,a^{-3}\sum_{\substack{I_N+I_\Sigma+I_\zeta=2(l-1)\\I_\zeta\leq2(l-2)}}\nabla^{I_N}(N+1)\ast\nabla^{I_\Sigma}\Sigma\ast\nabla^{I_\zeta+2}\zeta\\
&\,+\frac{\dot{a}}a\sum_{\substack{I_N+I_\zeta=2l\\ I_\zeta\geq 2}}\nabla^{I_N}N\ast\nabla^{I_\zeta}\zeta\\
&\,+a^{-3}\sum_{m=0}^{l-2}\sum_{\substack{I_N+I_\Sigma+I_\zeta+\mathcal{I}=2m}}\nabla^{I_N}(N+1)\ast\nabla^{I_\Sigma}\Sigma\ast\nabla^{I_1}\Ric[G]\ast\dots\ast\nabla^{I_{l-m-1}}\Ric[G]\ast\nabla^{I_\zeta+2}\zeta\\
&\,+a^{-3}\sum_{m=0}^{l-2}\sum_{\substack{I_N+I_\Sigma+I_\zeta+\mathcal{I}=2m\\I_1\neq 2l-4}}\nabla^{I_N}(N+1)\ast\nabla^{I_\Sigma}\Sigma\ast\nabla^{I_1+1}\Ric[G]\ast\dots\ast\nabla^{I_{l-m-1}}\Ric[G]\ast\nabla^{I_\zeta+1}\zeta\\
&\,+\frac{\dot{a}}a\sum_{m=0}^{l-1}\sum_{\substack{I_N+\mathcal{I}+I_\zeta=2m-1\\\,I_\zeta\neq2(l-1)}}\nabla^{I_N}N\ast\nabla^{I_1}\Ric[G]\ast\dots\ast\nabla^{I_{l-m-1}}\Ric[G]\ast\nabla^{I_\zeta+1}\zeta\\
\numberthis\label{eq:[del-t,nabla-Lap-l]zeta-junk}\mathfrak{J}([\del_t,\nabla\Lap^l]\zeta)=&\,\frac{\dot{a}}a\sum_{\substack{I_N+I_\zeta=2l\\\,I_\zeta\neq 0}}\nabla^{I_N}N\ast\nabla^{I_\zeta+1}\zeta+a^{-3}\sum_{\substack{I_N+I_\Sigma+I_\zeta=2(l-1)+1\\(I_\Sigma, I_\zeta)\neq(0,2(l-1)+1),(1,2(l-1))}}\nabla^{I_N}(N+1)\ast\nabla^{I_\Sigma}\Sigma\ast\nabla^{I_\zeta+1}\zeta\\
&\,+a^{-3}\sum_{m=0}^{l-2}\sum_{\substack{I_N+I_\Sigma+I_\zeta+\mathcal{I}=2m+1}}\nabla^{I_N}(N+1)\ast\nabla^{I_\Sigma}\Sigma\ast\nabla^{I_1}\Ric[G]\ast\dots\ast\nabla^{I_{l-m-1}}\Ric[G]\ast\nabla^{I_\zeta+2}\zeta\\
&\,+a^{-3}\sum_{m=0}^{l-2}\sum_{\substack{I_N+I_\Sigma+I_\zeta+\mathcal{I}=2m+1\\I_1\neq 2l-3}}\nabla^{I_N}(N+1)\ast\nabla^{I_\Sigma}\Sigma\ast\nabla^{I_1+1}\Ric[G]\ast\dots\ast\nabla^{I_{l-m-1}}\Ric[G]\ast\nabla^{I_\zeta+1}\zeta\\
&\,+\frac{\dot{a}}a\sum_{m=0}^{l-1}\sum_{\substack{I_N+\mathcal{I}+I_\zeta=2m\\\,I_\zeta\neq2(l-1)}}\nabla^{I_N}N\ast\nabla^{I_1}\Ric[G]\ast\dots\ast\nabla^{I_{l-m-1}}\Ric[G]\ast\nabla^{I_\zeta+1}\zeta\\}
\numberthis\label{eq:[del-t,Lap-l]T-junk}\mathfrak{J}([\del_t,\Lap^l])\mathfrak{T}=&\,a^{-3}\sum_{I_N+I_\Sigma+I_\mathfrak{T}=2l}\nabla^{I_N}N\ast\nabla^{I_\Sigma}\Sigma\ast\nabla^{I_\mathfrak{T}}\mathfrak{T}+a^{-3}\sum_{\substack{I_\Sigma+I_\mathfrak{T}=2l\\I_\Sigma,I_\mathfrak{T}\geq 2}}\nabla^{I_\Sigma}\Sigma\ast\nabla^{I_\mathfrak{T}}\mathfrak{T}\\
&\,+a^{-3}\sum_{m=0}^{l-1}\sum_{\substack{I_N+I_\Sigma+I_\mathfrak{T}+\mathcal{I}=2m\\ I_1<2l-3}}\nabla^{I_N}(N+1)\ast\nabla^{I_\Sigma}\Sigma\ast\nabla^{I_1}\Ric[G]\ast\dots\ast\nabla^{I_{l-m-1}}\Ric[G]\ast\nabla^{I_\mathfrak{T}}\mathfrak{T}\\
&\,+\frac{\dot{a}}a\sum_{I_N+I_\mathfrak{T}=2l,\,I_\mathfrak{T}\geq 2}\nabla^{I_N}N\ast\nabla^{I_\mathfrak{T}}\mathfrak{T}\\
\numberthis\label{eq:[del-t,nabla-Lap-l]T-junk}\mathfrak{J}([\del_t,\nabla\Lap^l]\mathfrak{T})=&\,a^{-3}\Sigma\ast\nabla N\ast\Lap^l\Ric[G]+a^{-3}N\ast\nabla\Sigma\ast\Lap^l\mathfrak{T}+\frac{\dot{a}}a\nabla N\ast\Lap^l\mathfrak{T}+\nabla\mathfrak{J}([\del_t,\Lap^l]\mathfrak{T})
\end{align*}}
\end{subequations}
\noindent We can extend the formulas to $l=1$ by dropping any term which would contain negative powers of $\Lap$ or a multiindex of negative order.
\end{lemma}
\begin{proof}
This follows by iteratively applying the commutators in Lemma \ref{lem:com-time-first}.
\end{proof}


While all of the above commutators will be essential for the mainline argument, the a priori estimates require the following commutators:

\begin{lemma}[Auxiliary commutators]\label{lem:aux-comm} Let $J\in\N$. Then, we have:
\begin{subequations}
\begin{align*}
\numberthis\label{eq:commutator-aux-scalar}[\del_t,\nabla^J]\zeta=&\,a^{-3}\sum_{I_N+I_\Sigma+I_\zeta=l-1,\,I_\zeta<J-1}\nabla^{I_N}(N+1)\ast\nabla^{I_\Sigma}\Sigma\ast\nabla^{I_\zeta+1}\zeta\\
&\,+\frac{\dot{a}}a\sum_{I_N+I_\zeta=J-1,I_N>0}\nabla^{I_N}N\ast\nabla^{I_\zeta+1}\zeta\\
\numberthis\label{eq:commutator-aux-tensor}[\del_t,\nabla^J]\mathfrak{T}=&\,a^{-3}\sum_{I_N+I_\Sigma+I_{\mathfrak{T}}=J,\,I_\mathfrak{T}< J}\nabla^{I_N}(N+1)\ast\nabla^{I_\Sigma}\Sigma\ast\nabla^{I_\mathfrak{T}}\mathfrak{T}\\
&\,+\frac{\dot{a}}a\sum_{I_N+I_\mathfrak{T}=J,I_N>0}\nabla^{I_N}N\ast\nabla^{I_\mathfrak{T}}\mathfrak{T}\\
\end{align*}
\end{subequations}
\end{lemma}
\begin{proof}
For $J=1$, this has already been shown in \eqref{eq:[del-t,nabla]zeta} and \eqref{eq:[del-t,nabla]T}. For higher orders, the formulas follow from a straightforward induction argument using that, in local coordinates, we schematically have
\begin{align*}
[\del_t,\nabla^J]\zeta=[\del_t,\nabla]\nabla^{J-1}\zeta+\nabla[\del_t,\nabla^{J-1}]\zeta=(\del_t\Gamma[G])\ast\nabla^{J-1}\zeta+\nabla[\del_t,\nabla^{J-1}]\zeta
\end{align*}
and analogously replacing $\zeta$ with $\mathfrak{T}$. 
\end{proof}

\subsection{Borderline and junk terms}\label{subsec:error-terms}

\begin{definition}[Error terms]\label{def:error-terms} Let $L\in2\N,\,L\geq 2$. Then, the error terms in the Laplace-commuted equations stated in Lemma \ref{lem:laplace-commuted-eq} take the following form:\\

\noindent For the constraint equations, we have
{\scriptsize \begin{subequations}
\begin{align*}
\numberthis\label{eq:comeq-mom-div-junk}\mathfrak{M}_{L,Junk}=&\,-8\pi(\Psi+C)\nabla\Lap^{\frac{L}2-2}\Ric[G]\ast\nabla^2\phi+\nabla^{L-2}\Ric[G]\ast\nabla\Sigma+\underbrace{\nabla^{L-3}\Ric[G]\ast\nabla^2\Sigma}_{\text{if }L\neq 2}\\
&\,+\sum_{I_\Psi+I_\phi=L,\,I_\Psi\neq 0}\nabla^{I_\Psi}\Psi\ast\nabla^{I_\phi+1}\phi+8\pi(\Psi+C)\mathfrak{J}([\Lap^{\frac{L}2},\nabla]\phi)-\mathfrak{J}([\Lap^{\frac{L}2},\div_G]\Sigma)\\
\numberthis\label{eq:comeq-mom-curl-junk}\change{\tilde{\mathfrak{M}}_{L,Junk}=}&\change{\,\underbrace{-\epsilonLC[G]\ast\nabla^{L-3}\Ric[G]\ast\nabla\Sigma}_{\text{if }L\neq 2}-\mathfrak{J}([\Lap^\frac{L}2,\curl_G]\Sigma)}\\
\mathfrak{H}_{L,Border}=&\,a^{-4}\left[\Sigma\ast\Lap^{\frac{L}2}\Sigma+\nabla\Sigma\ast\nabla^{L-1}\Sigma\right] \numberthis\label{eq:comeq-Ham-border}\\
\mathfrak{H}_{L,Junk}=&\,\sum_{I_1+I_2=L}\nabla^{I_1+1}\phi\ast\nabla^{I_2+1}\phi+a^{-4}\sum_{I_1+I_2=L,I_i\geq2}\nabla^{I_1}\Sigma\ast\nabla^{I_2}\Sigma \numberthis\label{eq:comeq-Ham-junk}\\
&\,+\Lap^\frac{L}2\left[\frac{4\pi}3\lvert\nabla\phi\rvert_G^2+\frac{8\pi}3a^{-4}\Psi^2+\frac{16\pi}3Ca^{-4}\Psi\right]\cdot G
\end{align*}
\end{subequations}}
\noindent The lapse equation error terms are
\begin{subequations}
{\scriptsize \begin{align*}
\mathfrak{N}_{L,Border}=&\,a^{-4}(N+1)\left(\Sigma\ast\Lap^{\frac{L}2}\Sigma+\nabla\Sigma\ast\nabla^{L-1}\Sigma+\Psi\ast\Lap^{\frac{L}2}\Psi+\nabla\Psi\ast\nabla^{L-1}\Psi\right) \numberthis \label{eq:comeq-lapse-border}\\
&\,+a^{-4}\left[\lvert\Sigma\rvert_G^2+\Psi^2+\Psi\right]\ast\Lap^{\frac{L}2}N+a^{-4}\nabla\left[\lvert\Sigma\rvert_G^2+\Psi^2+\Psi\right]\ast\nabla^{L-1}N\,\\
\mathfrak{N}_{L,Junk}=&\,a^{-4}\sum_{\substack{I_N+I_1+I_2=L;\\I_N\leq L-2;\,I_N>0\text{ or }I_1\leq I_2\leq L-2}}\nabla^{I_N}(N+1)\ast\left(\nabla^{I_1}\Sigma\ast\nabla^{I_2}\Sigma+\nabla^{I_1}\Psi\ast\nabla^{I_2}\Psi\right) \numberthis\label{eq:comeq-lapse-junk}\\
&\,+a^{-4}N\ast\Lap^{\frac{L}2}\Psi+a^{-4}\sum_{I_N+I_\Psi=L;\,I_\Psi\geq 2,\,I_N\geq 1}\nabla^{I_N}N\ast\nabla^{I_\Psi}\Psi\,,
\end{align*}}
as well as
{\scriptsize \begin{align*}
\numberthis \label{eq:comeq-lapse-border-odd}\mathfrak{N}_{L+1,Border}=&\,a^{-4}(N+1)\Bigr(\Sigma\ast\nabla\Lap^{\frac{L}2}\Sigma+\nabla\Sigma\ast\nabla^{L}\Sigma+\nabla^2\Sigma\ast\nabla^{L-1}\Sigma+\Psi\ast\nabla\Lap^{\frac{L}2}\Psi\\
&\,+\nabla\Psi\ast\nabla^{L}\Psi+\nabla^2\Psi\ast\nabla^{L-1}\Psi\Bigr)+a^{-4}\left[\lvert\Sigma\rvert_G^2+\Psi^2+\Psi\right]\ast\nabla\Lap^{\frac{L}2}N\\
&\,+\nabla\Psi\ast\nabla^{L}N+\nabla^2\Psi\ast\nabla^{L-1}\Psi\,\\
\mathfrak{N}_{L+1,Junk}=&\,a^{-4}\sum_{\substack{I_N+I_1+I_2=L;\\I_N<L+1;\,I_N>0\text{ or }I_1\geq I_2>2}}\nabla^{I_N}(N+1)\ast\left(\nabla^{I_1}\Sigma\ast\nabla^{I_2}\Sigma+\nabla^{I_1}\Psi\ast\nabla^{I_2}\Psi\right) \numberthis\label{eq:comeq-lapse-junk-odd}\\
&\,+a^{-4}N\ast\nabla\Lap^{\frac{L}2}\Psi+a^{-4}\sum_{I_N+I_\Psi=L+1;\,I_N,I_\Psi>2}\nabla^{I_N}N\ast\nabla^{I_\Psi}\Psi
\end{align*}}
\end{subequations}
whereas the scalar field error terms read
\begin{subequations}
{\scriptsize \begin{align*}
\numberthis\label{eq:comeq-Psi-even-border}\mathfrak{P}_{L,Border}=&\,-3\Psi\frac{\dot{a}}a\Lap^{\frac{L}2}N+\frac{\dot{a}}a\nabla\Psi\ast\nabla^{L-1}N +2a^{-3}(N+1)\langle\Sigma,\nabla^2\Lap^{\frac{L}2-1}\Psi\rangle_G+2a^{-3}(N+1)\nabla^{L-3}\Ric\ast\Sigma\ast\nabla\Psi\\
&\,\change{-2a^{-3}(N+1)\langle\div_G\Lap^{\frac{L}2-1}\Sigma,\nabla\Psi\rangle_G+a^{-3}(N+1)\nabla\Sigma\ast\nabla^3\Lap^{\frac{L}2-2}\Psi}\\
\mathfrak{P}_{L,Junk}=&\,\frac{\dot{a}}a\sum_{I_N+I_\Psi=L,\,I_\Psi\geq 2}\nabla^{I_N}N\ast\nabla^{I_\Psi}\Psi+a\sum_{I_N+I_\phi=L+1,\,I_N,I_\phi\neq 0}\nabla^{I_N}N\ast\nabla^{I_\phi+1}\phi+\mathfrak{J}([\del_t,\Lap^\frac{L}2]\Psi)\numberthis\label{eq:comeq-Psi-even-junk}\\
\numberthis\label{eq:comeq-Q-even-border}\mathfrak{Q}_{L,Border}=&\,a^{-3}\Psi\nabla\Lap^{\frac{L}2}N+a^{-3}(N+1)\Sigma\ast\nabla^3\Lap^{\frac{L}2-1}\phi\change{+a^{-3}(N+1)\nabla^{L}\phi\ast\nabla \Sigma}
\\
\mathfrak{Q}_{L,Junk}=&\,a^{-3}\sum_{I_N+I_\Psi=L+1,\,I_N,I_\Psi\neq 0}\nabla^{I_N}N\ast\nabla^{I_\Psi}\Psi+a^{-3}\nabla\Lap^{\frac{L}2-1}N\ast\nabla^2\phi\ast\Sigma\numberthis\label{eq:comeq-Q-even-junk}\\
&\,\change{+a^{-3}(N+1)\nabla^2\Lap^{\frac{L}2-1}\Sigma\ast\nabla\phi}+(N+1)a^{-3}\nabla\Lap^{\frac{L}2-1}\Ric[G]\ast\Sigma\ast\nabla\phi\\
&\,+a^{-3}\nabla^{L-2}\Ric[G]\ast\left((N+1)\ast\Sigma\ast\nabla\phi\right)+\frac{\dot{a}}a\langle\nabla^2\Lap^{\frac{L}2-1}N,\nabla\phi\rangle_G
\change{+\mathfrak{J}([\del_t,\nabla\Lap^{\frac{L}2}]\phi)}
\end{align*}}
and
{\scriptsize \begin{align*}
\mathfrak{P}_{L+1,Border}=&\,-3\Psi\frac{\dot{a}}a\nabla\Lap^{\frac{L}2}N+\frac{\dot{a}}a\nabla\Psi\ast\nabla^2\Lap^{\frac{L}2-1}N +2a^{-3}\langle\Sigma,\nabla^3\Lap^{\frac{L}2-1}\Psi\rangle_G\numberthis\label{eq:comeq-Psi-odd-border}\\
&\,+a^{-3}(N+1)\nabla\Sigma\ast\nabla^L\Psi+2a^{-3}\nabla^{L-2}\Ric\ast\Sigma\ast\nabla\Psi\\
\change{&\,+a^{-3}(N+1)\nabla^2\Lap^{\frac{L}2-1}\Sigma\ast\nabla\Psi}\\
\mathfrak{P}_{L+1,Junk}=&\,\frac{\dot{a}}a\sum_{I_N+I_\Psi=L+1,\,I_\Psi\geq 2}\nabla^{I_N}N\ast\nabla^{I_\Psi}\Psi+a\sum_{I_N+I_\phi=L+2,\,I_N,I_\phi\neq 0}\nabla^{I_N}N\ast\nabla^{I_\phi+1}\phi \numberthis\label{eq:comeq-Psi-odd-junk}\\
&\,
\change{+\mathfrak{J}([\del_t,\nabla\Lap^\frac{L}2]\Psi)}\\
\mathfrak{Q}_{L+1,Border}=&\,a^{-3}\Psi\Lap^{\frac{L}2+1}N+a^{-3}\nabla\Psi\ast\nabla\Lap^{\frac{L}2}N+a^{-3}(N+1)\Sigma\ast\nabla^2\Lap^{\frac{L}2}\phi+\frac{\dot{a}}a\nabla\Lap^\frac{L}2N\ast\nabla\phi\numberthis\label{eq:comeq-Q-odd-border}\\
&\change{+a^{-3}(N+1)\nabla\Sigma\ast\nabla^2\Lap^{\frac{L}2-1}\phi}\\
\mathfrak{Q}_{L+1,Junk}=&\,a^{-3}\sum_{I_N+I_\Psi=L+2,\,2\leq I_\Psi\leq L+1}\nabla^{I_N}N\ast\nabla^{I_\Psi}\Psi+a^{-3}(N+1)\nabla^{L-2}\Ric[G]\ast\Sigma\ast\nabla\phi\\
&\,
\change{+a^{-3}(N+1)\nabla^2\Lap^{\frac{L}2-1}\Sigma\ast\nabla\phi+\mathfrak{J}([\del_t,\Lap^{\frac{L}2+1}]\phi)}\numberthis\label{eq:comeq-Q-odd-junk}
\end{align*}}
as well as
{\scriptsize \begin{align*}
\numberthis\label{eq:comeq-Q-1-border}\mathfrak{Q}_{1,Border}=&\,a^{-3}\Psi\Lap N+a^{-3}(N+1)\Sigma\ast\nabla^2\phi\deletemath{+a^{-3}(N+1)(\Psi+C)\nabla\phi\ast\nabla\phi}\\
\numberthis\label{eq:comeq-Q-1-junk}\mathfrak{Q}_{1,Junk}=&\,a^{-3}\nabla\Psi\ast\nabla N
\change{+a^{-3}(N+1)\nabla\Sigma\ast\nabla\phi+\mathfrak{J}([\del_t,\Lap]\phi)}\,.
\end{align*}}
\end{subequations}
The commuted rescaled evolution equation for $\Sigma$ has the error terms
{\scriptsize \begin{subequations}
\begin{align*}
\numberthis\label{eq:comeq-Sigma-border}\mathfrak{S}_{L,Border}=&\,a^{-3}(N+1)\left(\Sigma\ast\nabla^2\Lap^{\frac{L}2-1}\Sigma+\nabla\Sigma\ast\nabla^3\Lap^{\frac{L}2-2}\Sigma\right)+a^{-3}\left(\Lap^{\frac{L}2}N\cdot(\Sigma\ast\Sigma)+\nabla\Lap^{\frac{L}2-1}N\ast\nabla\Sigma\ast\Sigma\right)\\
&\,+a^{-3}(N+1)\Sigma\ast\Sigma\ast\nabla^2\Lap^{\frac{L}2-2}\Ric[G]+\frac{\dot{a}}a\Lap^{\frac{L}2}N\ast\Sigma+\frac{\dot{a}}a\nabla\Lap^{\frac{L}2-1}N\ast\nabla\Sigma\\
&\,+\underbrace{a^{-3}[(N+1)\nabla\Sigma\ast\Sigma+\nabla N\ast\Sigma\ast\Sigma]\ast\nabla^{L-3}\Ric[G]}_{\text{not present for }L=2}\\
\numberthis\label{eq:comeq-Sigma-junk}\mathfrak{S}_{L,Junk}=&\,-a[\Lap^\frac{L}2,\nabla^2]N+a\sum_{I_N+I_\Ric=L,I_N\neq 0}\nabla^{I_N} N\ast\nabla^{I_\Ric}\Ric[G]+\frac{\dot{a}}a\sum_{I_N+I_\Sigma=L}\nabla^{I_N}N\ast\nabla^{I_\Sigma}\Sigma\\
&\,+a^{-3}\sum_{I_1+I_2=L,\,I_i>0}\nabla^{I_1}\Sigma\ast\nabla^{I_2}\Sigma+a^{-3}\sum_{I_N+I_{1}+I_{2}=L,\,I_N<L}\nabla^{I_N}N\ast\nabla^{I_1}\Sigma\ast\nabla^{I_2}\Sigma\\
&\,+a\sum_{I_N+I_1+I_2=L}\nabla^{I_N}(N+1)\ast\nabla^{I_1+1}\phi\ast\nabla^{I_2+1}\phi+\left(4\pi C^2a^{-3}+\frac13a\right)\Lap^\frac{L}2N\cdot G+\mathfrak{J}([\del_t,\Lap^{\frac{L}2}]\Sigma)
\end{align*}
\end{subequations}}
while the commuted Ricci tensor evolution equations have error terms, where $\mathcal{I}=\sum_{i=1}^{\nicefrac{L}2-m+1}I_i$,
{\scriptsize
\begin{subequations}
\begin{align*}
\mathfrak{R}_{L,Border}=&\,a^{-3}\left[\nabla^{L+2}N\cdot\Sigma+\nabla^{L+1}N\ast\nabla\Sigma+\Sigma\ast\nabla^2\Lap^{\frac{L}2-1}\Ric[G]+\nabla\Sigma\ast\nabla^{L-1}\Ric[G]\right] \numberthis\label{eq:comeq-Ric-even-border}\\
\mathfrak{R}_{L+1,Border}=&\,a^{-3}\left[\nabla^{L+3}N\cdot\Sigma+\nabla^{L+2}N\ast\nabla\Sigma+\Sigma\ast\nabla^3\Lap^{\frac{L}2-1}\Ric[G]+\nabla\Sigma\ast\nabla^L\Ric[G]\right] \numberthis\label{eq:comeq-Ric-odd-border}\\
\mathfrak{R}_{L,Junk}=&\,a^{-3}\sum_{I_N+I_\Sigma=L+2,\,I_\Sigma\geq 2}\nabla^{I_N}N\ast\nabla^{I_\Sigma}\Sigma \numberthis\label{eq:comeq-Ric-even-junk}\\
&+a^{-3}\sum_{\substack{I_N+I_\Sigma+I_\Ric=L\\ (I_\Sigma,I_\Ric)\neq (0,L),(1,L-1)}}\nabla^{I_N}(N+1)\ast\nabla^{I_\Sigma}\Sigma\ast\nabla^{I_\Ric}\Ric[G]\\
&\,+a^{-3}\sum_{m=0}^{\frac{L}2-1}\sum_{{I_N+I_\Sigma+\mathcal{I}=2m}}\nabla^{I_N}(N+1)\ast\nabla^{I_\Sigma}\Sigma\ast\nabla^{I_1}\Ric[G]\ast\dots\ast\nabla^{I_{\frac{L}2-m+1}}\Ric[G]\\
&\,+\frac{\dot{a}}a\left([\Lap^\frac{L}2,\nabla^2]N+\Lap^\frac{L}2N\ast\Ric[G]+\nabla^{L-1}N\ast\nabla\Ric[G]\right)+\mathfrak{J}([\del_t,\Lap^{\frac{L}2}]\Ric[G])\\
\mathfrak{R}_{L+1,Junk}=&\,a^{-3}\sum_{I_N+I_\Sigma=L+3,\,I_\Sigma\geq 2}\nabla^{I_N}N\ast\nabla^{I_\Sigma}\Sigma \numberthis\label{eq:comeq-Ric-odd-junk}\\
&+a^{-3}\sum_{\substack{I_N+I_\Sigma+I_\Ric=L\\ (I_\Sigma,I_\Ric)\neq (0,L+1),(1,L)}}\nabla^{I_N}(N+1)\ast\nabla^{I_\Sigma}\Sigma\ast\nabla^{I_\Ric}\Ric[G]\\
&\,+a^{-3}\sum_{m=0}^{\frac{L}2-1}\sum_{{I_N+I_\Sigma+\mathcal{I}=2m+1}}\nabla^{I_N}(N+1)\ast\nabla^{I_\Sigma}\Sigma\ast\nabla^{I_1}\Ric[G]\ast\dots\nabla^{I_{\frac{L}2-m+1}}\Ric[G]\\
&\,+\frac{\dot{a}}a\left(\nabla[\Lap^\frac{L}2,\nabla^2]N+\nabla\Lap^{\frac{L}2}N\ast\Ric[G]+\nabla^2\Lap^{\frac{L}2-1}N\ast\nabla\Ric[G]\right)+\mathfrak{J}([\del_t,\nabla\Lap^\frac{L}2]\Ric[G])\\
\end{align*}
\end{subequations}}

Finally, the Bel-Robinson evolution error terms are
{\scriptsize
\begin{subequations}
\begin{align*}
\mathfrak{E}_{L,Border}=&\,\frac{\tau}3\left(\Lap^{\frac{L}2}N\cdot\RE+\nabla^{L-1}N\ast\nabla\RE\right)-a^{-1}\left(\Lap^\frac{L}2\RE\times\Sigma+\RE\times\Lap^{\frac{L}2}\Sigma\right)\numberthis\label{eq:comeq-RE-border}\\
&\,+a^{-3}\epsilonLC[G]\ast\epsilonLC[G]\ast\left(\nabla^{L-1}\RE\ast\nabla\Sigma+\nabla\RE\ast\nabla^{L-1}\Sigma\right)\\
&\,+a^{-3}\Lap^{\frac{L}2}N\cdot(\RE\ast\Sigma)+a^{-3}\nabla\Lap^{\frac{L}2-1}N\ast[\nabla\RE\ast\Sigma+\RE\ast\nabla\Sigma]\\
&\,+a^{-3}\left(\Sigma\ast\nabla^2\Lap^{\frac{L}2-1}\RE+\nabla\Sigma\ast\nabla^3\Lap^{\frac{L}2-2}\RE+\nabla\RE\ast\nabla\Lap^{\frac{L}2-1}\Sigma+\RE\ast\Lap^{\frac{L}2}\Sigma\right)\\
&\,+a^{-3}\left[(N+1)\Sigma\ast\RE\ast\nabla^2\Lap^{\frac{L}2-2}\Ric[G]+\underbrace{\nabla\left((N+1)\ast\Sigma\ast\RE\right)\ast\nabla^{L-3}\Ric[G]}_{\text{if }L\neq 2}\right]\\
&\,+4\pi a^{-3}(\Psi+C)^2\Lap^\frac{L}2N\cdot\Sigma\,+4\pi a^{-3}\nabla^{L-1}N\ast\left[(\Psi+C)^2\nabla\Sigma+2(\Psi+C)\ast\nabla\Psi\ast\Sigma\right]\\
&\,+4\pi a^{-3}(N+1)\left[(\Psi^2+2C\Psi)\Lap^\frac{L}2\Sigma+2(\Psi+C)\Lap^\frac{L}2\Psi\cdot\Sigma\right]\\
&\,+4\pi a^{-3}\nabla^{L-1}\Sigma\ast\left[(\Psi+C)^2\nabla N+2(N+1)(\Psi+C)\nabla\Psi\right]\\
&\,+4\pi a^{-3}(\Psi+C)\nabla^{L-1}\Psi\ast\left[(N+1)\nabla\Sigma+\nabla N\ast\Sigma\right]\\
\mathfrak{E}_{L,top}=&\,a^{-1}(N+1)\epsilonLC[G]\ast \RB\ast\nabla\Lap^{\frac{L}2-1}\Ric[G]+a(N+1)(\Psi+C)\nabla\Lap^{\frac{L}2-1}\Ric[G]\ast\nabla\phi \numberthis\label{eq:comeq-RE-top}\\
\numberthis\label{eq:comeq-RE-junk}\mathfrak{E}_{L,Junk}=&\,\frac{\dot{a}}a\sum_{I_N+I_{\RE}=L,\,I_{N}\leq L-2}\nabla^{I_N}N\ast\nabla^{I_{\RE}}\RE\\
&\,+a^{-1}\epsilonLC[G]\ast\left[\sum_{I_N+I_{\RB}=L+1,\,I_{N},I_{\RB}\leq L}\nabla^{I_N}N\ast\nabla^{I_{\RB}}\RB+(N+1)\nabla^2\Lap^{\frac{L}2-2}\Ric[G]\ast\nabla\RB\right] \\
&\,+a^{-3}\epsilonLC[G]\ast\epsilonLC[G]\ast\sum_{\substack{I_N+I_{\RE}+I_\Sigma=L,\\ I_N\leq L-2;\,I_N>0\text{ or }I_{\RE},I_{\Sigma}\geq 2}}\nabla^{I_N}N\ast\nabla^{I_{\RE}}\RE\ast\nabla^{I_\Sigma}\Sigma\\
&\,+a\sum_{\substack{I_N+I_\Psi+I_\phi=L+1\\ I_N,I_\Psi,I_\phi\neq L+1}}\nabla^{I_N}(N+1)\ast\nabla^{I_\Psi}(\Psi+C)\ast\nabla^{I_\phi+1}\phi\\
&\,+a^{-3}\sum_{\substack{I_N+I_\Sigma+I_1+I_2=L \\ I_N,I_\Sigma,I_i\leq L-2}}\nabla^{I_N}(N+1)\ast\nabla^{I_\Sigma}\Sigma\ast\nabla^{I_1}(\Psi+C)\ast\nabla^{I_2}(\Psi+C)\\
&\,+\dot{a}a^3\sum_{I_N+I_1+I_2=L}\nabla^{I_N}(N+1)\ast\nabla^{I_1+1}\phi\ast\nabla^{I_2+1}\phi\\
&\,+a\sum_{I_N+I_\Sigma+I_1+I_2=L}\nabla^{I_N}(N+1)\ast\nabla^{I_\Sigma}\Sigma\ast\nabla^{I_1+1}\phi\ast\nabla^{I_2+1}\phi\\
&\,+a^{-1}\epsilon[G]\ast\RB\ast[\Lap^\frac{L}2,\nabla]N+4\pi a(N+1)(\Psi+C)\left[\nabla^{L-2}\Ric[G]\ast\nabla^2\phi+\mathfrak{J}([\Lap^\frac{L}2,\nabla^2]\phi)\right]\\
&\,+a\left\{(\Psi+C)[\Lap^\frac{L}2,\nabla]N+(N+1)[\Lap^\frac{L}2,\nabla]\Psi\right\}\ast\nabla\phi+\change{(N+1)a^{-1}\mathfrak{J}([\Lap^\frac{L}2,\curl_G]\RB)}+\mathfrak{J}([\del_t,\Lap^{\frac{L}2}]\RE)\\
&\,+\Lap^\frac{L}2\left[a^{-3}(N+1)\RE\ast\Sigma+\frac{2\pi}3a^6(N+1)\left(\del_0\left(a^{-6}(\Psi+C)^2+a^{-2}\lvert\nabla\phi\rvert_G^2\right)+4\pi\frac{\dot{a}}a(\Psi+C)^2\right)\right]\cdot G\\[2em]
\mathfrak{B}_{L,Border}=&\,\frac{\tau}3\left(\Lap^{\frac{L}2}N\cdot\RB+\nabla^{L-1}N\ast\nabla\RB\right)-a^{-1}\left(\Lap^{\frac{L}2}\RB\times\Sigma+\RB\ast\Lap^{\frac{L}2}\Sigma\right)\numberthis\label{eq:comeq-RB-border}\\
&\,+a^{-3}\epsilonLC[G]\ast\epsilonLC[G]\ast\left(\nabla^{L-1}\RB\ast\nabla\Sigma+\nabla\RB\ast\nabla^{L-1}\Sigma\right)\\
&\,+a^{-3}\Lap^{\frac{L}2}N\cdot(\RB\ast\Sigma)+a^{-3}\nabla\Lap^{\frac{L}2-1}N\cdot\left[\nabla\RB\ast\Sigma+\RB\ast\nabla\Sigma\right]\\
&\,+a^{-3}\left(\Sigma\ast\nabla^2\Lap^{\frac{L}2-1}\RB+\nabla\Sigma\ast\nabla^3\Lap^{\frac{L}2-2}\RB+\nabla\RB\ast\nabla\Lap^{\frac{L}2-1}\Sigma+\RB\ast\Lap^{\frac{L}2}\Sigma\right)\\
&\,+a^{-3}\left[(N+1)\Sigma\ast\RB\ast\nabla^{L-2}\Ric[G]+\underbrace{\nabla\left((N+1)\ast\Sigma\ast\RB\right)\ast\nabla^{L-3}\Ric[G]}_{\text{if }L\neq 2}\right]\\
&\,+a^{-1}(N+1)(\Psi+C)\cdot\epsilonLC[G]\ast\nabla\Lap^\frac{L}2\phi\ast\Sigma+a^{-1}\epsilonLC[G]\ast\nabla^2\nabla^L\phi\ast\nabla\left((N+1)(\Psi+C)\Sigma\right)\\
\mathfrak{B}_{L,top}=&\,a^3(N+1)\epsilonLC[G]\ast\nabla\Lap^{\frac{L}2-1}\Ric[G]\ast\nabla\phi\ast\nabla\phi+a^{-1}\epsilonLC[G]\ast\RE\ast\nabla\Lap^{\frac{L}2-1}\Ric[G] \numberthis\label{eq:comeq-RB-top}\\
\mathfrak{B}_{L,Junk}=&\,\frac{\dot{a}}a\sum_{I_N+I_{\RB}=L,I_{N}\leq L-2}\nabla^{I_N}N\ast\nabla^{I_{\RB}}\RB \numberthis\label{eq:comeq-RB-junk}\\
&\,+a^{-1}\epsilonLC[G]\ast\left[\sum_{I_N+I_{\RE}=L+1,I_{N},I_{\RE}\leq L}\nabla^{I_N}N\ast\nabla^{I_{\RE}}\RE+\nabla^2\Lap^{\frac{L}2-2}\Ric[G]\ast\nabla\RE\right] \\
&\,+a^{-3}\epsilonLC[G]\ast\epsilonLC[G]\ast\sum_{\substack{I_N+I_{\RE}+I_\Sigma=L,\\ I_N\leq L-2;\,I_N>0\text{ or }I_{\RB},I_{\Sigma}\geq 2}}\nabla^{I_N}(N+1)\ast\nabla^{I_{\RB}}\RB\ast\nabla^{I_\Sigma}\Sigma\\
&\,+a^3\epsilonLC[G]\ast\sum_{\substack{I_N+I_{1}+I_{2}=L,\\ \,I_N>0\text{ or }I_2< L}}\nabla^{I_N}(N+1)\ast\nabla^{I_{1}+1}\phi\ast\nabla^{I_{2}+2}\phi\\
&\,+a^{-1}\epsilonLC[G]\ast\sum_{\substack{I_N+I_\Psi+I_\phi+I_\Sigma=L\\I_\phi\leq L-2}}\nabla^{I_N}(N+1)\ast\nabla^{I_\Psi}(\Psi+C)\ast\nabla^{I_\phi+1}\phi\ast\nabla^{I_\Sigma}\Sigma\\
&\,+a^{-1}\epsilon[G]\ast\RE\ast[\Lap^{\frac{L}2},\nabla]N\\
&\,+a^{-1}(N+1)(\Psi+C)\cdot\epsilonLC[G]\ast\Sigma\ast\left[\Lap^\frac{L}2,\nabla\right]\phi\\
&\,+a^3(N+1)\epsilonLC[G]\ast\nabla\phi\ast\left(\nabla^2\Lap^{\frac{L}2-2}\Ric[G]\ast\nabla^2\phi+\mathfrak{J}([\Lap^{\frac{L}2},\nabla]\phi)\right)\\
&\,\change{-a^{-1}(N+1)\mathfrak{J}([\Lap^\frac{L}2,\curl_G]\RE)}+\mathfrak{J}([\del_t,\Lap^{\frac{L}2}]\RB)+\Lap^\frac{L}2\left[a^{-3}(N+1)\RB\ast\Sigma\right]\cdot G\\
&\,+\Lap^{\frac{L}2}\left[4\pi a^2\nabla^{\sharp m}\phi(\Psi+C)+\frac{2\pi}3a^5\Lap^{\frac{L}2}\nabla^{\sharp m}\left(a^{-6}(\Psi+C)^2+a^{-2}\lvert\nabla\phi\rvert_G^2\right)\right]\epsilonLC[G]_{(\cdot)m(\cdot)}
\end{align*}
\end{subequations}}
\end{definition}

\subsection{$L^2_G$ error term estimates}\label{subsec:L2-error-est}

In this subsection, we collect how the error terms can be controlled in terms of energies as well as homogeneous Sobolev norms of $\phi$. We do not claim that these estimates are optimal -- in particular, we note that at low order (like $L=2$), many of the curvature errors that appear in the estimates below could be avoided entirely: These arise as a result of applying the general estimates in Lemma \ref{lem:Sobolev-norm-equivalence-improved} where the Ricci tensor does not naturally occur in the respective equations, and can be avoided at low orders by applying \eqref{eq:APmidRic} on all curvature terms that occur. \\
Instead of optimality, we try to keep both notation and form of the error term estimates as simple as possible and the energy estimates between base and top level as unified as possible. In particular, we track the \enquote{worst} curvature energy occurring at high orders for all estimates below, even if these terms are added in artificially for low orders.

\begin{lemma}[Estimates for borderline error terms] Let $L\in 2\Z_+,\,L\leq 20$. Then, the following estimates hold:
\begin{subequations}
\begin{align*}
\numberthis\label{eq:L2-Border-H}\|\mathfrak{H}_{L,Border}\|_{L^2_G}\lesssim&\,\epsilon a^{-4}\sqrt{\E^{(L)}(\Sigma,\cdot)}+\epsilon a^{-4-c\sqrt{\epsilon}}\sqrt{\E^{(\leq L-2)}(\Sigma,\cdot)}+\epsilon^2a^{-4-c\sqrt{\epsilon}}\sqrt{\E^{(\leq L-2)}(\Ric,\cdot)}\\
\numberthis\label{eq:L2-Border-N}\|\mathfrak{N}_{L,Border}\|_{L^2_G}\lesssim&\,\epsilon a^{-4}\left[\sqrt{\E^{(L)}(\phi,\cdot)}+\sqrt{\E^{(L)}(\Sigma,\cdot)}\right]+\epsilon a^{-4}\sqrt{\E^{(L)}(N,\cdot)}\\
&\,+\epsilon a^{-4-c\sqrt{\epsilon}}\left[\sqrt{\E^{(\leq L-2)}(\phi,\cdot)}+\sqrt{\E^{(\leq L-2)}(\Sigma,\cdot)}+\sqrt{\E^{(\leq L-2)}(N,\cdot)}\right]\\
&\,+\underbrace{\epsilon^2 a^{-4}\sqrt{\E^{(\leq L-2)}(\Ric,\cdot)}+\epsilon a^{-4-c\sqrt{\epsilon}}\sqrt{\E^{(\leq L-4)}(\Ric,\cdot)}}_{\text{not present for }L=2}\\
\numberthis\label{eq:L2-Border-N-odd}\|\mathfrak{N}_{L+1,Border}\|_{L^2_G}\lesssim&\,\epsilon a^{-6}\left[\sqrt{a^4\E^{(L+1)}(\phi,\cdot)}+\sqrt{a^4\E^{(L+1)}(\Sigma,\cdot)}\right]+\epsilon a^{-6}\sqrt{a^4\E^{(L+1)}(N,\cdot)}\\
&\,+\epsilon a^{-4}\left[\sqrt{\E^{(L)}(\phi,\cdot)}+\sqrt{\E^{(L)}(\Sigma,\cdot)}+\sqrt{\E^{(L)}(N,\cdot)}\right]\\
&\,+\epsilon a^{-4-c\sqrt{\epsilon}}\left[\sqrt{\E^{(\leq L-2)}(\phi,\cdot)}+\sqrt{\E^{(\leq L-2)}(\Sigma,\cdot)}+\sqrt{\E^{(\leq L-2)}(N,\cdot)}\right]\\
&\,+\epsilon^2a^{-4-c\sqrt{\epsilon}}\sqrt{\E^{(\leq L-2)}(\Ric,\cdot)}\\
\numberthis\label{eq:L2-Border-P-even}\|\mathfrak{P}_{L,Border}\|_{L^2_G}\lesssim&\,
\change{\epsilon}a^{-3}\sqrt{\E^{(L)}(\phi,\cdot)}+\epsilon a^{-3}\sqrt{\E^{(L)}(N,\cdot)}+\epsilon a^{-3-c\sqrt\epsilon}\sqrt{\E^{(\leq L-2)}(N,\cdot)}\\
&\,+
\change{\epsilon}a^{-3-c\sqrt{\epsilon}}\sqrt{\E^{(\leq L-2)}(\phi,\cdot)}
\change{+\epsilon a^{-3}\sqrt{\E^{(L)}(\Sigma,\cdot)}+\epsilon a^{-3-c\sqrt{\epsilon}}\sqrt{\E^{(L-2)}(\Sigma,\cdot)}}\\
&\,+\underbrace{\epsilon^2 a^{-3-c\sqrt{\epsilon}}\sqrt{\E^{(\leq L-3)}(\Ric,\cdot)}}_{\text{not present for }L=2}\\
\numberthis\label{eq:L2-Border-Q-even}\|\mathfrak{Q}_{L,Border}\|_{L^2_G}\lesssim&\,\epsilon a^{-3}\sqrt{\E^{(L+1)}(N,\cdot)}
+
\change{\epsilon}a^{-3}\sqrt{a^{-4}\E^{(L)}(\phi,\cdot)}\\
&\,+
\change{\epsilon}a^{-3-c\sqrt{\epsilon}}\sqrt{a^{-4}\E^{(\leq L-2)}(\phi,\cdot)}\\
\numberthis\label{eq:L2-Border-P-odd}\|\mathfrak{P}_{L+1,Border}\|_{L^2_G}\lesssim&\,
\change{\epsilon}a^{-3}\sqrt{\E^{(L+1)}(\phi,\cdot)}+\epsilon a^{-3}\sqrt{\E^{(L+1)}(N,\cdot)}+\epsilon a^{-3-c\sqrt\epsilon}\sqrt{\E^{(\leq L-1)}(N,\cdot)}\\
&\,+\change{\epsilon}a^{-3-c\sqrt{\epsilon}}\sqrt{\E^{(\leq L-1)}(\phi,\cdot)}+\change{\epsilon a^{-3}\sqrt{\E^{(L+1)}(\Sigma,\cdot)}+\epsilon a^{-3-c\sqrt{\epsilon}}\sqrt{\E^{(L-1)}(\Sigma,\cdot)}}\\
&\,+\epsilon^2 a^{-3-c\sqrt{\epsilon}}\sqrt{\E^{(\leq L-2)}(\Ric,\cdot)}\\
\numberthis\label{eq:L2-Border-Q-odd}\|\mathfrak{Q}_{L+1,Border}\|_{L^2_G}\lesssim&\,\epsilon a^{-3}\sqrt{\E^{(L+2)}(N,\cdot)}+\epsilon a^{-3-c\sqrt{\epsilon}}\sqrt{\E^{(L+1)}(N,\cdot)}\\
&\,+\change{\epsilon} a^{-3}\sqrt{a^{-4}\E^{(L+1)}(\phi,\cdot)}+\change{\epsilon}a^{-3-c\sqrt{\epsilon}}\sqrt{a^{-4}\E^{(\leq L-1)}(\phi,\cdot)}\\
\numberthis\label{eq:L2-Border-S}\|\mathfrak{S}_{L,Border}\|_{L^2_G}\lesssim&\,\epsilon a^{-3}\sqrt{\E^{(L)}(\Sigma,\cdot)}+\epsilon a^{-3}\sqrt{\E^{(L)}(N,\cdot)}+\epsilon^2a^{-3}\sqrt{\E^{(L-2)}(\Ric,\cdot)}\\
&\,+\epsilon a^{-3-c\sqrt{\epsilon}}\sqrt{\E^{(\leq L-2)}(\Sigma,\cdot)}+\epsilon a^{-3-c\sqrt{\epsilon}}\sqrt{\E^{(\leq L-2)}(N,\cdot)}\\
&\,+\underbrace{\epsilon^2a^{-3-c\sqrt{\epsilon}}\sqrt{\E^{(\leq L-4)}(\Ric,\cdot)}}_{\text{not present for }L=2}\\
\numberthis\label{eq:L2-Border-R-even}\|\mathfrak{R}_{L,Border}\|_{L^2_G}\lesssim&\,\epsilon a^{-3}\sqrt{\E^{(L+2)}(N,\cdot)}+\epsilon a^{-3-c\sqrt{\epsilon}}\sqrt{\E^{(\leq L)}(N,\cdot)}\\
&\,+\epsilon a^{-3}\sqrt{\E^{(L)}(\Ric,\cdot)}+\epsilon a^{-3-c\sqrt{\epsilon}}\sqrt{\E^{(\leq L-2)}(\Ric,\cdot)}\\
\numberthis\label{eq:L2-Border-R-odd}\|\mathfrak{R}_{L+1,Border}\|_{L^2_G}\lesssim&\,\epsilon a^{-3}\sqrt{\E^{(L+3)}(N,\cdot)}+\epsilon a^{-3-c\sqrt{\epsilon}}\sqrt{\E^{(\leq L+1)}(N,\cdot)}\\
&\,+\epsilon a^{-3}\sqrt{\E^{(L+1)}(\Ric,\cdot)}+\epsilon a^{-3-c\sqrt{\epsilon}}\sqrt{\E^{(\leq L-1)}(\Ric,\cdot)}
\end{align*}
\vspace{-1em}
\begin{align*}
\numberthis\label{eq:L2-Border-BR}&\,\|\mathfrak{E}_{L,Border}\|_{L^2_G}+\|\mathfrak{B}_{L,Border}\|_{L^2_G}\\
\lesssim&\,\epsilon a^{-3}\sqrt{\E^{(L)}(\phi,\cdot)}+\epsilon a^{-3-c\sqrt{\epsilon}}\sqrt{\E^{(\leq L-2)}(\phi,\cdot)}\\
&\,+\epsilon a^{-3}\left(\sqrt{\E^{(L)}(N,\cdot)}+\sqrt{\E^{(L)}(\Sigma,\cdot)}\right)+\epsilon a^{-3}\sqrt{\E^{(L)}(W,\cdot)}\\
&\,+\epsilon a^{-3-c\sqrt{\epsilon}}\left(\sqrt{\E^{(\leq L-2)}(N,\cdot)}+\sqrt{\E^{(\leq L-2)}(\Sigma,\cdot)}+\sqrt{\E^{(\leq L-2)}(W,\cdot)}\right)\\
&\,+\epsilon^2a^{-3}\sqrt{\E^{(L-2)}(\Ric,\cdot)}+\underbrace{\epsilon^2a^{-3-c\sqrt{\epsilon}}\sqrt{\E^{(\leq L-4)}(\Ric,\cdot)}}_{\text{not present for }L=2}
\end{align*}
\delete{Further, we have the alternative estimates [...]}
\end{subequations}
\end{lemma}
\begin{proof}
All of these estimates follow from applying $L^2_G$-$L^\infty_G$-type Hölder estimates to the individual nonlinear terms. The lower order terms are either controlled by the zero order estimates in subsection \ref{subsec:APlow} or the a priori estimates in Lemma \ref{lem:AP}. Furthermore, we apply Lemma \ref{lem:Sobolev-norm-equivalence-improved}, along with again Lemma \ref{lem:AP}, to translate $L^2_G$-norms into energies up to additional curvature energy terms. For the sake of simplicity, we always estimate $\frac{\dot{a}}a$ by $a^{-3}$ up to constant (see \eqref{eq:Friedman}), and liberally apply \eqref{eq:ibp-trick} to deal with odd order energies and to distribute $a^{-c\sqrt{\epsilon}}$ factors to lower orders while updating $c>0$ wherever this is convenient.
\end{proof}

\begin{lemma}[Estimates for top order error terms]
\begin{subequations}
\begin{align*}
\numberthis\label{eq:L2-top-E}\|\mathfrak{E}_{L,top}\|_{L^2_G}\lesssim&\,\sqrt{\epsilon} a^{1-c\sqrt{\epsilon}}\sqrt{\E^{(L-1)}(\Ric,\cdot)}=\sqrt{\epsilon}a^{-1-c\sqrt{\epsilon}}\sqrt{a^4\E^{(L-1)}(\Ric,\cdot)}\\
\numberthis\label{eq:L2-top-B}\|\mathfrak{B}_{L,top}\|_{L^2_G}\lesssim&\,\epsilon a^{-1}\sqrt{\E^{(L-1)}(\Ric,\cdot)}=\change{\epsilon a^{-3}\sqrt{a^4\E^{(L-1)}(\Ric,\cdot)}}
\end{align*}
\end{subequations}
\end{lemma}
\begin{proof}
This follows directly using \eqref{eq:APE} and \eqref{eq:APmidB} for the Bel-Robinson terms as well as \eqref{eq:APmidphi}.
\end{proof}

\begin{lemma}[Junk terms] Recalling $\parallel$-notation from Remark \ref{rem:notation-parallel}, the following hold:
\begin{subequations}
\begin{align*}
\numberthis\label{eq:L2-junk-M}\|\mathfrak{M}_{L,Junk}\|_{L^2_G}\lesssim
&\,\epsilon a^{-2-c\sqrt{\epsilon}}\sqrt{\E^{(\leq L-1)}(\phi,\cdot)}+a^{-2-c\sqrt{\epsilon}}\sqrt{\E^{(\leq L-2)}(\phi,\cdot)}
\\
&\,+a^{-c\sqrt{\epsilon}}\sqrt{\E^{(\leq L-1)}(\Sigma,\cdot)}+\sqrt{\epsilon}a^{-c\sqrt{\epsilon}}\sqrt{\E^{(\leq L-2)}(\Ric,\cdot)}\\
\numberthis\label{eq:L2-junk-Mtilde}\|\tilde{\mathfrak{M}}_{L,Junk}\|_{L^2_G}\lesssim&\,\epsilon a^{-1-c\sqrt{\epsilon}}\sqrt{\E^{(\leq L-2)}(\Ric,\cdot)}+a^{-1-c\sqrt{\epsilon}}\sqrt{\E^{(\leq L-1)}(\Sigma,\cdot)}\\[1em]
\numberthis\label{eq:L2-junk-H-par}\|\mathfrak{H}_{L,Junk}^\parallel\|_{L^2_G}\lesssim&\,\epsilon a^{-4-c\sqrt{\epsilon}}\sqrt{\E^{(\leq L-2)}(\Sigma,\cdot)}+\sqrt{\epsilon}a^{-2-c\sqrt{\epsilon}}\sqrt{\E^{(\leq L)}(\phi,\cdot)}\\
&\,+\epsilon a^{-c\sqrt{\epsilon}}\sqrt{\E^{(\leq L-2)}(\Ric,\cdot)}\\
\numberthis\label{eq:L2-junk-N}\|\mathfrak{N}_{L,Junk}\|_{L^2_G}\lesssim&\,\epsilon a^{-4-c\sqrt{\epsilon}}\sqrt{\E^{(\leq L-2)}(N,\cdot)}+\epsilon a^{-c\sigma}\sqrt{\E^{(L)}(\phi,\cdot)}\\
&\,+\epsilon a^{-4-c\sqrt{\epsilon}}\left[\sqrt{\E^{(\leq L-2)}(\phi,\cdot)}+\sqrt{\E^{(\leq L-2)}(\Sigma,\cdot)}\right]\\
&\,+\underbrace{\epsilon a^{-4}\sqrt{\E^{(\leq L-2)}(\Ric,\cdot)}+\epsilon a^{-4-c\sqrt{\epsilon}}\sqrt{\E^{(\leq L-4)}(\Ric,\cdot)}}_{\text{not present for }L=2}\\
\numberthis\label{eq:L2-junk-N-odd}\|\mathfrak{N}_{L+1,Junk}\|_{L^2_G}\lesssim&\,\epsilon a^{-4-c\sqrt{\epsilon}}\sqrt{\E^{(\leq L-1)}(N,\cdot)}+\epsilon a^{-c\sigma}\left(\sqrt{\E^{(L+1)}(\phi,\cdot)}+\sqrt{\E^{(L+1)}(\Sigma,\cdot)}\right)\\
&\,+\epsilon a^{-4-c\sqrt{\epsilon}}\left[\sqrt{\E^{(\leq L-1)}(\phi,\cdot)}+\sqrt{\E^{(\leq L-1)}(\Sigma,\cdot)}\right]\\
&\,+\epsilon^2a^{-4}\sqrt{\E^{(\leq L-1)}(\Ric,\cdot)}+\underbrace{\epsilon^2a^{-4-c\sqrt{\epsilon}}\sqrt{\E^{(\leq L-3)}(\Ric,\cdot)}}_{\text{not present for }L=2}\\
\numberthis\label{eq:L2-junk-P-even}\|\mathfrak{P}_{L,Junk}\|_{L^2_G}\lesssim&\,\epsilon a^{1-c\sigma}\sqrt{\E^{(L)}(\phi,\cdot)}+\epsilon a^{-3-c\sqrt{\epsilon}}\sqrt{\E^{(\leq L-2)}(\phi,\cdot)}\\
&\,+\epsilon a^{-3-c\sqrt{\epsilon}}\sqrt{\E^{(\leq L-2)}(\Sigma,\cdot)}+\sqrt{\epsilon}a^{1-c\sqrt{\epsilon}}\sqrt{\E^{(L)}(N,\cdot)}\\
&\,+\left[\epsilon a^{-3-c\sqrt{\epsilon}}+\sqrt{\epsilon}a^{-1-c\sqrt{\epsilon}}\right]\sqrt{\E^{(\leq L-2)}(N,\cdot)}\\
&\,+\underbrace{\epsilon^2a^{1-c\sigma}\sqrt{\E^{(\leq L-2)}(\Ric,\cdot)}+\epsilon^2 a^{-3-c\sqrt{\epsilon}}\sqrt{\E^{(\leq L-3)}(\Ric,\cdot)}
}_{\text{not present for }L=2}\\
\numberthis\label{eq:L2-junk-P-odd}\|\mathfrak{P}_{L+1,Junk}\|_{L^2_G}\lesssim&\,\epsilon a^{1-c\sigma}\sqrt{\E^{(L+1)}(\phi,\cdot)}+\epsilon a^{-3-c\sqrt{\epsilon}}\sqrt{\E^{(\leq L-1)}(\phi,\cdot)}\deletemath{+\epsilon a^{-3-c\sqrt{\epsilon}}\|\nabla\phi\|_{H^{L-2}_G}}\\
&\,+\epsilon a^{-3-c\sqrt{\epsilon}}\sqrt{\E^{(\leq L-1)}(\Sigma,\cdot)}+\sqrt{\epsilon}a^{1-c\sqrt{\epsilon}}\sqrt{\E^{(L+1)}(N,\cdot)}\\
&\,+\left[\epsilon a^{-3-c\sqrt{\epsilon}}+\sqrt{\epsilon}a^{-1-c\sqrt{\epsilon}}\right]\sqrt{\E^{(\leq L-1)}(N,\cdot)}\\
&\,+\epsilon^2a^{-1-c\sigma}\sqrt{a^4\E^{(L-1)}(\Ric,\cdot)}\\
&\,+\left(\epsilon^2a^{-1-c\sigma}+\epsilon^2 a^{-3-c\sqrt{\epsilon}}\right)\sqrt{\E^{(\leq L-2)}(\Ric,\cdot)}\\
\numberthis\label{eq:L2-junk-Q-even}\|\mathfrak{Q}_{L,Junk}\|_{L^2_G}\lesssim&\,
\change{\epsilon a^{-1-c\sqrt{\epsilon}}\sqrt{a^{-4}\E^{(L)}(\phi,\cdot)}+\epsilon a^{-3-c\sqrt{\epsilon}}\sqrt{\E^{(\leq L-2)}(\phi,\cdot)}}\\
&\,+\change{\sqrt{\epsilon}a^{-3-c\sqrt{\epsilon}}\sqrt{\E^{(\leq L)}(\Sigma,\cdot)}}+\sqrt{\epsilon}a^{-3-c\sqrt{\epsilon}}\sqrt{\E^{(\leq L)}(N,\cdot)}\\
&\,+\underbrace{\epsilon a^{-3-c\sqrt{\epsilon}}\sqrt{\E^{(\leq L-2)}(\Ric,\cdot)}}_{\text{not present for }L=2}\\
\numberthis\label{eq:L2-junk-Q-1}\|\mathfrak{Q}_{1,Junk}\|\lesssim&\,\epsilon a^{-3-c\sqrt{\epsilon}}\sqrt{\E^{(1)}(N,\cdot)}+\epsilon^\frac32a^{-3-c\sqrt{\epsilon}}\|\nabla\phi\|_{L^2_G}+\change{\sqrt{\epsilon}a^{-3-c\sqrt{\epsilon}}\sqrt{\E^{(\leq 1)}(\Sigma,\cdot)}}\\
\numberthis\label{eq:L2-junk-Q-odd}\|\mathfrak{Q}_{L+1,Junk}\|_{L^2_G}\lesssim&\,
\change{\epsilon a^{-1-c\sqrt{\epsilon}}\sqrt{a^{-4}\E^{(L+1)}(\phi,\cdot)}+\epsilon a^{-3-c\sqrt{\epsilon}}\E^{(\leq L-1)}(\phi,\cdot)}\\
&\,\change{+\sqrt{\epsilon}a^{-3-c\sqrt{\epsilon}}\sqrt{\E^{(\leq L+1)}(\Sigma,\cdot)}}+\sqrt{\epsilon} a^{-3-c\sqrt{\epsilon}}\sqrt{\E^{(\leq L+1)}(N,\cdot)}\\
&\,+\epsilon a^{-3-c\sqrt{\epsilon}}\sqrt{\E^{(\leq L-1)}(\Ric,\cdot)}\\
\numberthis\label{eq:L2-junk-S}\|\mathfrak{S}_{L,Junk}^\parallel\|_{L^2_G}\lesssim&\,\epsilon a^{1-c\sigma}\sqrt{\E^{(L)}(\Sigma,\cdot)}+\epsilon a^{-3-c\sqrt{\epsilon}}\sqrt{\E^{(\leq L-2)}(\Sigma,\cdot)}+\sqrt{\epsilon}a^{-1-c\sqrt{\epsilon}}\sqrt{\E^{(L)}(\phi,\cdot)}\\
&\,+\left(\epsilon a^{-3}+a^{1-c\sqrt{\epsilon}}\right)\sqrt{\E^{(\leq L)}(N,\cdot)}+\epsilon a^{5-c\sigma}\sqrt{\E^{(\leq L-1)}(\Ric,\cdot)}\\
&\,+\underbrace{\epsilon a^{-3}\sqrt{\E^{(L-2)}(\Ric,\cdot)}+\epsilon a^{-3-c\sqrt{\epsilon}}\sqrt{\E^{(\leq L-4)}(\Ric,\cdot)}}_{\text{not present for }L=2}\\
\numberthis\label{eq:L2-junk-R-even}\|\mathfrak{R}_{L,Junk}\|_{L^2_G}\lesssim&\,\epsilon^2a^{1-c\sigma}\sqrt{\E^{(\leq L-1)}(\Ric,\cdot)}+\epsilon a^{-3-c\sqrt{\epsilon}}\sqrt{\E^{(\leq L-2)}(\Ric,\cdot)}\\
&\,+\epsilon a^{1-c\sigma}\sqrt{\E^{(\leq L+2)}(\Sigma,\cdot)}+a^{-3-c\sqrt{\epsilon}}\sqrt{\E^{(\leq L)}(\Sigma,\cdot)}\\
&\,+a^{-3-c\sqrt{\epsilon}}\sqrt{\E^{(\leq L)}(N,\cdot)}\\
\numberthis\label{eq:L2-junk-R-odd}\|\mathfrak{R}_{L+1,Junk}\|_{L^2_G}\lesssim&\,\epsilon^2a^{1-c\sigma}\sqrt{\E^{(\leq L)}(\Ric,\cdot)}+\epsilon a^{-3-c\sqrt{\epsilon}}\sqrt{\E^{(\leq L-1)}(\Ric,\cdot)}\\
&\,+\epsilon a^{1-c\sigma}\sqrt{\E^{(\leq L+3)}(\Sigma,\cdot)}+a^{-3-c\sqrt{\epsilon}}\sqrt{\E^{(\leq L+1)}(\Sigma,\cdot)}\\
&\,+a^{-3-c\sqrt{\epsilon}}\sqrt{\E^{(\leq L+1)}(N,\cdot)}
\end{align*}
\vspace{-1em}
\begin{align*}
\numberthis\label{eq:L2-junk-BR-par}\|\mathfrak{E}_{L,Junk}^\parallel\|_{L^2_G}+\|\mathfrak{B}_{L,Junk}^\parallel\|_{L^2_G}\lesssim&\,\epsilon a^{-1-c\sqrt{\epsilon}}\sqrt{\E^{(\leq L)}(W,\cdot)}+\epsilon a^{-3-c\sqrt{\epsilon}}\sqrt{\E^{(\leq L-2)}(W,\cdot)}\\
&\,+\epsilon a^{-1-c\sqrt{\epsilon}}\sqrt{\E^{(L)}(\phi,\cdot)}+\left(\epsilon a^{-3-c\sqrt{\epsilon}}+a^{-1-c\sqrt{\epsilon}}\right)\sqrt{\E^{(\leq L-2)}(\phi,\cdot)}\\
&\,+\sqrt{\epsilon} a^{1-c\sqrt{\epsilon}}\sqrt{\E^{(\leq L)}(N,\cdot)}+\epsilon a^{-3-c\sqrt{\epsilon}}\sqrt{\E^{(\leq L-2)}(N,\cdot)}\\
&\,+\epsilon a^{-1-c\sigma}\sqrt{\E^{(L)}(\Sigma,\cdot)}+\epsilon a^{-3-c\sqrt{\epsilon}}\sqrt{\E^{(\leq L-2)}(\Sigma,\cdot)}\\
&\,+\epsilon a^{-3}\sqrt{\E^{(L-2)}(\Ric,\cdot)}+\epsilon a^{-3-c\sqrt{\epsilon}}\underbrace{\sqrt{\E^{(\leq L-4)}(\Ric,\cdot)}}_{\text{not present for }L=2}
\end{align*}
\end{subequations}
\end{lemma}

\begin{proof}
Once again, this follows by applying the a priori estimates from subsection \ref{subsec:APlow} and Lemma \ref{lem:AP}, as well as the bootstrap assumption \eqref{eq:BsN} for the lapse, to deal with the lower order terms in the nonlinearities, and then applying Lemma \ref{lem:Sobolev-norm-equivalence-improved} as well as \eqref{eq:ibp-trick} wherever this is necessary. Further, especially in \eqref{eq:L2-junk-BR-par}, it is often more convenient to use the bootstrap assumption for $\|\nabla\phi\|_{C_G}$ instead of the a priori estimate \eqref{eq:APmidphi} to gain higher powers of $\epsilon$ in prefactors. \\
Recognizing that every low order curvature term can be estimated up to constant by $a^{-c\sqrt{\epsilon}}$ at worst (see \eqref{eq:APmidRic}), we also note that any of the highly nonlinear curvature terms in $\mathfrak{J}$-expressions turn out to be negligible after updating $c$ compared to Ricci energies arising from applying Lemma \ref{lem:Sobolev-norm-equivalence-improved} or compared to junk terms in which $\Ric[G]$ is tracked explicitly. 
\end{proof}

\section{Appendix -- Future stability}\label{sec:appendix-fut}

Here, we collect the commutators in CMCSH gauge necessary to study the commuted scalar-field equations:

\begin{lemma}[Commutator formulas for future stabilty]\label{lem:fut-comm-formula} Let $\zeta$ be a scalar function on $\Sigma_T$. Then, the following formulas hold:
\begin{align*}
[\fdel,\nabla]\zeta=&\,0\\
[\fdel,\fLap]\zeta=&\,(\fdel(\fg^{-1})^{ab})\nabla_a\nabla_b\zeta-2(\fg^{-1})^{ab}\left(\div_{\fg}(\fn\fk)_a-2\nabla_a\fn\right)\nabla_b\zeta\deletemath{+\fX^k[\nabla_k,\fLap]\zeta}
\end{align*}
Schematically, for $k\in\N$, this implies
\begin{subequations}
\change{\begin{align*}
\numberthis\label{eq:[fdel,Lapk]}[\fdel,\fLap^k]\zeta=&\,\sum_{I_{\fn}+I_{\fk}+I_{\zeta}=2k-1}\nabla^{I_{\fn}}\fn\ast_{\fg}\nabla^{I_{\fk}}\fk\ast_{\fg}\nabla^{I_{\zeta}+1}\zeta\changefinal{\,+\sum_{I_{\fN}+I_{\zeta}=2k-1}\nabla^{I_{\fN}}\fN\ast_{\fg}\nabla^{I_\zeta+1}\zeta}\\
\numberthis\label{eq:[fdel,nablaLapk]}[\fdel,\nabla\fLap^k]\zeta=&\,\sum_{I_{\fn}+I_{\fk}+I_{\zeta}=2k}\nabla^{I_{\fn}}\fn\ast_{\fg}\nabla^{I_{\fk}}\fk\ast_{\fg}\nabla^{I_{\zeta}+1}\zeta\changefinal{\,+\sum_{I_{\fN}+I_{\zeta}=2k}\nabla^{I_{\fN}}\fN\ast_{\fg}\nabla^{I_\zeta+1}\zeta}
\end{align*}}
\end{subequations}
\end{lemma}
\begin{proof}
This follows from straightfoward computations, similar to Lemma \ref{lem:com-time-first} for the low order commutators and to Lemma \ref{lem:com-time} for higher orders.
\end{proof}


\bibliographystyle{alpha}

\bibliography{bibliography}

\end{document}